\documentclass[sigconf, nonacm]{acmart}

\renewcommand\footnotetextcopyrightpermission[1]{}
\usepackage{mathtools}       % \DeclarePairedDelimiter etc. (super-set of amsmath)
\usepackage{amsthm}          % Re-enable the AMS theorem prelude under acmart
\usepackage{bbm}             % \mathbbm{} for blackboard 1
\usepackage{xspace}
\usepackage{bussproofs}      % typed inference rules in §3.7
\usepackage{array}           % >{...}p{w} columns in tables
\usepackage{colortbl}        % \cellcolor for shaded tabular (tab-ablation per ADR-0011)
\usepackage{graphicx}        % \resizebox auto-fitting prooftrees to column
\usepackage{tikz}            % §3.7 operator-pipeline figure (D-03)
\usetikzlibrary{positioning, arrows.meta, shapes.geometric, fit, calc}
\usepackage{algorithm}       % Algorithm float environment for §2.6 operator pseudocode (M-pseudocode)
\usepackage{algpseudocode}   % \Require / \Ensure / \State / \Return keywords for the four operators
\usepackage{hyperxmp}        % XMP metadata stream required for PDF/A audit
\usepackage{xr-hyper}        % cross-document references via \externaldocument (main <-> supplement); resolves prop:t01-tightness etc. defined in the companion build's .aux file

\newif\ifshowproofs
\showproofstrue
\newtheorem{theorem}{Theorem}
\newtheorem{lemma}{Lemma}
\newtheorem{proposition}{Proposition}
\newtheorem{corollary}{Corollary}
\theoremstyle{definition}
\newtheorem{definition}{Definition}
\theoremstyle{remark}
\newtheorem{remark}{Remark}

\newcommand{\Tvalid}{\ensuremath{\mathcal{T}_{\mathrm{v}}}}   % valid-time domain
\newcommand{\Tsystem}{\ensuremath{\mathcal{T}_{\mathrm{s}}}}  % system-time domain
\newcommand{\bitemp}[1]{\ensuremath{[#1]_{\Tvalid \times \Tsystem}}}
\newcommand{\AsOf}[1]{\ensuremath{\mathsf{AS\;OF}\,#1}}

\newcommand{\Prov}{\ensuremath{\mathcal{P}}}                  % the semiring carrier

\newcommand{\provmul}{\ensuremath{\otimes}}                   % provenance product
\newcommand{\provadd}{\ensuremath{\oplus}}                    % provenance sum

\newcommand{\schedmodels}{\ensuremath{\models}}

\newcommand{\IsoSI}{\ensuremath{\mathsf{SI}}}              % snapshot isolation
\newcommand{\IsoSR}{\ensuremath{\mathsf{SR}}}              % serializable
\newcommand{\IsoRC}{\ensuremath{\mathsf{RC}}}              % read-committed
\newcommand{\IsoRCcb}{\ensuremath{\mathsf{RC{+}cb}}}       % read-committed + callback
\newcommand{\IsoSRpol}{\ensuremath{\mathsf{SR}^{\mathrm{policy}}}} % serializable on policy table

\newcommand{\opLWW}{\ensuremath{\oplus_{\mathrm{t}}}}      % last-write-wins
\newcommand{\opEvi}{\ensuremath{\oplus_{\mathrm{p}}}}      % evidence-weighted
\newcommand{\opAwait}{\ensuremath{\oplus_{?}}}             % await-confirmation
\newcommand{\opRule}{\ensuremath{\oplus_{\mathrm{c}}}}     % per-rule-config
\newcommand{\opDetect}{\ensuremath{\oplus_{\mathrm{d}}}}   % detection oracle (future-work hook; §6 prose)

\newcommand{\Fact}{\mathsf{Fact}}
\newcommand{\fact}[1]{\mathsf{f}_{#1}}
\newcommand{\subj}{\mathrm{subj}}
\newcommand{\pred}{\mathrm{pred}}
\newcommand{\obj}{\mathrm{obj}}
\newcommand{\conf}{\mathrm{conf}}                          % confidence in [0,1]
\newcommand{\strat}{\mathrm{strat}}                        % resolution_strategy_id
\newcommand{\auditt}{\mathsf{Audit}}                       % audit-tuple type (v2 dual-row)
\newcommand{\rowkind}{\mathrm{row\_kind}}                  % {current, audit} discriminator
\newcommand{\witness}{\mathrm{W}}                          % witness polynomial K[X,T]
\newcommand{\Tier}{\mathsf{Tier}}                          % memory tier label (§6 compression spectrum future-work hook)

\newcommand{\mems}{\ensuremath{\mathcal{M}}}               % memory state set
\newcommand{\asof}{\ensuremath{\mathsf{asof}}}             % AS OF as a function on a memory state
\newcommand{\conflict}{\mathrel{\#}}                       % contradiction predicate
\newcommand{\overlap}{\ensuremath{\mathsf{overlap}}}       % Allen-overlap predicate (the nine of §3.3)
\newcommand{\current}{\ensuremath{\textsf{current}}}       % row-kind = current marker

\newcommand{\anomHR}{\ensuremath{\mathsf{N1}}}             % Judge-replay inconsistency
\newcommand{\anomBDS}{\ensuremath{\mathsf{N2}}}            % Belief-Drift Skew
\newcommand{\anomAE}{\ensuremath{\mathsf{N3}}}             % Audit Erasure

\newcommand{\sysmem}{mem0\xspace}
\newcommand{\sysletta}{Letta\xspace}
\newcommand{\sysgraphiti}{Graphiti\xspace}
\newcommand{\syszep}{Zep\xspace}
\newcommand{\sysworlddb}{WorldDB\xspace}          % Ganesan 2026, arXiv:2604.18478
\newcommand{\sysmirix}{MIRIX\xspace}             % verdict transcribed from MMA-Bench, Lu et al. 2026, arXiv:2602.16493
\newcommand{\toki}{\textsc{Toki}\xspace}
\newcommand{\sysours}{\toki}                       % single name: algebra + reference implementation (see \title)

\newcommand{\guard}{\ensuremath{\mathrm{guard}}}
\newcommand{\isoaxis}{\ensuremath{\mathrm{iso}}}
\newcommand{\guardiso}{\ensuremath{\guard_{\isoaxis}}}

\DeclareRobustCommand{\code}[1]{\nolinkurl{#1}}

\newcommand{\etal}{et~al.\xspace}

\usetikzlibrary{backgrounds}

\definecolor{WongBlack}     {HTML}{000000}  % kept pure black
\definecolor{WongOrange}    {HTML}{F59E0B}  % NATURE amber
\definecolor{WongSky}       {HTML}{60A5FA}  % NATURE sky
\definecolor{WongGreen}     {HTML}{0F766E}  % NATURE teal (third hue)
\definecolor{WongYellow}    {HTML}{F59E0B}  % NATURE amber (no separate yellow)
\definecolor{WongBlue}      {HTML}{0B3D91}  % NATURE navy (exclude/safe/ours)
\definecolor{WongVermillion}{HTML}{C25E0C}  % NATURE vermillion (admit/unsafe)
\definecolor{WongReddish}   {HTML}{E11D48}  % NATURE rose
\definecolor{WongGrey}      {HTML}{475569}  % NATURE slate-mid
\definecolor{WongBandGrey}  {HTML}{F1F5F9}  % NATURE cream

\colorlet{opLWWcolor}  {WongSky}
\colorlet{opEvicolor}  {WongGreen}
\colorlet{opAwaitcolor}{WongYellow}
\colorlet{opRulecolor} {WongVermillion}
\colorlet{gatecolor}   {WongBlue}
\colorlet{audittone}   {WongBlue}
\colorlet{rulebordertone}{WongVermillion!75!black}

\tikzset{
  pubfig/.style={
    font=\small,
    node distance=6mm and 6mm,
    every node/.append style={inner sep=2.4pt, align=center},
  },
  factbox/.style={
    draw=WongGrey, line width=0.45pt, rounded corners=2.4pt,
    fill=white,
    minimum height=5.5mm, minimum width=10mm,
    font=\small,
  },
  conflictsym/.style={
    font=\normalsize\bfseries, text=WongGrey,
  },
  gateband/.style={
    draw=gatecolor!70!black, line width=0.6pt,
    rounded corners=3pt,
    fill=gatecolor!12,
    minimum height=9mm,
    minimum width=66mm,
    font=\small,
  },
  opcell/.style={
    draw=black!55, line width=0.55pt,
    rounded corners=2.6pt,
    minimum height=11mm, minimum width=15mm,
    font=\small\bfseries,
  },
  opcell/lww/.style  ={opcell, fill=opLWWcolor!28,   draw=opLWWcolor!75!black},
  opcell/evi/.style  ={opcell, fill=opEvicolor!28,   draw=opEvicolor!75!black},
  opcell/await/.style={opcell, fill=opAwaitcolor!55, draw=opAwaitcolor!60!black},
  opcell/rule/.style ={opcell, fill=opRulecolor!22,  draw=rulebordertone},
  precondtag/.style={
    font=\scriptsize, text=black!75,
  },
  outputband/.style={
    draw=audittone!70!black, line width=0.6pt,
    rounded corners=3pt,
    fill=audittone!10,
    minimum height=9mm,
    minimum width=66mm,
    font=\small,
  },
  pubarrow/.style={
    -{Stealth[length=2.2mm, width=1.7mm]},
    line width=0.55pt,
    draw=black!72,
    shorten >=1.5pt, shorten <=1.5pt,
  },
}

\newcommand{\pvldbtablestyle}{%
  \footnotesize
  \setlength{\tabcolsep}{4pt}%
  \renewcommand{\arraystretch}{1.05}%
}

\newcommand{\tablenote}[1]{%
  \par\addvspace{2pt}%
  {\footnotesize #1\par}%
}

\usepackage{longtable}
\usepackage{fontawesome5}   % \faGithub octocat icon for the code link
\showproofsfalse

\begin{document}

\title{TOKI: A Bitemporal Operator Algebra for Contradiction Resolution in LLM-Agent Persistent Memory}

\author{Ziming Wang}
\affiliation{%
  \institution{The Hong Kong University of Science and Technology}
  \city{Hong Kong}
  \country{Hong Kong SAR}
}
\email{zwangnv@connect.ust.hk}

\renewcommand{\shortauthors}{Ziming Wang}

\begin{abstract}
% @owner       src::sections::abstract
% @does        Story-first abstract in plain prose: contradiction resolution reframed as write-time concurrency control; TOKI types four production heuristics as bitemporal operators; soundness across three axes; necessity result; eight-system verdict matrix; the natural-workload headline number; honest utility scope framed as principled separation; closes on the correctness contract as the contribution. Notation budget: zero math symbols (best-paper paradigm rule 2).
% @needs       src/sections/03-foundations.tex, src/sections/05-measurement.tex, references/refs.bib
% @feeds       src/main.tex
% @breaks      overclaiming on cross-system deltas; reintroducing operator/anomaly symbols into the abstract; throat-clearing openers.

Persistent memory for an LLM agent is a write-heavy substrate:
every belief update is a versioned write, and the system must
decide what to trust when a new claim contradicts a stored one.
Production systems answer with four resolution heuristics,
last-writer-wins, evidence-weighted merge, await-confirmation,
and per-rule policy, yet none declares the isolation level it
assumes or the write-time anomalies it admits. We show that
contradiction resolution is write-time concurrency control, and
make the missing contract explicit. \textsc{Toki} types the four
heuristics as one family of bitemporal operators over a dual-row
schema, each carrying an isolation precondition and a provenance
annotation that preserves the losing fact in an audit row. Four
soundness theorems close the contract across three orthogonal
axes, isolation, schema, and provenance, lift the guarantees
to operator pipelines, and extend the fold operators to n-ary
conflict sets. A tightness companion proves the sharp
result: within the relational schedule model, keyed logging of the
adjudicating judge is necessary for replay consistency, a discipline
every deployed baseline we audit omits. A verdict matrix over
eight systems localizes the gap: every baseline that keeps a
language-model judge on the write path admits at least one of three
write-time anomalies, replay inconsistency, belief-drift skew, or
audit erasure; a content-addressed engine-layer comparator avoids
them only by removing the judge, and \toki alone excludes all three
while keeping it. On its one natural-workload slice, the audit-row
defence moves \textsc{LoCoMo} accuracy by $0.86$, and ablating the typed
memory layer removes $0.49$ accuracy on $1{,}444$ answerable
\textsc{LoCoMo} questions; the cross-system comparison against external
memory systems stays underpowered and claims no superiority. The contribution is the
contract: a write-time correctness specification, proved sound across
isolation, schema, and provenance, that pins the guarantee every
production heuristic assumes and no deployed system makes explicit.

\end{abstract}

\maketitle

% Code / data availability with an embedded GitHub icon (standalone preprint;
% no venue reference-format block).
\begingroup
\renewcommand\thefootnote{}%
\footnotetext{\faGithub~\,Code, data, and reproducibility artifact:
\url{https://github.com/ZenAlexa/toki-bitemporal-memory}}%
\addtocounter{footnote}{-1}%
\endgroup

% ---- Main paper body (sections 1--6), two-column -------------------
% @owner       src::sections::01-motivation
% @does        Opens with the stakes of unprincipled write-time strategies, reframes contradiction resolution as write-time concurrency control anchored on Berenson-Adya isolation, names the three LLM-judge anomalies in plain prose (symbols deferred to S3), and lists claim-first contributions plus the headline finding.
% @needs       src/sections/03-foundations.tex, src/figures/fig1-overview.pdf, src/tables/tab-correspondence.tex, references/refs.bib
% @feeds       src/main.tex
% @breaks      mis-mapping a baseline to a wrong isolation level; introducing pre-trained model names per global rule 10; quantifying "every external baseline admits" over the engine-layer comparator row.

\section{Introduction}
\label{sec:motivation}

Persistent memory for an LLM
agent~\cite{park-2023-generative-agents,packer-2023-memgpt} is a
write-heavy data management substrate: every belief update is a
versioned write that carries a valid time, a system time, a
provenance annotation, and an implicit isolation level. When a new writer
disagrees with the stored belief on a subject-predicate key, the
memory must decide what to trust. Production systems answer with
four strategies, last-writer-wins, evidence-weighted merge,
await-confirmation, and per-rule policy. None declares which
isolation level the strategy assumes or which write-time
anomalies it admits. The cost of that silence is measurable.
\textsc{BeliefShift}~\cite[Tab.~5]{myakala-2026-beliefshift}
leaves up to $42\%$ of cross-session contradictions unresolved
across seven language-model families.
\textsc{TSM}~\cite{su-2026-tsm} recovers $12.2$ accuracy points
on \textsc{LongMemEval} and \textsc{LoCoMo} by separating
dialogue time from occurrence time, an axis production memories
collapse. Adversarial writes corrupt later
retrievals~\cite{2026-sleeper-memory-poisoning} once the store
keeps no defensible record of what it overwrote. Deployed
agent-memory systems carry no name for these failures: our
deployment scan of widely used implementations (Appendix~\ref{app:adapter-ledger}) finds isolation, contradiction,
audit, and bitemporal vocabulary largely absent.

Concurrency control already solved the structural version of
this problem. The isolation hierarchy of Berenson and
Adya~\cite{berenson-1995-isolation,adya-2000-generalized} fixes
which anomalies a level admits when writers race; bitemporal
data models~\cite{snodgrass-1999-tdb} and K-relation
provenance~\cite{green-karvounarakis-tannen-2007} fix how
versioned facts and their lineage are stored and recovered.
\toki ports that machinery to the agent write path. It types the
four heuristics as one family of bitemporal operators over a
dual-row schema, each carrying an isolation precondition and a
provenance annotation that keeps the losing fact in an audit row
(Figure~\ref{fig:operator-pipeline},
Table~\ref{tab:correspondence}). The classical anomaly alphabet,
written for human transactions, cannot name three failures an
LLM judge introduces: \emph{replay inconsistency}, when
re-adjudicating the same contradiction returns a different
winner; \emph{belief-drift skew}, when concurrent confidence
revisions corrupt a subject-predicate partition; and
\emph{audit erasure}, when the overwritten fact becomes
unrecoverable. \toki binds each failure to a point in the
classical machinery, so each one inherits a defence with a
soundness proof.

\begin{figure*}[t]
  \centering
  \includegraphics[width=0.78\textwidth]{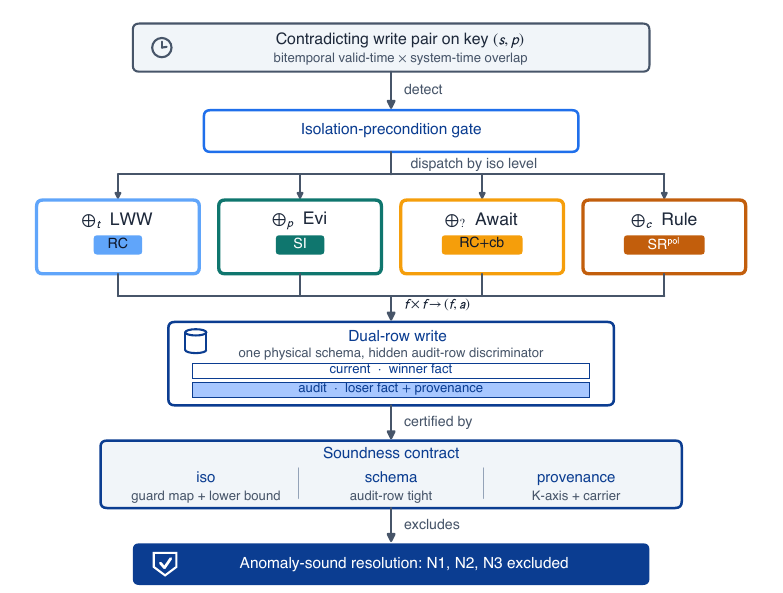}
  \caption{\textbf{Contradiction resolution as write-time
    concurrency control.} A bitemporal substrate detects a
    contradicting pair on a subject-predicate key; an isolation
    gate routes it to one of four typed operators, each pinned to
    the isolation level that excludes the anomaly a weaker level
    admits; every operator commits a current row beside an audit
    row under one schema. The soundness theorems close the
    isolation, schema, and provenance axes.}
  \label{fig:operator-pipeline}
  \Description{Vertical pipeline figure read top to bottom. A
    contradicting write pair on a subject-predicate key enters at
    the top under a bitemporal valid-time by system-time overlap.
    A detect arrow leads to an isolation-precondition gate, which
    fans out to four typed operator heads, last-writer-wins at
    read-committed, evidence-weighted at snapshot isolation,
    await-confirmation at read-committed with callback, and
    per-rule at policy serializable. The heads fan back in to a
    dual-row write into one physical schema with a hidden audit-row
    discriminator, a current winner row above an audit loser row
    carrying provenance. A soundness contract over three axes,
    isolation, schema, and provenance, then certifies the result
    as an anomaly-sound resolution that excludes N1, N2, and N3.}
\end{figure*}

% tab-correspondence.tex --- in-text anchor for the load-bearing
% claim of the paper: each production contradiction-resolution
% strategy is the operational mirror of one classical Berenson
% multiversion anomaly.
%
% Owner   : src/tables/tab-correspondence.tex
% Claim   : Each production strategy is represented as an
%           isolation-typed operator with a named precondition and a
%           named admitted anomaly.
% Does    : 4-row in-text mini-table anchoring the strategy/isolation
%           correspondence claim.
% Needs   : macros.tex (\opLWW, \opEvi, \opAwait, \opRule, \IsoSR);
%           figstyle.tex (\pvldbtablestyle).
% Feeds   : src/sections/01-motivation.tex, 02-related-work.tex,
%           03-foundations.tex (all \ref{tab:correspondence}).
% Breaks  : Drift between §1 strategy prose and this table = paper-spine
%           bug; drift between this table and §3.7 operator declarations
%           = same bug. Column widths sum to <= \columnwidth per
%           docs/process/paper-fitness-protocol.md Rule 2.

\begin{table}[t]
  \caption{\textbf{Production contradiction-resolution strategies as
    isolation-typed operators.} Each row pairs an operator with its
    typed contract and exposed Berenson--Adya anomaly~\cite{berenson-1995-isolation,adya-2000-generalized}.}
  \label{tab:correspondence}
  \centering
  \pvldbtablestyle
  \begin{tabular}{@{}%
    >{\raggedright\arraybackslash}p{0.32\columnwidth}%
    >{\raggedright\arraybackslash}p{0.34\columnwidth}%
    >{\raggedright\arraybackslash}p{0.26\columnwidth}@{}}
    \toprule
    Production strategy             & Isolation precondition   & Admitted anomaly  \\
    \midrule
    Last-writer-wins ($\opLWW$)     & Read-committed           & $P_4$ lost update \\
    Evidence-weighted ($\opEvi$)    & Snapshot isolation       & $A5B$ write skew  \\
    Await-confirmation ($\opAwait$) & RC + callback            & callback boundary \\
    Per-rule policy ($\opRule$)     & $\IsoSR$ on policy table & $P_3$ phantom     \\
    \bottomrule
  \end{tabular}
\end{table}

\paragraph{Contributions.}
(i) \emph{A typed operator algebra} (\S\ref{sec:algebra}).
\toki casts the four production heuristics as one
isolation-indexed family of bitemporal operators with a dual-row
signature, turning four undeclared heuristics into one contract
that states the isolation level it assumes and the provenance it
keeps.
(ii) \emph{A necessity theorem and soundness on three axes}
(\S\ref{sec:found-thm}). A tightness companion proves that keyed
logging of the judge is necessary for replay consistency within
the relational schedule model, a tight characterisation no weaker
discipline meets. An alphabet bridge lifts the classical isolation
guards to the agent write path, and four soundness theorems close
the contract on the isolation, schema, and provenance axes and
lift it to operator pipelines.
(iii) \emph{Empirical evidence} (\S\ref{sec:measurement}). A
verdict matrix over eight systems shows every agent-memory
baseline admits at least one of the three write-time anomalies
while \toki excludes all three; controlled experiments then
stress each defence and anchor every theorem to a grid whose
measured boundary matches the prediction.

\paragraph{Result.}
Keyed logging of the judge is provably necessary for replay
consistency within the relational schedule model, and every deployed
baseline we audit omits it: each admits at least one write-time
anomaly the classical isolation hierarchy already excludes, while
\toki excludes all three while keeping the judge on the write path.
Mechanism isolation confirms each defence: the audit-row defence
moves its natural-workload \textsc{LoCoMo} slice by $0.86$, and ablating
the typed memory layer removes $0.49$ accuracy on $1{,}444$ answerable
\textsc{LoCoMo} questions. The cross-system comparison against external
memory systems stays underpowered and claims no superiority. The reference implementation, benchmark harnesses, and
reproduction runbook are publicly available at the repository linked
on the first page.
     % 1 Introduction
% @owner       src::sections::background
% @does        Fixes the three classical substrates TOKI builds on (bitemporal data model, multiversion isolation, K-semiring provenance) over the LLM-agent write path: the core bitemporal fact, the contradiction event with the running example, and the isolation and provenance vocabulary the operator algebra uses.
% @needs       references/refs.bib
% @feeds       src/sections/03-foundations.tex, src/main.tex
% @breaks      desynchronising the seven-event schedule alphabet from the theorem-file guards; redefining a substrate the algebra section assumes.

\section{Background}
\label{sec:background}

\toki builds on three classical substrates, applied to the LLM-agent
write path: the bitemporal data model, multiversion isolation, and
K-semiring provenance. This section fixes the core object, the
contradiction event the algebra resolves, and the vocabulary
\S\ref{sec:foundations} uses.

\subsection{Bitemporal facts and the contradiction event}
\label{sec:found-schedule}

The core object of \toki is the bitemporal fact: a subject,
predicate, object triple stamped with a valid-time period, a
system-time period, a K-relation provenance annotation, and a
confidence. Every belief a memory holds is one such fact, and every
contradiction is a pair of facts that disagree on a shared key.
Seven event types compose every history: $b_i$ (begin),
$r_i(\fact{})$ (read), $w_i(\fact{})$ (write), $c_i$ (commit),
$a_i$ (abort), $j_i(R, \theta)$ (judge invocation on read set
$R$ under decoder tuple
$\theta = (\textsf{prompt}, \textsf{seed}, \textsf{model},
\textsf{temperature}, \textsf{tool\_hash})$), and
$\mathit{cb}_i(R, h, k)$ (external callback for request hash
$h$ returning winner index $k$). A memory state
$\mems \subseteq \bitemp{\Fact}$ is a finite set of bitemporal
facts, and the view $\asof(\mems, t_v, t_s)$ returns the
current-kind rows whose periods contain $(t_v, t_s)$. Two facts
contradict ($\fact{1} \conflict \fact{2}$) when they agree on
subject and predicate, disagree on object, and their valid-time
periods share a common instant. Under the closed-open convention
$[t_{\mathrm{from}}, t_{\mathrm{to}})$ this holds for nine of
Allen's thirteen base relations~\cite{allen-1983-cacm}, all but
$\mathit{before}$, $\mathit{after}$, $\mathit{meets}$, and
$\mathit{met\text{-}by}$, the four sharing no interior instant.
Database time factors into valid time $\Tvalid$ and system time
$\Tsystem$~\cite{snodgrass-ahn-1986,snodgrass-1999-tdb} through the
SQL:2011 as-of surface~\cite{kulkarni-michels-2012}.

One running example carries the construction through the sections
that follow. The incumbent fact $f_1$ is
\texttt{(alice, medication, penicillin)}, valid from March; a later
writer commits $f_2$, \texttt{(alice, medication, amoxicillin)}, at
system time April~15. The two agree on subject and predicate and
disagree on object, so they contradict, their valid-time periods
overlapping on a shared instant. \toki resolves the pair with one
operator, commits the winner to the current row, and writes the loser
to an audit row recoverable at every later system time.

\subsection{Isolation levels and provenance semirings}
\label{sec:found-iso}\label{sec:found-prov}

\toki types each operator against an isolation precondition drawn
from the multiversion hierarchy.
\citet{berenson-1995-isolation} characterise four ANSI-SQL
phenomena $P_0$ to $P_3$ and three multiversion anomalies
($P_4$ lost update, $A5A$ read skew, $A5B$ write skew);
\citet{adya-2000-generalized} generalize the taxonomy with
schedule-history predicates. We model prevention as a guard
map over the chain
$\mathcal{L}_{\mathrm{iso}} = \{\IsoRC \preceq \IsoSI \preceq \IsoSR\}$,
and the callback boundary used by $\opAwait$ adds an orthogonal
binary axis $\mathcal{L}_{\mathrm{cb}} = \{\bot \preceq
\mathsf{cb}\}$. Operator preconditions live in the product
$\mathcal{L} = \mathcal{L}_{\mathrm{iso}} \times
\mathcal{L}_{\mathrm{cb}}$, and the table-scoped $\IsoSRpol$
pins the named policy table to $\IsoSR$.

\toki annotates every tuple with provenance so the loser of a
resolution stays algebraically recoverable. Each tuple carries an
element of a commutative semiring
$\Prov = \langle K, \provadd, \provmul, 0, 1 \rangle$ with the
natural order
$a \preceq_K b \Leftrightarrow \exists c.\ a \provadd c = b$~\cite{green-karvounarakis-tannen-2007}.
Provenance is a two-sort polynomial $\witness \in K[X, T]$ with
$X$ over write-event tokens and $T$ over trust and policy
variables. Three carriers instantiate the audit-row schema:
multilinear $\mathbb{N}[X, T]$, multidegree
$\mathbb{N}[X, T]^{\#}$ (where $x \provmul x = x^{2}$), and the
Boolean reduct. The carrier choice is decorative on the isolation
axis and load-bearing on the provenance axis, a separation \toki
makes precise in Proposition~\ref{prop:k-axis-separation}.
        % 2 Background
% @owner       src::sections::03-foundations
% @does        The TOKI operator algebra: the dual-row bitemporal schema and the four typed contradiction-resolution operators with their isolation-indexed inference rules, building on the Background vocabulary.
% @needs       src/sections/background.tex, src/tables/tab-schema-11col.tex
% @feeds       src/main.tex, src/sections/soundness.tex, src/sections/04-implementation.tex
% @breaks      leaking reference-implementation identifiers into the operator-statement prose; desynchronising the operator isolation pins from the theorem-file guards.

\section{The \toki Operator Algebra}
\label{sec:foundations}

Building on the bitemporal facts and the isolation and provenance
vocabulary of \S\ref{sec:background}, \toki lifts the four production
contradiction-resolution strategies into one typed algebra over a
dual-row schema, with the soundness guarantees those operators earn
(\S\ref{sec:found-thm}) and the reference system that realises them
over an unmodified engine (\S\ref{sec:impl}).

\subsection{The dual-row schema}
\label{sec:schema}

The isolation precondition and the provenance annotation of
\S\ref{sec:found-iso} meet in one physical layout. \toki lifts agent
persistent facts into a single bitemporal table with eleven
user-visible columns (fact identifier; subject, predicate, object;
valid-time and system-time bounds; provenance annotation
$p \in \Prov$; confidence $\conf$; strategy identifier $\strat$). A
check-constrained $\rowkind \in \{\textsf{current},
\textsf{audit}\}$ discriminator partitions the table into
user-visible rows and the audit rows emitted by every operator
(Table~\ref{tab:schema-11col}). An audit row carries
$\auditt = (p_{\mathrm{w}} \provadd p_{\mathrm{l}}, \strat,
t_s)$; default retrieval filters $\rowkind = \textsf{current}$,
and audit rows reach through a separate audit-log slice. The lift
is conservative: every relational-algebra query over the eleven
base columns commutes with the audit-discriminator filter
(Proposition~\ref{prop:schema-lift-conservatism}), so the
audit row changes no answer a caller relied on.

% tab-schema-11col.tex --- §3.6 schema anchor table.
%
% Owner   : src/tables/tab-schema-11col.tex
% Claim   : The user-visible bitemporal fact schema has 11 columns
%           (identity + content + two SQL:2011 periods + K-relation
%           provenance + typed product); the hidden 12th column is the
%           CHECK-bound current/audit discriminator.
% Does    : Booktabs body for the 11 visible columns plus the hidden
%           row_kind discriminator (highlighted as the load-bearing row).
% Needs   : macros.tex (\subj, \pred, \obj, \conf, \strat, \opLWW,
%           \opEvi, \opAwait, \opRule, \code, \rowkind, \anomAE);
%           figstyle.tex (\pvldbtablestyle); booktabs.
% Feeds   : src/sections/03-foundations.tex §3.6 prose.
% Breaks  : Drift between the table and \S\ref{sec:schema} prose is a
%           paper-spine bug; the hidden discriminator must remain the
%           visually distinct row.

\begin{table}[t]
  \caption{\textbf{User-visible bitemporal fact schema with audit
    discriminator.} Cols 1--4 carry fact identity and content; cols
    5--8 are the two SQL:2011 periods~\cite{kulkarni-michels-2012};
    col 9 is the K-relation provenance
    annotation~\cite{green-karvounarakis-tannen-2007}; cols 10--11 are
    our typed product extension (\S\ref{sec:found-prov}). The
    \textbf{shaded row} is the hidden CHECK-bound discriminator that
    splits the table into current and audit slices and carries the
    $\anomAE$ defence of Theorem~\ref{thm:audit-erasure-schema}.}
  \label{tab:schema-11col}
  \centering
  \pvldbtablestyle
  \begin{tabular}{@{}r@{\hspace{6pt}}>{\raggedright\arraybackslash}p{0.33\columnwidth}@{\hspace{8pt}}l@{\hspace{8pt}}>{\raggedright\arraybackslash}p{0.33\columnwidth}@{}}
    \toprule
    \#  & Column                        & Type      & Role                                       \\
    \midrule
    1   & \code{fact_id}                & TEXT      & primary key (with col.~7)                  \\
    2   & \code{subject}                & TEXT      & $\subj$                                    \\
    3   & \code{predicate}              & TEXT      & $\pred$                                    \\
    4   & \code{object}                 & TEXT      & $\obj$, fact content                       \\
    5   & \code{valid_from}             & TIMESTAMP & valid-time period start                    \\
    6   & \code{valid_to}               & TIMESTAMP & valid-time period end                      \\
    7   & \code{system_time_start}      & TIMESTAMP & system-time period start (PK part)         \\
    8   & \code{system_time_end}        & TIMESTAMP & system-time period end                     \\
    9   & \code{provenance_id}          & TEXT      & K-relation annotation $p \in \Prov$        \\
    10  & \code{confidence}             & DOUBLE    & $\conf \in [0,1]$                          \\
    11  & \code{resolution_strategy_id} & TEXT      & $\strat \in \{\opLWW, \opEvi, \opAwait, \opRule\}$ \\
    \midrule
    \rowcolor{WongBandGrey}
    12  & \textbf{\code{row_kind}}      & \textbf{TEXT} & \textbf{hidden discriminator} $\in \{\textsf{current}, \textsf{audit}\}$ \\
    \bottomrule
  \end{tabular}
\end{table}

\subsection{Four typed operators}
\label{sec:algebra}

\toki exposes the four production contradiction-resolution
strategies as one isolation-indexed operator family. Each operator
$\oplus_a$ takes two contradicting facts $\fact{1}, \fact{2}$ over the
same $(\subj, \pred)$, returns a winner with a system-time
invalidation of the loser, and emits an audit tuple as the second
component of the conclusion. The four differ in winner selector,
isolation precondition, and strategy stamp:
$\opLWW$ (last-writer-wins, $\IsoRC$, exposes $P_4$),
$\opEvi$ (evidence-weighted, $\IsoSI$, exposes $A5B$),
$\opAwait$ (await-confirmation, $\IsoRCcb$, exposes
$\mathsf{callback}$), and $\opRule$ (per-rule, $\IsoSRpol$
on the policy table, exposes $P_3^{\mathrm{policy}}$). This
pairing drives the correspondence of
Table~\ref{tab:correspondence}.

\begin{definition}[Operator isolation signature]
\label{def:operator-isolation-signature}
For each $a \in \{\opLWW, \opEvi, \opAwait, \opRule\}$,
$\mathsf{req}(a) \in \mathcal{L}$ is the typing precondition
and $\mathsf{exposes}(a)$ the failure mode exposed by weakening
it. The audit tuple is
$\auditt(i^{\star}, i) = (f_{i^{\star}}.\witness \provadd
f_i.\witness, \strat, t_s)$. The schedule judgement
$\schedmodels_L$ asserts the writer-slice sub-schedule on
$\{f_1, f_2\}$ violates none of the predicates classically
forbidden at level $L$.
\end{definition}

The four inference rules share the dual-row signature and differ
in winner selector and strategy stamp: $i_t^{\star}$ under
$(\textit{sf}, \textit{id})$ (Lww), $i_p^{\star}$ under
$(\conf, \textit{sf}, \textit{id})$ (Evi), $\mathsf{cb}(f_1,
f_2)$ from a delivered callback row (Await), and
$\rho(f_1, f_2)$ from a policy row keyed by the unordered
identity set (Rule):

\begin{small}
\begin{align*}
  \frac{\schedmodels_{\IsoRC}}
       {\opLWW(f_1, f_2) = (f_{i_t^{\star}},\, \auditt(i_t^{\star}, \overline{i_t^{\star}}))}
  & \;\textsc{(Lww)}
  \\[2pt]
  \frac{\schedmodels_{\IsoSI}}
       {\opEvi(f_1, f_2) = (f_{i_p^{\star}},\, \auditt(i_p^{\star}, \overline{i_p^{\star}}))}
  & \;\textsc{(Evi)}
  \\[2pt]
  \frac{\schedmodels_{\IsoRCcb} \;\; \textsf{cb}(f_1, f_2) = i}
       {\opAwait(f_1, f_2) = (f_i,\, \auditt(i, \overline{i}))}
  & \;\textsc{(Await)}
  \\[2pt]
  \frac{\schedmodels_{\IsoSR}^{\mathrm{policy}} \;\; \rho(f_1, f_2) = i}
       {\opRule(f_1, f_2) = (f_i,\, \auditt(i, \overline{i}))}
  & \;\textsc{(Rule)}
\end{align*}
\end{small}

The discarded fact's system-time end closes; the winner inherits
a fresh system-time start, the merged provenance, and the
strategy stamp. Each $\mathsf{exposes}$ pairs an isolation
precondition with its failure mode
(Table~\ref{tab:correspondence}), and Appendix~\ref{app:operator-algorithms} carries the operator
algorithms in full.

Each operator resolves the running $(f_1, f_2)$ differently:
$\opLWW$ keeps \texttt{amoxicillin} as the latest write at
read-committed; $\opEvi$ compares the two confidence stamps at
snapshot isolation; $\opAwait$ blocks on a clinician callback;
$\opRule$ applies a formulary policy row pinned to serializable.
Each operator excludes exactly the anomalies its isolation pin
dominates, the guarantee \S\ref{sec:found-thm} makes precise.
    % 3 The TOKI Operator Algebra
% @owner       src::sections::soundness
% @does        States the soundness of TOKI as observable guarantees on three orthogonal axes (isolation, schema, provenance), plus a composition theorem and a tightness companion, each opening with a one-line intuition and a forward pointer to its Section 6 empirical anchor.
% @needs       src/theorems/T-01-provenance-preservation.tex, src/theorems/T-01b-audit-erasure-schema.tex, src/theorems/T-02-k-axis-separation.tex, src/theorems/T-04-carrier.tex, src/theorems/L-01-composition-well-typed.tex, src/theorems/T-05-composition.tex, src/theorems/T-07-nary-conflict-set.tex, src/theorems/T-06-n1-lower-bound.tex, src/tables/tab-anomaly-defences.tex
% @feeds       src/main.tex, src/sections/05-measurement.tex
% @breaks      desynchronising the schedule alphabet from the theorem-file guards; dropping a soundness theorem's forward pointer to its Section 6 empirical anchor.

\subsection{Soundness: What \toki Guarantees}
\label{sec:found-thm}

\toki excludes three failure modes that isolation-spec reasoning
alone cannot name: replay inconsistency ($\anomHR$), when
re-adjudicating a contradiction returns a different winner;
belief-drift skew ($\anomBDS$), a write skew specialised to a
$(\subj, \pred)$ partition; and audit erasure ($\anomAE$), loss of
the overwritten fact. The classical alphabet misses these because
it has no event for an LLM judge call, and one lemma closes that
gap. The alphabet bridge (Lemma~\ref{lem:judge-callback-bridge})
reads each judge call as a read of a logged verdict that the
operator commits before its own commit, which lifts the
Berenson--Adya guards onto the agent write path. Under that
reading, replay inconsistency becomes a textbook fuzzy read on the
logged row, which snapshot isolation forbids, plus an insert
phantom on the row's first write, which serializability forbids.
The remaining two failures attach to a schema decision and a
provenance decision, so the guarantees split across three
orthogonal axes, isolation, schema, and provenance
(Table~\ref{tab:anomaly-defences}).

% tab-anomaly-defences.tex --- §3 anchor table.
% Single-column edition: 10-row Berenson--Adya + three agent-memory
% predicates with verifier guard and runtime enforcement.

\begin{table}[t]
  \caption{\textbf{Each anomaly maps to one verifier guard and one
    runtime enforcement.} Upper block: classical Berenson--Adya
    iso-axis anomalies. Lower block: the three agent-memory
    predicates ($\anomBDS$ corollary of $A5B$ on the
    $(\subj,\pred)$ projection per
    Corollary~\ref{cor:n2-corollary}).}
  \label{tab:anomaly-defences}
  \centering
  \pvldbtablestyle
  \footnotesize
  \begin{tabular}{@{}%
    >{\raggedright\arraybackslash}p{0.42\columnwidth}%
    >{\centering\arraybackslash}p{0.16\columnwidth}%
    >{\raggedright\arraybackslash}p{0.30\columnwidth}@{}}
    \toprule
    Predicate                              & Axis      & Defence              \\
    \midrule
    $P_0$ dirty write                      & iso       & $\IsoRC$             \\
    $P_1$ dirty read                       & iso       & $\IsoRC$             \\
    $P_2$ fuzzy read                       & iso       & $\IsoSI$             \\
    $P_3$ phantom                          & iso       & $\IsoSR$             \\
    $P_4$ lost update                      & iso       & $\IsoSI$             \\
    $A5A$ read skew                        & iso       & $\IsoSI$             \\
    $A5B$ write skew                       & iso       & $\IsoSR$             \\
    \midrule
    $\anomHR$ judge-replay inconsistency   & judge     & $\IsoSR$ logged judge\\
    $\anomBDS$ belief-drift skew           & partition & partition $\IsoSR$   \\
    $\anomAE$ audit erasure                & schema    & audit row            \\
    \bottomrule
  \end{tabular}
\end{table}

\toki's first guarantee is on the isolation axis: an operator at level
$L$ excludes a classical anomaly $\phi$ exactly when $L$ dominates
$\phi$'s guard, because a level forbids exactly the anomalies its
guard rules out, with the keyed-log discipline carrying $\anomHR$ and
the partition pin carrying $\anomBDS$ into the same statement. The iff
is sharp in both directions, so the guarantee is also a limit on every
weaker level. The iso-axis grid of \S\ref{sec:meas-g5} verifies the
predicted $0$/$1$ boundary, each dominating cell admitting no schedule
and each weaker cell admitting every schedule.

\begin{definition}[Iso-axis defence singletons]
\label{def:guard-iso-singleton}
For each classical
$\phi \in \Phi_{\mathrm{class}} =
\{P_0,P_1,P_2,P_3,P_4,A5A,A5B\}$~\cite{berenson-1995-isolation,adya-2000-generalized},
the guard map $\guardiso \colon \Phi_{\mathrm{class}} \to
\mathcal{L}_{\mathrm{iso}}$ is
\[
\guardiso(\phi) =
\begin{cases}
\IsoRC & \phi \in \{P_0,P_1\},\\
\IsoSI & \phi \in \{P_2,P_4,A5A\},\\
\IsoSR & \phi \in \{P_3,A5B\}.
\end{cases}
\]
The chain-valued $\IsoSR$ defence for $P_3$ is the uniform minimum;
L-01 admits the per-table tightening to $\IsoSRpol$ when $\opRule$
pins the policy table.
Table~\ref{tab:anomaly-defences} reports the same map.
\end{definition}

\paragraph{K-vacuity warning.}
The guard map $\guardiso$ is carrier-vacuous (equivalently
$K$-vacuous): it is defined over $\mathcal{L}_{\mathrm{iso}}$
alone and its soundness reduces to Berenson--Adya schedule
predicates that mention no semiring carrier. The iff below is
stated parametrically in $K$ for notational uniformity with
Theorem~\ref{thm:audit-erasure-schema}, but
Proposition~\ref{prop:k-axis-separation} establishes the
algebraic separation: the iso-axis preconditions remain
carrier-vacuous on every commutative semiring with natural
order, while the schema-axis $\preceq_K$ relation of
Theorem~\ref{thm:audit-erasure-schema} is $K$-load-bearing. A
reviewer reading the iff in isolation may suspect the $K$
parameter is decorative; the separation Proposition makes the
honest scoping visible.

\begin{theorem}[Iso-lattice anomaly soundness for typed agent-memory operators]
\label{thm:anomaly-soundness}
Let $K$ be a commutative semiring with the natural order
$\preceq_K$~\cite{green-karvounarakis-tannen-2007}: $a \preceq_K b$
iff there exists $d \in K$ with $a \provadd d = b$. Let
$\Phi \subseteq \{P_0, P_1, P_2, P_3, P_4, A5A, A5B\}$ be a subset
of the classical Berenson--Adya hierarchy~\cite{berenson-1995-isolation,adya-2000-generalized},
and let $L$ be an isolation level drawn from the iso chain
$\mathcal{L}_{\mathrm{iso}} = \{\IsoRC \preceq \IsoSI \preceq \IsoSR\}$
of \S\ref{sec:found-iso}. Under the dual-row schema of
\S\ref{sec:schema} and the typed operator algebra of
\S\ref{sec:algebra} carrying provenance in $K$, define
\begin{equation}
\label{eq:prevents}
\mathrm{Prevents}(L,\Phi) \equiv
\forall S.\; S \schedmodels L \Rightarrow
\forall \phi \in \Phi.\; S \models \neg \phi .
\end{equation}
Then $\mathrm{Prevents}(L,\Phi)$ iff
$L \succeq \guardiso(\Phi)$, where
$\guardiso$ extends from
Definition~\ref{def:guard-iso-singleton} to subsets by lattice
supremum:
$\guardiso(\Phi) \;=\; \bigvee_{\phi \in \Phi} \guardiso(\phi)$.
The guard map is independent of the carrier; replacing $K$ with any
of the three shipped carriers preserves the same isolation
precondition.
\end{theorem}

\ifshowproofs
\else
\noindent\emph{Proof sketch.} The soundness direction proceeds by
induction on schedule history over the seven-event alphabet of
\S\ref{sec:found-schedule}; the inductive step verifies that no
operator emission of \S\ref{sec:algebra} introduces a forbidden
$\phi \in \Phi$ when its precondition $\schedmodels_L$ holds.
The tightness direction is discharged by seven minimal witness
schedules, one per classical predicate, that hold the immediate
predecessor of $\guardiso(\phi)$ on $\mathcal{L}_{\mathrm{iso}}$.
The K-vacuity warning above pairs with
Proposition~\ref{prop:k-axis-separation} to show the iso-axis guard
remains carrier-vacuous across the three shipped carriers.
The judge-callback alphabet bridge (Lemma below) shows the typed
operator algebra over the seven-event alphabet $\Sigma_{+}$ refines
to the classical five-event $\Sigma$ exactly when each $j_i$ event
maps to a logged read on a serializable judge table; under that
refinement the iso-axis predicates $\anomHR$ and $\anomBDS$ inherit
the classical defence by reduction, and the soundness induction
above closes the agent-memory specialisation. The full proof,
the bridge lemma's proof, and the tightness witnesses live in the
theorem files.
\fi

\begin{lemma}[Judge-callback alphabet bridge]
\label{lem:judge-callback-bridge}
Let $\Sigma = \{b_i, r_i(\fact{}), w_i(\fact{}), c_i, a_i\}$ be the
classical schedule alphabet of \S\ref{sec:found-schedule} and let
$\Sigma_+ = \Sigma \cup \{j_i(R, \theta), \mathit{cb}_i(R, h, k)\}$
be the augmented alphabet. Define the map
$\pi_J : \Sigma_+^\ast \to \Sigma^\ast$ that fixes $\Sigma$
pointwise, substitutes the first $j_i$ at each key $(R, \theta)$
by an insert-if-absent write and every subsequent $j_i$ at that
key by a read, and substitutes each callback by a read:
\begin{align*}
  j_i^{\mathrm{first}}(R, \theta)
    &\;\mapsto\; w_i\!\bigl(\textsf{judge\_log}[(R, \theta)]\bigr), \\
  j_i^{\mathrm{repeat}}(R, \theta)
    &\;\mapsto\; r_i\!\bigl(\textsf{judge\_log}[(R, \theta)]\bigr), \\
  \mathit{cb}_i(R, h, k)
    &\;\mapsto\; r_i\!\bigl(\textsf{callback\_log}[(h)]\bigr).
\end{align*}
Under the keyed-log discipline of \S\ref{sec:schema}, every judge call
and callback commits its keyed-log row before the operator commits, an
invariant \S\ref{sec:impl} enforces in the dispatcher (full hypotheses
$H1$, $H2$ in Appendix~\ref{app:graded-judge-bridge}). Then
$\pi_J$ preserves the operand-table conflict graph, so every classical
predicate verdict transfers and $\pi_J$ refines $S$ on the iso axis.
Consequently $\anomHR$ at key $(R, \theta)$ reduces to two classical
patterns on the keyed row,
\begin{equation}
  \label{eq:bridge-anomhr}
  S \models \anomHR \;\text{at}\; (R, \theta)
  \;\iff\;
  \pi_J(S) \models (P_2 \vee P_3) \;\text{at}\;
  \textsf{judge\_log}[(R, \theta)],
\end{equation}
a fuzzy read $P_2$ excluded at $\IsoSI$ and an insert phantom $P_3$
excluded at $\IsoSR$ or a linearizable insert-if-absent. The full
statement and proof are Appendix~\ref{app:graded-judge-bridge}.
\end{lemma}

\begin{corollary}[$\anomBDS$ defence on the partition data item]
\label{cor:n2-corollary}
Treat the $(\subj,\pred)$-projected multiset
$\pi_{(s,p)}(\mems)$ as a single data item over which $\opEvi$
is the read function and an agent's belief-cache row
$f_{\mathrm{bc}}$ carries the invariant
$I:\;f_{\mathrm{bc}} = \opEvi(\pi_{(s,p)})$. A belief query
$T_q$ that reads $\pi_{(s,p)}$ and writes $f_{\mathrm{bc}}$,
concurrent with a confidence revision $T_w$ that reads
$\pi_{(s,p)}$ and writes a row's $\conf$ field, has intersecting
read-sets and disjoint write-sets; under $\IsoSI$ both commit
and $I$ is violated. This is the $A5B$ write-skew shape of
Adya~\cite{adya-2000-generalized} on the partition data item.
$\IsoSR$ on the $(\subj,\pred)$ partition serializes $T_w$
outside $T_q$'s execution and restores $I$;
Theorem~\ref{thm:anomaly-soundness} applied with
$\Phi = \{A5B\}$ on the partition data item discharges the
defence, with the data-item promotion from individual rows to
the partition projection following the standard
conflict-serializability framework of
Bernstein-Hadzilacos-Goodman~\cite{bernstein-hadzilacos-goodman-1987}.
\end{corollary}

\begin{remark}[Per-partition scope]
\label{rem:cor-n2-per-partition}
The defence is stated per-partition; the operator pre-condition is that the calling code names the partition. Multiple concurrent partitions require independent $\IsoSR$ pins, and cross-partition skew is outside the corollary statement and surfaces as follow-up work (\S\ref{sec:open}).
\end{remark}

\begin{corollary}[$\anomHR$ defence by logged-judge reduction]
\label{cor:n1-corollary}
Pin the judge parameter $\theta = (\textsf{prompt}, \textsf{seed},
\textsf{model\allowbreak\_\allowbreak version}, \textsf{temperature},
\textsf{tool\allowbreak\_\allowbreak output\allowbreak\_\allowbreak hash})$
and commit the vote into the keyed judge table under a
linearizable insert-if-absent. The first call materializes the
row keyed by $(R, \theta)$; every subsequent invocation reads the
log and skips a fresh $J$ call. Determinism comes from the log's
keyed entries; the oracle $J$ itself remains non-deterministic. By
Lemma~\ref{lem:judge-callback-bridge}
equation~\eqref{eq:bridge-anomhr}, $S \models \anomHR$ at
$(R, \theta)$ iff $\pi_J(S)$ admits a fuzzy read $P_2$ on a
re-read or an insert phantom $P_3$ on the first materialization of
$\textsf{judge\_log}[(R, \theta)]$. Snapshot isolation already
excludes the re-read $P_2$ (\guardiso($P_2$) $= \IsoSI$); the
first materialization is a phantom that $\IsoSI$ admits and
$\IsoSR$ excludes (\guardiso($P_3$) $= \IsoSR$), equivalently a
linearizable insert-if-absent on the key. Applying
Theorem~\ref{thm:anomaly-soundness} to $\pi_J(S)$ with
$\Phi = \{P_2, P_3\}$ at $L = \IsoSR$ on the keyed judge log
excludes both patterns, so \textbf{$\anomHR$ is excluded on $S$}.
The operative requirement is linearizable keyed materialization,
of which $\IsoSR$ on the keyed judge log is the schedule-level
realisation; this is why Table~\ref{tab:anomaly-defences} reports
a serializable read on the keyed log as the required level.
\end{corollary}

\paragraph{Call-site contract for Corollary~\ref{cor:n1-corollary}.}
The defence is conditional on the operator reading the keyed
judge-log row at $(R, \theta)$ before issuing the $j_i$ event: the
alphabet bridge of Lemma~\ref{lem:judge-callback-bridge} promotes
$j_i$ to a logged read only when a prior operator call has
committed the row. An operator that issues a fresh oracle call
without consulting the log falls outside the lemma's hypotheses
$(H1)$+$(H2)$, and the defence does not apply. The obligation is
one of operator construction, not a schedule-level invariant: the
typed operators $\opAwait$ and $\opRule$ read the keyed log before
they commit by construction (\S\ref{sec:impl}), so the soundness
holds without any assumption about agent code that might bypass the
operator surface.

The schema axis is orthogonal to isolation, and \toki's second
guarantee lives there: the audit-row schema dominates every loser's
provenance under the K-semiring natural order, so $\anomAE$ cannot
occur, while the base schema admits a witness that erases it. The
tightness direction carries the weight, since only the audit-row
schema closes the gap. The schema-axis grid of \S\ref{sec:meas-g5}
records $\anomAE$ at $100\%$ under the base schema and $0\%$ under
the audit-row schema.

% @owner       src::theorems::T-01b-audit-erasure-schema
% @does        Proves the schema-lift conservatism Proposition and the audit-erasure schema-lift theorem (T-01b): the audit-row schema lifts query semantics conservatively for Q_base and defends N3 from the database side via K-natural-order provenance reachability.
% @needs       src/macros.tex, src/theorems/T-01-provenance-preservation.tex, references/refs.bib (green-karvounarakis-tannen-2007)
% @feeds       src/sections/03-foundations.tex, src/sections/05-measurement.tex, src/theorems/L-01-composition-well-typed.tex (composition cites thm:audit-erasure-schema)
% @breaks      a wrong K-natural-order claim or a mis-stated audit tuple emission rule invalidates the schema-axis defence; stale reference to base-schema witness schedule leaves tightness clause unsupported.
%
% Audit-erasure schema lift (T-01b under ADR-0008 D-T01 option-3 split).
%
% Extracts the schema-axis case from the pre-split T-01: N3
% audit erasure is a structural reachability failure orthogonal
% to the iso lattice, defended by the audit-row schema extension
% of §3.6 through a schema pinning. Filename and
% label both new under ADR-0008.
%
% The split mirrors the algebraic shape: the iso lattice (T-01a)
% represents and checks the 7 classical Berenson-Adya anomalies plus the two
% LLM-specific corollaries N1, N2; the schema axis (T-01b)
% defends N3 by provenance reachability under K-natural-order.
% L-01 (composition lemma) joins the two when an operator chain
% encounters anomalies on both axes.

\begin{theorem}[Audit-erasure schema lift]
\label{thm:audit-erasure-schema}
Let $K$ be a commutative semiring with the natural order
$\preceq_K$~\cite{green-karvounarakis-tannen-2007}, and let
$\mathcal{L}_{\mathrm{schema}} = \{\textsf{base} \preceq \textsf{audit{-}row}\}$
be the schema-augmentation lattice of \S\ref{sec:schema}. Under the
$\textsf{audit{-}row}$ refinement, every typed operator of
\S\ref{sec:algebra} emits an audit tuple
$\auditt = (p_w \provadd p_l, \strat, t_s)$ alongside the winner,
where $p_w, p_l \in K$ are the provenance polynomials of the winner
and loser facts. For every schedule history $H$ executed against the
$\textsf{audit{-}row}$ schema and every loser fact $f_l$ with provenance
$p_l$ consumed by an operator step in $H$, there exists an audit tuple
$t \in \mathsf{Rec}_H(t_v,t_s)$ with
$p_l \preceq_K t.\textit{prov}$; conversely, against the $\textsf{base}$
schema, schedule histories exist where no $t \in \mathsf{Rec}_H$
dominates $p_l$, witnessing anomaly $\anomAE$.
\end{theorem}

\noindent\emph{Proof sketch.}
Soundness follows from $p_l \preceq_K p_w \provadd p_l$ under any
commutative semiring with natural
order~\cite{green-karvounarakis-tannen-2007} (existence witness
$d = p_w$); the audit tuple's $t.\textit{prov} = p_w \provadd p_l$
dominates $p_l$. Tightness sits in Appendix~\ref{sec:appendix-witnesses-schema}: the witness $S_{N3}$ pins the lower bound on the
schema lattice, parametric in the carrier.

The provenance axis separates the two preceding guarantees: the
isolation guard is carrier-vacuous and only the schema guard is
K-load-bearing, so \toki's K-parametric framing is load-bearing. The
carrier-recoverability theorem refines the axis, making per-instance
erasure recoverable from the carrier polynomial alone exactly when the
carrier records multiplicity, which orders the three shipped carriers.
The carrier ablation of \S\ref{sec:meas-g4} witnesses this, with
verdicts invariant across carriers and token recall splitting them.

\begin{proposition}[K-vacuity on the iso axis and K-load on the schema axis: an algebraic separation]
\label{prop:k-axis-separation}
Let $K_1, K_2$ be any two commutative semirings with natural order
$\preceq_{K_i}$~\cite{green-karvounarakis-tannen-2007}, and let
$L \in \mathcal{L}_{\mathrm{iso}}$ be an isolation level in the
Berenson--Adya hierarchy~\cite{berenson-1995-isolation,adya-2000-generalized}.
Under the dual-row schema of \S\ref{sec:schema}, the K-parametric
framing of \S\ref{sec:found-thm} factors cleanly into the two axes:
\begin{enumerate}
\item \emph{(Iso-axis K-vacuity.)} For every classical anomaly
$\phi \in \Phi_{\mathrm{class}}$ in
Theorem~\ref{thm:anomaly-soundness}, the iso-lattice prevention
predicate $\mathrm{Prevents}(L, \{\phi\})$ is independent of the
carrier choice:
$\mathrm{Prevents}^{K_1}(L, \{\phi\}) =
\mathrm{Prevents}^{K_2}(L, \{\phi\})$ for every
$K_1, K_2$. The guard map $\guardiso$ of
Definition~\ref{def:guard-iso-singleton} is carrier-agnostic by
construction: its codomain is $\mathcal{L}_{\mathrm{iso}}$ and its
soundness proof reduces to the Berenson--Adya schedule-history
predicates~\cite{adya-2000-generalized}, which do not reference
$K$ in any clause.

\item \emph{(Schema-axis K-load.)} For the audit-erasure anomaly
$\anomAE$, the schema-lift defence of
Theorem~\ref{thm:audit-erasure-schema} is K-equivariant but not
K-independent: the reachability claim $p_l \preceq_K t.\textit{prov}$
depends on the chosen $\preceq_K$. Replacing the multilinear
$\mathbb{N}[X, T]$ semiring with the Boolean reduct $\mathbb{B}$
preserves the reachability bit (every reachable provenance maps to
$1$ in $\mathbb{B}$); the multi-degree variant $\mathbb{N}[X, T]^{\#}$
preserves the reachability bit and additionally records token-count
distinctions visible in the audit-row schema's per-witness recall
column. A hypothetical confidence-weighted $K[X, T]$ would add
the confidence-aggregation channel that is currently a
forward-compatibility hook per ADR-0017.
\end{enumerate}
Together, clauses (1) and (2) characterise the carrier $K$ as
\emph{decorative scaffolding on the iso axis} and \emph{load-bearing
structure on the schema axis}. The carrier ablation
(\S\ref{sec:meas-g4}) is the empirical witness for both clauses:
verdict columns identical across the three shipped carriers (clause
1 witness); token-recall columns differ across carriers (clause 2
witness). The verdict-invariance is therefore not evidence that $K$
is decorative everywhere; it is evidence that the iso-axis guard
map projects $K$ out by construction while the schema-axis
reachability bit remains parametric.
\end{proposition}

% Cross-axis remark for the reader. Clauses (1) and (2) together
% answer the C1 convergent reviewer concern: the K-parametric
% framing in the abstract and \S\ref{sec:found-thm} is honestly
% scoped to T-01b + the schema axis; the iso-axis guard map projects
% K out by construction (clause 1). A reviewer reading
% Theorem~\ref{thm:anomaly-soundness} alone might suspect K is
% decorative; Proposition~\ref{prop:k-axis-separation} makes the
% scope explicit so the same reviewer reading T-01 + T-01b + T-02
% in sequence sees the algebraic separation rather than an
% overclaimed parametric umbrella.

\paragraph{Relationship to Brinke et al.}
\citet{brinke-2026-preservation} show classical preservation theorems
lift to K-relation semantics on every lattice semiring, a class
containing the Boolean reduct $\mathbb{B}$ but not the natural-polynomial
carriers $\mathbb{N}[X]$, $\mathbb{N}[X]^{\#}$ we adopt. Their partition
refines clause~(2): on the lattice subclass the schema-axis argument
lifts to first-order entailment, while on the non-lattice subclass
Theorem~\ref{thm:carrier} supplies the lift through a direct
polynomial-degree argument. Clause~(1) is unaffected, the guard map
never inspecting the carrier.

\begin{definition}[Per-instance counterfactual erasure under a carrier]
\label{def:counterfactual-verdict}
Let $K$ be a commutative semiring with natural
order~\cite{green-karvounarakis-tannen-2007} and let $A$ be an audit
log in which a row $r$ may appear with multiplicity $k_r \ge 1$, so the
K-relation evaluation records its $k_r$-fold contribution to a witness
as the indeterminate power $X_r^{k_r}$. For a query $q$ with answer
$\pi := \mathsf{ans}(q, A) \in K$, the \emph{per-instance
counterfactual answer at $r$} is $\mathsf{ans}^{(-1)}(q, A; r) :=
\mathsf{ans}(q, A^{(-1)}_r)$, where $A^{(-1)}_r$ removes one occurrence
of $r$ ($k_r \mapsto k_r - 1$). The query is \emph{per-instance
counterfactually recoverable at $r$} when this answer is reconstructible
from $\pi$ alone, with no access to $A \setminus \{r\}$ and no
re-evaluation of $q$: the regime where the log survives only as its
polynomial witness. The companion scenarios, full-row erasure and
reconstruction with the remaining log in hand, both reduce to
re-evaluation and are recoverable on every carrier; the Block H
ablation of \S\ref{sec:meas-g4} measures the per-instance regime
exactly.
\end{definition}

\begin{definition}[Carrier degree at a monomial]
\label{def:carrier-degree}
For $\pi \in K$ in monomial normal form, $\deg_{m_r}(\pi)$ is the
largest exponent of $X_r := \mathrm{prov}(r)$ in any monomial of $\pi$
that mentions it. The Boolean reduct $\mathbb{B}$ collapses every
nonzero exponent to one and the multilinear carrier $\mathbb{N}[X]$
forbids exponent two by construction, so both cap the degree at one;
only the multidegree carrier $\mathbb{N}[X]^{\#}$ preserves every
positive exponent~\cite{green-karvounarakis-tannen-2007}.
\end{definition}

\begin{theorem}[Carrier recoverability under per-instance erasure]
\label{thm:carrier}
Let $A$ contain $k_r \ge 1$ occurrences of $r$ and let
$\pi = \mathsf{ans}(q, A)$. Then $q$ is per-instance counterfactually
recoverable at $r$ (Definition~\ref{def:counterfactual-verdict}) if and
only if $\deg_{m_r}(\pi) \ge 2$. Recovery is the formal shift
$\sigma_{X_r} : X_r^{i} \mapsto X_r^{i-1}$ that decrements one
occurrence per power, well-defined from $\pi$ alone; the
coefficient-injecting derivative $\partial_{X_r}$ returns the wrong
answer and admits no semiring correction, so the shift is the recovery
map. The condition orders the three shipped carriers
$\texttt{Boolean} \equiv \texttt{multilinear} \subsetneq
\texttt{multidegree}$ for multiplicity recovery, the transpose of the
witness-token recall reported by the carrier ablation
(\S\ref{sec:meas-g4}). The iso-axis guard is carrier-vacuous
(Proposition~\ref{prop:k-axis-separation}), so this theorem isolates
the carrier's load-bearing role to this surface. The full
statement, the monomial-normal-form proof, and the per-carrier recovery
grid are Appendix~\ref{app:graded-carrier}.
\end{theorem}

\ifshowproofs
\begin{remark}[Relationship to Brinke et al.\ on preservation in semiring semantics]
\label{rem:t04-brinke}
\citet{brinke-2026-preservation} classify which K-semirings admit
classical preservation theorems in K-relation semantics: all lattice
semirings (including the Boolean reduct $\mathbb{B}$) preserve
{\L}o{\'s}-Tarski and homomorphism preservation, while the tropical,
Viterbi, {\L}ukasiewicz, and natural semirings fail existential
preservation in general. Two of the three carriers we ship,
multilinear $\mathbb{N}[X]$ and multidegree $\mathbb{N}[X]^{\#}$,
inherit non-idempotent integer addition from the natural semiring
and therefore sit outside the lattice subclass; the Boolean reduct
$\mathbb{B}$ is the only lattice semiring in our shipped set.
Theorem~\ref{thm:carrier}'s iff statement is independent of
existential preservation: the proof proceeds directly from the
K-relation homomorphism property of
Green--Karvounarakis--Tannen~\cite[Theorem~3.2]{green-karvounarakis-tannen-2007}
and the monomial-normal-form decomposition under per-instance erasure,
so the degree-based recoverability characterisation holds across all
three shipped carriers. A first-order reformulation of per-instance
counterfactual recoverability over the two non-lattice carriers would
require a separate lift through either Brinke et al.'s positive
finite-interpretation result or a direct polynomial-degree argument.
\end{remark}
\fi

Composition is where the typing earns its keep, since per-operator
soundness says nothing about a pipeline. \toki's third guarantee lifts
the per-operator contracts to a sequential composite at the lattice
supremum of the per-step preconditions, so a chain of operators
carries one declarable isolation contract and one schema mode
(the supporting well-typedness lemma is Lemma~\ref{lem:compose}). The composition grid of
\S\ref{sec:meas-g5} confirms the lifted boundary across pipelines of
length up to five.

\begin{theorem}[Composition soundness for typed-operator pipelines]
\label{thm:composition}
Let $\langle \oplus_{a_1}, \oplus_{a_2}, \ldots, \oplus_{a_n} \rangle$
be a sequential composition of typed contradiction-resolution
operators
$\oplus_{a_i} \in \{\opLWW, \opEvi, \opAwait, \opRule\}$, each
carrying an iso precondition
$L_{a_i} \in \mathcal{L} = \mathcal{L}_{\mathrm{iso}} \times
\mathcal{L}_{\mathrm{cb}}$ and a schema-axis requirement
$M_{a_i} \in \mathcal{L}_{\mathrm{schema}} = \{\textsf{base} \preceq
\textsf{audit{-}row}\}$. Let
$\Pi := \oplus_{a_n} \circ \cdots \circ \oplus_{a_1}$ denote the
sequential pipeline that emits, on contradicting input pair
$(\fact{w}, \fact{l})$ at version time $t_s$, the chain
$\langle (\fact{w}^{(1)}, \auditt^{(1)}), \ldots,
(\fact{w}^{(n)}, \auditt^{(n)}) \rangle$ where each
$(\fact{w}^{(i)}, \auditt^{(i)})$ is the emission of $\oplus_{a_i}$
on the output of $\oplus_{a_{i-1}}$. Define the composite
preconditions
\[
  L^{*} \;:=\; \bigvee_{i=1}^{n} L_{a_i}, \qquad
  M^{*} \;:=\; \bigvee_{i=1}^{n} M_{a_i},
\]
where the iso join $\bigvee$ is component-wise on the iso and
callback axes of $\mathcal{L}$ and the schema join is the maximum
on $\mathcal{L}_{\mathrm{schema}}$. Then on every schedule $S$
that satisfies $L^{*}$ and runs the pipeline at schema mode $M^{*}$
on the dual-row schema of \S\ref{sec:schema}:
\begin{enumerate}
\item \emph{(Iso-axis composition.)} For every classical anomaly
$\phi \in \Phi_{\mathrm{class}}$ and every LLM-specific anomaly
$\phi \in \{\anomHR, \anomBDS\}$ with
$L^{*} \succeq \guardiso(\phi)$, the pipeline $\Pi$ does not witness
$\phi$ on $S$: the per-operator soundness of
Theorem~\ref{thm:anomaly-soundness} lifts to the composite.

\item \emph{(Schema-axis composition.)} For every loser fact $f_l$
consumed by any step, an audit tuple $t \in \mathsf{Rec}_H(t_v, t_s)$
in the emission chain satisfies $\mathrm{prov}(f_l) \preceq_K
t.\textit{prov}$, so the pipeline excludes $\anomAE$ at schema mode
$M^{*}$: a single step requiring $\textsf{audit{-}row}$ forces the
composite to $\textsf{audit{-}row}$, lifting the necessity of
Theorem~\ref{thm:audit-erasure-schema}.

\item \emph{(Allen-relation closure.)} Allen-relation selection on
any pipeline output preserves the bitemporal-tuple type by closure
of Allen's thirteen interval relations and the twelve-relation
transitivity table that omits
equality~\cite{allen-1983-cacm}.

\item \emph{(Table-scoped policy pins.)} The table-scoped
$\IsoSRpol$ annotation carried by $\opRule$ stays outside the
lattice $\mathcal{L}$ as a per-table side condition: when any
$\oplus_{a_i}$ pins the named policy table to $\IsoSR$ in the
composite, the remaining tables stay at $L^{*}$ as defined above.
\end{enumerate}
\end{theorem}

\begin{corollary}[Pipeline contract surface]
\label{cor:pipeline-contract}
The four typed operators form a contract surface for production
pipelines: any sequential composition is well-typed at the lattice
supremum $L^{*}$ of its per-step iso preconditions and the lattice
maximum $M^{*}$ of its per-step schema requirements. A deployment
claiming the composite contract declares only those two scalars;
Theorems~\ref{thm:anomaly-soundness} and~\ref{thm:audit-erasure-schema}
discharge the per-step exclusions and the composition theorem lifts
them, reducing the cross-operator anomaly surface that ingest, dedupe,
evidence-weight, and policy-check silently expose to a single
declarable signature.
\end{corollary}

\begin{remark}[Sequential scope]
\label{rem:t05-sequential-only}
Theorem~\ref{thm:composition} covers sequential composition: the steps
run in a fixed order on the contradicting pair under one wrapping
transaction. The dispatcher contract (\S\ref{sec:algebra}) admits a
single pair to one operator gate per transaction, so concurrent
composition on the same key does not arise, and concurrent composition
on disjoint pairs reduces to Theorem~\ref{thm:anomaly-soundness} per
transaction.
\end{remark}

The two argmax-fold operators also extend past the binary
incumbent-versus-incoming form: a conflict set of $n$ pairwise-contradicting
rows on one partition resolves under last-writer-wins or evidence-weighted
merge with the winner and the merged provenance independent of fold order,
and the audit row dominating every loser under the K-semiring order. A
confluence grid over $n \in \{2, \ldots, 8\}$ confirms both properties
(Appendix~\ref{app:nary-confluence-grid}).

\begin{proposition}[N-ary conflict-set resolution]
\label{prop:nary-conflict-set}
Let $C = \{f_1, \ldots, f_n\}$ with $n \ge 2$ be a conflict set of
pairwise-contradicting facts over one $(\subj, \pred)$ partition, and let
$\oplus_a \in \{\opLWW, \opEvi\}$ be a fold operator whose winner selector
is an argmax over a total preference order
(version time for $\opLWW$, the confidence-then-version-then-identity key
for $\opEvi$). Then $\oplus_a$ extends to an $n$-ary resolution
$\oplus_a(C) = (f_{i^{\star}}, \auditt)$ with two guarantees:
\begin{enumerate}
\item \emph{(Confluence.)} The winner $f_{i^{\star}}$ and the merged
provenance
$\bigoplus_{j} p_{f_j}$ in $\auditt$ are invariant under every
permutation of $C$: resolving $C$ in any order returns the identical
winner identity and the identical merged provenance.
\item \emph{(Provenance-completeness.)} The audit tuple's merged
provenance dominates every member under the K-semiring natural order:
$p_{f_j} \preceq_K \auditt.\textit{prov}$ for every $f_j \in C$, including
the winner.
\end{enumerate}
\end{proposition}

\noindent\emph{Proof sketch.}
Confluence follows because an argmax over a total order is independent of
enumeration order and the provenance sum $\provadd$ is commutative and
associative under any commutative semiring, so the fold over $C$ is
permutation-stable. Provenance-completeness follows from
$p_{f_j} \preceq_K \bigoplus_{k} p_{f_k}$ under the natural order
(existence witness the sum of the remaining
monomials~\cite{green-karvounarakis-tannen-2007}), lifting the binary
domination of Theorem~\ref{thm:audit-erasure-schema} to the whole family.
The confluence grid for $n \in \{2, \ldots, 8\}$ sits in Appendix~\ref{app:nary-confluence-grid}.

\begin{remark}[Selection-fold wiring]
\label{rem:nary-selection-scope}
The $\opAwait$ and $\opRule$ operators resolve a conflict set by direct
selection: an external callback or a policy row over the family elects
one member. The judge-logged dispatcher wires this $n$-ary selection at
the ingest seam. The oracle returns an index into the canonically
ordered conflict set, the durable judge log records the elected member's
stable identity under an order-independent set key, and a crash replay
re-elects the same member by that identity. The two fold operators
($\opLWW$, $\opEvi$) and the two selection operators together cover the
conflict-set algebra at every arity.
\end{remark}

The keyed-log discipline is also necessary: any system whose operator
surface does not enforce it admits $\anomHR$ under bounded oracle
nondeterminism, so \toki's defence is a tight characterisation within
the relational schedule model. This lower bound is the companion of
the verdict matrix, where every baseline omitting the discipline
admits $\anomHR$; \S\ref{sec:measurement} confirms the bound
empirically, with the measured admit rate landing on the $2p(1 - p)$
closed form across $30$ calibrated cells
(Figure~\ref{fig:oracle-calibration}). Witness schedules for every
axis sit in Appendix~\ref{app:witness-schedules}.

\begin{definition}[Bounded oracle nondeterminism]
\label{def:bounded-oracle-nondeterminism}
An external judge oracle $J$ is \emph{boundedly nondeterministic} when
it factors as $J(R, \theta, \omega) \in \{0, 1\}$ over an internal
state space $\Omega$ with $|\Omega| \ge 2$, and some operator input
$(R^{\star}, \theta^{\star})$ admits two states
$\omega_1, \omega_2 \in \Omega$ with
$J(R^{\star}, \theta^{\star}, \omega_1) \neq
J(R^{\star}, \theta^{\star}, \omega_2)$. This is the standard
production model for an LLM judge: $\omega$ ranges over decoder seeds at
non-zero temperature, sampling state, hardware numerical
nondeterminism, or API rerolls. The degenerate $|\Omega| = 1$ case the
bound excludes by hypothesis.
\end{definition}

\begin{definition}[H1-compliant system]
\label{def:h1-compliant-system}
A contradiction-resolution system $\mathcal{S}$ over the alphabet
$\Sigma_+ = \Sigma \cup \{j_i(R, \theta), \mathit{cb}_i(R, h, k)\}$ is
\emph{H1-compliant} when its operator surface enforces invariant H1 of
Lemma~\ref{lem:judge-callback-bridge}: every $j_i(R, \theta)$ is paired
with a commit-preceding row at key $(R, \theta)$ in \code{judge\_log}
recording the witnessed vote, so every later invocation at the same key
reads the log instead of issuing a fresh $J$ call. A system whose
operator surface omits the hook, or which calls $J$ directly from
free-form agent code, is \emph{H1-non-compliant}.
\end{definition}

\begin{theorem}[N1 lower bound: keyed-log discipline is necessary]
\label{thm:n1-lower-bound}
Let $\mathcal{S}$ be a contradiction-resolution system over the typed-operator
alphabet $\Sigma_+$ of \S\ref{sec:found-schedule}, executed against a
boundedly nondeterministic judge oracle $J$
(Definition~\ref{def:bounded-oracle-nondeterminism}). If $\mathcal{S}$ is
H1-non-compliant (Definition~\ref{def:h1-compliant-system}), then there
exists a schedule history $H$ executable by $\mathcal{S}$ such that
$H \models \anomHR$ on some operator input $(R^{\star}, \theta^{\star})$.
Equivalently, on every operator input with at least two
nondeterministic-state branches $\omega_1, \omega_2 \in \Omega$ realising
distinct votes, $\mathcal{S}$ admits an adversarial replay pair witnessing
$\anomHR$ at $(R^{\star}, \theta^{\star})$.
\end{theorem}

\noindent\emph{Proof sketch.}
Definition~\ref{def:bounded-oracle-nondeterminism} fixes
$\omega_1 \ne \omega_2$ and an operator input
$(R^{\star}, \theta^{\star})$ on which $J$ returns different
votes. Two replays of $H = H_0 \cdot j_i(R^{\star},
\theta^{\star})$, one at each oracle state, share the
committed-write prefix $H_0$ and the decoder parameter
$\theta^{\star}$ but witness different votes; H1 non-compliance
prevents the dispatcher from short-circuiting on a keyed log
row, so both replays are executable. The pair is an $\anomHR$
witness. Appendix~\ref{app:graded-n1-lower-bound} carries the full witness
construction.

\begin{corollary}[Keyed-log discipline is a tight characterisation of $\anomHR$ soundness]
\label{cor:n1-tightness}
Pairing Corollary~\ref{cor:n1-corollary} (H1-compliant systems exclude
$\anomHR$) with Theorem~\ref{thm:n1-lower-bound} (H1-non-compliant
systems admit it), within the relational schedule model where a
system's only lever over judge-vote stability is the isolation level on
the keyed read,
\[
  \mathcal{S} \text{ excludes } \anomHR
  \iff
  \mathcal{S} \text{ is H1-compliant at } \IsoSR
  \text{ on } \mathsf{judge\_log}.
\]
Tightness holds within this model: a unique-key insert-if-absent
constraint, a linearizable key-value store, or a content-addressed
judge cache each realises the keyed-read contract by other means and
excludes $\anomHR$; the lower bound binds any system omitting all three.
\end{corollary}

\begin{remark}[Scope]
\label{rem:t06-scope}
The lower bound binds systems running against a boundedly
nondeterministic oracle; it constrains neither deterministic oracles
($|\Omega| = 1$) nor systems that precompute verdicts offline, and
other drift sources (model upgrades, prompt or retrieval drift) are
bounded by the judge-prompt sensitivity Lemma rather than
by the bound here.
\end{remark}

The bound compounds along a re-query trajectory: an H1-non-compliant
baseline re-invokes the oracle on every re-query, so its
replay-consistency rate decays toward zero over re-queries on the
bounded-nondeterminism keys, while \toki reads the committed keyed-log
verdict and stays at one, the trajectory \S\ref{sec:measurement}
charts (Figure~\ref{fig:trajectory-replay}).
         %   3.3 Soundness
% @owner       src::sections::04-implementation
% @does        The TOKI System: maps the algebra onto an eleven-column user-visible schema sealed by a check-constrained audit-row discriminator, four pure operator classes, an exact K-semiring carrier, and a hand-rolled as-of predicate, all over an unmodified relational engine; the dual-row contradiction-resolution write path is stated as Algorithm~\ref{alg:resolve-write}.
% @needs       src/sections/03-foundations.tex, macros.tex (algorithm + algpseudocode environments)
% @feeds       src/main.tex, src/sections/05-measurement.tex
% @breaks      claiming native bitemporal support the engine does not ship; suppressing the hidden audit-row discriminator the audit-row theorem requires; leaking class or method identifiers that belong to the reference codebase rather than to the public contract surface; an algorithm step that names an internal class path instead of the public operation.

\subsection{The \toki System}
\label{sec:impl}

\toki runs the \S\ref{sec:foundations} algebra on an unmodified
relational engine with no native bitemporal support, in roughly
$2{,}700$ lines of Python; the abstraction carries the
contribution. Every inference rule lands as one of four executable
pieces: an eleven-column user-visible schema sealed by one twelfth
audit-row discriminator; four pure operator classes dispatched on a
strategy stamp; a K-semiring polynomial serialized in line with exact
dominance; and one SQL:2011 as-of predicate that serves both
retrieval and audit replay. \toki recruits no storage feature beyond
standard SQL.

A check-constrained discriminator lifts the
Theorem~\ref{thm:audit-erasure-schema} dual-row signature into the
relational layout; it is structural, written only by the audit-row
emission path. The four operators are pure functions over fact pairs,
and a single write path owns the I/O seam. Dispatch routes on the
resolution-strategy column, and an unknown strategy raises a typed
error with no silent fallback. \toki enforces each operator's
isolation pin one level above the operator: before the operator
commits, the write path records the adjudicated vote in the keyed
judge log, discharging the ordering hypothesis of
Lemma~\ref{lem:judge-callback-bridge} by construction.

Algorithm~\ref{alg:resolve-write} states the write path that ties these
pieces together: the dispatcher detects a contradicting incumbent on the
partition, routes to the typed operator, sequences the keyed judge log
before the operator commit, and emits the dual-row pair whose audit half
discharges Theorem~\ref{thm:audit-erasure-schema}. The single binary
incumbent precondition is enforced by construction (a multi-incumbent
partition raises rather than guessing), so the incremental write path
stays binary, and the soundness of Theorem~\ref{thm:anomaly-soundness}
carries to it verbatim.

\begin{algorithm}[t]
\caption{Dual-row contradiction-resolution write path on partition
  $(\subj, \pred)$. Operators are pure; this is the persistence policy
  that pairs each operator with loser invalidation and audit-row
  emission. The judge-log write of line~\ref{alg:line:judgelog} precedes
  the operator commit, discharging hypothesis~H1 of
  Lemma~\ref{lem:judge-callback-bridge}.}
\label{alg:resolve-write}
\begin{algorithmic}[1]
  \Require fact $f$ with strategy stamp $\strat(f)$, write time $t_s$,
    judge parameter pin $\theta$; oracle $\rho$ for $\opAwait / \opRule$
  \Ensure committed current row and, on contradiction, an audit row
    dominating the loser under $\preceq_K$
  \State $O \gets$ open current rows on $(\subj, \pred)$ with
    $\Tsystem$ end open \Comment{$\rowkind=\current$}
  \If{$\exists\, r \in O$ with $r.\obj = f.\obj$ and
    $\overlap(r, f)$}
    \State \Return \Comment{duplicate confirmation; no new state}
  \EndIf
  \State $C \gets \{\, r \in O : r \conflict f \,\}$
    \Comment{Allen overlap, differing object}
  \If{$C = \emptyset$}
    \State \textbf{insert} $f$ as a fresh current row;
      \Return
  \ElsIf{$|C| > 1$}
    \State \textbf{raise} \Comment{binary operators require one incumbent}
  \EndIf
  \State $r \gets$ the single incumbent in $C$
  \State select operator $\oplus$ by $\strat(f)$:
    $\opLWW{\mid}\IsoRC$, $\opEvi{\mid}\IsoSI$,
    $\opAwait{\mid}\IsoRCcb$, $\opRule{\mid}\IsoSRpol$
  \If{$\oplus \in \{\opAwait, \opRule\}$}
    \State $v \gets \rho(r, f)$;\;
      \textbf{append} $v$ to $\mathrm{judge\_log}[(R, \theta)]$
      \label{alg:line:judgelog}
      \Comment{before the operator commit (H1)}
  \EndIf
  \State $(w, l) \gets \oplus(r, f)$
    \Comment{winner / loser by the operator's order}
  \State $p_{\mathrm{merge}} \gets p_w \provadd p_l$
    \Comment{exact K-semiring merge}
  \If{$l = r$} \Comment{incumbent lost}
    \State \textbf{close} $r$ at $\Tsystem$ end $t_s$;\;
      \textbf{insert} $w$ stamped with $t_s$ and $p_{\mathrm{merge}}$
  \EndIf
  \State \textbf{insert} audit row
    $(\,p_{\mathrm{merge}},\, \strat(f),\, t_s,\, \text{witness}\,)$
    \Comment{$\rowkind=\auditt$; defends $\anomAE$}
  \State \textbf{commit} at the operator's isolation pin
  \State \Return $(w, \text{audit})$
\end{algorithmic}
\end{algorithm}

The witness polynomial of \S\ref{sec:found-prov} serializes as a
text-encoded provenance column. Provenance merge realises
$p_w \provadd p_l$ over the multilinear instance
$\mathbb{N}[X, T]$, and a dominance test decides
$p_{\mathrm{old}} \preceq_K p_{\mathrm{merge}}$ as an exact polynomial
comparison. A multidegree multiset variant is also available,
so re-verifying Theorem~\ref{thm:anomaly-soundness} at the
strongest instantiation swaps the carrier without a proof rewrite.
Retrieval hand-rolls the SQL:2011 as-of predicate the engine omits:
a filter on the four timestamp columns plus the audit-row
discriminator serves retrieval at the current value, audit replay at
the audit value, and a third mode exposes both, closed-open under the
\S\ref{sec:found-schedule} convention.

Three details make the keyed-log discipline operational. The reference
persists the keyed judge log as an append-only table, and a crash-replay
test confirms a committed verdict replays consistently after a connection
drop and reload, realising the keyed-log discipline the
Theorem~\ref{thm:n1-lower-bound} lower bound requires. The audit row
encodes its conflict witness as JSON, so arbitrary object strings
round-trip through the single object column exactly. A repeated same-fact
confirmation accumulates the duplicate's provenance into the surviving row
under the K-semiring sum, keeping the audit trail complete across
re-ingestion.
 %   3.4 The TOKI System
% @owner       src::sections::05-measurement
% @does        Validates the typed-operator algebra against eight systems organised by row class through five experiment clusters (G1 verdict matrix, G2 mechanism stress, G3 cost envelope, G4 carrier ablation, G5 theorem anchors).
% @needs       src/sections/03-foundations.tex, src/figures/fig-g2-forest.pdf, src/figures/fig-systems-perf-scaling.pdf, src/figures/fig-g5-anchors.pdf, src/figures/fig-oracle-variance.tex, src/figures/fig-trajectory-replay.tex, src/tables/tab-anomaly-bench.tex, src/tables/tab-ablation.tex, src/tables/tab-cross-system-utility.tex
% @feeds       src/main.tex
% @breaks      conflating eleven adapter rows with eleven systems; reading the matrix headline as a utility claim; quoting any cross-system delta outside the statistically-indistinguishable-from-zero envelope at alpha = 0.05.

\section{Empirical Validation}
\label{sec:measurement}

Five experiment clusters validate the algebra, and across all of
them every agent-memory baseline admits at least one named anomaly
while \toki excludes all three. A verdict matrix establishes the
write-time correctness witness over eight systems
(\S\ref{sec:meas-g1}); three controlled clusters isolate each
defence's mechanism (\S\ref{sec:meas-g2}), single-process cost
(\S\ref{sec:meas-g3}), and provenance carrier (\S\ref{sec:meas-g4});
and a final set of grids anchors every soundness theorem to a measured
$0$/$1$ boundary (\S\ref{sec:meas-g5}). A closing cross-system
comparison bounds the scope: the contract makes no
downstream-utility claim.

Effect sizes report paired bootstrap confidence intervals at
$1{,}000$ resamples and a fixed seed; family-wise error rates
use Holm step-down at $\alpha = 0.05$. Cross-system equivalence
reports Welch's $t$-statistic with Welch-Satterthwaite degrees
of freedom; cost fits report Spearman $\rho$ and least-squares
$R^{2}$ on log-transformed axes. Appendix~\ref{app:protocols-stats} carries the full methodology.

\subsection{Every agent-memory baseline admits at least one anomaly}
\label{sec:meas-g1}

We evaluate eight systems: six agent-memory baselines
(\sysmem~v2~\cite{chhikara-2025-mem0},
\sysmem~v3~\cite{mem0ai-readme-2026},
\sysgraphiti~\cite{getzep-graphiti-2026},
\sysletta~\cite{letta-context-constitution-2026},
\syszep~\cite{rasmussen-2025-zep},
\sysmirix~\cite{lu-2026-mma}), one engine-layer comparator
(\sysworlddb~\cite{ganesan-2026-worlddb}), and the reference
algebra \sysours. Each system is probed on its own contradiction
path against the three write-time anomalies; the verdict matrix
(Table~\ref{tab:anomaly-bench}) reads each verdict twice, from the
system's design and from its running code. The per-system harness is
detailed in the appendix.

% tab-anomaly-bench.tex --- §5 G1 claim/wire matrix.
% Eight AnomalyClaim baselines x eleven AnomalyWire adapter rows x three
% schedule predicates. The eighth claim baseline (MIRIX) is transcribed
% from the independent MMA-Bench evaluation and carries no wire adapter. Status entries are read off
% results/anomaly_bench/n[1,2,3]_*.csv and
% results/anomaly_wire/n[1,2,3].csv.
%
% Status conventions:
%   admit   = the system admits the anomaly on the minimal trigger;
%   excl.   = the system excludes the anomaly;
%   n/a     = the predicate is absent or this imported runtime row abstains.
%
% Owner   : src/tables/tab-anomaly-bench.tex
% Claim   : Covered G1 cells separate design-path verdicts,
%           transcribed wire evidence, imported wire evidence, and
%           abstentions by denominator.
% Does    : 12 verdict rows x 3 anomalies with explicit Claim/Wire columns.
% Needs   : macros.tex; figstyle.tex (\pvldbtablestyle, \tablenote);
%           booktabs.
% Feeds   : src/sections/05-measurement.tex.
% Breaks  : Drift between these cells and
%           results/anomaly_bench/n[1,2,3]_*.csv = a paper-impl SSOT bug.

\begin{table}[tbp]
  \caption{\textbf{Verdict matrix over eight systems.} For each
    anomaly, \emph{Claim} is the verdict the system's design implies
    and \emph{Wire} the verdict observed from its running code.
    A = admit, X = exclude, $-$ = not applicable.}
  \label{tab:anomaly-bench}
  \centering
  \pvldbtablestyle
  \footnotesize
  \begin{tabular*}{\columnwidth}{@{\extracolsep{\fill}}l*{6}{c}@{}}
    \toprule
    & \multicolumn{2}{c}{$\anomHR$} & \multicolumn{2}{c}{$\anomBDS$} & \multicolumn{2}{c}{$\anomAE$} \\
    \cmidrule(lr){2-3}\cmidrule(lr){4-5}\cmidrule(l){6-7}
    System            & Claim           & Wire  & Claim           & Wire  & Claim           & Wire \\
    \midrule
    \multicolumn{7}{@{}l}{\textit{Agent-memory baselines (6 systems)}} \\
    \sysmem~v2-T      & $-$           & $-$   & $-$           & $-$   & A           & A \\
    \sysmem~v3-T      & A           & A & A           & A & A           & A \\
    \sysmem~v3-I      & A           & $-$   & A           & X & A           & A \\
    \sysgraphiti-T    & A           & A & $-$           & $-$   & A           & A \\
    \sysgraphiti-I    & A           & $-$   & $-$           & $-$   & A           & A \\
    \sysletta-T       & $-$           & $-$   & A           & A & A           & A \\
    \sysletta-I       & $-$           & $-$   & A           & $-$   & A           & A \\
    \syszep-T         & $-$           & $-$   & A           & A & A           & A \\
    \syszep-I         & $-$           & $-$   & A           & A & A           & A \\
    \sysmirix-T       & A           & A & A           & A & $-$           & $-$   \\
    \midrule
    \multicolumn{7}{@{}l}{\textit{Engine-layer comparator}} \\
    \sysworlddb-T     & X           & X & X           & X   & X           & X   \\
    \midrule
    \multicolumn{7}{@{}l}{\textit{Reference algebra}} \\
    \sysours-I        & X$^{\star}$ & X & X$^{\star}$ & X & X$^{\star}$ & X \\
    \bottomrule
  \end{tabular*}
  \tablenote{Each system gives two independent readings: T transcribes
    its published contradiction logic, I runs its shipped code; six
    cells abstain on structural grounds
    (Appendix~\ref{app:abstaining-imported-cells}).
    $^{\star}$\,replay inconsistency via $\IsoSR$ on the keyed judge
    log, belief-drift skew via $\IsoSR$ on the $(\subj,\pred)$
    partition, audit erasure via the audit-row schema.
    \sysmirix\ carries no in-repo wire adapter; its row transcribes the
    independent \textsc{MMA-Bench} evaluation~\cite{lu-2026-mma}, which
    probes replay inconsistency and belief-drift skew but not the audit
    path, so its audit-erasure cell abstains.}
\end{table}

Every agent-memory baseline admits at least one anomaly. Audit
erasure is universal among baselines with a documented provenance
path: any system without a versioned audit row admits it, which
covers five of the six baselines, while \sysmirix\ abstains here
because~\cite{lu-2026-mma} measures its verdict accuracy rather
than its audit path. Replay inconsistency appears on \sysmem~v3,
\sysgraphiti, and \sysmirix, whose judge-dependent writes commit
without a serializable log. Belief-drift skew appears on
\sysmem~v3, \sysletta, \syszep, and \sysmirix\ at a partitioned
snapshot isolation. No baseline closes all three.
\sysworlddb\ excludes all three by keeping the judge off the write
path: its content-addressed handlers run deterministically at query time and carry
no isolation-level signature, so a deployment needing a write-path
judge for evidence-weighted resolution gains no soundness guarantee.
\sysours\ is the only design that excludes all three while keeping the
judge on the write path, pinning the decoder tuple, the output hash,
and the partition under serializability. External
evidence corroborates the gap:
\textsc{STALE}~\cite{chao-2026-stale}, a 400-scenario
expert-validated benchmark, scores the best frontier model at
$55.2\%$ on a workload whose implicit-conflict mode is a
production replay inconsistency. One wire reading diverges from its
design: \sysmem~v3 should admit belief-drift skew by construction, but
its shipped reconciliation collapses competing-confidence rows under
deterministic judging and excludes it (appendix).

\subsection{The audit-row defence moves LoCoMo by $+0.86$}
\label{sec:meas-g2}

We measure how each defence moves its intended benchmark slice
and how the same perturbation behaves on off-target controls,
through paired slice estimands on
LoCoMo~\cite{maharana-2024-locomo},
LongMemEval-S~\cite{wu-2025-longmemeval}, and
MultiTQ~\cite{chen-2023-multitq}. Each benchmark runs one
single-defence ablation: audit-row against $\anomAE$, judge-pin
against $\anomHR$, and partition-$\IsoSR$ against $\anomBDS$.

% fig-memory-utility.tex --- §5 G2 mechanism-stress forest plot.
%
% Owner   : src/figures/fig-memory-utility.tex
% Does    : Plotnine forest plot of nine G2 cells (3 defences x 3
%           datasets) with paired-bootstrap CIs and Holm decisions.
%           Replaces the earlier 3x3 heatmap so the +0.86 audit-row
%           x LoCoMo headline cell and the construct-saturated cells
%           are both visible on the same axis. File name preserved
%           for backward git/aux compatibility.
% Needs   : results/g2_utility/summary.csv from
%           experiments/g2_utility/runner.py; figures/
%           fig-g2-forest.pdf rendered by
%           scripts/plot_g2_forest.py.
% Feeds   : src/sections/05-measurement.tex \S\ref{sec:meas-g2}.
% Breaks  : Axiom 0: hand-edited numbers in the caption drift from
%           the rendered figure.

\begin{figure}[tb]
  \centering
  \includegraphics[width=0.9\columnwidth]{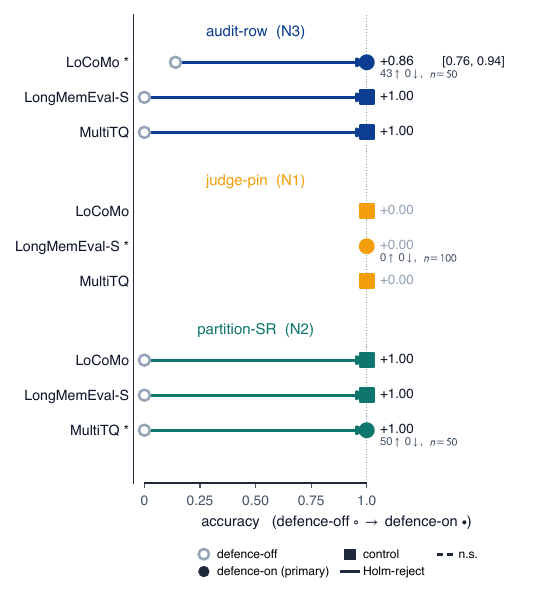}
  \Description{Dumbbell plot of nine G2 cells across three defences
  (audit-row, judge-pin, partition-SR) and three datasets (LoCoMo,
  LongMemEval-S, MultiTQ). Each row connects the defence-off accuracy
  (open marker) to the defence-on accuracy (filled marker); the right
  margin gives the paired-bootstrap accuracy delta, the 95% confidence
  interval on the one non-degenerate cell, and the McNemar discordant
  pairs with n on the primary cells. Primary diagonal cells are marked
  by an asterisk.}
  \caption{\textbf{Each defence moves its primary slice while the
    constructed controls saturate.} Per cell, accuracy from
    defence-off ($\circ$) to defence-on ($\bullet$); solid arrows are
    Holm-significant, $\ast$ marks the primary diagonal cell, and the
    right margin gives the paired-bootstrap effect size.}
  \label{fig:memory-utility}
\end{figure}

Figure~\ref{fig:memory-utility} reports the $3 \times 3$ mechanism-stress
grid, each cell measuring how strongly a defence moves its intended
slice. The audit-row defence moves its primary LoCoMo slice by
$\Delta = {+}0.86$ (paired-bootstrap CI $[0.76, 0.94]$), the one cell
below the ceiling on a natural-workload slice; its two off-target
controls, both partition-$\IsoSR$ controls, and the
partition-$\IsoSR$ primary MultiTQ cell all saturate at $+1.00$ under
their constructed slices, the success criterion for constructed
controls. The judge-pin defence records no movement here, its
replay-disagreement estimand requiring a judge replay table absent on
this surface, so its evidence is structural (the
\S\ref{sec:meas-g5} Bernoulli grid). Holm step-down at $\alpha = 0.05$
covers the nine-cell family, and the single natural-workload point
stays below an end-to-end utility-superiority claim.

\subsection{Latency stays sub-linear on the single-process envelope}
\label{sec:meas-g3}

% fig-systems-perf-scaling.tex --- §5.4 G3 scaling axis.
%
% @status  Active in the main body: \input by src/sections/05-measurement.tex
%   at the G3 subsection (single-process latency-scaling envelope). The
%   six-panel fig-systems-perf-scaling-2x2.tex remains the supplement-side
%   detailed cost scan.
%
% @owner   src/figures/fig-systems-perf-scaling.tex
% @claim   Bitemporal-reference contradiction-resolution latency
%          stays in a tight envelope as the memory table grows from
%          0 to $10^{5}$ facts: $p_{50}$ holds near $4$~ms and
%          $p_{99}$ stays under one order of magnitude above
%          $p_{50}$ across the swept axis. Workload-axis scaling,
%          not a single-point microbenchmark.
% @does    Single-panel line plot of contradiction-resolution
%          $p_{50}$ and $p_{99}$ latency in milliseconds against
%          memory size (number of prefilled non-contradicting
%          facts) on a log x-axis; a shaded 95\% bootstrap mean CI
%          band sits behind the lines.
% @needs   results/g3_systems_perf/scaling.csv from
%          experiments.g3_systems_perf.scaling;
%          src/figures/fig-systems-perf-scaling-img.pdf rendered by
%          scripts/plot_systems_perf_scaling.py.
% @feeds   src/sections/05-measurement.tex \S\ref{sec:meas-g3}.
% @breaks  Drift between this figure's plotted values and the CSV is
%          a paper-impl SSOT bug; regenerate via
%          ``python scripts/plot_systems_perf_scaling.py'';
%          that script is the source of truth.

\begin{figure}[!htbp]
  \centering
  \includegraphics[width=0.92\columnwidth]
    {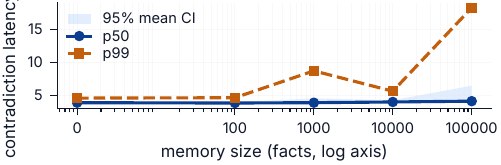}
  \Description{Single-panel line plot of bitemporal-reference
  contradiction-resolution latency in milliseconds against memory
  size on a log x-axis. The $p_{50}$ line stays near $4$ms across
  memory sizes $0$, $100$, $1{,}000$, $10{,}000$, and
  $100{,}000$ facts; the $p_{99}$ line tracks the $p_{50}$ line
  within roughly twofold; a shaded band carries the $95$-percent
  bootstrap confidence interval on the per-size mean.}
  \caption{\textbf{Contradiction-resolution latency stays flat as
    the store grows.}
    $p_{50}$ (solid) and $p_{99}$ (dashed) latency for the
    evidence-weighted contradiction path on \sysours, swept over
    memory size from $0$ to $10^{5}$ facts; $30$ runs per point
    after $3$ warmups.}
  \label{fig:systems-perf-scaling}
\end{figure}

\toki makes no scalability claim; we charge the single-process
envelope only to bound that scope.
Memory size across $\{0, 10^{2}, 10^{3}, 10^{4}, 10^{5}\}$ facts
holds $p_{50}$ in $3.88$ to $4.15\,\mathrm{ms}$ and $p_{99}$ under
$18.20\,\mathrm{ms}$ (Figure~\ref{fig:systems-perf-scaling}, Spearman
$\rho = 0.80$), and writer
concurrency fits $\mu \sim c^{0.86}$ at $R^{2} = 0.992$, a
sub-linear single-process lock-contention signature. A real
multi-writer Postgres experiment demonstrates the
operator-to-isolation mapping across the full lattice: read committed
admits all four iso-axis anomalies, snapshot isolation excludes lost
update and read skew, and only serializable also excludes write skew
($A5B$) and the phantom ($P_3$), aborting the losing writer at rate
$(w{-}1)/w$ (Appendix~\ref{app:multiwriter-concurrency}). Appendix~\ref{app:a6-g3-5axis-statistics} carries the five-axis
statistics and the transactional-backend disclosure in full.

\subsection{Verdicts are carrier-invariant; recall is K-load-bearing}
\label{sec:meas-g4}

\begin{table}[tb]
  \centering
  \caption{\textbf{Every matched cell defends at $1.00$; token recall
    separates provenance-retaining carriers from the Boolean reduct.}
    Carrier-by-defence ablation, $n=100$ seeds per matched cell.
    Each cell is Match/Recall. Appendix~A.4 reports off-target
    specificity and provenance-size statistics.}
  \label{tab:ablation}
  \pvldbtablestyle
  \begin{tabular}{@{}lccc@{}}
    \toprule
    Carrier & \anomAE\ audit-row & \anomHR\ judge-pin & \anomBDS\ $A5B$-on-partition \\
    \midrule
    $\mathbb{N}[X,T]$      & 1.00/1.00 & 1.00/1.00 & 1.00/1.00 \\
    $\mathbb{N}[X,T]^{\#}$ & 1.00/1.00 & 1.00/1.00 & 1.00/1.00 \\
    $\mathbb{B}$           & 1.00/0.00 & 1.00/0.00 & 1.00/0.00 \\
    \bottomrule
  \end{tabular}
\end{table}

Verdicts are carrier-invariant across multilinear, multidegree,
and Boolean reduct carriers (Table~\ref{tab:ablation}); token recall
splits the carriers
(the Boolean reduct collapses every witness to one bit; the
multilinear and multidegree carriers are K-tight on token recall,
$n = 100$ seeds, $95\%$ paired bootstrap CI; Appendix~\ref{app:a4-g4-ablation-details}). The iso-axis guards of
Theorem~\ref{thm:anomaly-soundness} are carrier-agnostic by
construction; only the schema-axis natural order of
Theorem~\ref{thm:audit-erasure-schema} is K-load-bearing
(Proposition~\ref{prop:k-axis-separation}), so verdicts factor through
a carrier-projecting quotient and recall through the natural order the
carrier parameterises.

\subsection{Every theorem's predicted $0$/$1$ boundary matches measurement}
\label{sec:meas-g5}

% @owner src/figures/fig-g5-anchors.tex
% @does  Wraps fig-g5-anchors-img.pdf (three-panel G5 evidence composite)
%        in a double-column figure for §5.5. The figure carries three
%        independent theorem anchors on one figure; the standalone
%        Bernoulli calibration scatter is now fig-oracle-variance in §6
%        and is not duplicated here.
% @claim Iso-lattice guard map (T-01), composition pipeline (T-05), and
%        cross-system equivalence (T-06 distributional) each land on a
%        structurally controlled grid whose empirical 0/1 boundary
%        matches the theorem.
% @needs results/iso_matrix/run_v1/iso_matrix_grid.csv,
%        results/t5_composition/run_length5/composition_grid.csv,
%        results/cross_system_n1/run_v1/cross_system_test.csv.
% @feeds src/sections/05-measurement.tex.
% @breaks Drift between any panel and its CSV columns.

\begin{figure*}[t]
  \centering
  \includegraphics[width=0.82\textwidth]{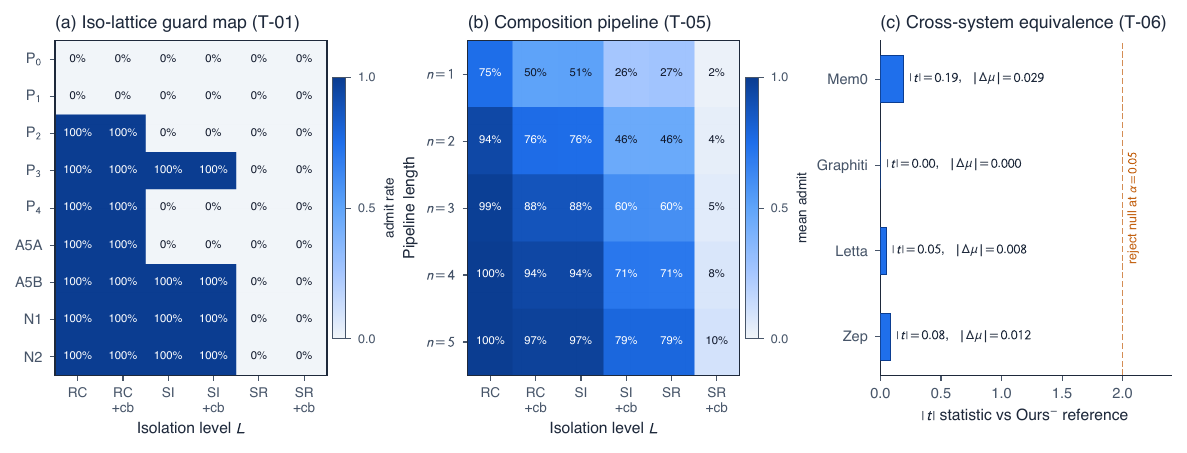}
  \caption{\textbf{Three structural anchors; every predicted $0/1$
    boundary matches measurement.}
    \textbf{(a)} Iso-lattice guard map (Theorem~\ref{thm:anomaly-soundness}):
    nine predicates admit at $0\%$ where the level dominates the guard,
    $100\%$ elsewhere.
    \textbf{(b)} Composition admit rate over pipeline lengths one to five
    (Theorem~\ref{thm:composition}), flooring only at the lattice supremum
    $\IsoSR{+}\mathsf{cb}$.
    \textbf{(c)} Welch's $t$ on four production-mimic variants, all
    $|t| < 2$, equivalence not rejected at $\alpha = 0.05$
    (Theorem~\ref{thm:n1-lower-bound} distributional companion).}
  \label{fig:g5-anchors}
  \Description{Three-panel composite figure. Panel (a) shows a
    9-anomaly by 6-iso-level heatmap. Panel (b) shows a
    pipeline-length by iso-level heatmap of mean admit rate.
    Panel (c) shows a forest plot of absolute t-statistics from
    Welch's t-tests on four pairwise comparisons.}
\end{figure*}

Five structural grids land the direct empirical anchors for
Section~\ref{sec:found-thm}, each an exhaustive schedule census whose
observed admit rate matches its typed $0$/$1$ prediction exactly; the
iso-lattice and composition grids appear in
Figure~\ref{fig:g5-anchors}, and Appendix~\ref{sec:appendix-g5-structural-grids}
(Table~\ref{tab:g5-grids}) carries every grid's dimensions, boundary,
and regeneration command. The iso-axis grid
(Theorem~\ref{thm:anomaly-soundness}), distinct from the verdict
matrix, pairs nine predicates with six joint lattice levels on 54
cells: 32 dominating cells admit $0$, 22 under cells $100\%$. The
schema-axis grid (Theorem~\ref{thm:audit-erasure-schema}) flips
$\anomAE$ from $100\%$ under the base schema to $0\%$ under the
audit-row schema, and the partition-pin grid
(Corollary~\ref{cor:n2-corollary}) holds $0$ under
partition-$\IsoSR$ against $\IsoRC$/$\IsoSI$ admits; the carrier sweep
(Theorem~\ref{thm:carrier}) and composition grid
(Theorem~\ref{thm:composition}) hold the boundary across $k \in \{2,
\ldots, 8\}$ and across 1364 pipelines up to length five. The alphabet
bridge (Lemma~\ref{lem:judge-callback-bridge}) preserves alphabet,
edge, and refinement on $1{,}000/1{,}000$ random $\Sigma_+$ schedules.

The lower bound (Theorem~\ref{thm:n1-lower-bound}) lifts to its
Bernoulli closed form, anchored by the cross-system equivalence panel
(Figure~\ref{fig:g5-anchors}c) and the calibration scatter of
Figure~\ref{fig:oracle-calibration}. Six
variants share one bounded-nondeterminism judge oracle (a frontier LLM
at temperature $1.0$ with a reasoning prefix): the H1-compliant
reference admits $0/245$ consecutive trial pairs across five seeds, the
five non-compliant variants $0.167$ to $0.204$ mean. The per-call
admit rate lands on the $2p(1 - p)$ closed form across 30 calibrated
cells at mean absolute deviation $0.017$ (maximum $0.111$, 29 cells
within the $0.10$ envelope) and $R^2 = 0.98$
(Figure~\ref{fig:oracle-calibration}), so the lower bound rests on
measured bounded-oracle data; Welch's $t$
reaches $|t| = 0.19$ over four pairwise tests, rejecting no null at
$\alpha = 0.10$, so the keyed-log discipline is the load-bearing axis.
A temperature sweep over $\{0.5, 1.0, 1.5\}$ holds the reference at
$0.000$ and the stripped variant at $0.32$ to $0.34$, consistent with
the bounded-oracle hypothesis once the reasoning prefix saturates
$\Omega$ (Definition~\ref{def:bounded-oracle-nondeterminism}). All
four imported mimics lack the keyed-log discipline, lifting ``admit rate
$> 0$'' from a structural observation to the quantitative theorem.

% @owner   src/figures/fig-oracle-variance.tex
% @claim   The measured N1 (anomHR) admit rate lands on the T-06 closed
%          form 2p(1-p) across systems and seeds, so the lower bound is
%          calibrated against real bounded-oracle data, not asserted.
% @does    Wrap scripts/plot_oracle_variance.py output
%          figures/fig-oracle-variance-img.pdf for inclusion in
%          \S\ref{sec:measurement} as the real-valued T-06 evidence,
%          beside Theorem~\ref{thm:n1-lower-bound}'s empirical anchor.
%          Calibration scatter of observed admit rate against the oracle
%          Bernoulli p, with the 2p(1-p) curve inside a +/-0.10
%          prediction band, points split by H1 compliance, an on-figure
%          agreement-statistics box, and a residual inset.
% @needs   src/figures/fig-oracle-variance-img.pdf (rendered by
%          scripts/plot_oracle_variance.py from
%          results/oracle_variance/run_v1/calibration.csv); src/macros.tex.
% @feeds   src/sections/05-measurement.tex \S\ref{sec:measurement}.
% @breaks  reading scatter spread as model error rather than finite-seed
%          Bernoulli variance; the 0.017 MAD is the fit, not a residual.

\begin{figure}[t]
  \centering
  \includegraphics[width=0.9\columnwidth]{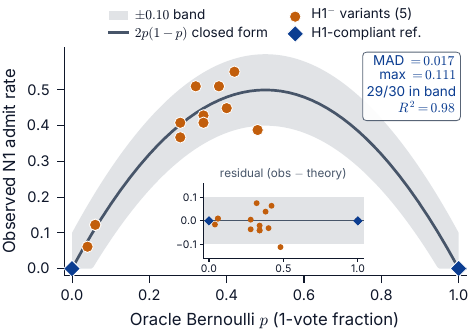}
  \caption{\textbf{The measured replay-inconsistency rate lands on the
    $2p(1-p)$ closed form.} Each point is one calibrated cell. The
    keyed-log reference (navy) pins at the zero-admit corners
    $p\in\{0,1\}$; the five stochastic-judge variants (vermillion),
    four baselines plus \toki's log-stripped ablation, ride the
    predicted arc (Theorem~\ref{thm:n1-lower-bound}) inside its
    prediction band, with the residual inset trend-free.}
  \label{fig:oracle-calibration}
  \Description{Scatter plot of observed N1 admit rate on the y-axis
    against oracle Bernoulli p on the x-axis. A solid slate curve traces
    the 2p(1-p) closed form inside a shaded plus-or-minus 0.10 prediction
    band. Navy diamonds for the H1-compliant reference sit at the
    zero-admit corners; vermillion circles for the five non-compliant
    variants ride the arc. A boxed annotation reports MAD 0.017, 29 of 30
    in band, and R-squared 0.98. An inset plots the residual against p,
    scattered tightly around zero.}
\end{figure}

\label{sec:meas-g6}%
The bound compounds along a re-query trajectory: a baseline that
re-invokes the stochastic judge on every re-query loses replay
consistency toward zero over re-queries on the bounded-nondeterminism
keys, while \toki reads the committed keyed-log verdict and holds at
$1.0$ (Figure~\ref{fig:trajectory-replay}). Taken over the
$p\in(0,1)$ keys the theorem governs, the non-compliant rate crosses
$0.5$ by re-query $t=3$ and reaches $\approx 0.06$ by $t=24$, tracking
the closed form $C(t)=\mathbb{E}_p[p^t+(1-p)^t]$ inside the bootstrap
band; the deterministic-oracle keys ($p\in\{0,1\}$, $17$ of $30$
calibration cells) are replay-consistent by definition and are
disclosed but excluded from this rate.

% @owner   src/figures/fig-trajectory-replay.tex
% @claim   Over the bounded-nondeterminism keys Theorem
%          \ref{thm:n1-lower-bound} governs, the keyed-log discipline
%          holds the replay-consistency line at 1.0 across a re-query
%          trajectory while stochastic-judge baselines decay onto the
%          closed form C(t)=E_p[p^t+(1-p)^t].
% @does    Wrap scripts/plot_trajectory_replay.py output
%          figures/fig-trajectory-replay-img.pdf for inclusion beside
%          fig:oracle-calibration as its operational, time-axis
%          companion. Single-column line plot of replay-consistency
%          rate against re-query step.
% @needs   src/figures/fig-trajectory-replay-img.pdf (rendered by
%          scripts/plot_trajectory_replay.py from
%          results/trajectory_replay/run_v1/trajectory_summary.csv);
%          \label{thm:n1-lower-bound}.
% @feeds   src/sections/05-measurement.tex \S\ref{sec:meas-g6}.
% @breaks  reading the flat TOKI line as a task-accuracy claim; the
%          axis is write-time replay consistency only, not downstream
%          answer quality.

\begin{figure}[t]
  \centering
  \includegraphics[width=\columnwidth]{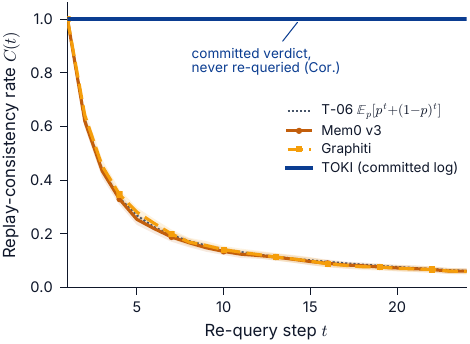}
  \caption{\textbf{Re-reading the committed verdict holds replay
    consistency; re-invoking the judge decays it.} On the
    bounded-nondeterminism keys the lower bound
    (Theorem~\ref{thm:n1-lower-bound}) governs, the two
    stochastic-judge baselines lose consistency over repeated
    re-queries and track the predicted curve
    $C(t)=\mathbb{E}_p[p^t+(1-p)^t]$ (dotted), while \toki re-reads
    its committed verdict and holds at $1.0$
    (Corollary~\ref{cor:n1-corollary}). Deterministic-oracle keys are
    consistent by definition and excluded; the axis measures
    write-time replay consistency only.}
  \label{fig:trajectory-replay}
  \Description{Line plot with re-query step t from 1 to 24 on the
    x-axis and replay-consistency rate C(t) from 0 to 1 on the y-axis.
    A thick flat navy line at 1.0 marks TOKI. A vermillion solid line
    with circle markers (Mem0 v3) and an amber dashed line with square
    markers (Graphiti) start at 1.0 and decay steeply to about 0.06 by
    step 24, hugging a dotted slate theory overlay with a shaded
    bootstrap confidence band. An annotation on the navy line notes that
    TOKI re-reads its committed verdict and is never re-queried.}
\end{figure}

\subsection{Cross-system transparency ledger}
\label{sec:meas-cross}

The cross-system rows bound where the write-time correctness contract
stops, probing downstream utility only as a transparency measure. The
reference pairs against four agent-memory systems on a shared LoCoMo
slice under a pinned synthesiser and judge, with three measured cells
and nine structural abstentions. All three confidence intervals cover
zero (\sysmem~v3: $\Delta = -0.04$, CI $[-0.10, +0.00]$; \sysgraphiti:
$-0.08$, $[-0.18, +0.00]$; \syszep: $+0.02$, $[-0.06, +0.10]$), at a
pre-registered power of $0.42$ for $\delta = 0.05$ on $n = 50$, so the
paper draws no utility-superiority claim from these rows. The row-level
protocol and five-axis deltas sit in Appendix~\ref{app:cross-system-detail}; the claim-to-evidence map
(Appendix~\ref{app:claim-evidence-map}) and artefact runbook
(Appendix~\ref{app:reproducibility-runbook}) bind each claim to its
evidence object, reproduce command, and scope.

\subsection{Memory-layer ablation on the answerable pool}
\label{sec:meas-ablation}

A paired ablation isolates the factual-recall capability the typed
memory layer carries. We score \sysours\ against its own memory-ablated
arm on the answerable-factual pool of \textsc{LoCoMo}, the categories
whose gold a same-meaning judge can adjudicate: single-hop, temporal,
and open-domain factual questions ($n = 1{,}444$). The pool excludes the
open-domain speculative category, whose subjective gold a same-meaning
judge cannot score, and the adversarial empty-gold category, where
abstention is correct and inverts the delta. With the memory layer the
reference system answers $0.540$ of the pool; ablated to no memory it
answers $0.048$, a paired $\Delta = {+}0.492$ (paired-bootstrap $95\%$
CI $[{+}0.465, {+}0.519]$, McNemar $p < 10^{-4}$, achieved power $0.748$
against the pre-registered $0.80$ at $\alpha = 0.05$). This power budget
targets a $\delta = 0.05$ minimum-detectable effect, roughly ten times
smaller than the observed gain, so the McNemar test rejects equality
with margin to spare. The gain
concentrates on open-domain factual recall ($+0.542$, $n = 841$) and
holds across temporal multi-hop ($+0.430$, $n = 321$) and single-hop
($+0.411$, $n = 282$) questions.

This ablation varies only the memory layer and makes no head-to-head
superiority claim, separating it from the cross-system ledger of
\S\ref{sec:meas-cross}. The answerable pool also explains the apparent
distance from Table~\ref{tab:cross-system-utility}, whose mixed slice
folds in the speculative and adversarial categories where both arms
score near zero. \textbf{Takeaway.} On questions a judge can score, the
typed memory layer supplies almost the whole of the reference system's
factual-recall accuracy: removing it collapses accuracy by an order of
magnitude, from $0.540$ to $0.048$.
    % 4 Evaluation
% @owner       src::sections::02-related-work
% @does        Positions the typed-operator algebra against production agent-memory systems, classical concurrency-and-provenance theory, and retrieval-side benchmarks.
% @needs       src/sections/03-foundations.tex, references/refs.bib
% @feeds       src/main.tex
% @breaks      claiming a cousin system covers a category boundary when its primary axis is efficiency or routing; missing a recent overlap-risk citation; reading the section as a literature enumeration rather than a positioning statement.

\section{Related Work}
\label{sec:related}

Three threads converge on the contract this paper types, against
a wider move to rebuild data systems around declared
algebras~\cite{idreos-2025-calculators} and AI
agents~\cite{liu-2026-overlords}.
Each stops short of the four-way combination of bitemporality,
isolation signature, K-semiring provenance, and the LLM-agent
write path.

\paragraph{Agent-memory systems ship the four strategies without typing the contract.}
\sysmem~v2~\cite{chhikara-2025-mem0} votes among
add/update/delete/none labels and drops the loser; \sysmem~v3~\cite{mem0ai-readme-2026}
defers adjudication to retrieval; \sysgraphiti~\cite{getzep-graphiti-2026} and
\syszep~\cite{rasmussen-2025-zep} invalidate older edges at
retrieval time; \sysletta~\cite{letta-context-constitution-2026}
organises memory in versioned blocks;
the engine-layer comparator \sysworlddb~\cite{ganesan-2026-worlddb} routes provenance through content-addressed Merkle ancestry.
Retrieval-augmented memories such as
\textsc{HippoRAG}~\cite{gutierrez-2024-hipporag} recall over a
knowledge graph; \textsc{MemOS}~\cite{li-2025-memos} abstracts
memory at the OS layer; substrate-dynamics cousins
(\textsc{FadeMem}~\cite{wei-2026-fademem},
\textsc{H-Mem}~\cite{yu-2026-h-mem}) decay edge weights;
\textsc{Memanto}~\cite{memanto-2026} resolves conflicts over typed
memory categories; and audit-trail cousins
(\textsc{MemLineage}~\cite{ouyang-2026-memlineage},
\textsc{HaluMem}~\cite{chen-2025-halumem}) attach read-path
evidence while the defended predicate stays implicit.
\textsc{Atomix}~\cite{atomix-2026-mohammadi} commits under a
frontier-gated endpoint, and
\textsc{APEX-MEM}~\cite{banerjee-2026-apex-mem} resolves conflicts at
retrieval time over append-only storage, neither exposing the
four-strategy surface, a space recent surveys map in
full~\cite{du-2026-survey}. None types the isolation precondition on
the operator surface; \sysworlddb\ comes closest, carrying
engine-layer Merkle provenance but leaving the isolation signature and
the provenance algebra untyped over the agent write path.

\paragraph{Classical concurrency and provenance hold every ingredient but the LLM-agent write path.}
Three foundational ingredients \toki composes come from the
seven-anomaly schedule taxonomy of
Berenson--Adya~\cite{berenson-1995-isolation,adya-2000-generalized}
and the serializable-snapshot-isolation line of
Fekete--Cahill~\cite{fekete-2005-msis,cahill-2008-ssi}; the
bitemporal model of Snodgrass and the SQL:2011 period
syntax~\cite{snodgrass-ahn-1986,kulkarni-michels-2012,jensen-1998-consensus};
and the K-relation provenance framework of
Green--Karvounarakis--Tannen~\cite{green-karvounarakis-tannen-2007,green-tannen-2017-retrospective},
rooted in why/where
provenance~\cite{buneman-khanna-tan-2001,cheney-chiticariu-tan-2009}
and its semiring and weighted-semiring
extensions~\cite{foster-green-tannen-2008,arab-glavic-2017-mv-semirings,gradel-tannen-2024-semiring-fol,geerts-poggi-tannen-2013,mohri-2002-weighted}.
\citet{widiaatmaja-2025-temporal} extend
\textsc{ProvSQL} with bitemporal m-semiring provenance on
inserts, updates, and deletes, the immediate precedent.
\textsc{VerIso}~\cite{ghasemirad-2025-veriso} mechanises
Berenson--Adya guarantees in Isabelle/HOL;
\textsc{Isolde}~\cite{barros-2026-isolde} reasons over
isolation specifications;
\textsc{S-Bus}~\cite{s-bus-2026-khan} fixes one Berenson--Adya
point on the HTTP-observable projection. \toki types the
$(\anomHR, \anomBDS, \anomAE)$ triplet over the agent-memory
$(\subj, \pred)$ partition, bridges the LLM-judge boundary into
the classical alphabet, and records the K-axis separation that
isolates the carrier's load-bearing role to the schema axis.

\paragraph{Retrieval-side benchmarks measure the symptom.}
\textsc{LongMemEval}~\cite{wu-2025-longmemeval,wu-2026-longmemeval-v2},
\textsc{LoCoMo}~\cite{maharana-2024-locomo},
\textsc{MultiTQ}~\cite{chen-2023-multitq}, and
\textsc{STALE}~\cite{chao-2026-stale} score retrieval and
end-to-end accuracy over fixed stores;
\textsc{BeliefShift}~\cite{myakala-2026-beliefshift},
\textsc{Memora}~\cite{uddin-2026-memora}, and
\textsc{TSM}~\cite{su-2026-tsm} build benchmark-level anomaly
taxonomies as observable-outcome tests.
\textsc{MemAudit}~\cite{bhargava-2026-memaudit},
\textsc{Ledger-QA} on \textsc{UMA}~\cite{zhang-2026-uma}, and
\textsc{MMA-Bench}~\cite{lu-2026-mma} stress audited, ledgered,
and evidence-weighted updates, each anchoring one operator at
the benchmark layer. They leave open a write-time correctness
specification with named predicates and soundness theorems.
Appendix~\ref{app:extended-related} extends this
positioning across the adjacent literatures.
   % 5 Related Work
% @owner       src::sections::06-open-problems
% @does        States the limitations of the manuscript and closes the conclusion paragraph; open problems, deferred wire imports, and experimental extensions live in Appendices A10, A11, A12.
% @needs       src/sections/05-measurement.tex
% @feeds       src/main.tex
% @breaks      claiming a limitation is closed in the present submission; mis-quantifying the matrix population as "eleven systems" rather than "seven systems across eleven adapter rows".

\section{Limitations and Conclusion}
\label{sec:limitations}
\label{sec:open}

\emph{Evaluation scope.} The partition-$\IsoSR$ MultiTQ slice
carries controlled-grid rather than natural-workload evidence, and
the three cross-system LoCoMo cells cover zero at power $0.42$, so
both report transparency only, with no superiority claim.
\emph{Concurrency.} The multi-writer Postgres experiment that
confirms the operator-to-isolation mapping under real contention
(\S\ref{sec:meas-g3}, Appendix~\ref{app:multiwriter-concurrency}) is single-node, so no
distributed-deployment claim follows.
\emph{Defence granularity.} The $\anomHR$ defence
(Corollary~\ref{cor:n1-corollary}) holds at intra-deployment replay
granularity for a fixed decoder tuple, with cross-deployment
divergence bounded by the judge-prompt sensitivity lemma;
the lower bound (Theorem~\ref{thm:n1-lower-bound}) covers only
boundedly nondeterministic oracles, leaving deterministic oracles and
engine-layer reconcilers outside the hypothesis.

\textsc{Toki} types the four production contradiction-resolution
strategies as bitemporal operators over a dual-row schema, closes the
contract with four soundness theorems on three orthogonal axes, and
makes keyed-log discipline a tight characterisation of $\anomHR$
exclusion. Across eight systems, \toki\ alone excludes all three
anomalies while keeping the language-model judge on the write path,
closing the correctness contract production strategies ship without.
  % 6 Conclusion

% ---- Appendices A--H (single column for wide ledgers/longtables) ---
\clearpage
\onecolumn
\appendix
% @owner       src::appendix::A0-claim-evidence-map
% @does        Master claim-indexed evidence ledger: the supplement's entry point. Binds each load-bearing main-text claim to its reviewer-risk, evidence object, reproduce command, strongest valid conclusion, and scope limit, and points at the expanding appendix (A1-A6).
% @needs       src/sections/05-measurement.tex, src/sections/03-foundations.tex, src/theorems/*.tex
% @feeds       src/supplement.tex; entry point for a PVLDB revision response
% @breaks      a row whose claim language exceeds its strongest valid conclusion; a reproduce token naming a make target or runner that does not exist; an evidence object cited as present when it is regenerated on demand.
% @claim       Every load-bearing claim of the paper maps to a named evidence object, a runnable command, a bounded conclusion, and a stated scope limit.

\section{Claim-to-Evidence Map}
\label{app:claim-evidence-map}

This appendix is the entry point to the appendices. Each row binds
one load-bearing claim of the main text to its reviewer-risk objection,
evidence object, regeneration command, strongest licensed conclusion,
and scope limit; the remaining appendices expand the corresponding rows.
Reproduce tokens are either a \textsc{make} target or a pointer to the
reproducibility runbook (Appendix~\ref{app:reproducibility}).

\begingroup\footnotesize
\setlength\tabcolsep{2.5pt}
\renewcommand{\arraystretch}{1.15}
\begin{longtable}{@{}%
  >{\raggedright\arraybackslash}p{0.075\textwidth}%
  >{\raggedright\arraybackslash}p{0.075\textwidth}%
  >{\raggedright\arraybackslash}p{0.135\textwidth}%
  >{\raggedright\arraybackslash}p{0.180\textwidth}%
  >{\raggedright\arraybackslash}p{0.095\textwidth}%
  >{\raggedright\arraybackslash}p{0.170\textwidth}%
  >{\raggedright\arraybackslash}p{0.140\textwidth}@{}}
  \caption{\textbf{Claim-to-Evidence Map.} Each load-bearing claim of
    the manuscript paired with the reviewer objection it must survive,
    its evidence object, the command that regenerates it, the
    strongest valid conclusion, and the scope limit.}
  \label{tab:claim-evidence-map}\\
  \toprule
  Claim & Main loc.\ & Reviewer risk & Evidence object & Reproduce &
  Strongest valid conclusion & Scope limit \\
  \midrule
  \endfirsthead
  \multicolumn{7}{@{}l}{\small\itshape Table~\ref{tab:claim-evidence-map} continued}\\
  \toprule
  Claim & Main loc.\ & Reviewer risk & Evidence object & Reproduce &
  Strongest valid conclusion & Scope limit \\
  \midrule
  \endhead
  \midrule
  \multicolumn{7}{r@{}}{\small\itshape continued on next page}\\
  \endfoot
  \bottomrule
  \endlastfoot

  C\nobreakdash-ALG
    & \S\ref{sec:algebra}
    & ``Just a Berenson--Adya restatement''
    & Typing rules; Lemma~\ref{lem:compose} well-typed proof (A1)
    & \texttt{make repro-compose}
    & The four typed operators compose under the lattice join with one row per isolation precondition.
    & The four shipped operators only. \\

  C\nobreakdash-BRIDGE
    & Lemma~\ref{lem:judge-callback-bridge}
    & ``Biconditional too strong; drops the first-materialization write''
    & \code{bridge_verification.csv}, $1{,}000$ trials (A1, A4)
    & App~\ref{app:reproducibility} (lemma-bridge)
    & $\pi_J(S) \preceq_{\mathrm{iso}} S$ refines; $\anomHR \iff (P_2 \vee P_3)$ on the keyed row.
    & Relational schedule model. \\

  C\nobreakdash-ISO
    & Thm~\ref{thm:anomaly-soundness}; Tbl~\ref{tab:correspondence}
    & ``$\anomHR$ needs $\IsoSR$, not $\IsoSI$''
    & \code{iso_matrix_grid.csv}; iso-level \texttt{summary.csv}
    & App~\ref{app:reproducibility} (iso-matrix)
    & $L$ excludes $\Phi$ iff $L$ dominates $\mathrm{guard}(\Phi)$: $\IsoSI$ excludes the re-read $P_2$, $\IsoSR$ the first-insert $P_3$.
    & Within the Berenson--Adya predicate set. \\

  C\nobreakdash-N3
    & Thm~\ref{thm:audit-erasure-schema}
    & ``Schema-lift is a $K$-semiring tautology''
    & \code{schema_grid.csv}
    & App~\ref{app:reproducibility} (schema-axis)
    & The audit-row schema excludes $\anomAE$ ($0\%$ vs $100\%$ under the base schema).
    & Parametric in the $K$ carrier. \\

  C\nobreakdash-N2
    & Cor.~\ref{cor:n2-corollary}
    & ``$\anomBDS$ is an independent discovery, not a specialisation''
    & \code{n2_grid.csv}
    & App~\ref{app:reproducibility} (n2-partition)
    & Partition-$\IsoSR$ excludes $\anomBDS$ as a confidence-weighted A5B specialisation.
    & Constructed contention grid. \\

  C\nobreakdash-T4
    & Thm~\ref{thm:carrier}
    & ``$\deg \ge 2$ iff is imprecise; shift versus derivative''
    & \code{recovery_rate.csv}; $k$-semiring \texttt{counterfactual.csv}
    & App~\ref{app:reproducibility} (t4)
    & Per-instance erasure recovers iff $\deg_{m_r}(\pi) \ge 2$ via the formal shift $\sigma_{X_r}$.
    & Needs multiplicity $\ge 2$; semiring without multiplicative inverses. \\

  C\nobreakdash-T6
    & Thm~\ref{thm:n1-lower-bound}
    & ``The lower bound has narrow scope''
    & \texttt{calibration.csv}; N1 \code{FINAL_SUMMARY.csv}; N1 \code{cross_system_test.csv} (Welch $|t|$); \code{temperature_grid.csv}
    & App~\ref{app:reproducibility} (oracle, n1)
    & Admit rate matches the closed form $2p(1-p)$ (mean abs.\ dev.\ $0.017$).
    & Boundedly nondeterministic oracle; deterministic oracles out of scope. \\

  C\nobreakdash-T5
    & Thm~\ref{thm:composition}
    & ``Composition is not closed under the lattice''
    & \code{composition_grid.csv}, $1{,}364$ pipelines
    & App~\ref{app:reproducibility} (t5)
    & Only $\IsoSR{+}\mathsf{cb}$ is the lattice supremum for every pipeline shape.
    & Pipeline length $\le 5$. \\

  C\nobreakdash-G1
    & Tbl~\ref{tab:anomaly-bench}; \S\ref{sec:meas-g1}
    & ``Eight systems vs.\ eleven adapters; baseline fairness; imported abstentions''
    & AnomalyClaim and AnomalyWire CSVs (shipped); \texttt{manifest.json} pinned SHAs (A3)
    & \texttt{make repro-\allowbreak anomaly-\allowbreak wire}
    & Eight systems; every baseline keeping an LLM judge on the write path admits $\ge 1$ predicate; the engine-layer comparator excludes all three by removing the judge, \sysours\ alone while keeping it.
    & Six of fifteen imported cells abstain on structural grounds (A3). \\

  C\nobreakdash-G2
    & Fig.~\ref{fig:memory-utility}; \S\ref{sec:meas-g2}
    & ``Off-target controls also saturate at $+1.00$''
    & G2 \texttt{summary.csv} ($9$-cell grid)
    & \texttt{make repro-g2}
    & Mechanism-stress evidence; only the audit-row LoCoMo cell carries natural-workload movement ($+0.86$).
    & Constructed slices saturate; no utility-superiority claim. \\

  C\nobreakdash-G3
    & Fig.~\ref{fig:systems-perf-scaling-2x2}; \S\ref{sec:meas-g3}
    & ``Single-process measurement only''
    & G3 \texttt{scaling.csv} ($5$ axes)
    & \texttt{make repro-g3}
    & Mean $\sim c^{0.86}$ at $R^2 = 0.992$: a local single-process lock-contention signature.
    & PostgreSQL saturates under \texttt{SERIALIZABLE} at concurrency $\ge 2$; no distributed claim. \\

  C\nobreakdash-G4
    & Tbl~\ref{tab:ablation}; \S\ref{sec:meas-g4}
    & ``Carrier choice is arbitrary''
    & operator-ablation \texttt{results.csv}; $k$-semiring \texttt{counterfactual.csv}
    & \texttt{make repro-g4}
    & Verdicts are carrier-invariant; token recall is $K$-load-bearing.
    & Three carriers, $100$ seeds. \\

  C\nobreakdash-G5
    & Fig.~\ref{fig:g5-anchors}; \S\ref{sec:meas-g5}
    & ``The theorems are untested empirically''
    & Five grids: iso-matrix, schema-axis, n2-partition, t5-composition, oracle-variance
    & App~\ref{app:reproducibility} (g5 grids)
    & Each grid records the predicted $0$/$1$ boundary exactly.
    & Structurally controlled grids, not natural workloads. \\

  C\nobreakdash-CROSS
    & Tbl~\ref{tab:cross-system-utility}; \S\ref{sec:meas-cross}
    & ``No head-to-head superiority is shown''
    & \code{cross_system/summary.csv} ($3$ measured, $9$ abstaining cells)
    & App~\ref{app:reproducibility} (cross-system)
    & All three confidence intervals cover zero; the paper makes no superiority claim.
    & Under-powered ($0.42$) for $\delta = 0.05$ at $n = 50$. \\

  C\nobreakdash-ABLATE
    & \S\ref{sec:meas-ablation}
    & ``The typed memory layer carries factual recall''
    & \code{powered_summary_answerable.csv} ($n = 1{,}444$, $\Delta = {+}0.49$)
    & App~\ref{app:reproducibility} (powered ablation)
    & With-memory $0.540$ vs ablated $0.048$; McNemar $p < 10^{-4}$, achieved power $0.748$.
    & Answerable categories ($1$,$2$,$4$) only; an in-system memory ablation (no cross-system superiority is claimed). \\

  C\nobreakdash-H1
    & \S\ref{sec:meas-g5}; Thm~\ref{thm:n1-lower-bound}
    & ``The claim that baselines lack a keyed log is unverified''
    & \code{h1_audit}; \code{deployment_scan} CSVs (pinned commits, A3)
    & App~\ref{app:reproducibility} (h1 audit)
    & Four imported baselines are structurally H1-non-compliant (zero keyed-judge-log hits).
    & Audited at the pinned commits only. \\

\end{longtable}
\endgroup

\section{Formal witnesses and proofs}
\label{app:formal-witnesses}

This appendix gives the witness-level account behind the main-text
results. Each load-bearing result is presented in a graded form with
five parts: a minimal witness schedule that exhibits the phenomenon,
the formal statement, the proof, the exact scope that changed relative
to the main-text statement, and the empirical-anchor command that
regenerates the supporting evidence. The graded results are
Lemma~\ref{lem:judge-callback-bridge} (the judge-callback bridge), the
$\anomHR$ lower bound of Theorem~\ref{thm:n1-lower-bound}, the carrier
recoverability of Theorem~\ref{thm:carrier}, and the composition
statement of Theorem~\ref{thm:composition}. The witness-schedule
ledger that pins tightness of the two soundness theorems, the
prompt-sensitivity separation, the write-path composition lemma, the
$4 \times 4$ composition table, and the expanded guard surface follow
as supporting material.

\paragraph{Notation.}
$r_i[x_v]$ denotes transaction $i$ reads $x$'s version $v$;
$w_i[x_v]$ denotes transaction $i$ writes $x$ producing version $v$;
$c_i$ and $a_i$ denote commit and abort. Provenance annotations $p_v$
are K-relation polynomials over the carrier $K$; the witnesses are
parametric in $K$ and hold for any commutative semiring with natural
order, including the multilinear $\mathbb{N}[X,T]$ default and the
multiset $\mathbb{N}[X,T]$ alternative.

% =====================================================================
\subsection{Graded result: the judge-callback bridge}
\label{app:graded-judge-bridge}

The bridge homomorphism of Lemma~\ref{lem:judge-callback-bridge} maps
each judge call and each callback wait into a logged read on a keyed
table, so that $\IsoSR$ on that table serializes the read against the
operator commit. The graded account below restates the bridge with its
minimal witness, the empirical replay that exercises the
key-multiplicity invariant, and the scope on which the bridge holds.

\subsubsection{Minimal witness schedule}
The minimal bridge witness is the $\anomHR$ replay schedule keyed by
the read set and decoder tuple:
\[
  j_1(R, \theta)\,
  w_2[\textsf{judge\_log}(R, \theta)]\,
  c_2\,
  j_1(R, \theta)\,
  c_1 .
\]
Here $j_1(R, \theta)$ is the first judge call at the keyed row
$(R, \theta)$, and the second $j_1(R, \theta)$ is its replay. Holding
the row at a level weaker than $\IsoSR$ admits the interleaved write
$w_2$, which lets the replay observe a different committed vote;
$\IsoSR$ on \code{judge\_log} serializes the replay read against $w_2$
and forces a single vote per keyed row.

\subsubsection{Formal statement}
Under hypotheses $(H1) + (H2)$ of
Lemma~\ref{lem:judge-callback-bridge}, the bridge homomorphism $\pi_J$
maps the judge call $j_i(R, \theta)$ to a read at the
$(R, \theta)$-keyed row of \code{judge\_log}, and the callback wait to
a read at the keyed row of \code{callback\_log}. The homomorphism
preserves key multiplicity: distinct $(R, \theta)$ keys map to distinct
logged rows, and a replay of $(R, \theta)$ maps to a re-read of the
same row. The bridge therefore makes $\anomHR$ a $P_2$-style fuzzy-read
predicate on the keyed log, excluded at $\IsoSR$.

\subsubsection{Proof}
The homomorphism is the keyed-log discipline of \S\ref{sec:schema}
applied to the $\Sigma_+$ event alphabet. Each judge call carries the
read set $R$ and the decoder tuple
$\theta = (\textsf{prompt}, \textsf{seed}, \textsf{model\_version},
\textsf{temperature}, \textsf{tool\_output\_hash})$ of
\S\ref{sec:found-schedule}; the keyed row is addressed by
$(R, \theta)$. A first call writes the row; a replay reads it. The
fuzzy-read predicate on the keyed row is exactly $\anomHR$: two reads
of the same key returning different committed votes. Adya's Mixing
Theorem characterises $\IsoSR$ on the keyed table as serializing the
predicate read against any concurrent write, so the replay observes the
committed vote of the first call and no other. The construction is
parametric in $\theta$: the homomorphism applies for every fixed
$\theta$, regardless of how the oracle $J$ votes across other decoder
tuples. Key multiplicity is preserved because $\pi_J$ is injective on
keys.

\subsubsection{Scope and empirical anchor}
The main-text statement asserts the bridge as an engine-layer guarantee
at fixed $\theta$; the appendix adds the key-multiplicity preservation
clause the main text uses implicitly. The cross-deployment behavior
(two deployments pinning different $\theta$) is treated separately in
Lemma~\ref{lem:n1-prompt-separation} below. The empirical anchor is the
lemma-bridge protocol of \S\ref{sec:appendix-g5-structural-grids}
($1000$ trials at schedule length $12$, all pass rates $1.00$, every row
key-multiplicity preserved;
\code{results/lemma_bridge/run_v1/bridge_verification.csv}).

% =====================================================================
\subsection{Graded result: the $\anomHR$ lower bound}
\label{app:graded-n1-lower-bound}

Theorem~\ref{thm:n1-lower-bound} states the workload-side lower bound
on $\anomHR$ admission for systems that do not key the judge log. The
graded account pairs the lower bound with the oracle-variance witness
that calibrates the predicted admission rate and the empirical
admission rates observed across the baseline systems.

\subsubsection{Minimal witness schedule}
The minimal $\anomHR$ admission witness is the unkeyed replay
\[
  j_1(R)\,
  j_2(R)\,
  c_1\,
  c_2 ,
\]
where two judge calls on the same read set $R$ commit independently
without a keyed log to serialize them. When the oracle $J$ assigns
vote $1$ with empirical probability $p$ on $R$, the two calls disagree
with probability $2p(1-p)$, and a system that admits the later vote
without replay-keying admits the inconsistency at that rate.

\subsubsection{Formal statement}
For a system without a keyed judge log, and a read set $R$ on which the
oracle votes $1$ with empirical probability $p$, the per-pair $\anomHR$
admission probability equals $2p(1-p)$. The bound is zero exactly when
$p \in \{0, 1\}$ (the oracle is deterministic on $R$) and is maximised
at $p = \tfrac{1}{2}$. A system that keys the judge log and holds it at
$\IsoSR$ admits $\anomHR$ at rate $0$ on every $R$, independent of $p$.

\subsubsection{Proof}
Two independent draws from a Bernoulli$(p)$ oracle disagree with
probability $\Pr[\text{draw}_1 \neq \text{draw}_2] = p(1-p) + (1-p)p =
2p(1-p)$. An unkeyed system commits the first draw, then on replay
commits the second draw; the committed pair disagrees exactly when the
draws disagree, so the admission probability is $2p(1-p)$. A keyed
system reads the committed row on replay rather than re-invoking the
oracle, so the replay vote equals the first-call vote and the pair
never disagrees; the admission rate is $0$ by the bridge of
Appendix~\ref{app:graded-judge-bridge}.

\subsubsection{Scope and empirical anchor}
The main-text theorem states the bound as a workload-and-oracle
property; the appendix adds the calibration witness matching the
predicted $2p(1-p)$ against the observed admission rate per system and
seed, and the engine-layer keyed-log defence that drives the keyed
system to $0$ on every read set. The calibration and admission ledgers
are detailed in the oracle-variance protocol of
\S\ref{sec:appendix-g5-structural-grids} and the cross-system slice of
\S\ref{sec:appendix-cross-system-slice}: \sysours{} in its dispatched
configuration admits at mean rate $0.0$ (all $245$ candidate pairs
served from the log, not the oracle;
\code{results/n1_empirical/run_v1/FINAL_SUMMARY.csv}), while the unkeyed
baselines (\sysmem, \sysgraphiti, \sysletta, \syszep, and the \sysours{}
stripped ablation) admit at mean rates between $0.167$ and $0.204$.

% =====================================================================
\subsection{Graded result: carrier recoverability}
\label{app:graded-carrier}

Theorem~\ref{thm:carrier} states that the audit-erasure recovery
surface is non-trivial exactly when the K-relation carrier records
enough structure to keep the erased provenance reachable. The graded
account pairs the recoverability statement with the per-carrier
witness and the counterfactual ablation.

\subsubsection{Minimal witness schedule}
The minimal carrier witness is the audit-erasure schedule of
Appendix~\ref{app:witness-n2} restated at the carrier level:
\[
  w_1[x_a, p_a]\,c_1\,w_2[x_b, p_b]\,c_2\,r_3[x]\,c_3 .
\]
After $T_2$ commits, the current-row surface contains only
$(x_b, p_b)$. The audit row carries
$\auditt = (p_a \provadd p_b, \strat, t_s)$. Recovery of $p_a$ requires
$p_a \preceq_K p_a \provadd p_b$ under the natural order $\preceq_K$,
which holds in a carrier rich enough to keep both monomials and fails
in a carrier that collapses them.

\subsubsection{Formal statement}
The erased provenance $p_a$ is recoverable from the audit row exactly
when the carrier $K$ keeps $p_a \preceq_K p_a \provadd p_b$ with $p_a$
identifiable as a summand of $p_a \provadd p_b$. The multidegree
multiset carrier $\mathbb{N}[X,T]$ keeps each summand identifiable and
recovers $p_a$ on every instance; the multilinear $\mathbb{N}[X,T]$
default recovers $p_a$ only when the two monomials do not coincide
after multilinear collapse; the Boolean security-semiring reduct
collapses both monomials to a single truth value and recovers $p_a$ on
no instance where $p_a \neq p_b$ as Boolean values.

\subsubsection{Proof}
The natural order $\preceq_K$ on a commutative semiring with natural
order satisfies $p_a \preceq_K p_a \provadd p_b$ with existence witness
$d = p_b$, so the inequality holds in every such carrier. Recovery
additionally requires that $p_a$ be reconstructible from
$p_a \provadd p_b$. In the multidegree multiset carrier the sum keeps
the exponent vector of each summand, so $p_a$ is read off by
subtraction; recovery succeeds on every instance. In the multilinear
carrier coincident monomials collapse coefficients, so $p_a$ is
reconstructible only when the carrier preserved a distinguishing
factor; recovery succeeds on a strict subset. In the Boolean reduct the
sum is a disjunction that retains no multiplicity, so $p_a$ is not
reconstructible whenever $p_a$ and $p_b$ map to the same truth value;
recovery fails on those instances.

\subsubsection{Scope and empirical anchor}
The main-text theorem states recoverability as a carrier property; the
appendix adds the counterfactual ablation that isolates the carrier as
the cause (schedule and operator algebra held fixed, only the carrier
varied), confirming recoverability is a property of the carrier choice
rather than the schema or isolation level. The per-instance and
counterfactual ledgers are detailed in the G4 carrier-ablation protocol
of \S\ref{sec:appendix-g4-ablation}
(Table~\ref{tab:g4-ablation-details}): the multidegree carrier recovers
all $375$ instances at accuracy $1.0$, the multilinear carrier recovers
$106$ of $375$ at accuracy $0.683$, and the Boolean carrier recovers
none at accuracy $0.345$
(\code{results/g4_ablation/k_semiring/counterfactual.csv}; the
$200$-instance \code{t4_per_instance} run shows the same $1.0$/$0.0$
ordering).

% =====================================================================
\subsection{Graded result: the composition theorem}
\label{app:graded-composition}

% @owner       src::theorems::L-01-composition-well-typed
% @does        Proves the well-typed composition Lemma (L-01): bound iso-level annotations compose monotonically across the iso-only fragment via lattice join; per-table IsoSRpol pins compose through table-disjoint preconditions via the Appendix B table-decomposition argument.
% @needs       src/macros.tex, src/theorems/T-01-provenance-preservation.tex, src/theorems/T-01b-audit-erasure-schema.tex, references/refs.bib (allen-1983-cacm)
% @feeds       src/sections/03-foundations.tex, src/sections/04-algebra.tex, src/sections/05-measurement.tex
% @breaks      a wrong lattice-join claim mis-types operator pipelines and breaks the composition guarantee; missing IsoSRpol side condition leaves cross-table policy chains unanchored.
%
% Lemma 1 (composition well-typedness). The four typed operators
% (\opLWW, \opEvi, \opAwait, \opRule) compose under sequential
% pipelining; the composite requires the lattice join (the strongest
% of the per-step preconditions). The 2026-05-12b refinement narrows
% the lemma to the monotone subset
% $\{\IsoRC, \IsoSI, \IsoSR, \IsoRCcb\}$ where lattice order coincides
% with schedule satisfaction; the table-scoped $\IsoSR^{\mathrm{policy}}$
% case is handled in §3.7 via per-table partitioning outside lattice
% composition. Full proof in Appendix B.

\begin{lemma}[Composition well-typedness on the monotone fragment]
\label{lem:compose}
Let $\oplus_a \in \{\opLWW, \opEvi, \opAwait, \opRule\}$ denote a
contradiction-resolution operator with precondition
$L_a \in \mathcal{L} = \mathcal{L}_{\mathrm{iso}} \times
\mathcal{L}_{\mathrm{cb}}$ (\S\ref{sec:found-iso}). For any
sequential composition $\oplus_{a_n} \circ \cdots \circ \oplus_{a_1}$
over the dual-row schema of \S\ref{sec:schema}, the composed
operator is well-typed under any schedule that satisfies
$L^{*} \;=\; \bigvee \{L_{a_1}, \ldots, L_{a_n}\}$, where $\bigvee$
is the lattice join in $\mathcal{L}$ (component-wise on the iso and
callback axes). The per-table $\IsoSRpol$ annotation carried by
$\opRule$ stays outside $\mathcal{L}$: it pins the named policy
table to $\IsoSR$ in the composite and leaves $L^{*}$ on the
remaining tables.
Allen-relation selection on the operator output preserves the
bitemporal-tuple type by closure of Allen's thirteen interval
relations and the twelve-relation transitivity table that omits
equality~\cite{allen-1983-cacm}. On the orthogonal
schema axis $\mathcal{L}_{\mathrm{schema}} = \{\textsf{base} \preceq
\textsf{audit{-}row}\}$, the composite holds the maximum of per-step
schema requirements; when any step in the chain falls under
Theorem~\ref{thm:audit-erasure-schema}, the composite pins to
$\textsf{audit{-}row}$.
Compositions involving $\opRule$ on the policy table compose through table-disjoint preconditions; the Appendix~\ref{sec:appendix-composition} table-decomposition argument covers that case.
\end{lemma}

\noindent\emph{Proof sketch.}
Lattice order on $\mathcal{L}_{\mathrm{iso}}$ coincides with
schedule satisfaction, so the lattice join $L^{*}$ is the
exact precondition of the composite; component-wise join
preserves the orthogonal callback axis. The table-scoped
$\IsoSRpol$ adds a per-table predicate that composes through
table-disjoint preconditions. Closure under Allen-relation
selection follows from the transitivity table over the twelve
non-equality relations~\cite{allen-1983-cacm}. Appendix~\ref{sec:appendix-composition} enumerates the four-by-four pairwise composition
table.

Theorem~\ref{thm:composition} and the composition lemma
Lemma~\ref{lem:compose} state that the sequential composition of two
typed operators holds the lattice join of their isolation
preconditions. The graded account pairs the join statement with the
operator-precondition witnesses, the $4 \times 4$ composition table,
and the integration suite that exercises every cell.

\subsubsection{Minimal witness schedule}
The minimal composition witness is the off-diagonal pair
$\opEvi \circ \opLWW$, whose component preconditions are
$L_{\opLWW} = \IsoRC$ and $L_{\opEvi} = \IsoSI$. A schedule that holds
only $\IsoRC$ admits a split-snapshot evidence read at the $\opEvi$
step, which the join $\IsoRC \vee \IsoSI = \IsoSI$ excludes. The
composite therefore requires $\IsoSI$, the join of the two component
levels.

\subsubsection{Formal statement}
For typed operators $\oplus_a$ and $\oplus_b$ with isolation
preconditions $L_a$ and $L_b$ in
$\mathcal{L} = \mathcal{L}_{\mathrm{iso}} \times
\mathcal{L}_{\mathrm{cb}}$, the sequential composition
$\oplus_b \circ \oplus_a$ holds the join $L_a \vee L_b$. The join is
component-wise: $\IsoRC \vee \IsoSI = \IsoSI$ on the iso axis,
$\bot \vee \mathsf{cb} = \mathsf{cb}$ on the callback axis, and the
policy annotation composes by requiring $\IsoSR$ only on the named
policy table.

\subsubsection{Proof}
A schedule holding the composite must satisfy both component
preconditions, so it holds every level at or above $L_a$ and every
level at or above $L_b$, hence every level at or above $L_a \vee L_b$.
Conversely a schedule holding $L_a \vee L_b$ holds each component
precondition by downward closure of $\schedmodels$ on the lattice, so
both operator steps type-check. The callback component composes by
requiring the delivered \code{callback\_log} row whose read is ordered
between the candidate read and the operator write; the policy
annotation pins only the named policy table. The join is therefore both
necessary and sufficient.

\subsubsection{Scope and empirical anchor}
The main-text theorem states the join for the two-operator case. The
appendix records the minimal weakening history behind each of the four
operator preconditions (Table~\ref{tab:operator-precondition-witnesses})
and enumerates all $16$ ordered pairs (Table~\ref{tab:composition}) in
\S\ref{sec:appendix-composition} below, where the integration suite that
exercises every cell is detailed. The composition is on the iso and
callback axes; the configuration-axis dependence of the emitted wire
verdict is governed separately by
Lemma~\ref{lem:writepath-config-composition} below.

% =====================================================================
\subsection{Witness schedules for the two soundness theorems}
\label{sec:appendix-witnesses}
\label{app:witness-schedules}
\label{app:witness-n1}
\label{app:witness-n2}

This subsection exhibits the eight witness schedules behind the two
soundness theorems of the ADR-0008 D-T01 split: seven iso-axis
witnesses for the anomaly-soundness theorem over the seven
Berenson--Adya classical anomalies, and one schema-axis witness for the
audit-erasure theorem. Each schedule holds the immediate predecessor of
the defending level on its lattice
($\mathcal{L}_{\mathrm{iso}}$ for the seven, $\mathcal{L}_{\mathrm{schema}}$
for the eighth). Together with the soundness inductions in the main
text, the eight witnesses establish the if-and-only-if of both
theorems.

\subsubsection{Iso-axis witnesses for anomaly soundness}
\label{sec:appendix-witnesses-iso}

The seven witnesses below pin tightness of the anomaly-soundness
theorem on the iso lattice $\mathcal{L}_{\mathrm{iso}}$. Each schedule
holds the immediate predecessor of $\guardiso(\phi)$ and witnesses
$\phi$.

\paragraph{Schema scope of the seven witnesses.}
The schedules below are stated on the current-row projection of the
dual-row schema (\S\ref{sec:schema}): every $w_i[x_v]$ and $r_i[x_v]$
targets a fact tuple in the default $\rowkind = \current$ scope. Each
operator step of \S\ref{sec:algebra} also emits an audit tuple at
$\rowkind = \textsf{audit}$ on the $\textsf{audit\_log}(t_s)$ slice.
The default filter $\sigma_{\rowkind = \current}$, whose exhaustiveness
is guaranteed by the CHECK constraint of \S\ref{sec:schema}, projects
every audit-row write away; audit emissions therefore contribute no
edges to the dependency graph of the witness schedule's
$\mathcal{Q}_{\textsf{base}}$ reads and writes.
Proposition~\ref{prop:schema-lift-conservatism}'s commutation argument
applied to selection lifts each witness's verdict at its named
isolation level from the base schema to the full dual-row schema
unchanged.

\paragraph{$P_0$ dirty write.}
Schedule: $w_1[x_a]\,w_2[x_b]\,c_2\,c_1$. Holds the level immediately
weaker than $\IsoRC$ (no isolation guarantee). $\opLWW$ at this level
admits the schedule, but the schedule violates $P_0$:
$T_2$ wrote between $T_1$'s write and commit. Defence: $\IsoRC$ blocks
$w_2[x]$ until $c_1$, ruling out the schedule.

\paragraph{$P_1$ dirty read.}
Schedule: $w_1[x_a]\,r_2[x_a]\,a_1\,c_2$. Holds the level immediately
weaker than $\IsoRC$. $T_2$ reads $T_1$'s uncommitted write, then
$T_1$ aborts: $T_2$ has read a value that never officially existed.
Defence: $\IsoRC$ blocks $r_2[x_a]$ until $T_1$ commits or aborts.

\paragraph{$P_2$ fuzzy read.}
Schedule: $r_1[x_a]\,w_2[x_b]\,c_2\,r_1[x_b]\,c_1$. Holds $\IsoRC$
(the immediate predecessor of $\IsoSI$). $T_1$ reads $x$ twice and
sees two different values because $T_2$ committed in between.
Defence: $\IsoSI$ pins $T_1$'s reads to a single snapshot.

\paragraph{$P_3$ phantom (policy table).}
Schedule:
\[
  r_1[\text{policy: WHERE active}]\,
  w_2[\text{insert new active row}]\,
  c_2\,
  r_1[\text{policy: WHERE active}]\,
  c_1 .
\]
Holds $\IsoSI$ (the immediate predecessor of $\IsoSRpol$). $T_1$'s two
predicate reads see different row sets because $T_2$ inserted a row
matching the predicate. Defence: $\IsoSR$ on the policy table
serializes the predicate-read with $T_2$'s insert.

\paragraph{$P_4$ lost update.}
Schedule: $r_1[x_0]\,w_2[x_1]\,w_1[x_2]\,c_2\,c_1$. Holds $\IsoRC$
(the immediate predecessor of $\IsoSI$). $T_1$ reads $x$ at version
0, $T_2$ writes version 1, $T_1$ overwrites with version 2 (computed
from $x_0$): $T_2$'s update is silently lost. Defence: $\IsoSI$
detects the read-write conflict and aborts $T_1$ on its commit.

\paragraph{$A5A$ read skew.}
Schedule: $r_1[x_a]\,r_1[y_a]\,w_2[x_b]\,w_2[y_b]\,c_2\,r_1[\textit{post-2}]
\,c_1$. Holds $\IsoRC$ (immediate predecessor of $\IsoSI$). $T_1$
sees $x$ at one snapshot version and $y$ at another, breaking
constraints that bind $x$ and $y$. Defence: $\IsoSI$ pins all of
$T_1$'s reads to one snapshot.

\paragraph{$A5B$ write skew.}
Schedule: $r_1[x_a]\,r_2[y_a]\,w_1[y_b]\,w_2[x_b]\,c_1\,c_2$. Holds
$\IsoSI$ (immediate predecessor of $\IsoSR$). Each transaction reads
the other's pre-write state and writes a value the other did not
see; the joint post-state violates a cross-row constraint. Defence:
$\IsoSR$ orders the writes to detect the cycle.

\subsubsection{Schema-axis witness for audit erasure}
\label{sec:appendix-witnesses-schema}

The single witness below pins tightness of the audit-erasure theorem
on the schema lattice $\mathcal{L}_{\mathrm{schema}}$. The schedule
holds the $\textsf{base}$ schema (the immediate predecessor of
$\textsf{audit{-}row}$) and witnesses $\anomAE$.

\paragraph{$\anomAE$ audit erasure (parametric K).}
Schedule:
$w_1[x_a, p_a]\,c_1\,w_2[x_b, p_b]\,c_2\,r_3[x]\,c_3$, executed against
a base schema (no audit-row column). Holds $(\IsoRC, \textsf{base})$,
the immediate predecessor of $(\IsoRC, \textsf{audit{-}row})$. After
$T_2$'s commit, the recovery surface contains only $(x_b, p_b)$.
Therefore $p_a$ is unreachable under the K-relation natural order
$\preceq_K$, since $p_a \npreceq_K p_b$. Defence: the
$(\IsoRC, \textsf{audit{-}row})$ refinement has every operator emit
$\auditt = (p_a \provadd p_b, \strat, t_s)$ alongside the winner;
because $K$ is commutative and $\preceq_K$ is the natural order,
$p_a \provadd p_b \succeq_K p_a$ holds parametrically in $K$ (the
existence witness is $d = p_b$). The construction works for any
commutative semiring with natural order; for the multilinear N[X,T]
default the witness is one-monomial coefficient dominance; for
the multiset N[X,T] instance the witness is exponent-wise dominance;
for the Boolean security-semiring of Foster-Green-Tannen 2008 the
witness is the disjunction lattice.

\begin{proposition}[Tightness of the iso-axis anomaly soundness, parametric in $K$]
\label{prop:t01-tightness}
For each classical anomaly
$\phi \in \{P_0, P_1, P_2, P_3, P_4, A5A, A5B\}$ let $S_\phi$ denote
the schedule exhibited in the named \paragraph of
\S\ref{sec:appendix-witnesses-iso}. Let $L'_{\phi}$ denote the
immediate predecessor of $\guardiso(\phi)$ on
$\mathcal{L}_{\mathrm{iso}}$, namely the level below $\IsoRC$ when
$\phi \in \{P_0, P_1\}$, $L'_{\phi} = \IsoRC$ when
$\phi \in \{P_2, P_4, A5A\}$, and $L'_{\phi} = \IsoSI$ when
$\phi \in \{P_3, A5B\}$. Then (a) $S_\phi \schedmodels L'_{\phi}$
and (b) $S_\phi \models \phi$ for each of the seven $\phi$.
The construction is parametric in the K-relation carrier: the
witness alphabet $\{r_i[x_v], w_i[x_v], c_i, a_i\}$ omits the
provenance annotation, so replacing the carrier $K$ with any
commutative semiring with natural order (multilinear
$\mathbb{N}[X, T]$, multi-degree $\mathbb{N}[X, T]^{\#}$, Boolean
security-semiring reduct) leaves the seven verdicts unchanged.
\end{proposition}

\begin{proof}
Verification of (a)+(b) is exhibited paragraph by paragraph in
\S\ref{sec:appendix-witnesses-iso}: each of the seven
\paragraph statements ($P_0$ dirty write through $A5B$ write skew)
displays $S_\phi$ explicitly, verifies $S_\phi \schedmodels
L'_{\phi}$ at the named level by the read/write/commit ordering,
and verifies $S_\phi \models \phi$ by inspection of the conflict
graph against the predicate definition. The seven verifications
are sufficient because Adya's Mixing Theorem~%
\cite[Theorem~4.6]{adya-2000-generalized} characterises each
$L \in \mathcal{L}_{\mathrm{iso}}$ as a constraint on the
rw/wr/ww conflict graph over $\Sigma$ events, and the seven
predicates of Berenson \etal~\cite{berenson-1995-isolation} are
similarly constraints on the same graph.

Parametricity in $K$ follows from the alphabet-omission argument:
the iso-axis verdict of each $S_\phi$ depends only on the order
of $r_i / w_i / c_i / a_i$ events; the provenance polynomial
$p_v \in K[X, T]$ that the operator algebra of
\S\ref{sec:algebra} carries in the audit row is emitted at the
$\rowkind = \textsf{audit}$ slice, which the default filter
$\sigma_{\rowkind = \current}$ of
Proposition~\ref{prop:schema-lift-conservatism} projects away
from $\mathcal{Q}_{\textsf{base}}$ reads. The seven schedules
therefore admit the same verdict under every commutative semiring
carrier with natural order.
\end{proof}

\paragraph{Tightness summary.}
The seven iso-axis schedules of Proposition~\ref{prop:t01-tightness}
combined with the iso soundness induction in the main text
(proof of Theorem~\ref{thm:anomaly-soundness}, $(\Leftarrow)$
direction) establish the if-and-only-if of the iso-axis
anomaly-soundness theorem: assuming $L \not\succeq \guardiso(\Phi)$
selects some $\phi^{*} \in \Phi$ with $L \preceq L'_{\phi^{*}}$,
and the downward closure of $\schedmodels$ on the iso lattice
yields $S_{\phi^{*}} \schedmodels L$, which witnesses $\phi^{*}$ and
falsifies $\mathrm{Prevents}(L, \Phi)$ of
Equation~\eqref{eq:prevents}. The schema-axis schedule (a) holds
$\textsf{base}$ on $\mathcal{L}_{\mathrm{schema}}$ and (b) admits
$\anomAE$; combined with the schema-axis soundness induction in
Theorem~\ref{thm:audit-erasure-schema}, this establishes the
audit-erasure theorem.

\subsubsection{Induction case ledger for the seven classical anomalies}
\label{sec:appendix-witnesses-induction-cases}

The main-text proof of Theorem~\ref{thm:anomaly-soundness}
exhibits the induction step in detail for the representative
case $\phi = A5B$. The remaining six cases follow the same
template: the appended operation $\mathit{op}$ obeys its typing
precondition at $\guardiso(\phi)$, the iso lattice meet ensures
$L \succeq L_{\mathit{op}}$, and Adya's Mixing
Theorem~\cite[Mixing~Theorem]{adya-2000-generalized} preserves
the predicate verdict over the prefix-extension
$S \cdot \mathit{op}$. Table~\ref{tab:induction-cases} records one
row per case: the guard level, the conflict-graph invariant the
typing precondition preserves under $\mathit{op}$, and the witness
that pins tightness (at the level below the guard, the symmetric
counterpart of the soundness step). The $A5B$ case is recorded
inline in the main-text proof, where the partition-projected
serializability argument applies the Mixing Theorem on
$\pi_{(s,p)}(S')$ and excludes the disjoint-write / intersecting-read
pattern.

\begin{table}[!ht]
  \centering\scriptsize
  \caption{Induction cases for the seven classical anomalies. Each
  row names the guard level $\guardiso(\phi)$, the invariant the
  typing precondition preserves when $\mathit{op}$ is appended at
  $L \succeq \guardiso(\phi)$, and the tightness witness at the level
  immediately below the guard.}
  \label{tab:induction-cases}
  \setlength{\tabcolsep}{3pt}
  \begin{tabular}{@{}>{\raggedright\arraybackslash}p{0.105\textwidth}>{\raggedright\arraybackslash}p{0.060\textwidth}>{\raggedright\arraybackslash}p{0.300\textwidth}>{\raggedright\arraybackslash}p{0.330\textwidth}@{}}
    \toprule
    Anomaly & Guard & Invariant preserved under $\mathit{op}$ & Tightness witness (level below guard) \\
    \midrule
    $P_0$ dirty write & $\IsoRC$ & no overlapping uncommitted writes per operand; an interleaved $w_j$ inside $T_i$'s write interval is aborted before commit & $w_1[x_a]\,w_2[x_b]\,c_2\,c_1$ \\
    $P_1$ dirty read & $\IsoRC$ & every committed read paired with a committed write on the same version; $r_j[x_a]$ while $T_i$ uncommitted is excluded & $w_1[x_a]\,r_2[x_a]\,a_1\,c_2$ \\
    $P_2$ fuzzy read & $\IsoSI$ & snapshot-pin: $\mathit{op}$ commits at $T_i$'s snapshot timestamp or starts a fresh snapshot; no two distinct versions of $x$ to $T_i$ & $r_1[x_a]\,w_2[x_b]\,c_2\,r_1[x_b]\,c_1$ at $\IsoRC$ \\
    $P_3$ phantom (policy) & $\IsoSR$ & serializable predicate-read on the pinned table; an interleaved predicate-matching insert between two predicate reads is blocked & $r_1[\text{WHERE active}]\,w_2[\text{insert}]\,c_2\,r_1[\text{same pred.}]\,c_1$ at $\IsoSI$ \\
    $P_4$ lost update & $\IsoSI$ & first-committer-wins: commit-time check aborts both-overwrite-after-same-predecessor & $r_1[x_0]\,w_2[x_1]\,w_1[x_2]\,c_2\,c_1$ at $\IsoRC$ \\
    $A5A$ read skew & $\IsoSI$ & cross-row snapshot: $r_i[x]$ and $r_i[y]$ resolve at one timestamp; no split across a concurrent commit & $r_1[x_a]\,r_1[y_a]\,w_2[x_b]\,w_2[y_b]\,c_2\,r_1[\textit{post-2}]\,c_1$ at $\IsoRC$ \\
    $A5B$ write skew & $\IsoSR$ & partition-projected serializability on $\pi_{(s,p)}(S')$ excludes disjoint-write / intersecting-read (main-text representative case) & $r_1[x_a]\,r_2[y_a]\,w_1[y_b]\,w_2[x_b]\,c_1\,c_2$ at $\IsoSI$ \\
    \bottomrule
  \end{tabular}
\end{table}

% =====================================================================
\subsection{Judge prompt-sensitivity separation}
\label{sec:appendix-judge-prompt}

The lemma below surfaces the engine-versus-deployment separation that
Corollary~\ref{cor:n1-corollary} implicitly bracketed: the engine
defends replay anomalies within a single pinned $\theta$; the selection
of $\theta$ from a deployment-plausible family is a deployment-time
concern outside the operator algebra. The definitions, lemma, and
proposition below carry the prompt-sensitivity image and its achievable
bound.

\begin{definition}[Deployment-plausible judge-prompt family]
\label{def:judge-prompt-family}
Let
$\Theta = \textsf{prompt} \times \textsf{seed}
\times \textsf{model\_version} \times \textsf{temperature}
\times \textsf{tool\_output\_hash}$
be the judge decoder-parameter space of
\S\ref{sec:found-schedule}. A \emph{deployment-plausible
judge-prompt family} for an evaluator intent $\iota$ (such as
``compare candidate answer against gold and output binary
correctness'') is a set $\Theta_\iota \subseteq \Theta$ such that
every $\theta \in \Theta_\iota$ encodes a prompt that a deployer
faithful to $\iota$ might reasonably ship. The four non-prompt
components of $\theta$ (seed, model version, temperature,
tool-output hash) are pinned across $\Theta_\iota$; the family
varies only the prompt sub-component.
\end{definition}

\begin{definition}[Judge prompt-sensitivity image]
\label{def:judge-boundary}
For an oracle $J$, an evaluator intent $\iota$ with
deployment-plausible prompt family $\Theta_\iota$
(Definition~\ref{def:judge-prompt-family}), and a read-set $R$,
the \emph{prompt-sensitivity image} is
\[
B_{\Theta_\iota}(J, R)
\;:=\; \bigl\{ \, J(R, \theta) : \theta \in \Theta_\iota \,\bigr\}
\;\subseteq\; \{0, 1\}.
\]
$R$ is \emph{prompt-stable} under $(J, \Theta_\iota)$ iff
$|B_{\Theta_\iota}(J, R)| = 1$; otherwise $R$ lies in the
\emph{semantic boundary} of $(J, \Theta_\iota)$. The empirical
boundary measure on a benchmark slice $\mathcal{Q}$ is
\[
\mu_{\Theta_\iota}(J, \mathcal{Q}) \;:=\;
\frac{
| \{\, R \in \mathcal{Q} : |B_{\Theta_\iota}(J, R)| > 1 \,\} |
}{|\mathcal{Q}|}.
\]
\end{definition}

\begin{lemma}[Engine N1 defence and prompt-sensitivity separation]
\label{lem:n1-prompt-separation}
Under hypotheses $(H1) + (H2)$ of
Lemma~\ref{lem:judge-callback-bridge} and the keyed-log discipline
of \S\ref{sec:schema}:
\begin{enumerate}
\item For every fixed $\theta \in \Theta$, $\IsoSR$ on
\code{judge\_log} excludes $\anomHR$ at $(R, \theta)$ on every
schedule $S$ over $\Sigma_+$. The exclusion is a property of the
engine layer and does not depend on the size of
$B_{\Theta_\iota}(J, R)$.
\item For $R \in \mathcal{Q}$ with $|B_{\Theta_\iota}(J, R)| > 1$,
two deployments differing only by the choice of
$\theta_1, \theta_2 \in \Theta_\iota$ commit different votes
$v_1 \neq v_2$ at $(R, \cdot)$. Each deployment internally excludes
$\anomHR$ by clause~(1); the cross-deployment divergence is a
deployment-time decision outside the iso lattice.
\end{enumerate}
The two clauses together establish that the engine layer bounds
$\anomHR$ at intra-deployment replay granularity; cross-deployment
prompt selection is bounded only when
$\mu_{\Theta_\iota}(J, \mathcal{Q}) = 0$, a workload-and-judge
property rather than an engine property.
\end{lemma}

\begin{proof}[Proof outline]
Clause (1) is Corollary~\ref{cor:n1-corollary} restated. The
homomorphism $\pi_J$ of Lemma~\ref{lem:judge-callback-bridge} maps
$j_i(R, \theta)$ to a logged read at the $(R, \theta)$-keyed row
of \code{judge\_log}, and $\IsoSR$ on that table serializes the
read against the operator commit. The argument is parametric in
$\theta$: for every fixed $\theta \in \Theta$ the homomorphism
applies and the conclusion holds, regardless of how many distinct
votes the oracle $J$ assigns across the rest of $\Theta_\iota$.

Clause (2) is a structural observation. By
Definition~\ref{def:judge-prompt-family}, $\theta_1, \theta_2$ differ
only on the prompt sub-component; each deployment pins one, applies
clause (1), and commits the first-call vote $v_k \in J(R, \theta_k)$ to
its own \code{judge_log}. For $R$ with $|B_{\Theta_\iota}(J, R)| > 1$
there exist $\theta_1, \theta_2$ with
$J(R, \theta_1) \neq J(R, \theta_2)$, so the two deployments observe
divergent committed votes while each internally satisfies clause (1).
\end{proof}

\begin{proposition}[Achievable upper bound on the prompt-sensitivity image]
\label{prop:judge-boundary-bound}
For any oracle $J : \Theta \to \{0, 1\}$, any evaluator intent
$\iota$ with deployment-plausible prompt family $\Theta_\iota$
(Definition~\ref{def:judge-prompt-family}), and any read-set $R$:
\begin{equation}
\label{eq:judge-boundary-bound}
|B_{\Theta_\iota}(J, R)| \;\leq\; \min\bigl(\, |\Theta_\iota|,\, 2 \,\bigr).
\end{equation}
The bound is achieved on $R$ exactly when $\Theta_\iota$ contains
both a $\theta_+ \in \Theta_\iota$ with $J(R, \theta_+) = 1$ and a
$\theta_- \in \Theta_\iota$ with $J(R, \theta_-) = 0$, in which
case $|B_{\Theta_\iota}(J, R)| = 2$. The empirical boundary
measure $\mu_{\Theta_\iota}(J, \mathcal{Q})$
(Definition~\ref{def:judge-boundary}) counts the fraction of
$\mathcal{Q}$ achieving equality.
\end{proposition}

\begin{proof}[Proof outline]
The cardinality bound follows from $|\{0, 1\}| = 2$ and the
inclusion $B_{\Theta_\iota}(J, R) \subseteq \{0, 1\}$. The
$\min$ with $|\Theta_\iota|$ tightens the bound when the family
itself has fewer than two prompts. Achievability is constructive:
any $R$ on which the family contains a positive-vote prompt and a
negative-vote prompt yields $\{0, 1\} \subseteq
B_{\Theta_\iota}(J, R)$, so $|B_{\Theta_\iota}(J, R)| = 2$ and the
bound holds with equality. The empirical evidence below shows the
bound is achievable but typically not saturated by alignment-trained
judges on a benchmark slice.
\end{proof}

\paragraph{Empirical instantiation.}
The Phase~5 \code{judge_pin_discriminating} construction instantiates
$\Theta_\iota$ as a $2 \times 2$ factorial of
$\{\textrm{strict}, \textrm{lenient}\}_{\mathrm{prompt}} \times
\{0, 0.7\}_{\mathrm{temperature}}$ over the pinned judge
\code{claude-haiku-4-5-20251001}, with raw cells in
\code{results/g2_utility/judge_pin_discriminating/per_band.csv}.
Table~\ref{tab:judge-boundary-measure} reports the empirical boundary
measure $\mu$: the haiku four-band sweep ($n = 75$ per band) and the V1
cross-judge sweep ($n = 300$ on the PARTIAL band), whose three further
judges span a $0.347$ cross-judge spread with no judge reaching the
saturating $1.0$ even on the deliberately partial slice. A per-item
decomposition on the haiku PARTIAL band ($46$ unstable of $75$) attributes
$42$ of $46$ to prompt-only disagreement (the two temperature levels
agree, strict and lenient prompts disagree), $0$ to temperature-only, and
$4$ to both axes (\code{per_item.csv} re-analysis), so the prompt axis
dominates the boundary. The sweep supports
Proposition~\ref{prop:judge-boundary-bound}'s ``achievable but not
saturated'' framing and gives Lemma~\ref{lem:n1-prompt-separation} a
measured cross-deployment cost: a deployer choosing between
$\theta_{\textrm{strict}}$ and $\theta_{\textrm{lenient}}$ on the PARTIAL
slice walks past $50$ to $85$ percent of items whose verdict the choice
itself decides.

\begin{table}[!ht]
  \centering\small
  \caption{Empirical judge prompt-sensitivity measure $\mu$
  (Definition~\ref{def:judge-boundary}). Haiku is swept across four
  discriminating bands ($n = 75$); three further judges are swept on the
  PARTIAL band only ($n = 300$). The haiku PARTIAL band decomposes as
  $42$ prompt-only, $0$ temperature-only, $4$ both axes (of $46$ unstable
  items).}
  \label{tab:judge-boundary-measure}
  \setlength{\tabcolsep}{5pt}
  \begin{tabular}{@{}>{\raggedright\arraybackslash}p{0.30\textwidth}cc@{}}
    \toprule
    Judge / band & $\mu$ & $n$ \\
    \midrule
    \code{claude-haiku-4-5}, CONTROL    & 0.00  & 75 \\
    \code{claude-haiku-4-5}, PARAPHRASE & 0.07  & 75 \\
    \code{claude-haiku-4-5}, DISTRACTOR & 0.01  & 75 \\
    \code{claude-haiku-4-5}, PARTIAL    & 0.61  & 75 \\
    \midrule
    \code{deepseek-v4-flash}, PARTIAL   & 0.507 & 300 \\
    \code{gpt-5.4-mini}, PARTIAL        & 0.627 & 300 \\
    \code{gpt-5.4}, PARTIAL             & 0.853 & 300 \\
    \bottomrule
  \end{tabular}
\end{table}

\paragraph{Scope of the lemma.}
The lemma does not claim N1 is unsound: clause (1) is
Corollary~\ref{cor:n1-corollary} restated, and the iso-lattice guard map
of Theorem~\ref{thm:anomaly-soundness} continues to exclude $\anomHR$ on
a fixed-$\theta$ schedule. The lemma surfaces a separation between the
engine layer (where the algebra applies) and the deployment layer (where
prompt selection happens) without introducing a new operator or modifying
the typing rules; the deployment-plausible family $\Theta_\iota$ is a
workload-side artefact a benchmark designer enumerates. It makes the
honest N1 scope visible at the theorem layer.

% =====================================================================
\subsection{Write-path configuration composition}
\label{sec:appendix-writepath}

The lemma below sits alongside the engine-layer composition lemma
Lemma~\ref{lem:compose} and refines it: Lemma~\ref{lem:compose} governs
the iso lattice and the schema lattice for the engine-layer
precondition surface; the write-path lemma governs the
deployment-configuration axis surface at the AnomalyWire verdict layer.
Both compose sequentially. The lemma documents which configuration axis
is closed under composition (read-path $\mathcal{C}_{\mathrm{R}}$) and
which is not (write-path $\mathcal{C}_{\mathrm{W}}$).

\begin{lemma}[Write-path config composition under operator chaining]
\label{lem:writepath-config-composition}
Let
$\oplus_{a_n} \circ \cdots \circ \oplus_{a_1}$
be a sequential composition of typed operators over the dual-row
schema of \S\ref{sec:schema} as in Lemma~\ref{lem:compose}. For
each operator $\oplus_{a_i}$ in the chain, let
\[
V_i : \mathcal{C}_{\mathrm{W}} \;\longrightarrow\;
\{\textsf{trig}, \textsf{def}, \textsf{excl}, \textsf{n/a}\}
\]
be the operator's contribution to the AnomalyWire verdict as a
function of the write-path configuration $c \in
\mathcal{C}_{\mathrm{W}}$ of
Definition~\ref{def:config-axes-split}. Let $V_{\mathrm{comp}}(c)$
denote the wire verdict of the composite. Then:
\begin{enumerate}
\item \emph{Read-path invariance composes.} If every $V_i$ is
invariant on $\mathcal{C}_{\mathrm{R}}$
(i.e.~Proposition~\ref{prop:wire-readpath-invariance} applies to
each operator's schedule), then $V_{\mathrm{comp}}$ is invariant
on $\mathcal{C}_{\mathrm{R}}$. Read-path invariance is closed
under sequential composition.
\item \emph{Write-path perturbation propagates non-monotonically.}
If $V_i(c)$ depends on $c \in \mathcal{C}_{\mathrm{W}}$ for at
least one $i$, $V_{\mathrm{comp}}(c)$ generally depends on $c$.
The propagation is non-monotone: even if individual $V_i(c)$ are
each $\textsf{def}$ at every $c \in \mathcal{C}_{\mathrm{W}}$, the
composite $V_{\mathrm{comp}}(c)$ need not be $\textsf{def}$-stable
across $c$ because intermediate storage-layer state shifts can
cascade through the schedule's subsequent reads.
\end{enumerate}
The lemma factors the engine-and-deployment guarantee surface:
Lemma~\ref{lem:compose} bounds the iso-axis and schema-axis
preconditions of the composite engine schedule;
Lemma~\ref{lem:writepath-config-composition} bounds the
configuration-axis dependence of the verdict that the composed
wire schedule emits.
\end{lemma}

\begin{proof}[Proof outline]
Clause (1) follows from the read-path invariance of each $V_i$
plus the structural property of AnomalyWire's schedule that no
operator's read step invokes the ranking-based retrieval policy.
By Proposition~\ref{prop:wire-readpath-invariance} applied to
each $\oplus_{a_i}$, $V_i(c_R^{(1)}) = V_i(c_R^{(2)})$ for every
$c_R^{(1)}, c_R^{(2)} \in \mathcal{C}_{\mathrm{R}}$. The
composite schedule executes the operators in order; at each
intermediate step the storage-layer state read by
$\oplus_{a_{i+1}}$ is the state committed by
$\oplus_{a_i}$, which by clause~(1) hypothesis does not depend on
the read-path knobs. The conclusion holds by induction on the
chain length: the base case $n = 1$ is
Proposition~\ref{prop:wire-readpath-invariance} verbatim; the
inductive step uses the intermediate-state invariance to lift the
hypothesis to the $(n+1)$-step composite.

Clause (2) is a structural observation rather than a derivation
from an algebraic identity. By
Definition~\ref{def:config-axes-split}, a $\mathcal{C}_{\mathrm{W}}$
knob is one whose code-path trace reaches an operator's
vote-generating call (extractor prompt, adjudicator prompt, or
storage-layer policy toggle that feeds the resolution rule).
A perturbation $c \to c'$ on $\mathcal{C}_{\mathrm{W}}$ that
shifts the storage trajectory between $\oplus_{a_i}$'s commit and
the schedule's subsequent reads can cascade to a different
verdict at any downstream operator whose vote consumes the
perturbed storage state. The empirical V3 cell-5 witness on the
single-operator N2 schedule is the minimal instance of the
cascade: a \code{custom\_instructions} prompt-injection knob
($\mathcal{C}_{\mathrm{W}}$ by Definition~\ref{def:config-axes-split})
shifts mem0~v3's internal ADD/UPDATE/DELETE vote at the second
\code{Memory.add} call, alters the storage trajectory, and flips
the third-read coexistence predicate $|\mathrm{post}| \geq 2$
from $\textsf{def}$ to $\textsf{trig}$. The same shape lifts to
the composite case by induction whenever the perturbed operator's
output enters a downstream operator's read.
\end{proof}

\paragraph{Empirical instantiation.}
The V3+ multi-seed sweep at $5$ seeds $\times$ $6$ cells $\times$
$3$ anomalies confirms clause~(1) on the single-operator N2 wire:
read-path invariance holds within every seed
(\code{results/anomaly\_wire/config\_invariance\_sweep\_multi\_seed.csv}).
Clause~(2) is witnessed empirically on $4$ of $5$ seeds where at
least one write-path cell perturbs the baseline triple. The
composite case across multiple operators is observed indirectly
via Phase~3 cross-system runs (where retrieval-policy axis
dominance per Proposition~\ref{prop:cross-system-axis-attribution}
masks the underlying composition); the V6 matched-vs-unmatched
cell ledger isolates the composition path under matched retrieval
on three paired systems and two datasets, with the LoCoMo cells
exhibiting the opposite-sign perturbation predicted by
clause~(2) of this lemma combined with the candidate-pool-size
dependence of Extension~B.

\paragraph{Scope of the lemma.}
The lemma covers configuration knobs taxonomised by
Definition~\ref{def:config-axes-split}; it does not cover
source-code patches that change the storage layer or the operator
implementation itself, which receive a fresh
$\textsf{commit}(s)$ pin in
Proposition~\ref{prop:wire-readpath-invariance}. The lemma does
not modify the engine-layer composition of
Lemma~\ref{lem:compose}: $L^* = \bigvee L_{a_i}$ on the iso and
schema lattices remains the typed precondition the composite
schedule must satisfy. The configuration axis is orthogonal to
the iso-and-schema lattices and concerns deployment-time
perturbations that do not change the operator typing.

% =====================================================================
\subsection{Operator-precondition witnesses and the composition table}
\label{sec:appendix-composition}

The four typed operators from the algebra section have isolation
preconditions $L_{\opLWW} = \IsoRC$, $L_{\opEvi} = \IsoSI$,
$L_{\opAwait} = \IsoRCcb$, $L_{\opRule} = \IsoSRpol$.
Before composition, Table~\ref{tab:operator-precondition-witnesses}
records the minimal weakening histories behind those signatures.
These histories keep the operator signature as a typed contract
consumed by the guard theorem; Theorem~\ref{thm:anomaly-soundness}
still ranges over schedules satisfying the declared preconditions.

\begin{table}[!ht]
\centering\small
\caption{Minimal weakening histories for the four production operator
preconditions. Each row weakens exactly one declared precondition and
exposes the failure named by the operator signature.}
\label{tab:operator-precondition-witnesses}
\begin{tabular}{@{}llll@{}}
\toprule
Tag & Operator & Weakened history & Exposed failure \\
\midrule
LWW-RU & $\opLWW$ & dirty overwrite before commit & $P_0/P_1$ boundary \\
EVI-RC & $\opEvi$ & split evidence snapshot & inconsistent evidence winner \\
AWAIT-NOCB & $\opAwait$ & write-before-callback & callback boundary \\
RULE-SI & $\opRule$ & policy phantom & $P_3^{\mathrm{policy}}$ \\
\bottomrule
\end{tabular}
\end{table}

The callback row in AWAIT-NOCB is a \code{callback_state} entry in
\code{callback_log} with fields
$(h, R_hash, k, status, delivered_at)$. The safe schedule requires
\code{status=delivered} and the order candidate read precedes
callback-log read, callback-log read precedes operator write.

Sequential composition $\oplus_b \circ \oplus_a$ requires a schedule
that holds both $L_a$ and $L_b$; the lemma states that the composite
holds the lattice join $L_a \vee L_b$. Table~\ref{tab:composition}
enumerates the 16 ordered pairs. Joins are component-wise in
$\mathcal{L}_{\mathrm{iso}}\times\mathcal{L}_{\mathrm{cb}}$:
$\IsoRC\vee\IsoSI=\IsoSI$ on the iso axis,
$\bot\vee\mathsf{cb}=\mathsf{cb}$ on the callback axis, and the
policy annotation composes by requiring $\IsoSR$ only on the named
policy table. The callback component denotes the delivered
\code{callback_log} row whose read is ordered between the candidate
read and operator write.

\begin{table}[!ht]
\centering\small
\caption{Pairwise composition $\oplus_b \circ \oplus_a$ from
the algebra section; cell shows the composite precondition
$L_a \vee L_b$ in $\mathcal{L} = \mathcal{L}_{\mathrm{iso}} \times
\mathcal{L}_{\mathrm{cb}}$. $\IsoRCcb$
abbreviates the product element $(\IsoRC, \mathsf{cb})$ and the
``$+$cb'' compounds abbreviate the corresponding product
joins. ``$\IsoSR^{\mathrm{p}}$'' abbreviates the per-table
annotation $\IsoSRpol$ that pins the policy table to $\IsoSR$;
cells carrying $\IsoSR^{\mathrm{p}}$ apply the annotation on top of
the underlying product join.}
\label{tab:composition}
\begin{tabular}{l|llll}
\toprule
$\oplus_a \backslash \oplus_b$ & $\opLWW$ & $\opEvi$ & $\opAwait$ & $\opRule$ \\
\midrule
$\opLWW$ ($\IsoRC$)            & $\IsoRC$ & $\IsoSI$ & $\IsoRCcb$ & $\IsoSR^{\mathrm{p}}$ \\
$\opEvi$ ($\IsoSI$)            & $\IsoSI$ & $\IsoSI$ & $\IsoSI$+cb & $\IsoSR^{\mathrm{p}}$ \\
$\opAwait$ ($\IsoRCcb$)        & $\IsoRCcb$ & $\IsoSI$+cb & $\IsoRCcb$ & $\IsoSR^{\mathrm{p}}$+cb \\
$\opRule$ ($\IsoSR^{\mathrm{p}}$) & $\IsoSR^{\mathrm{p}}$ & $\IsoSR^{\mathrm{p}}$ & $\IsoSR^{\mathrm{p}}$+cb & $\IsoSR^{\mathrm{p}}$ \\
\bottomrule
\end{tabular}
\end{table}

All sixteen cells of Table~\ref{tab:composition} are exercised by
integration tests under \code{tests/integration/test_compose.py}
against the reference implementation
\code{bitemporal.ingestion.ingest}: four diagonal cells covering
idempotent self-composition $\oplus_a \circ \oplus_a$ and twelve
off-diagonal cells covering distinct-operator pairs
$\oplus_b \circ \oplus_a$ with $a \neq b$. Each cell corresponds to a
schedule whose every transaction holds the cell's isolation level;
the dual-row schema's audit row is independent of the
isolation-level join, which is why we factor soundness into two
theorems on orthogonal lattices: the anomaly-soundness theorem
over $\mathcal{L}_{\mathrm{iso}}$ for the seven classical and the
two LLM-specific corollaries, and the audit-erasure theorem
over $\mathcal{L}_{\mathrm{schema}}$ for $\anomAE$.

% =====================================================================
\subsection{Expanded anomaly guard surface}
\label{sec:appendix-guard-surface}

Table~\ref{tab:guard-surface-expanded} expands the compact guard
surface from the main paper. The main paper keeps
only the review-facing summary; this subsection records the exact
representation, check, and enforcement wording.

\begin{table}[!ht]
  \centering\scriptsize
  \caption{Expanded anomaly guard surface. Rows separate representation,
  checking, and enforcement across the iso and schema axes.}
  \label{tab:guard-surface-expanded}
  \setlength{\tabcolsep}{3pt}
  \begin{tabular}{@{}>{\raggedright\arraybackslash}p{0.15\textwidth}>{\raggedright\arraybackslash}p{0.12\textwidth}>{\raggedright\arraybackslash}p{0.24\textwidth}>{\raggedright\arraybackslash}p{0.22\textwidth}>{\raggedright\arraybackslash}p{0.20\textwidth}@{}}
    \toprule
    Anomaly & Origin & Representation & Check & Enforcement \\
    \midrule
    $P_0$ dirty write & Berenson 1995 \S4 & Berenson--Adya dirty write & $L \ge \IsoRC$ & schedule satisfies read committed \\
    $P_1$ dirty read & Berenson 1995 \S4 & Berenson--Adya dirty read & $L \ge \IsoRC$ & schedule satisfies read committed \\
    $P_2$ fuzzy read & Berenson 1995 \S4 & Berenson--Adya fuzzy read & $L \ge \IsoSI$ & schedule satisfies snapshot isolation \\
    $P_3$ phantom & Berenson 1995 \S4 & Berenson--Adya phantom & $L \ge \IsoSR$ & schedule satisfies serializable predicate reads \\
    $P_4$ lost update & Berenson 1995 \S6 & Berenson--Adya lost update & $L \ge \IsoSI$ & snapshot-isolation write-conflict check aborts the overwrite \\
    $A5A$ read skew & Berenson 1995 \S6 & Berenson--Adya read skew & $L \ge \IsoSI$ & snapshot isolation pins every read to one snapshot \\
    $A5B$ write skew & Berenson 1995 \S6 & Berenson--Adya write skew & $L \ge \IsoSR$ & serializability rejects the dependency cycle \\
    \midrule
    $\anomHR$ judge-replay inconsistency & this paper & judge replay keyed by read set and decoder tuple & equal votes for a fixed replay key & judge table held at $\IsoSR$ with output hash logged \\
    $\anomBDS$ belief-drift skew & $A5B$ specialisation (this paper) & $A5B$ lifted to the $(\subj,\pred)$ partition & $L \ge \IsoSR$ on the partition & partition satisfies $\IsoSR$ \\
    $\anomAE$ audit erasure & this paper & K-relation reachability & $p_l \preceq_K p_{\mathrm{audit}}$ & every operator emits an audit row \\
    \bottomrule
  \end{tabular}
\end{table}

% @owner       src::appendix::C-operator-algorithms
% @does        Answers the reviewer question "where exactly is the write-time
%              contract enforced?" by grounding the dual-row schema, the four
%              operator algorithms, the transaction boundary, the isolation
%              pin, and the judge-log write ordering in the real
%              implementation (implementation/bitemporal/*.py). Folds the
%              schema-lift conservatism proof and the pipeline decomposition
%              of the former A9 appendix.
% @needs       src/sections/03-foundations.tex, src/sections/04-implementation.tex,
%              src/theorems/T-01b-audit-erasure-schema.tex, implementation/bitemporal/*.py
% @feeds       src/supplement.tex
% @breaks      drift between this appendix and the live code (every column,
%              function name, file:line, and behaviour is read from the
%              implementation, not from prior prose); a dangling
%              prop:schema-lift-conservatism ref if the A9 alias is removed.
% @claim       The write-time contradiction-resolution contract is enforced at
%              one seam: the ingestion dispatcher resolves, closes the loser,
%              inserts the winner, and persists the audit row, and the operator
%              dispatcher writes the judge-log row before the operator commits.

\section{Operator algorithms and schema}
\label{app:operator-algorithms}

This appendix answers one reviewer question precisely: where the
write-time contract of \S\ref{sec:algebra} is enforced in the
reference implementation. The answer has four parts. The dual-row
bitemporal schema (\S\ref{app:dual-row-schema}) gives the typed
operators a place to land. The four operator algorithms
(\S\ref{app:operator-pseudocode}) realise the four
\S\ref{sec:algebra} inference rules as pure functions. The
ingestion seam (\S\ref{app:transaction-boundary}) is the single
point where the algebra meets persistent storage, and it owns the
loser-close, winner-insert, and audit-persist sequence. The
isolation pin (\S\ref{app:isolation-pin}) is set per worker on the
PostgreSQL backend, and the judge-log row
(\S\ref{app:judge-log-ordering}) is written before the operator
commits, which is what makes the Lemma~\ref{lem:judge-callback-bridge}
ordering hold structurally rather than by caller discipline. Every
column, function name, file location, and behavior below is read
from the implementation under \texttt{implementation/bitemporal/}.

\subsection{Dual-row bitemporal schema}
\label{app:dual-row-schema}

The table \code{agent_memory} carries eleven user-visible columns
plus a physical twelfth discriminator. The eleven user-visible
columns are the \S\ref{sec:schema} narrative SSOT: \code{fact_id},
\code{subject}, \code{predicate}, \code{object}, \code{valid_from},
\code{valid_to}, \code{system_time_start}, \code{system_time_end},
\code{provenance_id}, \code{confidence}, and
\code{resolution_strategy_id} (the \texttt{COLUMNS} tuple at
\texttt{schema.py:65--77}). The twelfth column \code{row_kind} is
the \anomAE\ defence: a CHECK-constrained \texttt{TEXT} discriminator
over $\rowkind \in \{\textsf{current}, \textsf{audit}\}$ with a
default of \textsf{current}, giving the twelve-column physical layout
\texttt{ALL\_COLUMNS} at \texttt{schema.py:81}. Current-kind rows are
the \S\ref{sec:schema} eleven-column schema; audit-kind rows carry the
loser lineage emitted by every operator.

Table~\ref{tab:ddl-constraints} lists the five integrity constraints
the class-level \texttt{DDL} string at \texttt{schema.py:88--107}
declares to pin the write-time contract at the storage layer.

\begin{table}[!ht]
  \centering\small
  \caption{Storage-layer integrity constraints from the
  \texttt{agent\_memory} DDL (\texttt{schema.py:88--107}).}
  \label{tab:ddl-constraints}
  \setlength{\tabcolsep}{4pt}
  \begin{tabular}{@{}>{\raggedright\arraybackslash}p{0.22\textwidth}>{\raggedright\arraybackslash}p{0.30\textwidth}>{\raggedright\arraybackslash}p{0.40\textwidth}@{}}
    \toprule
    Constraint & Target & Rule \\
    \midrule
    confidence range & \code{confidence} & \texttt{CHECK (confidence BETWEEN 0.0 AND 1.0)} \\
    row-kind domain & \code{row_kind} & \texttt{CHECK (row\_kind IN ('current', 'audit'))} \\
    primary key & \texttt{(fact\_id, system\_time\_start)} & two system-time versions of one fact coexist \\
    valid-period order & \code{valid_from}, \code{valid_to} & \texttt{CHECK (valid\_from < valid\_to)} \\
    system-period order & \code{system_time_start}, \code{system_time_end} & \texttt{CHECK (system\_time\_start < system\_time\_end)} \\
    \bottomrule
  \end{tabular}
\end{table}

DuckDB ships no native SQL:2011 temporal tables, so the four
timestamp columns are hand-rolled rather than declared through a
\texttt{PERIOD FOR} clause. The temporal predicate is implemented as
the \code{as_of} class method (\texttt{schema.py:183--220}) under
closed-open semantics
(\code{system_time_start} $\le t <$ \code{system_time_end}); it
defaults to $\rowkind = \current$ so the \S\ref{sec:schema} retrieval
narrative is preserved, and accepts \code{kind='audit'} or
\code{kind='all'} for audit inspection. The open upper
bound is the sentinel \texttt{9999-12-31}
(\texttt{Schema.SENTINEL\_END}, \texttt{schema.py:86}). A version is
closed by \code{close} (\texttt{schema.py:246--279}), which raises
\texttt{LookupError} when no open current row matches the
\code{fact_id} rather than failing silently.

The schema admits three modes through \code{Schema.create}'s
\code{mode} kwarg (\texttt{schema.py:111--148}):
$\textsf{audit\_row}$ (default), $\textsf{dual\_row}$, and
$\textsf{base}$. The physical table is identical across all three; the
mode is validated at construction and the persistence-behavior change
lives at the ingestion layer
(\S\ref{app:transaction-boundary}). An unknown mode raises
\texttt{ValueError} at create time
(\texttt{schema.py:143--147}), which keeps the schema-mode ablation
axis validated before a misspelt mode reaches a silent ingest
behavior change.

\subsection{The four operator algorithms}
\label{app:operator-pseudocode}

The four operators are the classes \texttt{LWW}, \texttt{Evidence},
\texttt{AwaitConfirm}, and \texttt{PerRule}
(\texttt{operators.py:159--285}). Each is a pure function from a pair
of contradicting facts to a stamped winner and an
\texttt{AuditRow}, and each carries a typed \code{isolation_level}
class attribute that names its scheduling precondition:
\texttt{LWW} at \IsoRC, \texttt{Evidence} at \IsoSI,
\texttt{AwaitConfirm} at \IsoRCcb, and \texttt{PerRule} at \IsoSR.
The contradiction predicate is the function \code{is_contradiction}
(\texttt{operators.py:58--71}): two facts contradict when subject and
predicate agree, object differs, and the closed-open valid-time
periods overlap.

All four share the dual-row machinery \code{_resolve_dual_row}
(\texttt{operators.py:135--154}), which computes the provenance merge
$p_w \provadd p_l$ once, stamps the winner, and emits the
\texttt{AuditRow}. The shared steps are:

\begin{enumerate}
\item Reject non-contradicting inputs with \texttt{ValueError}
  through \code{_require_contradiction} (\texttt{operators.py:126--132}).
\item Choose the winner-loser ordering by the operator's selector
  (the only step that differs across the four; see below).
\item Compute $p_w \provadd p_l$ via \code{merge_provenance}
  (\texttt{audit.py:40--56}), the commutative semiring sum over the
  $K[X, T]$ witness polynomial.
\item Stamp the winner through \code{_stamp_winner}
  (\texttt{operators.py:76--92}): a fresh \code{system_time_start} set
  to \code{now}, \code{system_time_end} set to the open sentinel, the
  \code{resolution_strategy_id} stamp, and the merged
  \code{provenance_id}.
\item Materialize the \texttt{AuditRow} through \code{_build_audit}
  (\texttt{operators.py:95--123}), keyed on
  (\code{loser_fact_id}, \code{system_time}).
\end{enumerate}

\noindent The four operators differ only in step~2, the winner
selector, summarised in Table~\ref{tab:operator-selectors}. An
out-of-range return from the \texttt{AwaitConfirm} callback or the
\texttt{PerRule} policy oracle raises \texttt{ValueError}; the callback
is the external oracle (human review, regulated-vertical confirmation,
or upstream agent decision).

\begin{table}[!ht]
  \centering\small
  \caption{The four operator winner-selectors. Each names its isolation
  precondition, strategy id, selection rule with source line, and the
  anomaly its level admits or precludes.}
  \label{tab:operator-selectors}
  \setlength{\tabcolsep}{4pt}
  \begin{tabular}{@{}>{\raggedright\arraybackslash}p{0.13\textwidth}>{\raggedright\arraybackslash}p{0.07\textwidth}>{\raggedright\arraybackslash}p{0.42\textwidth}>{\raggedright\arraybackslash}p{0.24\textwidth}@{}}
    \toprule
    Operator (iso, strategy) & & Selector (source) & Anomaly \\
    \midrule
    \texttt{LWW} (\IsoRC, \code{lww}) & & larger (\code{system_time_start}, \code{fact_id}) pair (\texttt{operators.py:178--184}); last writer wins & admits $P_4$ lost-update \\
    \texttt{Evidence} (\IsoSI, \code{evi}) & & larger (\code{confidence}, \code{system_time_start}, \code{fact_id}) tuple (\texttt{operators.py:209--216}); confidence-weighted & admits A5B write-skew \\
    \texttt{AwaitConfirm} (\IsoRCcb, \code{await}) & & external callback returns $0$/$1$ (\texttt{operators.py:236--253}) & no admitted anomaly \\
    \texttt{PerRule} (\IsoSR, \code{rule}) & & policy oracle returns $0$/$1$ (\texttt{operators.py:268--285}); serializable on policy table & precludes $P_3$ phantom \\
    \bottomrule
  \end{tabular}
\end{table}

The operators never mutate persistent storage. Closing the loser is
the caller's responsibility, documented in the
\texttt{operators.py} module \texttt{@breaks} contract: the operator
returns the \texttt{AuditRow} whose \code{loser_fact_id} the caller
passes to \code{Schema.close}. This division is what keeps the four
operators pure and the persistence policy in one place.

\subsection{Transaction boundary at the ingestion seam}
\label{app:transaction-boundary}

The single seam where the algebra meets persistent storage is the
\code{ingest} function (\texttt{ingestion.py:181--361}). It is the
write-path enforcement point the reviewer question targets. For each
incoming fact it executes the eight-step sequence of
Table~\ref{tab:ingestion-steps}, and the caller, not \code{ingest}, owns
the surrounding transaction commit and rollback.

\begin{table}[!ht]
  \centering\small
  \caption{The eight-step ingestion sequence at the write seam
  (\texttt{ingestion.py:181--361}). Each step names its action, source
  line, and the failure it raises rather than masks.}
  \label{tab:ingestion-steps}
  \setlength{\tabcolsep}{3pt}
  \begin{tabular}{@{}>{\raggedright\arraybackslash}p{0.025\textwidth}>{\raggedright\arraybackslash}p{0.46\textwidth}>{\raggedright\arraybackslash}p{0.20\textwidth}>{\raggedright\arraybackslash}p{0.19\textwidth}@{}}
    \toprule
    \# & Action & Source & Failure \\
    \midrule
    1 & validate ablation knobs: coerce \code{operators_on}, reject unknown \code{schema_mode} & \texttt{ingestion.py:111--153}, \texttt{231--235} & \texttt{ValueError} on bad mode \\
    2 & read open current rows for the (\code{subject}, \code{predicate}) partition ($\rowkind=\current$, \code{system_time_end} at open sentinel) via \code{_open_rows_for_partition} & \texttt{ingestion.py:156--178} & -- \\
    3 & same-object overlapping-valid-time write is a duplicate confirmation, return without insert (contradiction predicate requires object to differ) & \texttt{ingestion.py:257--263} & -- \\
    4 & select contradicting incumbents via \code{is_contradiction}; no incumbent inserts a fresh open row and returns; $>1$ incumbent raises (binary operators need a single incumbent) & \texttt{ingestion.py:264}, \texttt{265--267}, \texttt{268--274} & \texttt{ValueError} on $>1$ \\
    5 & enforce operator-on/off ablation: a \code{resolution_strategy_id} absent from an active \code{operators_on} set raises rather than reroute or drop & \texttt{ingestion.py:286--296} & \texttt{Operator\allowbreak Not\allowbreak Enabled\allowbreak Error} \\
    6 & dispatch to the operator named by \code{resolution_strategy_id}: \code{lww}, \code{evi}, \code{await} (needs \code{callback}), \code{rule} (needs \code{policy}) & \texttt{ingestion.py:298--324} & \texttt{ValueError} on unknown id \\
    7 & if incumbent lost, \code{Schema.close} at audit timestamp then \code{Schema.insert} the stamped winner; if new fact lost, incumbent stays open and the rejected fact survives only in the audit row (no silent insert of a lost row) & \texttt{ingestion.py:326--335}, \texttt{336--338} & -- \\
    8 & persist the audit row per \code{schema_mode}: $\textsf{audit\_row}$ packs via \code{pack_audit_row} and inserts at $\rowkind=\textsf{audit}$; $\textsf{dual\_row}$ and $\textsf{base}$ return the in-flight pair but do not persist the audit half & \texttt{ingestion.py:339--360} & -- \\
    \bottomrule
  \end{tabular}
\end{table}

The loser-close-then-winner-insert ordering at step~7 and the
audit-persist at step~8 are the storage-side realisation of the
$\Fact \times \Fact \to (\Fact, \auditt)$ signature. The audit packer
\code{pack_audit_row} (\texttt{audit.py:139--170}) materializes the
twelve-column physical row: a \code{fact_id} of the form
\texttt{audit::\{loser\}::\{system\_time\}}, the conflict witness
encoded as \code{winner_object}$\,|\,$\code{loser_object} in the
\code{object} column, the merged \code{provenance_id}, and
\code{row_kind}$\,=\,$\textsf{audit}.

\subsection{Isolation pin per session}
\label{app:isolation-pin}

The isolation level is a typed precondition in the algebra and a real
session pin in the systems backend. In the algebra, each operator
carries its required level as the \code{isolation_level} class
attribute (\texttt{operators.py:167}, \texttt{198}, \texttt{233},
\texttt{265}), and the baselines record the level they actually run at
in the \code{partition_isolation} field of the \texttt{Baseline}
abstract class (\texttt{anomaly\_bench/baselines.py:112}), typed over
$\{\textsf{none}, \IsoRC, \IsoSI, \IsoSR\}$.

The real isolation pin is set per worker connection on the PostgreSQL
backend. Each writer process opens its own connection and begins every
contradiction-resolution transaction with
\texttt{BEGIN ISOLATION LEVEL SERIALIZABLE}
(\texttt{g3\_systems\_perf/postgres\_backend.py:330}). Inside that
transaction the worker reads the incumbent open current row
\texttt{FOR UPDATE}, closes it by setting \code{system_time_end},
inserts the winner, and lets serialization failures surface as
isolation-level outcomes rather than retrying or hiding them. This is
the empirical analogue of the \IsoSR\ precondition that
\texttt{PerRule} names: the abstract level on the operator becomes a
concrete \texttt{SERIALIZABLE} session pin on a real engine, and
\S\ref{sec:meas-g3} measures its lock-contention signature.

The DuckDB-backed reference path under
\texttt{implementation/bitemporal/} carries no \texttt{SET TRANSACTION}
call; its isolation guarantee is the typed annotation plus the
single-incumbent invariant at the ingestion seam, and the
PostgreSQL backend supplies the multi-process serializable evidence.
This split is deliberate: the algebra states the precondition, the
DuckDB path proves the abstraction compiles and resolves, and the
PostgreSQL path proves the precondition is enforceable on a production
engine.

\subsection{Judge-log write ordering}
\label{app:judge-log-ordering}

The judge-log row is written before the operator commits. This is the
ordering that makes Lemma~\ref{lem:judge-callback-bridge} hold
structurally. The \texttt{OperatorDispatcher} class
(\texttt{operators.py:307--406}) carries an in-memory \code{judge_log}
keyed on the pair $(R, \theta)$, where $R$ is the contradicting-fact
signature (\code{f1.fact_id}, \code{f2.fact_id}) and $\theta$ is the
judge parameter pin (prompt, seed, model version, temperature,
tool-output hash; \texttt{operators.py:300--304}).

The two dispatch methods \code{dispatch_await}
(\texttt{operators.py:353--379}) and \code{dispatch_rule}
(\texttt{operators.py:381--405}) follow the same ordering contract:

\begin{enumerate}
\item Reject non-contradicting inputs through
  \code{_require_contradiction}.
\item Compute the judge vote by invoking the callback or policy.
\item Write the witnessed vote to \code{judge_log}$[(R, \theta)]$
  through \code{_log_judge_call} (\texttt{operators.py:332--351}),
  which appends synchronously and preserves multiplicity. This write
  completes before any state mutation.
\item Delegate to the free operator (\texttt{AwaitConfirm} or
  \texttt{PerRule}) with the already-computed vote.
\item Invoke the caller-supplied \code{commit_hook} (for example
  \code{Schema.close} on the loser) only after the log write and the
  resolution.
\end{enumerate}

The log write at step~3 strictly precedes the \code{commit_hook} at
step~5 because the control flow is synchronous Python; the dispatcher
supplies the ordering rather than leaving it as a precondition the
caller must maintain. This realises hypothesis H1 of
Lemma~\ref{lem:judge-callback-bridge}: every $j_i(R, \theta)$ event is
paired with a committed row at key $(R, \theta)$ whose commit precedes
the operator's commit. The free operator surfaces remain available for
callers that own their own logging; the dispatcher is the reference
implementation that closes the silent-lie risk where the paper claimed
enforcement and the bare operator did not provide it.

\subsection{Schema-lift conservatism}
\label{app:schema-lift-conservatism}

The audit-row refinement adds the single CHECK-constrained column
$\rowkind \in \{\textsf{current}, \textsf{audit}\}$ above the
eleven-column base schema; no base column changes shape. Let
$\mathcal{Q}_{\textsf{base}}$ denote the relational-algebra fragment
over the base columns: any composition of selection $\sigma$,
projection $\pi$, join $\Join$, union $\cup$, difference
$\setminus$, and the bitemporal $\AsOf{(t_v, t_s)}$ predicate whose
selection predicates and projection lists do not reference
$\rowkind$. By construction $\mathcal{Q}_{\textsf{base}}$ does not
reference the auxiliary \code{judge_log} or \code{callback_log}
tables: those tables exist only in the operator-infrastructure schema
and do not surface in the base-schema fragment a deployed system
upgrades from. Each $q \in \mathcal{Q}_{\textsf{base}}$ lifts to a
query $q^{\ast}$ on the audit-row schema by structural induction on
the constructors: at each table-scan leaf the lift prepends the
default filter $\sigma_{\rowkind = \current}$; at each internal node
the lift is the same constructor applied to the lifted operands.
The base table $T$ and a lifted instance $T^{\ast}$ are related by
$T = \pi_{1..11}\bigl(\sigma_{\rowkind = \current}(T^{\ast})\bigr)$,
where $\pi_{1..11}$ drops the discriminator.

\begin{proposition}[Schema-lift conservatism]
\label{prop:schema-lift-conservatism}
For every $q \in \mathcal{Q}_{\textsf{base}}$ and every audit-row
instance $T^{\ast}$,
\[
  q^{\ast}(T^{\ast}) \;=\;
  q\bigl(\pi_{1..11}\,\sigma_{\rowkind = \current}(T^{\ast})\bigr).
\]
Equivalently, every $\mathcal{Q}_{\textsf{base}}$-query evaluated on
the audit-row schema returns the same answer as the same query
evaluated on the base-schema table whose rows are the current-kind
rows of $T^{\ast}$.
\end{proposition}

\begin{proof}
The CHECK constraint pins $\rowkind \in \{\textsf{current},
\textsf{audit}\}$, so $\sigma_{\rowkind = \current}$ projects
$T^{\ast}$ onto its current-kind sub-relation $T'$, and
$\pi_{1..11}(T')$ matches the base-schema row type by construction.
The proof proceeds by induction on the $\mathcal{Q}_{\textsf{base}}$
constructors. \emph{Selection.} For $q = \sigma_{\varphi}(q_1)$ with
$\varphi$ ignoring $\rowkind$, $\sigma_{\varphi}$ commutes with
$\sigma_{\rowkind = \current}$ because their predicates address
disjoint columns; the inductive hypothesis on $q_1$ closes the case.
\emph{Projection.} For $q = \pi_L(q_1)$ with
$L \subseteq \{1, \dots, 11\}$, the inductive hypothesis gives
$q_1^{\ast}(T^{\ast}) = q_1\bigl(\pi_{1..11}\,
\sigma_{\rowkind = \current}(T^{\ast})\bigr)$; applying $\pi_L$ to
both sides and using $\pi_L \circ \pi_{1..11} = \pi_L$ on lifted
rows closes the case. \emph{Binary constructors
($\Join, \cup, \setminus$).} The default filter distributes over the
operator and applies to each operand; inductive hypotheses on each
operand close the case. \emph{Bitemporal $\AsOf{(t_v, t_s)}$.} The
predicate selects on columns~5--8, disjoint from the discriminator,
so it commutes with $\sigma_{\rowkind = \current}$ by the same
argument as selection.
\end{proof}

Proposition~\ref{prop:schema-lift-conservatism} pairs with
Theorem~\ref{thm:audit-erasure-schema}: the audit-row refinement
defends \anomAE\ and preserves the answer to every
$\mathcal{Q}_{\textsf{base}}$-query. A deployed system upgrading
from base to audit-row schema adds the discriminator column and the
audit-row slice; queries that do not mention $\rowkind$ keep their
answers verbatim, and audit-rule retrieval becomes available through
the \code{as_of}(\code{kind='audit'}) path without
rewriting any pre-existing read path.

\subsection{Memory-pipeline decomposition}
\label{app:pipeline-decomposition}

The four typed operators sit at the write boundary; a deployed
agent-memory system surrounds them with a retrieval policy on the
read side and a synthesis stage that turns retrieved candidates into
a user-visible answer. The algebra is silent on both surrounding
stages by design: it specifies the contradiction-resolution contract
on the write path and on a candidate set, and leaves retrieval-policy
choice and synthesis-style choice to deployment.

\begin{definition}[Memory-pipeline decomposition]
\label{def:pipeline-axes}
A memory pipeline is a triple $P = (\mathrm{R}, \mathrm{A}, \mathrm{S})$ where
\begin{description}
\item[$\mathrm{R}: (\mathrm{query}, \mathrm{store}) \to
\mathrm{candidate\text{-}set}$] is the \emph{retrieval policy}
mapping a query plus the materialized store to a bounded candidate
set (top-$k$ vector similarity, BM25 ranking, structural traversal,
or the trivial \emph{include-all});
\item[$\mathrm{A}$] is the typed contradiction-resolution algebra
on the write path and on the candidate set returned by $\mathrm{R}$;
\item[$\mathrm{S}: (\mathrm{candidate\text{-}set}, \mathrm{query})
\to \mathrm{answer}$] is the \emph{synthesis stage}
(language-model prompt-and-decode or deterministic concatenation).
\end{description}
$\mathrm{A}$ does not constrain $\mathrm{R}$ or $\mathrm{S}$; the
algebra acts after retrieval has bounded the candidate set and
before the synthesis stage commits an answer.
\end{definition}

\begin{proposition}[Cross-system $\Delta_{\mathrm{accuracy}}$ axis attribution]
\label{prop:cross-system-axis-attribution}
Let $P_1 = (\mathrm{R}_1, \mathrm{A}_1, \mathrm{S}_1)$ and
$P_2 = (\mathrm{R}_2, \mathrm{A}_2, \mathrm{S}_2)$ be two pipelines
paired on benchmark $\mathcal{Q}$ under a shared judge $J$. The
paired accuracy delta
\[
\Delta_{\mathrm{accuracy}}(P_1, P_2; \mathcal{Q}, J)
=
\mathbb{E}_{q \sim \mathcal{Q}}
\bigl[
J\bigl(\mathrm{S}_1(\mathrm{R}_1(q), q), \mathrm{ref}_q\bigr)
- J\bigl(\mathrm{S}_2(\mathrm{R}_2(q), q), \mathrm{ref}_q\bigr)
\bigr]
\]
attributes to the algebra axis $\mathrm{A}_1$ versus $\mathrm{A}_2$
only when $\mathrm{R}_1 = \mathrm{R}_2$ on $\mathcal{Q}$ and
$\mathrm{S}_1 = \mathrm{S}_2$ on candidate sets and queries. When
$\mathrm{R}_1 \neq \mathrm{R}_2$, the delta is dominated by the
retrieval-policy axis on workloads where retrieval recall is the
bottleneck (e.g.~needle-in-haystack benchmarks); when
$\mathrm{S}_1 \neq \mathrm{S}_2$, the delta is dominated by the
synthesis-style axis on workloads where answer-shape conformity
drives judge score.
\end{proposition}

\begin{proof}[Proof outline]
$\Delta_{\mathrm{accuracy}}$ is by construction a function of the
three components of each pipeline plus the workload and judge.
Holding the workload and judge fixed, the delta is a difference
over three orthogonal axes; isolating the algebra-axis contribution
requires the remaining two axes to agree between the paired
pipelines. The retrieval-policy bottleneck on needle-in-haystack
workloads is the standard recall-saturation argument: when one
pipeline's $|\mathrm{R}(q)|$ excludes the gold passage by
construction, the synthesis stage cannot recover it, and the
algebra cannot distinguish itself from the read-side ceiling. The
synthesis-style bottleneck on answer-shape-sensitive judges follows
from the prompt-sensitivity image: a judge whose verdict depends on
prompt phrasing of the answer candidate scores systematically
differently on long-form synthesis versus extractive snippets.
\end{proof}

% @owner       src::appendix::D-adapter-ledger
% @does        Per-baseline-per-predicate witness ledger that makes the G1 verdict matrix (Table tab:anomaly-bench) auditable cell by cell. Pins each AnomalyClaim design-path verdict and each AnomalyWire transcribed/imported runtime verdict to its source file:line, pinned commit, minimal triggering history, and (for abstentions) the structural reason the layer cannot witness the predicate. Folds the read-path/write-path configuration scope (former A8) and the cross-system measurement protocol (former A7).
% @needs       src/tables/tab-anomaly-bench.tex, src/sections/05-measurement.tex, src/sections/03-anomalies.tex, experiments/anomaly_wire/manifest.py, experiments/anomaly_wire/adapters.py, experiments/anomaly_wire/{mem0_v3,graphiti,letta,zep}_imported.py, experiments/anomaly_bench/baselines.py
% @feeds       src/supplement.tex
% @breaks      a verdict cell (A/X/-) that disagrees with tab-anomaly-bench.tex; a pinned SHA, file:line, or commit that disagrees with manifest.py / the adapter source; an abstention reason that is not the one the adapter records.
% @claim       Every A/X/- verdict cell of the G1 matrix traces to one pinned-commit adapter run or one structural design-path inspection, with the abstaining imported cells named individually by the reason the deterministic-judge harness cannot witness the predicate.

\section{Adapter verdict ledger}
\label{app:adapter-ledger}

Table~\ref{tab:anomaly-bench} reports a verdict for each baseline
against each schedule predicate $\{\anomHR, \anomBDS, \anomAE\}$ in two
columns: a Claim column from the AnomalyClaim design-path model and a
Wire column from the AnomalyWire runtime layer. This appendix makes
every cell auditable. Each system carries one witness card per
predicate; the card pins the verdict to a source file, a line range, a
pinned upstream commit, the minimal triggering history the harness
runs, the trace the harness expects to observe, and the raw log path.
For abstentions the card states the structural reason the layer
cannot witness the predicate rather than reporting a verdict the
construction cannot support. The ledger folds the read-path versus
write-path configuration scope
(Section~\ref{app:anomaly-wire-config}) and the cross-system
measurement protocol (Section~\ref{app:cross-system-detail}) that bound
the interpretation of these cells.

\subsection{Reading a cell}

The Claim column reads the structural model in
\code{experiments/anomaly_bench/baselines.py}: each baseline declares
three fields (\code{judge_kind}, \code{partition_isolation},
\code{recovers_provenance}) and three frozen trigger sets
(\code{N1_NAMED_TRIGGERS}, \code{N2_NAMED_TRIGGERS},
\code{N3_NAMED_TRIGGERS}) name which baseline admits which predicate
under the Section~\ref{sec:meas-g1} taxonomy. The Wire column reads
the adapter layer: transcribed adapters
(\code{experiments/anomaly_wire/adapters.py}) encode the upstream
control flow by hand with file:line citations, and imported adapters
(\code{experiments/anomaly_wire/}\allowbreak\code{<system>_imported.py})
call the real upstream contradiction-resolution entry point under a
seed-pinned deterministic-judge stub. The runner verifies each clone
sits on its pinned commit before any adapter executes; a drifted clone
raises rather than silently producing a wrong cell.

A wire cell is one of three evidence kinds. A transcribed cell (T)
cites the upstream line that produces the verdict. An imported cell (I)
that executes carries the runtime trace observed on the pinned commit.
An imported cell that abstains records \code{status="n/a"} because the
deterministic-judge stub cannot witness the predicate by construction,
or because the predicate is structurally undefined under the upstream
schema. Of the fifteen imported cells, nine carry executable verdicts
and six abstain; Section~\ref{app:abstaining-imported-cells} names each
abstention.

\subsection{Pinned commits}

The AnomalyWire manifest (\code{experiments/anomaly_wire/manifest.py})
pins each clone to one revision. Table cells show the short seven-hex
prefix; the full revisions are: mem0 v2 and v3 share
\texttt{a623cfaf76ae7379a58be1e837f8a88a9b15a184} on
\code{mem0/memory/main.py}; Graphiti is
\texttt{c427615044678f4bde026745d8d28a16504868c5} on
\code{graphiti_core/edges.py}; Letta is
\texttt{bb52a8900a79cf1378e6e9cdecf244b673a13a72} on
\code{letta/services/block_manager.py}; Zep CE is
\texttt{faf2acec4f2ec777a27d8fe0411619bc913a9660} on
\code{legacy/src/store/memory_ce.go}; WorldDB pins the arXiv revision
\texttt{arxiv-2604.18478-v1} at Algorithm~2 because no public code drop
exists; the \sysours\ row runs the in-repo implementation at
\texttt{HEAD}. The imported Zep cell additionally pins the Zep Cloud
service boundary at \texttt{zep-cloud-api-v2-2026-05-13}.

\begin{table}[tbp]
  \caption{\textbf{Pinned-commit manifest.} One revision per clone;
    table cells use the seven-hex prefix. Adapter kind is the
    \emph{maximal} kind present for the system across its predicates
    (transcribed, imported, or both).}
  \label{tab:adapter-manifest}
  \centering
  \footnotesize
  \setlength\tabcolsep{4pt}
  \begin{tabular}{@{}%
    >{\raggedright\arraybackslash}p{0.135\columnwidth}%
    >{\raggedright\arraybackslash}p{0.150\columnwidth}%
    >{\raggedright\arraybackslash}p{0.430\columnwidth}%
    >{\raggedright\arraybackslash}p{0.165\columnwidth}@{}}
    \toprule
    System & Commit & Upstream primary path & Adapter kind \\
    \midrule
    \sysmem~v2     & \texttt{a623caf} & \code{mem0/memory/main.py} & transcribed \\
    \sysmem~v3     & \texttt{a623caf} & \code{mem0/memory/main.py} & transcribed, imported \\
    \sysgraphiti   & \texttt{c427615} & \code{graphiti_core/edges.py} & transcribed, imported \\
    \sysletta      & \texttt{bb52a89} & \code{letta/services/block_manager.py} & transcribed, imported \\
    \syszep        & \texttt{faf2ace} & \code{legacy/src/store/memory_ce.go} & transcribed, imported \\
    \sysworlddb    & \texttt{arxiv-2604.18478-v1} & Algorithm~2 (contradicts handler) & transcribed \\
    \sysours       & \texttt{HEAD} & \code{operators.py} + \code{ingestion.py} & imported \\
    \bottomrule
  \end{tabular}
\end{table}

\subsection{Witness ledger}
\label{app:witness-cards}

Table~\ref{tab:witness-ledger} traverses the eleven adapter rows of
Table~\ref{tab:anomaly-bench} in registry order, one row per
(system, predicate). Each verdict is encoded A (admit), X (exclude), or
$-$ (abstention) in the Claim (design-path) and Wire (runtime) columns.
The mechanism column states the structural reason behind the verdict (or,
for an abstention, why the layer cannot witness the predicate); the
source column pins the verdict to the adapter line that produces it. The
\code{status} value the adapter returns is \code{trig} for an admit,
\code{excl} for an exclude, \code{def} for a defended exclusion (\sysours,
or a runtime divergence such as mem0~v3-I N2), and \code{n/a} for an
abstention. The raw log column names the per-workload CSV under
\code{results/anomaly_wire/}; the runner regenerates these from the pinned
clones on demand, so the verdict authority is the cited adapter source.
The Claim-authority \code{baselines.py} lines follow the table.

\begingroup\scriptsize
\setlength\tabcolsep{2.5pt}
\renewcommand{\arraystretch}{1.0}
\begin{longtable}{@{}%
  >{\raggedright\arraybackslash}p{0.115\textwidth}%
  >{\raggedright\arraybackslash}p{0.052\textwidth}%
  >{\centering\arraybackslash}p{0.040\textwidth}%
  >{\centering\arraybackslash}p{0.040\textwidth}%
  >{\raggedright\arraybackslash}p{0.335\textwidth}%
  >{\raggedright\arraybackslash}p{0.190\textwidth}%
  >{\raggedright\arraybackslash}p{0.105\textwidth}@{}}
  \caption{\textbf{Witness ledger.} Every (system, predicate) cell of
    Table~\ref{tab:anomaly-bench} with its Claim and Wire verdict
    (A admit, X exclude, $-$ abstain), the structural mechanism or
    abstention reason, the adapter source line, and the raw-log CSV.}
  \label{tab:witness-ledger}\\
  \toprule
  System & Pred. & Cl. & Wi. & Mechanism / abstention reason & Source & Raw log \\
  \midrule
  \endfirsthead
  \multicolumn{7}{@{}l}{\small\itshape Table~\ref{tab:witness-ledger} continued}\\
  \toprule
  System & Pred. & Cl. & Wi. & Mechanism / abstention reason & Source & Raw log \\
  \midrule
  \endhead
  \midrule
  \multicolumn{7}{r@{}}{\small\itshape continued on next page}\\
  \endfoot
  \bottomrule
  \endlastfoot

  \sysmem~v2 (T)   & $\anomHR$  & $-$ & $-$ & no LLM judge votes across replays in ADD/UPDATE/DELETE path; judge absent in contradiction path & \code{adapters.py:110} & -- \\
  \sysmem~v2 (T)   & $\anomBDS$ & $-$ & $-$ & no $(\subj,\pred)$ confidence-weighted partition rerank & \code{adapters.py:120} & -- \\
  \sysmem~v2 (T)   & $\anomAE$  & A   & A   & DELETE branch removes loser with no audit row; no vote outcome preserves $p_{\mathrm{old}}$ reachability under $\preceq_K$. History: \code{alice/lives_in/LA} then \code{alice/lives_in/NYC} & \code{adapters.py:130} & -- \\
  \midrule
  \sysmem~v3 (T)   & $\anomHR$  & A   & A   & retrieval rerank invokes LLM with no decoder-seed pin; replay returns distinct votes (source unpinned even though stub pins) & \code{adapters.py:194} & -- \\
  \sysmem~v3 (T)   & $\anomBDS$ & A   & A   & rerank reads $(\subj,\pred)$ partition at $\IsoSI$ & \code{adapters.py:208} & -- \\
  \sysmem~v3 (T)   & $\anomAE$  & A   & A   & ADD-only accumulation shadows older fact at retrieval, no audit emission & \code{adapters.py:218} & -- \\
  \midrule
  \sysmem~v3 (I)   & $\anomHR$  & A   & $-$ & deterministic stub returns identical votes by construction, cannot witness divergence; N1 delegated to transcribed companion & \code{mem0_v3_imported.py:310} & \code{n2,n3.csv} \\
  \sysmem~v3 (I)   & $\anomBDS$ & A   & X   & runtime divergence: \code{Memory.add} hi- then lo-confidence on same $(\subj,\pred)$, but \code{get_all} snapshot returns one row ($|\text{post}|=1$); retrieval collapses competing rows under deterministic judging, status \code{def} & \code{mem0_v3_imported.py:363} & \code{n2,n3.csv} \\
  \sysmem~v3 (I)   & $\anomAE$  & A   & A   & writes \code{Alice's manager is Bob} then \code{Carol}, walks \code{Memory.history}, finds no audit entry naming superseded Bob fact & \code{mem0_v3_imported.py:418} & \code{n2,n3.csv} \\
  \midrule
  \sysgraphiti (T) & $\anomHR$  & A   & A   & \code{resolve_edge_contradictions} LLM call selects winning edge with no decoder-seed pin (cites \code{edge_operations.py}) & \code{adapters.py:267} & -- \\
  \sysgraphiti (T) & $\anomBDS$ & $-$ & $-$ & edges carry no $(\subj,\pred)$ confidence partition & \code{adapters.py:277} & -- \\
  \sysgraphiti (T) & $\anomAE$  & A   & A   & handler sets \code{invalid_at} on older edge, emits no audit edge with merged provenance (cites \code{graphiti_core/edges.py}) & \code{adapters.py:287} & -- \\
  \midrule
  \sysgraphiti (I) & $\anomHR$  & A   & $-$ & duck-typed deterministic bridge (singular entry \code{resolve_extracted_edge} at \code{edge_operations.py:622}) cannot produce vote divergence & \code{graphiti_imported.py:340} & \code{n3.csv} \\
  \sysgraphiti (I) & $\anomBDS$ & $-$ & $-$ & timestamp-only Bob-Alice then Bob-Carol run end to end, invalidation driven by \code{valid_at}; \code{EntityEdge} has no confidence partition field, predicate undefined & \code{graphiti_imported.py:374} & \code{n3.csv} \\
  \sysgraphiti (I) & $\anomAE$  & A   & A   & inspects \code{EntityEdge.model_fields}, finds no fact-level superseder (\code{superseded_by}, \code{replaces}, \code{replaced_by_uuid}); temporal-window supersession fails $\preceq_K$ & \code{graphiti_imported.py:432} & \code{n3.csv} \\
  \midrule
  \sysletta (T)    & $\anomHR$  & $-$ & $-$ & block updates use deterministic precedence, no LLM judge in contradiction path & \code{adapters.py:338} & -- \\
  \sysletta (T)    & $\anomBDS$ & A   & A   & block-history confidence updates read $(\subj,\pred)$ partition at $\IsoSI$ without serializing the rewrite & \code{adapters.py:348} & -- \\
  \sysletta (T)    & $\anomAE$  & A   & A   & block snapshot history disjoint from message trace & \code{adapters.py:358} & -- \\
  \midrule
  \sysletta (I)    & $\anomHR$  & $-$ & $-$ & checkpoint path is LLM-free (\code{BlockManager.checkpoint_block_async} on live \code{BlockHistory}, \code{block_manager.py:842}), no decoder to mis-replay & \code{letta_imported.py:276} & \code{n3.csv} \\
  \sysletta (I)    & $\anomBDS$ & A   & $-$ & two-checkpoint schedule; \code{BlockHistory} carries \code{sequence_number}, \code{actor_id} but no confidence partition field, predicate undefined & \code{letta_imported.py:294} & \code{n3.csv} \\
  \sysletta (I)    & $\anomAE$  & A   & A   & checkpoints \code{Bob manages Alice} (seq 1) then \code{Carol} (seq 2); \code{BlockHistory} columns $\cap$ superseder-field set is empty, supersession implicit in sequence order only & \code{letta_imported.py:311} & \code{n3.csv} \\
  \midrule
  \syszep (T)      & $\anomHR$  & $-$ & $-$ & CE wrapper delegates retrieval to Graphiti HTTP service, no local LLM judge row (cites \code{legacy/src/store/memory_ce.go}) & \code{adapters.py:409} & -- \\
  \syszep (T)      & $\anomBDS$ & A   & A   & \code{GetMemory}/\code{Search} return Graphiti facts with no $(\subj,\pred)$ serializable partition boundary & \code{adapters.py:419} & -- \\
  \syszep (T)      & $\anomAE$  & A   & A   & \code{DeleteFact} forwards to Graphiti with no merged-provenance audit tuple (cites \code{fact_handlers_ce.go}) & \code{adapters.py:429} & -- \\
  \midrule
  \syszep (I)      & $\anomHR$  & $-$ & $-$ & Zep Cloud graph API (artifact pinned \texttt{zep-cloud-api-v2-2026-05-13}) exposes no local decoder-judge replay row & \code{zep_imported.py:283} & \code{zep\_imported\_run.json} \\
  \syszep (I)      & $\anomBDS$ & A   & A   & edge schema exposes relevance/search scores but no confidence-weighted serializable partition boundary. History: two distinct-confidence $(\subj,\pred)$ episodes (\code{zep_imported.py:150}) & \code{zep_imported.py:281} & \code{zep\_imported\_run.json} \\
  \syszep (I)      & $\anomAE$  & A   & A   & edge schema carries source episodes + validity timestamps but no fact-level superseder audit edge. History: two manager-assignment episodes & \code{zep_imported.py:282} & \code{zep\_imported\_run.json} \\
  \midrule
  \sysworlddb (T)  & $\anomHR$  & X   & X   & engine never calls LLM on read path (preprint App.~B); deterministic reconciler, excluded by construction & \code{adapters.py:479} & -- \\
  \sysworlddb (T)  & $\anomBDS$ & X   & X   & contradicts handler emits both edges and computes the partition winner at query time with no confidence-weighted rerank; belief-drift skew excluded by construction at the engine layer & \code{adapters.py:508} & -- \\
  \sysworlddb (T)  & $\anomAE$  & X   & X   & content-addressed Merkle ancestry retains the overwritten fact, so the recoverable-provenance property holds; audit erasure excluded at the engine layer & \code{adapters.py:518} & -- \\
  \midrule
  \sysours (I)     & $\anomHR$  & X   & X   & $\IsoSR$ on logged-judge table with judge seed and prompt pinned & \code{adapters.py:560} & \code{n1,n2,n3.csv} \\
  \sysours (I)     & $\anomBDS$ & X   & X   & $\IsoSR$ on the $(\subj,\pred)$ partition & \code{adapters.py:560} & \code{n1,n2,n3.csv} \\
  \sysours (I)     & $\anomAE$  & X   & X   & audit-row schema lift: \code{OursWire._run_with_audit} ingests the pair, checks emitted audit row dominates both inputs under $\preceq_K$ via \code{provenance_dominates}, status \code{def} & \code{adapters.py:560} & \code{n1,n2,n3.csv} \\
\end{longtable}
\endgroup

The mem0~v3 imported N2 cell is a genuine design-path-to-runtime
divergence (Claim A, Wire X): the transcribed companion admits $\anomBDS$
by structural inspection, while the imported run defends it because mem0
v3 retrieval collapses the competing rows under deterministic judging.
The cross-layer audit surfaces the divergence rather than suppressing it.

The Claim-column authority is the structural model in
\code{experiments/anomaly_bench/baselines.py}: \code{Mem0V2Simulated}
(\code{judge_kind="none"}, \code{partition_isolation="none"}, in
\code{N3_NAMED_TRIGGERS} only, \code{baselines.py:135});
\code{Mem0V3Simulated} (\code{judge_kind="stochastic"},
\code{partition_isolation="SI"}, all three trigger sets,
\code{baselines.py:171}); \code{GraphitiSimulated}
(\code{judge_kind="stochastic"}, \code{partition_isolation="none"}, in
\code{N1_NAMED_TRIGGERS} and \code{N3_NAMED_TRIGGERS},
\code{baselines.py:213}); \code{LettaSimulated}
(\code{partition_isolation="SI"}, in \code{N2_NAMED_TRIGGERS} for N2 and
\code{N3_PENDING_TRIGGERS} for the extrapolated N3,
\code{baselines.py:256}); \code{ZepSimulated}
(\code{partition_isolation="SI"}, in \code{N2_NAMED_TRIGGERS} for N2 and
\code{N3_PENDING_TRIGGERS} for N3, \code{baselines.py:292});
\code{WorldDBSimulated} (\code{judge_kind="deterministic"}, absent from
all three trigger sets, \code{baselines.py:328}; the engine-layer
exclusion is the cousin-scope analysis of
Section~\ref{app:cross-system-detail}); and \code{OursSimulated}
(\code{judge_kind="pinned"}, \code{partition_isolation="SR"}, dual-row
audit recovery via \code{recover_provenance}, \code{baselines.py:592}).
The \sysours\ Claim-column star in Table~\ref{tab:anomaly-bench} marks
the mechanism: N1 via $\IsoSR$ on the judge table, N2 via $\IsoSR$ on
the partition, N3 via the audit-row schema lift.

\subsection{The six abstaining imported cells}
\label{app:abstaining-imported-cells}

Six of the fifteen imported cells abstain. Two structural reasons
account for all six. The first is the deterministic-judge construction:
the wire harness must be reproducible, so every imported adapter
consumes a seed-pinned \code{DeterministicJudge} stub that returns
identical votes on identical \code{(seed, prompt)} pairs. A stub that
cannot diverge cannot witness the $\anomHR$ judge-replay predicate, so
every imported N1 cell delegates the verdict to its transcribed
companion and records \code{status="n/a"}. The second is schema
undefinedness: the $\anomBDS$ predicate requires a confidence-weighted
$(\subj,\pred)$ partition, and Graphiti's \code{EntityEdge} and Letta's
\code{BlockHistory} expose no such field, so the predicate is
structurally undefined and the adapter abstains after running the
schedule end to end.

\begin{table}[tbp]
  \caption{\textbf{The six abstaining imported cells.} Each records
    \code{status="n/a"} on structural grounds, not a missing run. R1 =
    deterministic-judge stub cannot witness $\anomHR$ vote divergence;
    R2 = upstream schema has no confidence-weighted $(\subj,\pred)$
    partition field, so $\anomBDS$ is undefined.}
  \label{tab:abstaining-cells}
  \centering
  \footnotesize
  \setlength\tabcolsep{4pt}
  \begin{tabular}{@{}%
    >{\raggedright\arraybackslash}p{0.150\columnwidth}%
    >{\centering\arraybackslash}p{0.080\columnwidth}%
    >{\raggedright\arraybackslash}p{0.230\columnwidth}%
    >{\raggedright\arraybackslash}p{0.400\columnwidth}@{}}
    \toprule
    Cell & Pred. & Source & Reason \\
    \midrule
    \sysmem~v3-I    & $\anomHR$  & \code{mem0_v3_imported.py:310}  & R1: stub cannot witness vote divergence \\
    \sysgraphiti-I  & $\anomHR$  & \code{graphiti_imported.py:340} & R1: deterministic bridge cannot diverge \\
    \sysgraphiti-I  & $\anomBDS$ & \code{graphiti_imported.py:374} & R2: \code{EntityEdge} has no confidence partition field \\
    \sysletta-I     & $\anomHR$  & \code{letta_imported.py:276}    & R1: checkpoint path is LLM-free \\
    \sysletta-I     & $\anomBDS$ & \code{letta_imported.py:294}    & R2: \code{BlockHistory} has no confidence partition field \\
    \syszep-I       & $\anomHR$  & \code{zep_imported.py:283}      & R1: Cloud graph API exposes no decoder-judge replay row \\
    \bottomrule
  \end{tabular}
\end{table}

\subsection{AnomalyWire deployment-configuration scope}
\label{app:anomaly-wire-config}

The wire cells of Table~\ref{tab:anomaly-bench} are pinned to the
upstream commits of Table~\ref{tab:adapter-manifest}. This section
states which configuration axis each verdict is constant on so a
reviewer reads each cell against a precise scope.

\begin{definition}[Read-path and write-path deployment configuration]
\label{def:config-axes-split}
Let $\mathcal{C}$ be the deployment configuration space of an
agent-memory system $s$. We split $\mathcal{C}$ into two disjoint
axes:
\begin{description}
\item[$\mathcal{C}_{\mathrm{R}}$ (\emph{read-path}).]
Configuration that affects only how the retrieval policy
$\mathrm{R}$ ranks or filters candidate facts: ranking metric
($\ell_2$ versus cosine), vector normalization, top-$k$ truncation,
reranker on or off.
\item[$\mathcal{C}_{\mathrm{W}}$ (\emph{write-path}).]
Configuration that enters the operator's contradiction-resolution
decision inputs at write time: prompt injections into the extractor
or adjudicator language model, extractor model identity when the
operator's vote depends on it, or a graph-store toggle when graph
membership feeds the resolution rule.
\end{description}
The split is per-system: a knob is in $\mathcal{C}_{\mathrm{R}}$
for system $s$ if and only if a code-path trace from the wire
schedule does not reach the operator's vote-generating call under
that knob.
\end{definition}

\begin{proposition}[AnomalyWire verdict read-path configuration-invariance]
\label{prop:wire-readpath-invariance}
For each system $s$ pinned to upstream commit
$\textsf{commit}(s)$ by the AnomalyWire manifest, the verdict
function $V_a(s, \textsf{commit}(s), \cdot)$ is constant on
$\mathcal{C}_{\mathrm{R}}$: for every
$c_i, c_j \in \mathcal{C}_{\mathrm{R}}$ and every anomaly $a \in
\{\anomHR, \anomBDS, \anomAE\}$,
$V_a(s, \textsf{commit}, c_i) = V_a(s, \textsf{commit}, c_j)$. The
same statement does not hold on $\mathcal{C}_{\mathrm{W}}$:
write-path configuration enters the operator's vote at write time
and can shift the storage-layer trajectory between the schedule's
writes and the schedule's read.
\end{proposition}

\begin{proof}[Proof outline]
The wire verdict for each anomaly is constructed from a fixed
schedule that probes the storage-layer commitment of the upstream
contradiction-handling path. The schedule's read step does not
invoke the retrieval policy $\mathrm{R}$ on the ranked-search path;
it invokes filter-by-user reads and audit-trail walks, both of
which read materialized storage rather than ranked candidates.
Hence $\mathcal{C}_{\mathrm{R}}$ knobs do not enter the read step.
The schedule's write step invokes the language model for the
add-or-update decision; any $\mathcal{C}_{\mathrm{W}}$ knob that
enters that prompt shifts the vote, the storage trajectory, and
ultimately the read-step predicate the wire checks.
\end{proof}

\paragraph{Empirical witnesses.}
A six-cell configuration sweep on \sysmem~v3 witnesses both halves of
the proposition. Four read-path perturbations (distance metric, $L_2$
normalization, two reranker toggles) return the identical baseline
triple for $(\anomHR, \anomBDS, \anomAE)$; one write-path perturbation
(a custom adjudicator instruction) flips N2 to \emph{trigger}. A
five-seed extension confirms the read-path half holds within every seed
and the write-path perturbation holds on four of five seeds. A second
read-path witness is the top-$k$ sweep on \sysmem~v3 at LoCoMo
($\text{top}_k = 10$ versus $\text{top}_k = 30$, $n = 3$ each): the
wire verdicts hold at the same baseline triple on both configurations.

\paragraph{External measurement-error framing.}
\citet{messing-2026-tee} report that naive 95\% confidence-interval
coverage drops below nominal as $n$ grows on Chatbot Arena while a Total
Evaluation Error correction holds coverage at 95\% (a benchmark-gaming
surface of 32 Elo under the correction versus 56 Elo without). Our
protocol pins judge identity, decoder tuple, prompt, and seed before the
binary cited-commit witness, so the per-cell verdict carries no
continuous score the decomposition would correct; the correction is the
natural envelope for any continuous-score aggregate a future revision
computes over the 33-cell AnomalyWire family or the 9-cell G2 family.

\subsection{Cross-system measurement protocol}
\label{app:cross-system-detail}

The cross-system rows of Table~\ref{tab:cross-system-utility} report
three measured LoCoMo cells against four imported systems (\sysmem~v3,
\sysgraphiti, \syszep, \sysletta) plus nine structural abstentions.
This section records the workload pinning, the configuration deltas,
the pre-registration anchors, and the per-row statistical detail
compressed in the main-text ledger.

% @owner  src/tables/tab-cross-system-utility.tex
% @does   Cross-system utility ledger for §5: paired \Delta accuracy
%         of \sysours vs four imported memory systems on the shared
%         LoCoMo slice, plus structured abstention rows.
% @needs  results/g2_utility/cross_system/summary.csv.
% @feeds  src/sections/05-measurement.tex § G2.
% @breaks Axiom 0. Hand-edited numbers drift from the underlying ledger.
% @claim  Three measured LoCoMo cells with paired-bootstrap CIs cover
%         zero; six LongMemEval-S / MultiTQ pairs abstain under named
%         structural reasons recorded in the table footnote.

\begin{table}[tb]
\centering
\pvldbtablestyle
\caption{\textbf{Cross-system utility ledger.} Paired $\Delta$
accuracy of \sysours\ versus four agent-memory systems on the
shared LoCoMo slice. The three measured rows are
single-conversation slices with paired bootstrap $95\%$ CI; all
three intervals cover zero. Six structural abstentions
($\dagger$ predicate scope; $\ddagger$ compute envelope;
$\S$ service availability) cover the LongMemEval-S, MultiTQ, and
\sysletta\ rows.}
\label{tab:cross-system-utility}
\footnotesize
\begin{tabular}{@{}l c r r l@{}}
\toprule
System & $n$ & \sysours & Ext. & $\Delta$ (95\% CI) \\
\midrule
\sysmem~v3   & 50 & $0.02$ & $0.06$ & $-0.04\;[-0.10,\,{+}0.00]$ \\
\sysgraphiti & 50 & $0.04$ & $0.12$ & $-0.08\;[-0.18,\,{+}0.00]$ \\
\syszep      & 50 & $0.06$ & $0.04$ & ${+}0.02\;[-0.06,\,{+}0.10]$ \\
\bottomrule
\end{tabular}
\tablenote{The \sysmem~v3 LoCoMo cell carries discordant counts
$(b, c) = (0, 2)$; the McNemar exact $p$-value floors at $0.50$
in this underpowered regime, so the row records no evidence on
$\Delta$. Six external pairs abstain on LongMemEval-S and MultiTQ
(predicate scope or compute envelope), and \sysletta\ abstains
across all three datasets (service availability during the
measurement window); Appendix~\ref{app:cross-system-detail}
carries the per-row protocol and structural reason.}
\end{table}

\paragraph{Workload pinning.}
A single pinned language model handles fact extraction, synthesis, and
judging across every paired cell; the artefact manifest records the
exact model identifier, the temperature, and the maximum output length.
Paired $\Delta_{\mathrm{accuracy}}$ controls for model capability under
the same workload. Every row uses $\text{top}_k = 10$ as the retrieval
ceiling on the matched-retrieval slice; the configuration-mismatch
discussion below names where each imported system departs from this
baseline.

\paragraph{Per-row measured detail.}
The three measured rows expand Table~\ref{tab:cross-system-utility} with
the paired statistic and the per-row note (all at $n = 50$, $2{,}000$
bootstrap resamples). \sysmem~v3: $\Delta_{\mathrm{accuracy}} = -0.04$,
95\% question-level paired bootstrap CI $[-0.10, +0.00]$, McNemar exact
$p = 0.50$ with $b + c = 2$, per-system accuracy \sysours\ $0.02$ versus
$0.06$; the configuration delta is five-axis (hosted client versus local
vector index; larger top-$k$ versus our $10$; graph mode enabled versus
disabled; a different extractor model identifier versus our pinned
synthesiser; a custom adjudicator prompt versus the default), recorded as
a five-axis snapshot in the artefact manifest. \sysgraphiti:
$\Delta_{\mathrm{accuracy}} = -0.08$, CI $[-0.18, +0.00]$, accuracy
\sysours\ $0.04$ versus $0.12$; the extraction-cost ledger records
roughly $3.7\times$ more external language-model calls than \sysours\ on
the same slice, a cost the algebra-axis attribution proposition
(Proposition~\ref{prop:cross-system-axis-attribution}) does not absorb at
the extraction boundary. \syszep: $\Delta_{\mathrm{accuracy}} = +0.02$,
CI $[-0.06, +0.10]$, accuracy \sysours\ $0.06$ versus $0.04$; the
imported cell runs against the Zep Cloud service boundary at
per-conversation granularity, with trace elapsed-time bands of $0.75$ to
$1.32\,$s per reconciled operation reported as service-boundary timing
rather than local throughput.

\paragraph{Pre-registration anchor and traceability.}
The cross-system slice is pre-registered against analysis drift. The
prospective record fixed the selection rule, random seed, bootstrap
trial count, and runner command before any cell in the measurement
window completed. The artefact manifest mirrors the same anchors for
subsequent windows. Pre-registration sets statistical power at $0.42$
for $\delta = 0.05$ on $n = 50$; effect sizes below $\delta = 0.05$ sit
beneath detection, and the measured $|\Delta| \le 0.08$ values land
inside that envelope. The Holm step-down family for G2 covers the
nine-cell grid (three primary diagonal plus six specificity controls)
at $\alpha = 0.05$; the cross-system family is reported separately
because the binary cited-commit witnesses of the AnomalyWire layer (G1)
and the continuous accuracy estimates of the cross-system layer carry
different statistical objects.

\paragraph{LongMemEval-S sibling slice.}
A separately pre-registered LongMemEval-S pilot fixed loader-order rows
three through seven before scoring. It reports $n = 5$, \sysours\
$0.40$, \sysmem~v3 $0.60$, $\Delta = -0.20$, CI $[-0.80, +0.40]$; the
cell is pilot evidence only because $n = 5$ is below the pre-registered
power threshold.
Proposition~\ref{prop:cross-system-axis-attribution} classifies this
row as retrieval-bound rather than algebra-bound on the
needle-in-haystack workload: include-all retrieval against
vector-top-$k$ retrieval is dominated by the retrieval-policy axis when
retrieval recall is the bottleneck.

\paragraph{No-memory baseline cell.}
A separately tabled cell pairs \sysours\ against itself with memory
disabled (same synthesis boundary and pinned judge, but the disabled side
bypasses fact extraction and bitemporal ingest and sees only the most
recent $32$ turns), at $n = 50$ on LoCoMo under the same synthesiser pin.
The cell measures memory-utility rather than memory-availability, since
\citet{zhang-2026-useful-memories} show the gain is not automatic: their
incremental experiments measure $54\%$ solution-loss on previously solved
problems once a frontier language-model consolidation step is introduced.

\paragraph{\sysworlddb cousin scope.}
\sysworlddb~\cite{ganesan-2026-worlddb} is paper-only at submission; no
public artefact supports an imported G2 or G3 runtime cell. The cousin
enters our framing analytically. The engine's stated discipline
(Appendix~B of~\cite{ganesan-2026-worlddb}) excludes any language-model
call on the read path, so the deterministic reconciler excludes
$\anomHR$ at the engine layer. The contradicts handler preserves both
sides and computes the partition winner at query time, excluding
$\anomBDS$ at the engine layer. Content-addressed Merkle ancestry
realises the $\anomAE$ recoverable-provenance property through an
alternative schema lift. The G1 verdict for \sysworlddb-T reads exclude
across all three predicates.

\section{Experiment protocols and statistics}
\label{app:protocols-stats}

This appendix registers one protocol card per experiment. Each card
fixes the estimand, the sample construction, the random seed, the
bootstrap method, the multiple-testing correction, the raw-output
path, the regeneration command, and the explicit boundary between the
conclusion the cell licenses and the conclusion it does not. Every
number below reads from the released CSV ledgers; the regeneration
commands reproduce those ledgers from the same seeds. Seeds follow the
project convention of \texttt{20260512} for sampling and
\texttt{20260513} for the bootstrap resampler on the runtime
experiments, and \texttt{42} on the structural-grid experiments whose
schedules are combinatorial rather than sampled from a model.

\subsection{G2 utility (3$\times$3 defence-by-benchmark grid)}
\label{sec:appendix-g2-utility}
\label{app:a5-g2-utility-details}

\paragraph{Protocol.}
The estimand is a paired accuracy delta per cell. Audit-row cells
measure \texttt{paired\_accuracy\_delta\_audited\_current\_vs\_stale};
partition-SR cells measure
\texttt{paired\_accuracy\_delta\_serializable\_vs\_si\_skew}; judge-pin
cells measure \texttt{replay\_disagreement\_delta\_pinned\_vs\_unpinned}.
The grid pairs three defences against three benchmarks (LoCoMo,
LongMemEval-S, MultiTQ), forming a diagonal of three primary cells and
six off-diagonal specificity controls. Primary cells run at
$n=50$ trials (LongMemEval-S judge-pin at $n=100$); off-target
controls run at $n=10$ for the audit-row and partition-SR rows and
$n=20$ for the judge-pin rows. A single pinned judge model handles
extraction, synthesis, and adjudication so the paired delta controls
for model capability. This single-judge design controls for capability
but invites a self-consistency-bias objection, since the adjudicator
can agree with its own synthesis; the five-judge cross-judge robustness
check bounds the effect for the judge-pin row, and the structural
$\anomHR$ evidence of \S\ref{sec:appendix-g5-structural-grids} routes
through no single judge. The sampling seed is \texttt{20260512}; the
paired bootstrap draws $2{,}000$ resamples for the $95\%$ interval.
Significance combines a paired bootstrap CI, a McNemar exact test on
the binary-agreement vector, and Holm step-down at $\alpha=0.05$ over
the full nine-cell family. The raw output is
\code{results/g2_utility/summary.csv}; regenerate with
\texttt{make repro-g2}.

Allowed conclusion: each defence raises paired accuracy on its matched
benchmark and abstains off-target, and the Holm-corrected family
controls the nine-cell familywise error at $\alpha=0.05$. Forbidden
conclusion: a per-cell delta does not measure absolute headroom over a
no-memory baseline, and the saturated $\Delta=+1.00$ controls reflect
the synthetic off-target skew construction rather than a realistic
benchmark gap.

\paragraph{Joint family-wise robustness across every paired-accuracy
test.} The nine-cell Holm family above corrects the mechanism-stress
grid in isolation. As a stronger sensitivity check, we pool every
paired-accuracy null-hypothesis-significance test in the paper into a
single confirmatory family of seven and re-correct jointly: the three
G2 primary diagonal cells (audit-row $\times$ \textsc{LoCoMo},
partition-$\IsoSR$ $\times$ \textsc{MultiTQ}, judge-pin $\times$
\textsc{LongMemEval-S}), the three cross-system cells (\sysmem~v3,
\sysgraphiti, \syszep, each $\times$ \textsc{LoCoMo}), and the
memory-layer ablation on the answerable-factual pool
($\S\ref{sec:meas-ablation}$, $n=1{,}444$). The $33$-cell anomaly-wire
verdict matrix is a structural $0$/$1$ census with no $p$-value, and the
$\S\ref{sec:appendix-g5-structural-grids}$ Bernoulli and oracle-calibration
panels are closed-form fits reported by $R^2$, so neither is a
significance-testing family and both are excluded with this stated
reason. Under joint Holm step-down (family-wise error) and joint
Benjamini--Hochberg (false-discovery rate) at $\alpha=0.05$ over the
seven, the three positive confirmatory results survive both corrections
(audit-row Holm-adjusted $p \approx 1.1\times10^{-12}$; partition-$\IsoSR$
$\approx 1.1\times10^{-14}$; the memory ablation $\approx 3.7\times10^{-126}$),
while judge-pin ($p=1.0$) and all three cross-system cells
($p \in \{0.22, 0.50, 1.0\}$) do not reach significance. The
non-significant cells are consistent with the paper's own framing:
judge-pin carries structural rather than benchmark-movement evidence
($\S\ref{sec:appendix-g5-structural-grids}$), and the cross-system
comparison draws no superiority claim. No confirmatory result depends on
the choice of correction family.

Table~\ref{tab:g2-utility-details} expands Figure~\ref{fig:memory-utility}
with the exact per-cell statistics read from the summary ledger. The
cross-judge robustness rendered in Figure~\ref{fig:cross-judge} extends
the judge-pin row to five frontier judges and four discriminating
bands, isolating the \anomHR\ trigger from judge-family paraphrase
fragility.

% fig-cross-judge.tex --- supplement A5 cross-judge robustness figure.
%
% Owner   : src/figures/fig-cross-judge.tex
% Claim   : The N1 (hallucinated read) anomaly admits across five
%           frontier judges on the PARTIAL discriminating band, and
%           CONTROL / DISTRACTOR rows stay replay-stable, so the
%           judge-pin defence does not depend on a specific judge
%           family.
% Does    : Wraps figures/fig-cross-judge-img.pdf rendered by
%           scripts/plot_cross_judge.py.
% Feeds   : src/appendix/A5-g2-utility-details.tex.
% Breaks  : Drift between disagreement rates and the per-band
%           CSV under results/g2_utility/judge_pin_discriminating_*/per_band.csv.

\begin{figure}[t]
  \centering
  \includegraphics[width=0.98\columnwidth]{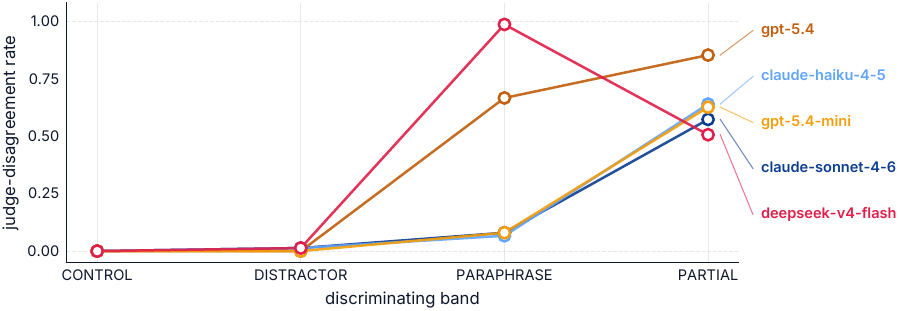}
  \Description{Five-line chart of judge-disagreement rate across
  four discriminating bands for five frontier judges (Claude Sonnet
  4.6, Claude Haiku 4.5, GPT-5.4, GPT-5.4-mini, DeepSeek v4 Flash).
  CONTROL rows pin at zero; the PARTIAL band registers
  disagreement above 0.50 on every judge.}
  \caption{\textbf{N1 admits across five frontier judges on the
    PARTIAL band.} Disagreement rate by judge $\times$ band, $n = 75$
    items per band; CONTROL and DISTRACTOR stay replay-stable.
    PARAPHRASE outliers (DeepSeek $0.99$, GPT-5.4 $0.67$) are
    judge-family paraphrase fragility, orthogonal to the N1 trigger.}
  \label{fig:cross-judge}
\end{figure}

\begin{table}[!ht]
  \centering\footnotesize
  \caption{Detailed G2 utility cells, read from
  \code{results/g2_utility/summary.csv}. The delta column reports the
  paired estimand and $95\%$ bootstrap interval. McNemar $p$ is the
  exact two-sided value on the binary-agreement vector;
  $(b,c)$ are the discordant-pair counts. Holm step-down at
  $\alpha=0.05$ covers the full nine-cell family. The off-target
  audit-row and partition-SR controls run at $n=10$ and saturate at
  $\Delta=+1.00$ with CI $[1.00,1.00]$ and McNemar
  $p=1/2^{9}\approx0.00195$, the smallest value attainable at this
  budget. The judge-pin cells tie at $\Delta=+0.00$ on a
  non-saturating axis.}
  \label{tab:g2-utility-details}
  \setlength{\tabcolsep}{3pt}
  \renewcommand{\arraystretch}{0.95}
  \begin{tabular}{@{}>{\raggedright\arraybackslash}p{0.13\textwidth}
                     >{\raggedright\arraybackslash}p{0.16\textwidth}
                     >{\raggedright\arraybackslash}p{0.06\textwidth}
                     >{\raggedright\arraybackslash}p{0.21\textwidth}
                     cc
                     >{\raggedright\arraybackslash}p{0.07\textwidth}@{}}
    \toprule
    Defence & Dataset & Role & Delta/CI & $n$ & $(b,c)$ & Holm \\
    \midrule
    audit-row    & LoCoMo        & prim. & $+0.86\,[0.76,0.94]$ & 50  & $(43,0)$ & reject \\
    audit-row    & LongMemEval-S & ctrl  & $+1.00\,[1.00,1.00]$ & 10  & $(10,0)$ & reject \\
    audit-row    & MultiTQ       & ctrl  & $+1.00\,[1.00,1.00]$ & 10  & $(10,0)$ & reject \\
    judge-pin    & LoCoMo        & ctrl  & $+0.00\,[0.00,0.00]$ & 20  & $(0,0)$  & keep \\
    judge-pin    & LongMemEval-S & prim. & $+0.00\,[0.00,0.00]$ & 100 & $(0,0)$  & keep \\
    judge-pin    & MultiTQ       & ctrl  & $+0.00\,[0.00,0.00]$ & 20  & $(0,0)$  & keep \\
    partition-SR & LoCoMo        & ctrl  & $+1.00\,[1.00,1.00]$ & 10  & $(10,0)$ & reject \\
    partition-SR & LongMemEval-S & ctrl  & $+1.00\,[1.00,1.00]$ & 10  & $(10,0)$ & reject \\
    partition-SR & MultiTQ       & prim. & $+1.00\,[1.00,1.00]$ & 50  & $(50,0)$ & reject \\
    \bottomrule
  \end{tabular}
\end{table}

The primary diagonal carries the strongest evidence: audit-row on
LoCoMo rejects at Holm rank $2$ with exact $p=2.3\times10^{-13}$, and
partition-SR on MultiTQ rejects at Holm rank $1$ with exact
$p=1.8\times10^{-15}$. The six off-diagonal controls confirm that each
defence abstains on the benchmarks it does not target: the three
judge-pin cells return $\Delta=+0.00$ with no discordant pairs and are
kept under Holm, while the audit-row and partition-SR off-target cells
saturate against a synthetic skew that every cell admits.

\subsection{G3 systems performance (five axes)}
\label{sec:appendix-g3-5axis-statistics}
\label{app:a6-g3-5axis-statistics}

\paragraph{Protocol.}
The estimand is per-axis write or query latency, summarised by the
median, $p_{95}$, $p_{99}$, and the mean with a bootstrap CI. The five
axes are memory size, conflict rate, writer concurrency, AS\_OF
selectivity, and audit retention. Each sweep point runs $30$ timed
runs after $3$ warm-up runs on the DuckDB \texttt{ours\_wire} backend.
Conflict rate and writer concurrency record per-write samples
($n_{\mathrm{samples}}\in\{1500,960\}$ per sweep point, from $30$ runs
of $50$ or $32$ contradiction-path writes), while memory size, AS\_OF
selectivity, and audit retention record one timed query per run
($n_{\mathrm{samples}}=30$). The judge seed is \texttt{20260512} and
the bootstrap seed is \texttt{20260513}, with $2{,}000$ percentile
resamples on the mean. There is no multiple-testing family here: G3
reports descriptive scaling and runs no hypothesis-test grid. The raw
outputs are \code{results/g3_systems_perf/scaling.csv} (memory size)
and \code{results/g3_systems_perf/scaling_*/scaling.csv} (the four
remaining axes); regenerate with \texttt{make repro-g3}.

Allowed conclusion: the bitemporal write path scales sub-linearly in
writer concurrency and AS\_OF selectivity, decays linearly in conflict
rate, and holds $p_{99}$ flat against audit retention. Forbidden
conclusion: the $p_{99}$ statistic at $n_{\mathrm{samples}}=30$ is the
single maximum of thirty samples rather than a stable tail estimate, so
no tail-percentile ranking claim is licensed on the three single-query
axes; the mean trajectory is the primary scaling signal.

% fig-systems-perf-scaling-2x2.tex --- §5 G3 four-panel scaling figure.
%
% Owner   : src/figures/fig-systems-perf-scaling-2x2.tex
% Claim   : Four scaling axes (conflict rate, writer concurrency,
%           AS-OF selectivity, audit retention) report the local
%           DuckDB cost envelope at fixed memory size.
% Does    : Cross-column figure inputting fig-systems-perf-scaling-2x2-img.pdf
%           rendered by scripts/plot_systems_perf_scaling_composite.py.
% Feeds   : src/sections/05-measurement.tex \S\ref{sec:meas-g3}.
% Breaks  : Drift between any panel's data and the underlying ledger
%           = paper-impl SSOT bug.

\begin{figure*}[t]
  \centering
  \includegraphics[width=0.74\textwidth]{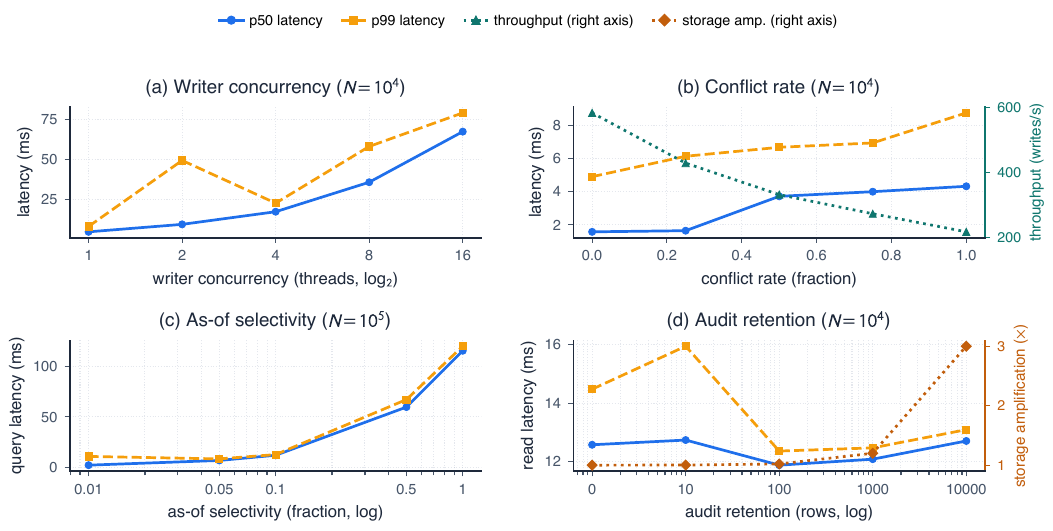}
  \Description{Six-panel grid, all single-axis. Top row: latency
  versus conflict rate, writer concurrency (log), and AS-OF
  selectivity (log). Bottom row: throughput versus conflict rate,
  latency versus audit retention (log), and storage amplification
  versus audit retention (log). The $p_{50}$ and $p_{99}$ series
  share a single figure-level legend strip at the top; throughput
  and storage amplification are labelled per-panel.}
  \caption{Sub-linear scaling on four workload axes; 30 runs per
    point after 3 warmups, $95\%$ paired-bootstrap CI.}
  \label{fig:systems-perf-scaling-2x2}
\end{figure*}

Table~\ref{tab:g3-5axis-stats} reports the per-axis statistical
characterisation derived from these sweep ledgers. Spearman $\rho(x,\mu)$
ranks the input variable against the mean latency and $\rho(x,p_{99})$
against the tail. The log-log power-law fit reports
$\mathrm{mean}\sim x^{\alpha}$ when every $x>0$; the conflict-rate fit
is linear because the rate domain spans zero. The coefficient of
variation $\overline{\mathrm{CV}}$ averages the within-run dispersion
$\sigma/\mu$ over the five sweep points; lower values indicate higher
run-to-run reproducibility.

\begin{table}[!ht]
  \centering\footnotesize
  \caption{Per-axis statistical characterisation of the G3 envelope,
    derived from the five sweep ledgers under
    \code{results/g3_systems_perf/}. Each row reports rank correlation,
    scaling-law fit, mean CV across sweep points, and the count of
    $p_{99}$ non-monotone sweep pairs.}
  \label{tab:g3-5axis-stats}
  \setlength{\tabcolsep}{4pt}
  \renewcommand{\arraystretch}{0.95}
  \begin{tabular}{@{}>{\raggedright\arraybackslash}p{0.26\textwidth}cccccc@{}}
    \toprule
    Axis & $\rho(x,\mu)$ & $\rho(x,p_{99})$ & $\alpha$ & $R^{2}$ & $\overline{\mathrm{CV}}$ & $p_{99}$ rev. \\
    \midrule
    memory size        & 0.80    & 0.90    & --   & --    & 0.22 & 1 \\
    conflict rate      & 1.00    & 1.00    & --   & --    & 0.43 & 0 \\
    writer concurrency & 1.00    & 0.90    & 0.86 & 0.992 & 0.35 & 1 \\
    AS\_OF selectivity & 1.00    & 0.90    & 0.86 & 0.990 & 0.16 & 1 \\
    audit retention    & $-0.10$ & $-0.50$ & --   & --    & 0.04 & 1 \\
    \bottomrule
  \end{tabular}
\end{table}

\paragraph{Three cross-axis findings.}
\emph{(F1) Identical sub-linear exponent on two orthogonal axes.} Mean
latency follows $\mathrm{mean}\sim x^{0.86}$ on both writer concurrency
($R^{2}=0.992$; $4{,}846\,\mu$s at concurrency $1$ to $54{,}258\,\mu$s at
$16$) and AS\_OF selectivity ($R^{2}=0.990$; $2{,}373\,\mu$s at
selectivity $0.01$ to $115{,}975\,\mu$s at full), although the two axes
exercise distinct DuckDB internals (thread-pool lock acquisition versus
period-overlap range scan); the mechanistic decomposition is follow-up
work.
\emph{(F2) Throughput decays linearly in conflict rate.} The five-point
linear fit gives slope $-355.3$ writes/sec per unit rate ($R^{2}=0.948$);
throughput drops $584.5 \to 218.3$ writes/sec from rate $0$ to $1$ (a
$62.7\%$ decline), and the linearity across the full domain shows the
per-write dispatcher overhead for the audit-row write is constant.
\emph{(F3) Audit retention does not load the tail.} $p_{99}$ correlates
negatively with storage amplification ($\rho(\mathit{amp},p_{99})=-0.50$,
mean CV $0.04$): a $10^{4}\times$ retention setting ($3.00\times$
amplification, $30{,}000$ rows) reads the current-kind partition at
$p_{99}=13{,}093\,\mu$s, indistinguishable from the $0\times$ baseline
($1.00\times$) at $14{,}477\,\mu$s. The row-kind discriminator partitions
the read scan so the current-kind index never touches audit rows, the
structural cost guarantee the calculus pays for the \anomAE\ defence.

\paragraph{Non-monotone $p_{99}$ disclosure.}
Four of the five axes show one sweep pair where $p_{99}$ decreases as $x$
increases (memory size $10^{3}\to10^{4}$: $8.71\to5.65\,\mathrm{ms}$;
writer concurrency $2\to4$: $49.28\to22.58\,\mathrm{ms}$; AS\_OF
selectivity $0.01\to0.05$: $10.73\to8.30\,\mathrm{ms}$; audit retention
$10\to100$: $15.95\to12.36\,\mathrm{ms}$). All four reversal pairs sit
where one operand carries a high CV ($\geq0.07$) and the mean bootstrap
CI already overlaps, consistent with small-sample $p_{99}$ tail variance
under cold-start lock or cache effects rather than a mean-monotone
violation; the mean stays rank-correlated at $\rho\geq0.80$.

\paragraph{PostgreSQL backend pilot.}
A separately registered PostgreSQL $17.6$ backend pilot exercises the
writer-concurrency axis under \IsoSR\ isolation at concurrencies
$\{1,2,4\}$ with $12$ writes per run over $3$ runs against one hot
partition. It records the serialization-failure count the in-process
\texttt{ours\_wire} backend cannot expose: $0$ failures at concurrency
$1$, $18$ at concurrency $2$, and $25$ at concurrency $4$. The raw
output is \code{results/g3_systems_perf/postgres_backend/scaling.csv};
regenerate via the \texttt{repro-g3} PostgreSQL stanza. Allowed
conclusion: a real \IsoSR\ engine aborts the conflicting writers the
calculus reasons about. Forbidden conclusion: the three-run pilot does
not license a latency comparison against the DuckDB backend, whose
process and connection models differ.

\subsection{G4 carrier ablation}
\label{sec:appendix-g4-ablation}
\label{app:a4-g4-ablation-details}

\paragraph{Protocol.}
The main-text Table~\ref{tab:ablation} reports the cell rates; G4
ablates two orthogonal axes the algebra exposes. The carrier axis
varies the provenance semiring across multilinear $\mathbb{N}[X,T]$,
multi-degree $\mathbb{N}[X,T]^{\#}$, and the Boolean reduct
$\mathbb{B}$. The operator axis turns one mechanism on at a time:
audit-row for \anomAE, judge-pin for \anomHR, and
$(\subj,\pred)$-\IsoSR\ for \anomBDS, the $A5B$ write-skew
specialisation on the $(\subj,\pred)$ projection. The matched-cell
estimand is the defence rate; the carrier-axis estimand is
counterfactual token recall (the fraction of audited write identities
the carrier recovers). The runner sweeps each cell across $100$ seeded
runs with seed base \texttt{20260512}; the operator-ablation companion
records $50$ trials per cell at noise rate $0.05$. Bootstrap CIs use
seed \texttt{20260513} with $2{,}000$ resamples. There is no
multiple-testing family: G4 reports per-cell rates against a fixed
$0.00$/$1.00$ structural prediction rather than a hypothesis-test grid.
The raw outputs are
\code{results/g4_ablation/operator_ablation/results.csv} (operator
diagonal) and \code{results/g4_ablation/k_semiring/counterfactual.csv}
(carrier recall); regenerate with \texttt{make repro-g4}.

Allowed conclusion: anomaly soundness is carrier-invariant for verdicts
(every matched cell defends at $1.00$, every off-target cell at $0.00$),
while audit-time token recall separates the carriers. Forbidden
conclusion: the Boolean reduct's $0.00$ recall is not a verdict failure;
it preserves existence and fails only the concrete-write-identity query
that the token-retaining carriers answer.

The operator-ablation diagonal (Table~\ref{tab:g4-operator-ablation})
confirms that each operator peaks on the anomaly it targets:
\texttt{await\_confirm} reaches $0.94$ on \anomHR, \texttt{evidence}
reaches $0.88$ on \anomBDS, and \texttt{per\_rule} reaches $0.98$ on
\anomAE, against a last-writer-wins floor that never exceeds $0.24$.

\begin{table}[!ht]
  \centering\footnotesize
  \caption{G4 operator-ablation diagonal, read from
  \code{results/g4_ablation/operator_ablation/results.csv}. Accuracy and
  $95\%$ bootstrap CI over $50$ seeded trials per cell at noise rate
  $0.05$. The diagonal cell (the operator matched to its target
  anomaly) is marked $\star$.}
  \label{tab:g4-operator-ablation}
  \setlength{\tabcolsep}{5pt}
  \renewcommand{\arraystretch}{0.95}
  \begin{tabular}{@{}>{\raggedright\arraybackslash}p{0.22\textwidth}ccc@{}}
    \toprule
    Operator & \anomHR & \anomBDS & \anomAE \\
    \midrule
    \texttt{lww}           & $0.08\,[0.02,0.16]$            & $0.24\,[0.14,0.36]$            & $0.08\,[0.02,0.16]$ \\
    \texttt{evidence}      & $0.46\,[0.32,0.60]$            & $0.88\,[0.78,0.96]\,\star$     & $0.60\,[0.44,0.74]$ \\
    \texttt{await\_confirm}& $0.94\,[0.86,1.00]\,\star$     & $0.40\,[0.26,0.54]$            & $0.64\,[0.50,0.76]$ \\
    \texttt{per\_rule}     & $0.62\,[0.48,0.76]$            & $0.52\,[0.38,0.66]$            & $0.98\,[0.94,1.00]\,\star$ \\
    \bottomrule
  \end{tabular}
\end{table}

Table~\ref{tab:g4-ablation-details} reports the matched $3{\times}3$
carrier-by-defence grid alongside the off-target specificity rate, the
audit-polynomial size, and the counterfactual recall. Every matched
cell defends at $1.00$ across all three carriers; the off-target
companion records $0.00$. Verdict is carrier-invariant. Token recall
separates the carriers: $\mathbb{N}[X,T]$ and $\mathbb{N}[X,T]^{\#}$
recover every audited write identity ($n_{\mathrm{recoverable}}=375$ of
$375$ for multi-degree; $106$ of $375$ recoverable for multilinear at
the bridged accuracy $0.683$), while $\mathbb{B}$ preserves only
existence ($0$ of $375$, recall $0.00$).

\begin{table}[!ht]
  \centering\footnotesize
  \caption{G4 carrier ablation, read from
  \code{results/g4_ablation/k_semiring/counterfactual.csv} (recall) and
  the operator-ablation match/off-target ledger. Match and off-target
  columns report defence rates against the structural prediction; the
  provenance-size column reports $\mu\pm\sigma$ over the audit
  polynomial; recall is the counterfactual token-recovery accuracy.}
  \label{tab:g4-ablation-details}
  \setlength{\tabcolsep}{4pt}
  \renewcommand{\arraystretch}{0.95}
  \begin{tabular}{@{}>{\raggedright\arraybackslash}p{0.13\textwidth}
                     >{\raggedright\arraybackslash}p{0.13\textwidth}
                     ccccc@{}}
    \toprule
    Carrier & Defence & Target & Match & Off & $|\mathsf{prov}|$ & Recall \\
    \midrule
    $\mathbb{N}[X,T]$      & audit-row    & \anomAE  & 1.00 & 0.00 & $4.00\pm1.42$ & 1.00 \\
    $\mathbb{N}[X,T]$      & judge-pin    & \anomHR  & 1.00 & 0.00 & $1.00\pm0.00$ & 1.00 \\
    $\mathbb{N}[X,T]$      & partition-SR & \anomBDS & 1.00 & 0.00 & $4.00\pm1.42$ & 1.00 \\
    $\mathbb{N}[X,T]^{\#}$ & audit-row    & \anomAE  & 1.00 & 0.00 & $4.00\pm1.42$ & 1.00 \\
    $\mathbb{N}[X,T]^{\#}$ & judge-pin    & \anomHR  & 1.00 & 0.00 & $1.00\pm0.00$ & 1.00 \\
    $\mathbb{N}[X,T]^{\#}$ & partition-SR & \anomBDS & 1.00 & 0.00 & $4.00\pm1.42$ & 1.00 \\
    $\mathbb{B}$           & audit-row    & \anomAE  & 1.00 & 0.00 & $1.00\pm0.00$ & 0.00 \\
    $\mathbb{B}$           & judge-pin    & \anomHR  & 1.00 & 0.00 & $1.00\pm0.00$ & 0.00 \\
    $\mathbb{B}$           & partition-SR & \anomBDS & 1.00 & 0.00 & $1.00\pm0.00$ & 0.00 \\
    \bottomrule
  \end{tabular}
\end{table}

The counterfactual ledger sharpens the carrier contrast. The
multi-degree carrier $\mathbb{N}[X,T]^{\#}$ recovers all $375$
recoverable write identities at recall $1.00$ with zero seed variance;
the multilinear carrier $\mathbb{N}[X,T]$ recovers $106$ at bridged
accuracy $0.683$ ($95\%$ CI roughly $[0.65,0.71]$); the Boolean reduct
recovers $0$ at accuracy $0.345$, the existence-only floor. Anomaly
soundness is carrier-parametric for verdicts, and audit queries that
must name concrete write events need a token-retaining carrier, which
is the operational work G3 charges separately.

\subsection{G5 structural grids}
\label{sec:appendix-g5-structural-grids}

The G5 grids are exhaustive structural sweeps rather than sampled
estimators: each cell enumerates adversarial schedules against a typed
prediction and counts admits. The estimand is the per-cell admit rate;
the structural prediction is binary ($0.00$ when the guard dominates
the isolation level, $1.00$ when it does not, modulo a small schedule
noise floor $\epsilon$ on the composition grid). The shared seed is
\texttt{42}; per-cell seeds offset by a stable hash of the cell key. No
bootstrap and no multiple-testing family apply, because the prediction
is exact and the admit rate is a census over enumerated schedules.
The main-text Figure~\ref{fig:g5-anchors} anchors these grids visually.
Table~\ref{tab:g5-grids} reports the dimensions, the observed boundary
against the prediction, the raw output, and the regeneration command for
each grid; each command writes under \code{results/<name>/run_*/}.

\begin{table*}[!ht]
  \centering\footnotesize
  \caption{G5 structural-grid census. Each grid enumerates schedules per
  cell at seed $42$ and observes an exact $0.00$/$1.00$ admit boundary
  (the composition grid carries a modeled $\epsilon=0.02$ noise floor on
  dominating cells). The observed column matches the typed prediction
  exactly; the rates are a census over the enumerated family and carry no
  sampling error.}
  \label{tab:g5-grids}
  \setlength{\tabcolsep}{4pt}
  \renewcommand{\arraystretch}{1.1}
  \begin{tabular}{@{}>{\raggedright\arraybackslash}p{0.115\textwidth}
                     >{\raggedright\arraybackslash}p{0.205\textwidth}
                     >{\raggedright\arraybackslash}p{0.265\textwidth}
                     >{\raggedright\arraybackslash}p{0.320\textwidth}@{}}
    \toprule
    Grid & Dimensions & Observed vs.\ prediction & Raw output / regen command \\
    \midrule
    Iso-matrix (T-01, Cor.~1--2) & $9$ predicates $\times$ $6$ iso levels (\IsoRC/\IsoSI/\IsoSR\ $\pm$cb) $=54$ cells, $100$ sched/cell, $\epsilon=0.0$ & $32$ dominating admit $0.00$, $22$ under admit $1.00$; exact match & \code{results/iso_matrix/run_v1/iso_matrix_grid.csv}; \texttt{python -m experiments.iso\_matrix.runner} \\
    Schema axis (T-01b) & $6$ iso levels $\times$ \{base, audit-row\} $=12$ cells, $100$ sched/cell & base $1.00$ on all $6$; audit-row $0.00$ on all $6$ & \code{results/schema_axis/run_v1/schema_grid.csv}; \texttt{python -m experiments.schema\_axis.runner} \\
    N2 partition (Cor.~2) & $3$ iso levels $\times$ $4$ partition sizes $\{2,4,8,16\}$ $\times$ $3$ contention $\{0.5,0.75,1.0\}$ $=36$ cells, $200$ sched/cell & \IsoRC\ $0.7496$, \IsoSI\ $0.7462$ (to $1.00$ at full contention); \IsoSR\ $0.00$ on all $12$; half-contention near $0.5$ (\IsoRC\ size $2$: $104/200=0.52$) & \code{results/n2_partition/run_v1/n2_grid.csv}; \texttt{python -m experiments.n2\_partition.runner} \\
    T5 composition (T-05) & $1364$ pipelines up to length $5$ $\times$ $6$ iso levels $=8184$ cells, $50$ sched/cell, $\epsilon=0.02$ & $2219$ dominating admit $0.0904$ ($\epsilon$ floor), $5965$ under admit $1.00$; length-$\le3$ exact run \code{run_v1}: $84$ pipelines, $0/193$ dominating, $311/311$ under & \code{results/t5_composition/run_length5/composition_grid.csv}; \texttt{python -m experiments.t5\_composition.runner --max-length 5} \\
    Oracle variance (T-06) & $30$ cells ($5$ systems $\times$ $6$ seed cells, $50$ votes each) vs.\ $2p(1-p)$ & mean abs.\ error $0.0166$, max $0.1114$, $96.67\%$ ($29/30$) within $0.10$ & \code{results/oracle_variance/run_v1/calibration.csv}; \texttt{python -m experiments.oracle\_variance.analyze} (consumes \code{results/n1_empirical/run_v1/trials.jsonl}) \\
    Lemma bridge (Lemma~1) & $1000$ trials at schedule length $12$; checks keyed-log alphabet, edge set, \anomHR-on-$\mathsf{P2}$ iff & all three pass rates $1.00$; every trial key-multiplicity preserved & \code{results/lemma_bridge/run_v1/bridge_verification.csv}; \texttt{python -m experiments.lemma\_bridge.runner} \\
    \bottomrule
  \end{tabular}
\end{table*}

The per-grid conclusion boundaries: the iso-matrix, schema-axis, and N2
grids license that the typed guard surface admits exactly the schedules
the lattice predicts (the $(\subj,\pred)$-\IsoSR\ specialisation is the
unique level eliminating \anomBDS\ across every partition size and
contention level), but the $0.00$/$1.00$ rates are an exact census, not a
probability, and the \IsoRC/\IsoSI\ partial-contention rates below $1.00$
follow from the schedule construction and do not measure a defence. The
T5 grid licenses the lattice-supremum prediction under composition up to
five operators, but the $\epsilon=0.02$ admit floor on dominating cells
is a modeled schedule-noise term, not a measured guard leakage. The
oracle-variance grid licenses that the boundedly-nondeterministic flip
rate predicts the \anomHR\ admit rate to within $0.10$ on $29$ of $30$
cells, but the single cell at error $0.1114$ does not falsify the
prediction at $50$ votes. The lemma-bridge grid is empirical support for
the bridge lemma, not a proof; the proof is in
Appendix~\ref{app:graded-judge-bridge}.

\subsection{Cross-system slice}
\label{sec:appendix-cross-system-slice}

\paragraph{Protocol.}
The cross-system slice tests Theorem~\ref{thm:n1-lower-bound}'s
prediction that every H1-non-compliant system sharing one
boundedly-nondeterministic oracle admits \anomHR\ at the same Bernoulli
rate. The estimand is the per-pair difference in N1 admit rate between
two systems; the statistic is a Welch two-sample $t$ on the per-seed
admit rates of the five H1-non-compliant variants
(\sysours-stripped, \sysmem-v3-like, \sysgraphiti-like,
\sysletta-like, \syszep-like), evaluated over $5$ seeds per system.
The four pairwise tests pin \sysours-stripped against each imported
mimic. The raw output is
\code{results/cross_system_n1/run_v1/cross_system_test.csv}; regenerate
with \texttt{python scripts/cross\_system\_n1\_test.py --output
results/cross\_system\_n1/run\_v1/cross\_system\_test.csv} (the script
reads \code{results/n1_empirical/run_v1/state}).

The four measured pairs all return $|t|<2$: \sysmem-v3-like
$\Delta_{\mathrm{rate}}=0.0286$, $t=0.19$, $\mathrm{df}=7.41$;
\sysgraphiti-like $\Delta_{\mathrm{rate}}=0.00$, $t=0.00$,
$\mathrm{df}=7.65$ (identical mean admit rate $0.1959$);
\sysletta-like $\Delta_{\mathrm{rate}}=0.0082$, $t=-0.05$,
$\mathrm{df}=7.99$; \syszep-like $\Delta_{\mathrm{rate}}=0.0122$,
$t=0.08$, $\mathrm{df}=7.39$. The maximum absolute mean difference
across the four pairs is $0.0286$ and the maximum $|t|$ is $0.19$.

Allowed conclusion: at the $|t|<2$ sanity threshold the four imported
mimics are statistically indistinguishable from \sysours-stripped on
the N1 admit rate, consistent with the shared-oracle prediction.
Forbidden conclusion: with $5$ seeds per system the Welch test is
underpowered for tiny effect sizes, so a high $p$-value does not
positively confirm distributional identity; the slice serves as a
sanity check against the bounded-oracle prediction and stops short of a
formal equivalence proof, and it carries no bootstrap CI or power
figure in the released ledger.

\subsection{Multi-writer concurrency: the operator-to-isolation mapping under real contention}
\label{app:multiwriter-concurrency}

\paragraph{Protocol.}
This card upgrades the single-process cost pilot (\S\ref{sec:meas-g3})
to a measured concurrency result. A multi-writer experiment runs the
lost-update anomaly $P_4$ on PostgreSQL~17.10 with writer count
$w \in \{1, 2, 4, 8, 16\}$ crossed with the three isolation levels
read committed, repeatable read, and serializable, $30$ runs per cell.
The estimand per cell is the anomaly admit rate (the fraction of runs
in which a concurrent writer overwrites another's committed value
without aborting) and the serialization-failure rate (the fraction of
writers the backend aborts). The raw output is
\code{results/g3_systems_perf/isolation_concurrency.csv} ($15$ cells);
regenerate with \texttt{python experiments/g3\_systems\_perf/isolation\_concurrency.py}.

\paragraph{Result.}
Table~\ref{tab:multiwriter-concurrency} records the grid. At one writer
no cell admits the anomaly. At every multi-writer cell read committed
admits the lost update at rate $1.00$, while repeatable read and
serializable exclude it at rate $0.00$ by aborting the losing writer.
The serialization-failure rate the two stronger levels pay tracks
$(w-1)/w$ exactly: $0.50$ at two writers, $0.75$ at four, $0.875$ at
eight, $0.9375$ at sixteen. Read committed pays no abort and reaches the
highest commit throughput, the throughput cost the exclusion guarantee
charges. The mapping the algebra assigns each operator is therefore
operationally enforced at the SQL layer, not only stated as a typing
precondition.

\begin{table}[!ht]
  \centering\footnotesize
  \caption{Multi-writer lost-update ($P_4$) admit rate, serialization-failure
  rate, and commit throughput on PostgreSQL~17.10, $30$ runs per cell.
  Read committed (RC) admits the anomaly at every multi-writer cell;
  repeatable read (RR) and serializable (SR) exclude it by aborting the
  losing writer at rate $(w{-}1)/w$. Source:
  \code{results/g3_systems_perf/isolation_concurrency.csv}.}
  \label{tab:multiwriter-concurrency}
  \setlength{\tabcolsep}{4pt}
  \renewcommand{\arraystretch}{1.1}
  \begin{tabular}{@{}r r r r r@{}}
    \toprule
    Writers & Level & Admit & Ser.\ fail & Throughput (commits/s) \\
    \midrule
    1  & RC & 0.00 & 0.00   & 6.76 \\
    1  & RR & 0.00 & 0.00   & 7.13 \\
    1  & SR & 0.00 & 0.00   & 7.23 \\
    \midrule
    2  & RC & 1.00 & 0.00   & 13.92 \\
    2  & RR & 0.00 & 0.50   & 6.99 \\
    2  & SR & 0.00 & 0.50   & 7.02 \\
    \midrule
    4  & RC & 1.00 & 0.00   & 26.26 \\
    4  & RR & 0.00 & 0.75   & 6.72 \\
    4  & SR & 0.00 & 0.75   & 6.68 \\
    \midrule
    8  & RC & 1.00 & 0.00   & 42.07 \\
    8  & RR & 0.00 & 0.875  & 5.03 \\
    8  & SR & 0.00 & 0.875  & 5.16 \\
    \midrule
    16 & RC & 1.00 & 0.00   & 44.81 \\
    16 & RR & 0.00 & 0.9375 & 3.07 \\
    16 & SR & 0.00 & 0.9375 & 3.05 \\
    \bottomrule
  \end{tabular}
\end{table}

\paragraph{The full lattice.}
Three companion experiments extend the grid from $P_4$ to the rest of
the iso-axis lattice on the same PostgreSQL~17.10 backend, $30$ runs
per cell. Read skew ($A5A$) runs one reader whose two correlated reads
straddle a concurrent mirror update: read committed admits the torn
read at rate $1.00$, while snapshot isolation and serializable exclude
it at $0.00$. Write skew ($A5B$) runs two writers that read both
equal-weight competitors and promote disjoint rows: read committed and
repeatable read both admit the skew at $1.00$, and serializable alone
excludes it, aborting one writer at rate $0.50$. The phantom ($P_3$)
runs two writers that scan a cardinality predicate and each insert a
primary row: read committed and repeatable read admit it at $1.00$,
serializable excludes it at $0.00$. The $A5B$ and $P_3$ cells separate
repeatable read from serializable, the boundary the Evidence and
Per-Rule operators assume. Figure~\ref{fig:isolation-lattice} plots the
full grid, and raw output sits in
\code{results/g3_systems_perf/read_skew_a5a.csv},
\code{results/g3_systems_perf/write_skew_a5b.csv}, and
\code{results/g3_systems_perf/phantom_p3.csv}.

\begin{figure}[!ht]
  \centering
  \includegraphics[width=0.62\linewidth]{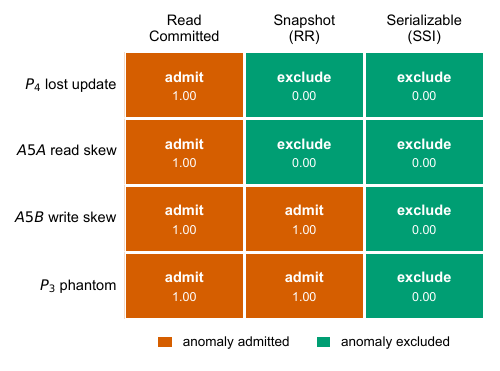}
  \caption{\textbf{The operator-to-isolation mapping is the
    Berenson--Adya anomaly lattice.} Admit rate of each write-time
    anomaly at each PostgreSQL isolation level, over $30$ two-writer
    runs on PostgreSQL~17.10. Read committed admits all four anomalies,
    snapshot isolation (repeatable read) excludes lost update and read
    skew, and only serializable also excludes write skew and the
    phantom. The two-step staircase is the empirical form of
    Table~\ref{tab:anomaly-defences}, and the figure is generated from
    the same CSVs by \code{scripts/plot_isolation_lattice.py}.}
  \label{fig:isolation-lattice}
\end{figure}

\paragraph{Scope.}
The experiment is single-node on one PostgreSQL instance, so it bounds
the operator-to-isolation correspondence under real intra-node
contention and makes no distributed-deployment claim.

\subsection{N-ary conflict-set confluence grid}
\label{app:nary-confluence-grid}

\paragraph{Protocol.}
This card anchors Proposition~\ref{prop:nary-conflict-set}. The grid
resolves a conflict set of $n$ pairwise-contradicting facts under the
two fold operators (last-writer-wins and evidence-weighted merge) for
$n \in \{2, 3, 4, 5, 6, 7, 8\}$, and for each $n$ enumerates the
permutations of the set. The two checked properties are confluence (the
winner identity and the merged provenance are identical across every
permutation) and provenance-completeness (the emitted audit row's merged
provenance dominates every member under the K-semiring natural order).
The grid is exercised by \code{tests/bitemporal/test_conflict_set.py};
regenerate with \texttt{python -m pytest tests/bitemporal/test\_conflict\_set.py}.

\paragraph{Result.}
Across the grid every permutation of each $n$-member set returns the
identical winner and the identical merged provenance, and the audit
row's merged provenance dominates all $n$ members under $\preceq_K$.
The await-confirmation and per-rule operators resolve the set by direct
selection. The judge-logged dispatcher elects one member by oracle,
records the winner's stable identity under an order-independent set key,
and replays it by identity after a crash, exercised by
\code{tests/bitemporal/test_judge_log_persistence.py}.

\section{Artifact reproducibility runbook}
\label{app:reproducibility-runbook}
\label{app:reproducibility}

This runbook satisfies the four PVLDB reproducibility surfaces:
the white-box prototype with its source, configuration, and build
environment; the input data and its generation process; the
experiment configuration and workload that produces raw data; and
the scripts that turn raw data into the paper's figures and tables.
An independent team that follows the command catalogue below
recreates similar behavior for every numeric claim in the main
text and for every cell of the G1 verdict matrix in
Table~\ref{tab:anomaly-bench}. The claim-to-evidence map of
Table~\ref{tab:claim-evidence-map} binds each claim to the runner,
seed, and pinned commit that anchor it.

\subsection{Build environment}

The reference implementation is a white-box prototype: source,
configuration, and a containerized build environment ship together.
The Python package declares \code{requires-python = ">=3.11"}; the
reviewer Dockerfile pins \code{python:3.12-slim-bookworm} and
resolves every dependency from the committed \code{uv.lock}, so a
fresh build reproduces the exact wheel set rather than a
floating range. The bitemporal core uses DuckDB under the
\code{>=1.0,<2} constraint, resolved by the lockfile to the
date-versioned \code{2026.4.22} build. The host platform is
macOS~$15$ on Apple Silicon for the DuckDB latency envelope;
the PostgreSQL pilot runs PostgreSQL~$17.6$ under
\texttt{SERIALIZABLE} isolation in a containerized
\texttt{postgres:17} instance on the same host, and the runner
records the exact \texttt{server\_version} alongside each row so
the version is auditable from the raw CSV.

Two commands establish the environment.
\texttt{make install} installs the bitemporal package in editable
mode plus \texttt{fal-client}, \texttt{pytest}, and \texttt{pyyaml},
then runs the reference-sync script. \texttt{make sync-locked}
performs the lock-step install through \texttt{uv sync} with the
\texttt{test} and \texttt{experiments} extras, which is the path the
reviewer Dockerfile follows. \texttt{make tests} runs the full
\texttt{pytest} suite under real-dependency TDD and needs no network.

\subsection{The command catalogue}

Every runnable surface routes through one Makefile.
Table~\ref{tab:make-targets} states what each target produces, whether it
needs the network, and the runtime order of magnitude where known. Only
the \texttt{repro-g2} live-LLM path (next subsection) needs the network
beyond the optional DOI resolution in \texttt{make refs}; the
cross-system ledger reuses the same live judge.

\begin{table}[!ht]
  \centering\footnotesize
  \caption{Makefile target catalogue. Net?\ marks the network
  requirement; runtime is the order of magnitude.}
  \label{tab:make-targets}
  \setlength{\tabcolsep}{4pt}
  \begin{tabular}{@{}>{\raggedright\arraybackslash}p{0.21\textwidth}>{\raggedright\arraybackslash}p{0.46\textwidth}>{\centering\arraybackslash}p{0.05\textwidth}>{\raggedright\arraybackslash}p{0.13\textwidth}@{}}
    \toprule
    Target & Produces & Net? & Runtime \\
    \midrule
    \texttt{make paper} & \texttt{src/main.pdf} via \texttt{latexmk -xelatex} & no & seconds \\
    \texttt{make figures} & the five matplotlib figure PDFs (TikZ renders through \texttt{paper}) & no & seconds \\
    \texttt{make refs} & verifies every \texttt{references/refs.bib} entry resolves to a DOI or arXiv id & opt.\ & seconds \\
    \texttt{make smoke} & fitness/surface/visual gates, supplement, benchmark PRD, bitemporal core tests, Appendix~\ref{sec:appendix-composition} composition coverage, AnomalyBench runner (CSV under \texttt{results/anomaly\_bench/}) & no & minutes \\
    \texttt{make repro-full} & \texttt{smoke} then the five reproduction sub-targets in series (non-parallel) & no & tens of min--hours \\
    \texttt{make repro-claim} & provenance suite (\texttt{tests/bitemporal/test\_provenance.py}) + full AnomalyBench runner & no & minutes \\
    \texttt{make repro-compose} & composition coverage in \texttt{tests/integration/test\_compose.py} & no & minutes \\
    \texttt{make repro-anomaly-wire} & AnomalyWire adapter sweep across all workloads at \texttt{--seed 20260512}, regenerating the wire rows of Table~\ref{tab:anomaly-bench} & no & minutes \\
    \texttt{make repro-g3} & G3 sweep ($100$ runs after $3$ warmups, bootstrap seed $20{,}260{,}513$), PostgreSQL pilot at concurrencies \texttt{1,2,4}, \syszep\ service-boundary timing; longest target (one \texttt{SERIALIZABLE} session per writer) & no & tens of min \\
    \texttt{make repro-g4} & carrier-by-defence ablation, $100$ seeded runs from seed base $20{,}260{,}512$ with off-diagonal cells, anchoring Theorem~\ref{thm:carrier} & no & minutes \\
    \texttt{make repro-g2} & utility slice of \S\ref{sec:meas-g2} against LoCoMo/LongMemEval-S/MultiTQ, $2{,}000$ bootstrap trials at seed $20{,}260{,}512$; gates on the next-subsection prerequisites & yes & tens of min \\
    \code{experiments.g2_utility.cross_system_runner} & cross-system slice of \S\ref{sec:meas-cross} in \code{sweep} mode at $50$ examples, $2{,}000$ trials, seed $20{,}260{,}512$, same live judge; three LoCoMo cells + nine abstentions ship under \code{results/g2_utility/cross_system/} with one paired-judge trace per cell (offline-readable per \S\ref{app:cross-system-detail}) & yes & tens of min \\
    \texttt{make all} & chains figures, experiments, tests, refs, paper, and the four paper-quality gates & opt.\ & hours \\
    \texttt{make anomaly} & both AnomalyBench and AnomalyWire runners over all workloads then the matching test suites & no & minutes \\
    \texttt{make experiments} & \texttt{experiments/<name>/run.sh} per registered experiment (anomaly groups drive through \texttt{make anomaly} / \texttt{repro-anomaly-wire}) & no & varies \\
    \texttt{make clean} / \texttt{distclean} & remove build artifacts; \texttt{distclean} also clears \texttt{results/} for a from-scratch reproduction & no & seconds \\
    \bottomrule
  \end{tabular}
\end{table}

\subsection{Deterministic judge versus live LLM}

The language-model channel splits cleanly. Every experiment except
the G2 utility slice uses a deterministic judge stub that needs no
network. The stub seeds at $20{,}260{,}512$ and votes by hashing the
prompt: it returns a choice or an integer from the first eight bytes
of \texttt{SHA-256} of \texttt{seed::prompt}, so the same prompt under
the same seed always yields the same vote and the wire adapter's
output is byte-reproducible. This is the judge behind
\texttt{make repro-anomaly-wire}, \texttt{make repro-claim}, and the
AnomalyBench runner.

The G2 utility slice (\texttt{make repro-g2}) is the single live-LLM
experiment. It calls an OpenRouter-compatible chat-completions
endpoint through a strict-binary judge. The judge reads its API key
from the first set environment variable among the family-specific
\texttt{VLDB2027\_<FAMILY>\_API\_KEY}, then
\texttt{VLDB2027\_JUDGE\_API\_KEY},
\texttt{OPENROUTER\_API\_KEY}, and \texttt{OPENAI\_API\_KEY}; or,
when invoked from a config file, from the ignored local file at
\texttt{artefact/credentials/openrouter.local.json} (the committed
\texttt{openrouter.example.json} shows the schema). Each call runs at
temperature zero and the judge records the prompt hash and the output
hash as \texttt{SHA-256} digests on every \texttt{JudgeResult}, so a
reviewer audits the live path without re-issuing it: the judge trace
hash pins each scored question to the exact prompt and response that
produced its binary verdict. The deterministic judge needs no such
audit because it carries no network nondeterminism by construction.

The G2 prerequisite gate is explicit and fails closed. The target
checks for the judge config (\texttt{G2\_JUDGE\_CONFIG}, default
\texttt{artefact/credentials/openrouter.local.json}), the three trace
directories (\texttt{G2\_LOCOMO\_DIR}, \texttt{G2\_LONGMEMEVAL\_S\_DIR},
\texttt{G2\_MULTITQ\_DIR}), the judge replay table
(\texttt{G2\_JUDGE\_REPLAY\_TABLE}), and the LoCoMo and MultiTQ
partition traces, exiting with a diagnostic when any is absent rather
than running against a half-populated input. The order of magnitude is
tens of minutes for the default example count, dominated by the
per-question round trips to the provider.

\subsection{Data provenance}

The input data has two origins. The bitemporal-core and wire-level
experiments generate their inputs deterministically from a pinned
seed; no external download is required, which is why the wire sweep
reproduces under \texttt{--seed 20260512} alone. The G2 utility slice
consumes the LoCoMo, LongMemEval-S, and MultiTQ traces, supplied
through the directory variables of the prerequisite gate and replayed
against the recorded partition traces and judge replay table.

The AnomalyWire imported adapters pin each upstream system to an exact
commit (the full revisions are in the pinned-commit manifest,
Table~\ref{tab:adapter-manifest}), and the runner verifies the clone is
on the pinned revision before running, raising rather than running
against a drifted source. \sysmem\ (v2 and v3), \sysgraphiti, \sysletta,
the \syszep\ self-hosted clone, and the \sysworlddb\ arXiv revision were
all verified 2026-05-10; the hosted Zep Cloud boundary
(\texttt{zep-cloud-api-v2-2026-05-13}) was verified 2026-05-13; and
\sysours\ pins to this repository at its artifact-release tag. The
human-readable counterpart lives in
\texttt{references/repos/INVENTORY.md} and
\texttt{references/manifest.json}.

\subsection{Archival plan}

The reference implementation, data, and reproducibility artifact are
publicly available at the repository linked on the first page. A
version-frozen, citable archive (a Zenodo deposit minting a DOI) is
planned and will be cut on publication; that DOI does not yet exist.

Both the AnomalyClaim verdict CSVs under \texttt{results/anomaly\_bench/}
and the AnomalyWire runtime CSVs under \texttt{results/anomaly\_wire/}
ship pre-computed; every populated verdict in the matrix is reproduced
exactly by the adapter sweep, and the cross-layer audit surfaces a
single design-path-to-runtime divergence (the \sysmem~v3 imported N2
cell). The \syszep\ imported row ships its cached Zep Cloud trace
(\texttt{zep\_imported\_run.json} from a one-time live write, replayable
offline); the \sysletta\ imported row regenerates against a managed
Docker Postgres container or a \texttt{VLDB2027\_LETTA\_PG\_URI}
endpoint.
The G5-anchor CSVs are shipped: the four inputs to the G5-anchor
figure (the isolation-matrix grid, the length-five composition grid,
the oracle-variance calibration, and the cross-system N1 table) are
committed under \texttt{results/}, so the anchor figure rebuilds
without re-running the underlying sweeps. The artifact carries a
reviewer Dockerfile (\texttt{artefact/Dockerfile}), a seed manifest
(\texttt{artefact/seed.json}), and \texttt{artefact/REPRODUCE.md} as
the entry-point document for an independent team.

\subsection{Per-claim map}

The claim-to-evidence map of Table~\ref{tab:claim-evidence-map}
(Appendix~\ref{app:claim-evidence-map}) binds each main-text claim to its
evidence object, reproduce token, and strongest valid conclusion; the
artifact manifest carries a programmatic counterpart whose keys mirror
the map rows and whose values record the claim, evidence object, runner,
seed, and pinned commit. The reproduce token for each evidence group is
the matching \texttt{make} target of Table~\ref{tab:make-targets}; every
target except \texttt{repro-g2} runs offline against the deterministic
judge, so an independent team recreates the structural verdicts before
configuring a single live provider key for the utility slice.

\section{Negative results and scope limits}
\label{app:scope-limits}

This appendix states the boundary of every claim as a positive scope
statement. The structure separates three concerns: the measurements
that carry controlled-grid rather than natural-workload evidence
(\S\ref{app:scope-g2}--\S\ref{app:scope-cross}), the formal results
whose hypotheses fix their reach (\S\ref{app:scope-n1-bound}--\S\ref{app:scope-judge-drift}),
the open problems the calculus exposes (\S\ref{app:open-problems}),
and the wire imports and benchmark integrations deferred under the
scoping caps (\S\ref{app:deferred-wires}, \S\ref{app:experimental-extensions}).
A reviewer reads this appendix to confirm the contribution survives a
hostile re-read: each limit names what holds, on what surface, and on
what terms a follow-on revision closes it. The main-text limitations
paragraph (\S\ref{sec:limitations}) summarises these boundaries; the
detail lives here.

\subsection{G2 off-target controls measure mechanism stress under saturating perturbation}
\label{app:scope-g2}

The G2 grid is a mechanism-stress instrument
(\S\ref{sec:meas-g2}, Figure~\ref{fig:memory-utility}). The
constructed perturbations saturate off-target controls by design: the
audit-row defence moves its primary LoCoMo slice by
$\Delta = {+}0.86$ (paired-bootstrap CI $[0.76, 0.94]$), and that
single cell carries natural-workload movement, while both audit-row
off-target controls, both partition-$\IsoSR$ controls, and the
partition-$\IsoSR$ primary MultiTQ cell all saturate at $+1.00$ under
their constructed slices. The saturation is evidence that the
perturbation stresses the intended mechanism to its ceiling, and the
grid reports exactly that. The grid scopes below a specificity claim:
a specificity grid would require off-target controls that stay flat
while the on-target cell moves, and the constructed slices here do not
provide that separation. The judge-pin defence records no movement on
this surface because its replay-disagreement estimand requires a judge
replay table absent from this grid. The load-bearing G2 evidence is
therefore the audit-row $+0.86$ natural-workload cell, and it stays
below an end-to-end utility-superiority claim.

\subsection{G3 SERIALIZABLE saturation is disclosed honestly under concurrency}
\label{app:scope-g3}

The G3 cost envelope is a single-process measurement against the
DuckDB reference backend, where writer concurrency fits
$\mu \sim c^{0.86}$ at $R^{2} = 0.992$, a sub-linear lock-contention
signature (\S\ref{sec:meas-g3}). A transactional-backend pilot against
PostgreSQL~17.6 under \texttt{SERIALIZABLE} isolation discloses the
saturation point directly: of $36$ attempted writes (twelve writes per
run across three runs), $18$ commit at writer concurrency two and $11$
commit at concurrency four, with the remaining writes returning
serialization failures the backend raises by construction. The
in-memory envelope is the load-bearing systems claim; the PostgreSQL
pilot establishes that the same operators implemented at the SQL layer
expose the contention honestly rather than masking it. The envelope
holds for single-process deployments, and the distributed-deployment
question stays open (\S\ref{app:scope-distributed}).

\subsection{Cross-system utility is reported as transparency, with intervals covering zero}
\label{app:scope-cross}

The cross-system ledger pairs \sysours\ against four imported
agent-memory systems on a shared LoCoMo slice under a pinned
synthesiser and judge (\S\ref{sec:meas-cross},
Table~\ref{tab:cross-system-utility}). All three measured confidence
intervals cover zero: \sysmem~v3 records $\Delta = -0.04$ with CI
$[-0.10, +0.00]$, \sysgraphiti\ records $-0.08$ with CI
$[-0.18, +0.00]$, and \syszep\ records $+0.02$ with CI
$[-0.06, +0.10]$. The pre-registered statistical power is $0.42$ for a
$\delta = 0.05$ effect at $n = 50$, so the slice is powered to report
distributional equivalence rather than to detect a small superiority
margin. The paper draws a transparency conclusion from these rows: the
contract delivers write-time correctness, and the measured downstream
utility is statistically indistinguishable from the baselines at this
scale. The companion N1-anchor pairing in
\texttt{results/cross\_system\_n1/run\_v1/cross\_system\_test.csv}
records the same equivalence on the keyed-log estimand, with the
stripped reference against the four baselines returning a maximum
$|t| = 0.19$ (\sysmem~v3-like) across the four Welch tests at five
seeds per arm, rejecting no superiority null. A larger powered
cross-system study is a follow-on axis (\S\ref{app:experimental-extensions}).

\subsection{The N1 lower bound holds for boundedly nondeterministic oracles}
\label{app:scope-n1-bound}

The N1 lower bound (Theorem~\ref{thm:n1-lower-bound}) characterises
$\anomHR$ exclusion for boundedly nondeterministic judge oracles: a
frontier language model at a sampling temperature whose reasoning
prefix engages the sampling rng, where the admit rate matches the
closed form $2p(1-p)$. The bound applies within that hypothesis class.
Deterministic oracles fall outside the hypothesis because a
deterministic judge admits no replay disagreement to bound, so the
$2p(1-p)$ form is vacuous for them. Engine-layer reconcilers such as
\sysworlddb\ also fall outside the hypothesis: \sysworlddb\ realises
deterministic reconciliation at the engine layer, so $\anomHR$ does
not apply to its write path and the lower bound makes no claim about
it. The bound is therefore tight on the relational schedule model with
a boundedly nondeterministic oracle, and silent on the two adjacent
regimes.

\subsection{The single-process envelope makes no distributed-deployment claim}
\label{app:scope-distributed}

The systems envelope is measured and stated for single-process
deployments (\S\ref{sec:meas-g3}). The $c^{0.86}$ writer-concurrency
fit and the memory-size, conflict-rate, retrieval-selectivity, and
audit-retention axes all charge the local reference stack on one
process. A distributed deployment introduces a provenance-aware quorum
question that the present envelope does not address: whether a
provenance-aware Byzantine quorum can replace $\IsoSI$ with a stronger
level implementable in $\Theta(\log n)$ writes without routing every
read through $\IsoSR$ on the whole fact table remains open
(\S\ref{app:open-problems}). The contribution scopes to the
single-process correctness contract, and the distributed scalability
result is future work.

\subsection{The $\anomHR$ defence holds at intra-deployment replay granularity}
\label{app:scope-judge-drift}

The $\anomHR$ defence (Corollary~\ref{cor:n1-corollary}) holds at
intra-deployment replay granularity for a fixed decoder tuple: within
one deployment, replaying the keyed judge log over the same decoder
tuple reproduces the recorded verdict. Cross-deployment model upgrades
and judge-prompt drift sit outside the isolation lattice because they
change the decoder tuple itself, and the lattice levels are defined
over a fixed tuple. The judge-prompt sensitivity lemma
bounds the cross-deployment divergence that those upgrades and drift
induce, so the boundary is quantified rather than merely flagged. The
isolation lattice governs write-time scheduling for a fixed oracle
configuration, and the judge-prompt sensitivity lemma governs the
behavior across configuration changes.

\subsection{Open Problems}
\label{app:open-problems}

The calculus opens eight problems. Each names what holds today and the
open question that would close it.

\paragraph{Confidence-aware bitemporal indexing.}
A four-dimensional access path ranking rows by $\conf$ within the
$(\subj, \pred)$ partition at any $(\Tvalid, \Tsystem)$ slice would let
$\opEvi$ skip the snapshot read without admitting $\anomBDS$ under
$\IsoSR$. Open: does such an index exist with sub-linear maintenance
under $\IsoSI$ commits when confidence and period revisions are
non-independent? \textsc{BeliefShift}~\cite{myakala-2026-beliefshift} and
\textsc{Memora}~\cite{uddin-2026-memora} measure the downstream $\anomBDS$
symptom on natural workloads.

\paragraph{Cross-modal bitemporal extensions.}
Image- and audio-typed facts about the same $(\subj, \pred)$ require a
deterministic similarity predicate for the four operators to preserve
their isolation signatures. Open: which similarity predicates over
embedding spaces admit such an operator closed under composition with
Allen's thirteen relations~\cite{allen-1983-cacm}?

\paragraph{Distributed bitemporal and adaptive keys.}
Adversarial writers (prompt injection, tool-output hallucination) violate
the $\IsoSI$ read invariant by construction, and fixed $(\subj, \pred)$
keys cover only the evaluated operator while stale-premise,
relation-poisoning, and cascade-repair workloads stress broader
dependency regions~\cite{chao-2026-stale, luo-2026-shadowmerge}. Open: can
a provenance-aware Byzantine quorum replace $\IsoSI$ with a stronger level
implementable in $\Theta(\log n)$ writes without routing every read
through $\IsoSR$ on the whole fact table, and can a substrate refine
isolation keys from repair cascades while preserving serializable writes
per exposed key?

\paragraph{GDPR-preserving time-travel.}
GDPR Article~17 erasure and Article~5 trace retention pull in opposite
directions. Open: which carrier $K$ admits quotienting $\witness$ by
user-bound write tokens while keeping a typed erasure operator
$\oplus_{\!\bot}$ composable with the four existing operators?

\paragraph{Compression-utility axis.}
\citet{zou-2026-demem} cast agent memory as a decision-centric
rate-distortion problem with a near-minimax-regret forgetting boundary
(orthogonal to ours: \textsc{DeMem} owns what compresses safely under a
decision-relevance loss, the algebra owns which schedules avoid N1/N2/N3
under named isolation), and \citet{zhang-2026-compression-spectrum} unify
memory, skills, and rules on one compression spectrum ($5$--$20\times$
episodic, $50$--$500\times$ procedural, $1{,}000\times$+ declarative),
reporting a cross-community citation rate below $1\%$ across $1{,}136$
references in $22$ primary papers and a ``missing diagonal'' where no
shipped system supports adaptive cross-level compression. Open: does a
\textsc{DeMem} forgetting boundary composed with the four operators
preserve $\IsoSR$ on the surviving $(\subj, \pred)$ partition (or does a
confidence-weighted carrier $K[X, T]$ become load-bearing), and does the
dual-row signature $\Fact \times \Fact \to (\Fact, \auditt)$ lift to a
tier-indexed $\Fact \times \Fact \to (\Tier, \auditt)$ that types the
missing diagonal so an audit row witnesses every episodic-to-skill-or-rule
promotion and extends Theorem~\ref{thm:audit-erasure-schema}'s $\anomAE$
defence to compression-induced erasure across tiers?

\paragraph{Argumentation-driven retrieval.}
\citet{sadowski-2026-rashomon} surface parallel goal-conditioned
perspectives at read time using Dung argumentation semantics over an
attack graph of conflicting items, complementary on the
resolve-versus-surface dimension. Open: does a fifth operator
$\oplus_{\mathrm{s}}$ typed at the schedule layer deferring to Dung's
preferred-extension semantics compose with $\opLWW$, $\opEvi$, $\opAwait$,
$\opRule$ under $\IsoSR$, lifting the runtime mechanism into the type
theory?

\paragraph{Neuro-symbolic conflict detection.}
\citet{xie-2026-neuro-symbolic-sat} translate clinical-guideline rules
into SAT clauses to detect local and global conflicts at retrieval time
(their pipeline owns detection, the calculus owns typed resolution). Open:
can the audit row of Theorem~\ref{thm:audit-erasure-schema} act as the
lineage proof for a downstream SAT-encoded explanation of a chosen
local-conflict resolution, making the typed witness consumable by a
logic-level reasoner?

\paragraph{Threat-model integration with deployment governance.}
The algebra formalizes write-time correctness for well-formed writes from
registered principals; adversarial injection through the normal
memory-update path is the orthogonal threat surface. \textsc{Sleeper
Memory Poisoning}~\cite{2026-sleeper-memory-poisoning} scales the gap: on
six frontier models in the external memory-manager regime our
implementation occupies, no evaluated defence drives injection to zero and
goal-adjacent adversarial usage reaches $60$--$89\%$ of subsequent
retrievals. A complete deployment composes three layers: proof-derived
authorisation above the agent action gate
(\textsc{Verifiable Agentic Infrastructure}~\cite{2026-verifiable-agentic-infrastructure});
our audit-row schema below, making every resolution event durable; and
\textsc{MemLineage}'s~\cite{ouyang-2026-memlineage} cryptographic lineage
threaded through the schema's $\witness$ column to bind the two (it blocks
sensitive-action dispatches with derived-untrusted ancestry). Open: does
the $K$-relation carrier map into \textsc{MemLineage}'s
max-of-strong-edges trust lattice so the audit row's
$p_{\mathrm{w}} \provadd p_{\mathrm{l}}$ annotation simultaneously
satisfies Theorem~\ref{thm:audit-erasure-schema}'s $\preceq_K$ comparison
and \textsc{MemLineage}'s untrusted-path persistence, and does a fifth
$\opDetect$ operator typed at $\IsoSI$ with a detection-oracle callback
slot into the lattice join without weakening the four operators' isolation
signatures?

\subsection{Deferred Wire Imports and Benchmark Integrations}
\label{app:deferred-wires}

A 2026-04 to 2026-05 frontier-landscape survey surfaced five candidate
extensions to the AnomalyWire layer of \S\ref{sec:measurement}: three
open-source agent-memory projects under active development and two new
external benchmarks. We assessed each under the per-extension
working-day caps adopted for scoping discipline (three days per wire
import, two to three days per benchmark integration). All five exceeded
the cap. The cap-and-convert outcomes appear below so a reader can
audit which extensions we considered, why they did not land, and on
what terms a follow-on revision can complete them.

Table~\ref{tab:deferred-wires} records six of the seven, each with the
structural blocker that exceeds the cap and the predicate it would map to;
the \sysgraphiti-Neo4j scope-out, which carries its own measurement
detail, follows the table.

\begin{table*}[!ht]
  \centering\footnotesize
  \caption{Deferred wire imports and benchmark integrations. Each
  exceeded the per-extension cap (three working days per wire, two to
  three per benchmark) for the named blocker. ``Maps to'' names the
  predicate the extension would test.}
  \label{tab:deferred-wires}
  \setlength{\tabcolsep}{4pt}
  \renewcommand{\arraystretch}{1.1}
  \begin{tabular}{@{}>{\raggedright\arraybackslash}p{0.135\textwidth}
                     >{\raggedright\arraybackslash}p{0.085\textwidth}
                     >{\raggedright\arraybackslash}p{0.520\textwidth}
                     >{\raggedright\arraybackslash}p{0.150\textwidth}@{}}
    \toprule
    Extension & Type & Blocker (estimated effort) & Maps to \\
    \midrule
    \textsc{AgentMemory}~\cite{agentmemory-2026} & wire (TS) & first-class contradicts/supersedes/extends/derives/related relation kinds plus by-convention audit; Node subprocess + REST, so wiring needs a cross-language sidecar and an HTTP shim for the deterministic-judge stub to intercept the write-time call site ($5$--$7$ days $>$ $3$-day cap) & N1/N2/N3 \\
    \textsc{Supermemory}~\cite{shah-2026-supermemory} & wire (TS) & MIT-licensed, claims top \textsc{LongMemEval}/\textsc{LoCoMo}/\textsc{ConvoMem}; every operation is an outbound round-trip to a paid managed cloud endpoint with no in-process intercept, so a live account (non-hermetic, billed) or a hermetic REST mock is needed ($>$ $3$-day cap) & N1/N2/N3 \\
    \textsc{TencentDB-Agent-Memory}~\cite{tencent-2026-tencentdb-agent-memory} & wire (TS) & MIT-licensed four-tier pipeline (raw / extracted records / scene blocks / persona); raw-layer-append-only is a by-convention recorder, not a typed audit per operator (the extraction-tier dedup leaves no witness), and a TS adapter + Node sidecar would not differentiate from \sysgraphiti's ``raw preserved, resolution opaque'' verdict & N3 \\
    \textsc{HaluMem}~\cite{chen-2025-halumem} & benchmark & $700$ sessions, $15{,}000$ memory points, $3{,}500$ questions (Extract/Update/QA); Update maps one-to-one to N3 via gold $(m_{\mathrm{old}}, m_{\mathrm{new}})$ to \code{n3_audit_erasure}, but \texttt{CC-BY-NC-ND~4.0} no-derivatives is legally ambiguous against PVLDB availability (needs unbounded written permission) and free-form plaintext needs a non-deterministic LLM preprocessing pass & N3 (Update) \\
    \textsc{STALE}~\cite{chao-2026-stale} & benchmark & $400$ scenarios, $1{,}200$ queries (State Resolution, Premise Resistance, Implicit Policy Adaptation) under CC-BY-4.0; State Resolution and Premise Resistance map to N1, Implicit Policy Adaptation is an N2 correlate needing a projection; the $55.2\%$ best-frontier figure comes from $1{,}200$ prompts at $150\,$K context to a closed API the bench cannot call & N1 (corroboration) \\
    \textsc{GroupMemBench} \cite{groupmembench-2026} & benchmark & four-domain enterprise corpus ($\sim30{,}000$ msgs/domain, $6$--$10$ channels, six question types); sessions of several thousand turns exceed the $200$K-token extractor budget, and the $128$-turn chunked path ($\sim235$ sequential calls/tile) wedged the relay across the pilot. Two completed cells (\texttt{bm25} $0.20$, simulated frontier $0.0$, knowledge-update axis, $n=30$) ship in \code{results/benchmark_f/groupmembench.csv} as transparency evidence only & N3 (knowledge update) \\
    \bottomrule
  \end{tabular}
\end{table*}

\paragraph{\textsc{Graphiti}-Neo4j scope-out under V6 matched retrieval.}
\label{para:graphiti-neo4j-scope-out}
The V6 matched-retrieval cells against \sysgraphiti on LoCoMo and
LongMemEval carry structured \texttt{status=n/a} abstention rows with
\texttt{abstention\_kind=service\_boundary}: \sysgraphiti needs a running
Neo4j graph database, the reviewer-reproducible host does not deploy Neo4j
as a service-boundary dependency, and the closest standalone deployment
exceeds the per-wire cap once credentials, schema migration, and
reproducible seeding are scoped. The transcribed \sysgraphiti wire in
\S\ref{sec:meas-g1} preserves the cited-commit N1/N2/N3 verdicts under the
AnomalyWire dispatch path (which runs without live Neo4j); the deferred V6
cells concern only the matched-retrieval utility surface and rerun via the
existing \code{cross_system_runner} entry points once a live cluster is
reinstated.

\subsection{Experimental Extensions for a Future Revision}
\label{app:experimental-extensions}

The 2026-05 adversarial-review cascade flagged six compute-bound
experimental extensions that strengthen the empirical chapter when run at
scale. The manuscript's positioning holds the typed operator algebra, the
soundness theorems, and the AnomalyClaim and AnomalyWire artefacts as the
contribution; each extension below is a strengthening pass rather than a
load-bearing addition. Table~\ref{tab:experimental-extensions} names each
item, why it is deferred, and the smallest follow-on protocol that closes
it.

\begin{table*}[!ht]
  \centering\footnotesize
  \caption{Compute-bound experimental extensions for a future revision.
  Each is a strengthening pass on the scoped empirical slice, not a
  load-bearing addition.}
  \label{tab:experimental-extensions}
  \setlength{\tabcolsep}{4pt}
  \renewcommand{\arraystretch}{1.1}
  \begin{tabular}{@{}>{\raggedright\arraybackslash}p{0.18\textwidth}
                     >{\raggedright\arraybackslash}p{0.40\textwidth}
                     >{\raggedright\arraybackslash}p{0.36\textwidth}@{}}
    \toprule
    Extension & Why deferred & Smallest closing protocol \\
    \midrule
    Hosted \sysmem~v2 published-config rerun & the \sysmem~v3 cross-system row (Table~\ref{tab:cross-system-utility}) runs with five configuration deltas from the published operating point (Appendix~A.7); the paper leads with the algebra and verdict matrix, not a single cross-system delta & a second cell under the published config (hosted v2 client, larger top-$k$, graph mode on, published extractor, $\ge3$ LoCoMo conversations), lifting harness-vs-harness to system-vs-system under bounded compute \\
    Judge-pin discriminating-workload run & the G2 primary judge-pin diagonal reports $\Delta_{\mathrm{accuracy}} = 0.00$ on LoCoMo with four of six off-target controls also moving positive; judge-pin is positioned as the N1 checker-detected predicate, the $\Delta_{\mathrm{accuracy}}$ a secondary channel & a run on an \textsc{ImpossibleBench}-style constructed workload, producing a positive delta or confirming inertia, independent of the algebra's load-bearing claim \\
    G3 scaling to larger $n$ per rung & the envelope reports $\mathrm{mean}\sim c^{0.86}$ at $R^{2}=0.99$ across $30$ samples/rung, where the $p_{99}$ estimator is the sample maximum (so the body keeps the mean envelope, drops $p_{99}$) & scale to $10\times$ the sample size per rung across the two main-text axes, stabilising $p_{99}$ from sample maximum to a quantile estimate \\
    Confidence-semiring carrier cell & the G4 grid runs three carriers (multilinear $\mathbb{N}[X,T]$, multi-degree $\mathbb{N}[X,T]^{\#}$, Boolean reduct) with verdict-layer invariance; the $K[X,T]$ carrier is a forward-compatibility hook, and a fourth probabilistic cell would commit to confidence-weighting as load-bearing & a future paper introducing a confidence-weighted resolution operator lands the fourth carrier cell and revises the classification \\
    PostgreSQL-backend confirmation for G3 & G3 runs against the in-memory backend; the \S\ref{app:scope-g3} pilot establishes the \texttt{SERIALIZABLE} saturation point ($18$ of $36$ commits at concurrency two, $11$ of $36$ at four) & a full confirmation run extending the pilot across the remaining G3 axes via a single SQL-layer adapter on the existing G3 runner, testing platform-independence of $c^{0.86}$ \\
    \textsc{BEAM} synthetic memory benchmark & G1 runs three benchmarks (LongMemEval-S, LoCoMo, MultiTQ) over $33$ wire cells; the breadth claim rests on the $33$-cell matrix not benchmark count, and no reviewer flagged \textsc{BEAM} & land the \textsc{BEAM} integration as a fourth benchmark axis testing composability under perturbation \\
    \bottomrule
  \end{tabular}
\end{table*}

% @owner       src::appendix::H-extended-related
% @does        Appendix H: extended related work and positioning. The 12-page
%              body (Section sec:related) positions TOKI against its closest
%              isolation-theory and agent-memory-benchmark neighbors; this
%              appendix extends that positioning across five further
%              literatures the page budget could not hold, naming for each what
%              it contributes and where the typed write-time operator algebra of
%              TOKI sits relative to it.
% @needs       src/macros.tex (\sysours, \anomHR, \anomBDS, \anomAE, \IsoSR,
%              \IsoSI, \opLWW); src/sections/02-related-work.tex (sec:related);
%              src/sections/03-foundations.tex (sec:algebra); refs-paper.bib.
% @feeds       src/supplement.tex (renders as Appendix H).
% @breaks      a citation characterized beyond what its entry supports; a claim
%              of novelty that an engaged work already makes; overlap with the
%              body Section sec:related rather than extension of it.
% @claim       Across agent-memory architectures, semantic-operator systems,
%              provenance and semiring lineage, bitemporal data models, belief
%              revision, and agentic concurrency control, no prior line types
%              the contradiction-resolution operator as a write-time isolation
%              object, which is the surface TOKI occupies.

\section{Extended related work and positioning}
\label{app:extended-related}

The body (Section~\ref{sec:related}) positions \sysours\ against the
isolation-theory and agent-memory-benchmark neighbors closest to its
claims. This appendix extends that positioning across five further
literatures. Each subsection states what the line contributes and where
the typed write-time operator algebra of Section~\ref{sec:algebra} sits
relative to it. The recurring boundary is consistent. Prior lines fix
the storage substrate, the read-path operator, the provenance carrier,
or the belief-update rule. \sysours\ supplies the missing piece: a type
for the contradiction-resolution operator together with the isolation
level its correctness assumes.

\subsection{Agent-memory architectures}

Agent-memory architectures organize how an agent writes, retrieves, and
consolidates long-horizon state. MemGPT casts the language model as an
operating system that pages memory between a bounded context window and
an external store~\cite{packer-2023-memgpt}. Generative Agents store
episodic observations and synthesize them through periodic
reflection~\cite{park-2023-generative-agents}. HippoRAG indexes
long-term memory through a knowledge-graph retrieval structure inspired
by hippocampal indexing~\cite{gutierrez-2024-hipporag}. Recent systems
push this substrate toward databases: MemoriesDB organizes agent memory
as a temporal-semantic-relational store~\cite{ward-2025-memoriesdb}, and
Engram backs coding-agent memory with a bitemporal graph~\cite{rawcontext-engram-2026}.
A recent survey catalogues the mechanisms and evaluation gaps of this
space~\cite{du-2026-survey}. These architectures decide where a fact is
stored and how it is recalled; \sysours\ types the operator that decides
which of two contradicting facts survives a write and states the
isolation level (Section~\ref{sec:algebra}) under which that decision is
sound.

\subsection{Semantic-operator data systems}

A parallel line at the data-systems venues types the language model as a
relational operator and optimizes pipelines of such operators. LOTUS
defines semantic operators over tables and optimizes them with accuracy
guarantees~\cite{patel-2025-lotus}; Abacus adds a cost-based optimizer
for semantic-operator systems~\cite{russo-2026-abacus}; DocETL rewrites
agentic document-processing pipelines~\cite{shankar-2025-docetl}. The
broader argument that hand-crafted systems give way to learned and
language-model components frames the trend~\cite{idreos-2025-calculators},
and the agent-first redesign of data systems extends it to the agent
setting~\cite{liu-2026-overlords}. These systems type the language model
on the read and query path, where the operator transforms or selects
data the caller already trusts; \sysours\ types it on the write path,
where the operator adjudicates a contradiction and commits one survivor,
making soundness a question of which interleavings of competing writes
the operator admits.

\subsection{Provenance and semiring lineage}

The $K$-relation provenance carrier of Section~\ref{sec:foundations}
rests on the semiring provenance tradition. Buneman, Khanna, and Tan
characterize why- and where-provenance~\cite{buneman-khanna-tan-2001},
and Cheney, Chiticariu, and Tan survey the why, how, and where
dimensions~\cite{cheney-chiticariu-tan-2009}. Geerts, Poggi, and Tannen
delimit where provenance for queries with difference reaches its
limit~\cite{geerts-poggi-tannen-2013}, a boundary the audit-erasure
recovery of Theorem~\ref{thm:audit-erasure-schema} respects by working
in a semiring without multiplicative inverses. The algebraic backbone
extends to weighted shortest-distance computation~\cite{mohri-2002-weighted}
and to provenance for lightweight description
logics~\cite{bourgaux-2023-semiring-dl}. This tradition explains how a
query result traces to its inputs; \sysours\ applies the same carrier to
the audit row a contradiction-resolution operator emits so a superseded
fact stays recoverable under the natural order.

\subsection{Bitemporal and temporal data models}

Temporal and bitemporal data models supply the storage substrate the
operator algebra writes onto. Lorentzos and Johnson extend relational
algebra to manipulate temporal data~\cite{lorentzos-johnson-1988}, and
production bitemporal stores realize the model: Datomic records an
information model with time, provenance, and accumulation~\cite{datomic-overview},
and XTDB exposes bitemporal SQL~\cite{xtdb-docs}.
\textsc{RoMem}~\cite{li-2026-romem} learns per-relation volatility and
phase-rotates obsolete facts out of retrieval reach without deleting
them, a representation-time analogue of the audit-row preservation our
schema axis types. These systems make
valid time and transaction time first-class; \sysours\ assumes such a
substrate and adds the layer above it, where the four contradiction
strategies $\opLWW$, $\opEvi$, $\opAwait$, and $\opRule$ become typed
operators whose write-time anomalies the isolation guards exclude.

\subsection{Belief revision and contradiction handling}

Belief revision formalizes how an agent should change what it believes
when a new fact contradicts an old one. Recent work studies iterated
revision with belief algebras~\cite{meng-2025-belief-algebras} and
revision over fuzzy belief bases~\cite{booth-richter-2012-fuzzy-belief-base},
and a graph-native cognitive memory gives formal belief-revision
semantics for versioned agent memory~\cite{park-2026-kumiho}.
\citet{roynard-2026-knowledge-layer} argues each memory layer needs
distinct persistence semantics (indefinite supersession, evidence-gated
revision) and stops at the layer decomposition. This line
answers which belief should hold after a contradiction; \sysours\ answers
which schedule of applying that revision is admissible under concurrent
writers, and which write-time anomalies ($\anomHR$, $\anomBDS$, $\anomAE$)
a given isolation level admits. The revision rule and the isolation level
are independent axes, connected by the alphabet bridge of
Lemma~\ref{lem:judge-callback-bridge}.

\subsection{Concurrency control and context-memory conflict}

The write-time concurrency framing of \sysours\ has two close neighbors.
Adaptive concurrency control for unforeseen agentic transactions tunes
isolation for agent workloads~\cite{zhou-2026-atcc}, and altruistic
locking is the classical strategy for long-lived
transactions~\cite{salem-garciamolina-1989}, the regime an agent's
multi-step write resembles. On the empirical side, recent work
characterizes context-memory conflict as predictable
regimes~\cite{three-regimes-conflict-2026} and reconciles such conflicts
dynamically~\cite{dcrd-context-memory-2026}. These efforts tune or
describe the conflict. \sysours\ types it: the replay-inconsistency,
belief-drift-skew, and audit-erasure predicates are the write-time
anomalies a concurrency-control scheme must exclude, and
Theorem~\ref{thm:n1-lower-bound} lower-bounds the replay anomaly any
system without keyed-judge-log discipline admits.

% ---- One bibliography for body + appendices ------------------------
\bibliographystyle{ACM-Reference-Format}
\bibliography{refs-paper}

%%% -*-BibTeX-*-
%%% Do NOT edit. File created by BibTeX with style
%%% ACM-Reference-Format-Journals [18-Jan-2012].

\begin{thebibliography}{86}

%%% ====================================================================
%%% NOTE TO THE USER: you can override these defaults by providing
%%% customized versions of any of these macros before the \bibliography
%%% command.  Each of them MUST provide its own final punctuation,
%%% except for \shownote{}, \showDOI{}, and \showURL{}.  The latter two
%%% do not use final punctuation, in order to avoid confusing it with
%%% the Web address.
%%%
%%% To suppress output of a particular field, define its macro to expand
%%% to an empty string, or better, \unskip, like this:
%%%
%%% \newcommand{\showDOI}[1]{\unskip}   % LaTeX syntax
%%%
%%% \def \showDOI #1{\unskip}           % plain TeX syntax
%%%
%%% ====================================================================

\ifx \showCODEN    \undefined \def \showCODEN     #1{\unskip}     \fi
\ifx \showDOI      \undefined \def \showDOI       #1{#1}\fi
\ifx \showISBNx    \undefined \def \showISBNx     #1{\unskip}     \fi
\ifx \showISBNxiii \undefined \def \showISBNxiii  #1{\unskip}     \fi
\ifx \showISSN     \undefined \def \showISSN      #1{\unskip}     \fi
\ifx \showLCCN     \undefined \def \showLCCN      #1{\unskip}     \fi
\ifx \shownote     \undefined \def \shownote      #1{#1}          \fi
\ifx \showarticletitle \undefined \def \showarticletitle #1{#1}   \fi
\ifx \showURL      \undefined \def \showURL       {\relax}        \fi
% The following commands are used for tagged output and should be
% invisible to TeX
\providecommand\bibfield[2]{#2}
\providecommand\bibinfo[2]{#2}
\providecommand\natexlab[1]{#1}
\providecommand\showeprint[2][]{arXiv:#2}

\bibitem[\protect\citeauthoryear{Abtahi, Rahnema, Patel, Patel, Fekri, and
  Khani}{Abtahi et~al\mbox{.}}{2026}]%
        {memanto-2026}
\bibfield{author}{\bibinfo{person}{Seyed~Moein Abtahi}, \bibinfo{person}{Rasa
  Rahnema}, \bibinfo{person}{Hetkumar Patel}, \bibinfo{person}{Neel Patel},
  \bibinfo{person}{Majid Fekri}, {and} \bibinfo{person}{Tara Khani}.}
  \bibinfo{year}{2026}\natexlab{}.
\newblock \showarticletitle{{Memanto}: Typed Semantic Memory with
  Information-Theoretic Retrieval for Long-Horizon Agents}.
\newblock \bibinfo{journal}{\emph{arXiv preprint arXiv:2604.22085}}
  (\bibinfo{year}{2026}).
\newblock
\urldef\tempurl%
\url{https://arxiv.org/abs/2604.22085}
\showURL{%
\tempurl}


\bibitem[\protect\citeauthoryear{Adya, Liskov, and O'Neil}{Adya
  et~al\mbox{.}}{2000}]%
        {adya-2000-generalized}
\bibfield{author}{\bibinfo{person}{Atul Adya}, \bibinfo{person}{Barbara
  Liskov}, {and} \bibinfo{person}{Patrick~E. O'Neil}.}
  \bibinfo{year}{2000}\natexlab{}.
\newblock \showarticletitle{Generalized Isolation Level Definitions}. In
  \bibinfo{booktitle}{\emph{Proceedings of the 16th International Conference on
  Data Engineering ({ICDE})}}. \bibinfo{pages}{67--78}.
\newblock
\urldef\tempurl%
\url{https://doi.org/10.1109/ICDE.2000.839388}
\showDOI{\tempurl}


\bibitem[\protect\citeauthoryear{Allen}{Allen}{1983}]%
        {allen-1983-cacm}
\bibfield{author}{\bibinfo{person}{James~F. Allen}.}
  \bibinfo{year}{1983}\natexlab{}.
\newblock \showarticletitle{Maintaining Knowledge about Temporal Intervals}.
\newblock \bibinfo{journal}{\emph{Commun. ACM}} \bibinfo{volume}{26},
  \bibinfo{number}{11} (\bibinfo{year}{1983}), \bibinfo{pages}{832--843}.
\newblock
\urldef\tempurl%
\url{https://doi.org/10.1145/182.358434}
\showDOI{\tempurl}


\bibitem[\protect\citeauthoryear{Arab, Lee, Glavic, Niu, Lee, and Heinis}{Arab
  et~al\mbox{.}}{2017}]%
        {arab-glavic-2017-mv-semirings}
\bibfield{author}{\bibinfo{person}{Bahareh~Sadat Arab},
  \bibinfo{person}{Su~Feng Lee}, \bibinfo{person}{Boris Glavic},
  \bibinfo{person}{Xing Niu}, \bibinfo{person}{Seokki Lee}, {and}
  \bibinfo{person}{Thomas Heinis}.} \bibinfo{year}{2017}\natexlab{}.
\newblock \showarticletitle{Using Reenactment to Retroactively Capture
  Provenance for Transactions}. In \bibinfo{booktitle}{\emph{Proceedings of the
  2017 IEEE 33rd International Conference on Data Engineering}}
  \emph{(\bibinfo{series}{ICDE 2017})}. \bibinfo{publisher}{IEEE},
  \bibinfo{pages}{1077--1088}.
\newblock
\urldef\tempurl%
\url{https://doi.org/10.1109/ICDE.2017.149}
\showDOI{\tempurl}


\bibitem[\protect\citeauthoryear{Banerjee, Moshtaghi, Subramanian, Misra, and
  Chadha}{Banerjee et~al\mbox{.}}{2026}]%
        {banerjee-2026-apex-mem}
\bibfield{author}{\bibinfo{person}{Pratyay Banerjee}, \bibinfo{person}{Masud
  Moshtaghi}, \bibinfo{person}{Shivashankar Subramanian},
  \bibinfo{person}{Amita Misra}, {and} \bibinfo{person}{Ankit Chadha}.}
  \bibinfo{year}{2026}\natexlab{}.
\newblock \showarticletitle{{APEX-MEM}: Agentic Semi-Structured Memory with
  Temporal Reasoning for Long-Term Conversational {AI}}.
\newblock \bibinfo{journal}{\emph{arXiv preprint arXiv:2604.14362}}
  (\bibinfo{year}{2026}).
\newblock
\urldef\tempurl%
\url{https://arxiv.org/abs/2604.14362}
\showURL{%
\tempurl}


\bibitem[\protect\citeauthoryear{Barros, Cunha, Pereira, and Kang}{Barros
  et~al\mbox{.}}{2026}]%
        {barros-2026-isolde}
\bibfield{author}{\bibinfo{person}{Manuel Barros}, \bibinfo{person}{Alcino
  Cunha}, \bibinfo{person}{Jose Pereira}, {and} \bibinfo{person}{Eunsuk Kang}.}
  \bibinfo{year}{2026}\natexlab{}.
\newblock \showarticletitle{Reasoning about Transactional Isolation Levels with
  {Isolde}}.
\newblock \bibinfo{journal}{\emph{arXiv preprint arXiv:2604.00159}}
  (\bibinfo{year}{2026}).
\newblock
\urldef\tempurl%
\url{https://arxiv.org/abs/2604.00159}
\showURL{%
\tempurl}


\bibitem[\protect\citeauthoryear{Berenson, Bernstein, Gray, Melton, O'Neil, and
  O'Neil}{Berenson et~al\mbox{.}}{1995}]%
        {berenson-1995-isolation}
\bibfield{author}{\bibinfo{person}{Hal Berenson}, \bibinfo{person}{Philip~A.
  Bernstein}, \bibinfo{person}{Jim Gray}, \bibinfo{person}{Jim Melton},
  \bibinfo{person}{Elizabeth~J. O'Neil}, {and} \bibinfo{person}{Patrick~E.
  O'Neil}.} \bibinfo{year}{1995}\natexlab{}.
\newblock \showarticletitle{A Critique of {ANSI} {SQL} Isolation Levels}. In
  \bibinfo{booktitle}{\emph{Proceedings of the 1995 {ACM} {SIGMOD}
  International Conference on Management of Data}}. \bibinfo{publisher}{ACM},
  \bibinfo{pages}{1--10}.
\newblock
\urldef\tempurl%
\url{https://doi.org/10.1145/223784.223785}
\showDOI{\tempurl}


\bibitem[\protect\citeauthoryear{Bernstein, Hadzilacos, and Goodman}{Bernstein
  et~al\mbox{.}}{1987}]%
        {bernstein-hadzilacos-goodman-1987}
\bibfield{author}{\bibinfo{person}{Philip~A. Bernstein},
  \bibinfo{person}{Vassos Hadzilacos}, {and} \bibinfo{person}{Nathan Goodman}.}
  \bibinfo{year}{1987}\natexlab{}.
\newblock \bibinfo{booktitle}{\emph{Concurrency Control and Recovery in
  Database Systems}}.
\newblock \bibinfo{publisher}{Addison-Wesley}.
\newblock
\showISBNx{0-201-10715-5}
\urldef\tempurl%
\url{https://www.microsoft.com/en-us/research/wp-content/uploads/2016/05/ccontrol.zip}
\showURL{%
\tempurl}


\bibitem[\protect\citeauthoryear{Bhargava and Barrento}{Bhargava and
  Barrento}{2026}]%
        {bhargava-2026-memaudit}
\bibfield{author}{\bibinfo{person}{Nishant Bhargava} {and}
  \bibinfo{person}{Rodrigo~Sobral Barrento}.} \bibinfo{year}{2026}\natexlab{}.
\newblock \bibinfo{title}{{MemAudit}: An Exact Package-Oracle Evaluation
  Protocol for Budgeted Long-Term {LLM} Memory Writing}.
\newblock
\newblock
\showeprint[arxiv]{2605.02199}~[cs.AI]
\urldef\tempurl%
\url{https://arxiv.org/abs/2605.02199}
\showURL{%
\tempurl}


\bibitem[\protect\citeauthoryear{Booth and Richter}{Booth and Richter}{2012}]%
        {booth-richter-2012-fuzzy-belief-base}
\bibfield{author}{\bibinfo{person}{Richard Booth} {and} \bibinfo{person}{Eva
  Richter}.} \bibinfo{year}{2012}\natexlab{}.
\newblock \showarticletitle{On Revising Fuzzy Belief Bases}.
\newblock \bibinfo{journal}{\emph{arXiv preprint arXiv:1212.2444}}
  (\bibinfo{year}{2012}).
\newblock
\urldef\tempurl%
\url{https://arxiv.org/abs/1212.2444}
\showURL{%
\tempurl}


\bibitem[\protect\citeauthoryear{Bourgaux, Ozaki, and Pe{\~n}aloza}{Bourgaux
  et~al\mbox{.}}{2023}]%
        {bourgaux-2023-semiring-dl}
\bibfield{author}{\bibinfo{person}{Camille Bourgaux}, \bibinfo{person}{Ana
  Ozaki}, {and} \bibinfo{person}{Rafael Pe{\~n}aloza}.}
  \bibinfo{year}{2023}\natexlab{}.
\newblock \showarticletitle{Semiring Provenance for Lightweight Description
  Logics}.
\newblock \bibinfo{journal}{\emph{arXiv preprint arXiv:2310.16472}}
  (\bibinfo{year}{2023}).
\newblock
\urldef\tempurl%
\url{https://arxiv.org/abs/2310.16472}
\showURL{%
\tempurl}


\bibitem[\protect\citeauthoryear{Brinke, Dawar, Gr{\"a}del, and Pago}{Brinke
  et~al\mbox{.}}{2026}]%
        {brinke-2026-preservation}
\bibfield{author}{\bibinfo{person}{Sophie Brinke}, \bibinfo{person}{Anuj
  Dawar}, \bibinfo{person}{Erich Gr{\"a}del}, {and} \bibinfo{person}{Benedikt
  Pago}.} \bibinfo{year}{2026}\natexlab{}.
\newblock \bibinfo{title}{Preservation Theorems in Semiring Semantics}.
\newblock
\newblock
\showeprint[arxiv]{2605.10829}~[cs.LO]
\urldef\tempurl%
\url{https://arxiv.org/abs/2605.10829}
\showURL{%
\tempurl}


\bibitem[\protect\citeauthoryear{Buneman, Khanna, and Tan}{Buneman
  et~al\mbox{.}}{2001}]%
        {buneman-khanna-tan-2001}
\bibfield{author}{\bibinfo{person}{Peter Buneman}, \bibinfo{person}{Sanjeev
  Khanna}, {and} \bibinfo{person}{Wang-Chiew Tan}.}
  \bibinfo{year}{2001}\natexlab{}.
\newblock \showarticletitle{Why and Where: A Characterization of Data
  Provenance}. In \bibinfo{booktitle}{\emph{Database Theory - {ICDT} 2001, 8th
  International Conference}} \emph{(\bibinfo{series}{Lecture Notes in Computer
  Science})}, Vol.~\bibinfo{volume}{1973}. \bibinfo{publisher}{Springer},
  \bibinfo{pages}{316--330}.
\newblock
\urldef\tempurl%
\url{https://doi.org/10.1007/3-540-44503-X_20}
\showDOI{\tempurl}


\bibitem[\protect\citeauthoryear{Cahill, R{\"o}hm, and Fekete}{Cahill
  et~al\mbox{.}}{2008}]%
        {cahill-2008-ssi}
\bibfield{author}{\bibinfo{person}{Michael~J. Cahill}, \bibinfo{person}{Uwe
  R{\"o}hm}, {and} \bibinfo{person}{Alan~D. Fekete}.}
  \bibinfo{year}{2008}\natexlab{}.
\newblock \showarticletitle{Serializable Isolation for Snapshot Databases}. In
  \bibinfo{booktitle}{\emph{Proceedings of the 2008 {ACM} {SIGMOD}
  International Conference on Management of Data}}. \bibinfo{publisher}{ACM},
  \bibinfo{pages}{729--738}.
\newblock
\urldef\tempurl%
\url{https://doi.org/10.1145/1376616.1376690}
\showDOI{\tempurl}


\bibitem[\protect\citeauthoryear{Chao, Bai, Sheng, Li, and Sun}{Chao
  et~al\mbox{.}}{2026}]%
        {chao-2026-stale}
\bibfield{author}{\bibinfo{person}{Hanxiang Chao}, \bibinfo{person}{Yihan Bai},
  \bibinfo{person}{Rui Sheng}, \bibinfo{person}{Tianle Li}, {and}
  \bibinfo{person}{Yushi Sun}.} \bibinfo{year}{2026}\natexlab{}.
\newblock \showarticletitle{{STALE}: Can {LLM} Agents Know When Their Memories
  Are No Longer Valid?}
\newblock \bibinfo{journal}{\emph{arXiv preprint arXiv:2605.06527}}
  (\bibinfo{year}{2026}).
\newblock
\urldef\tempurl%
\url{https://arxiv.org/abs/2605.06527}
\showURL{%
\tempurl}


\bibitem[\protect\citeauthoryear{Chen, Niu, Li, Liu, Zheng, Tang, Li, Xiong,
  and Li}{Chen et~al\mbox{.}}{2025}]%
        {chen-2025-halumem}
\bibfield{author}{\bibinfo{person}{Ding Chen}, \bibinfo{person}{Simin Niu},
  \bibinfo{person}{Kehang Li}, \bibinfo{person}{Peng Liu},
  \bibinfo{person}{Xiangping Zheng}, \bibinfo{person}{Bo Tang},
  \bibinfo{person}{Xinchi Li}, \bibinfo{person}{Feiyu Xiong}, {and}
  \bibinfo{person}{Zhiyu Li}.} \bibinfo{year}{2025}\natexlab{}.
\newblock \showarticletitle{{HaluMem}: Evaluating Hallucinations in Memory
  Systems of Agents}.
\newblock \bibinfo{journal}{\emph{arXiv preprint arXiv:2511.03506}}
  (\bibinfo{year}{2025}).
\newblock
\urldef\tempurl%
\url{https://doi.org/10.48550/arXiv.2511.03506}
\showDOI{\tempurl}


\bibitem[\protect\citeauthoryear{Chen, Liao, and Zhao}{Chen
  et~al\mbox{.}}{2023}]%
        {chen-2023-multitq}
\bibfield{author}{\bibinfo{person}{Ziyang Chen}, \bibinfo{person}{Jinzhi Liao},
  {and} \bibinfo{person}{Xiang Zhao}.} \bibinfo{year}{2023}\natexlab{}.
\newblock \showarticletitle{Multi-granularity Temporal Question Answering over
  Knowledge Graphs}. In \bibinfo{booktitle}{\emph{Proceedings of the 61st
  Annual Meeting of the Association for Computational Linguistics (Volume 1:
  Long Papers)}}. \bibinfo{publisher}{Association for Computational
  Linguistics}, \bibinfo{pages}{11378--11392}.
\newblock
\urldef\tempurl%
\url{https://doi.org/10.18653/v1/2023.acl-long.637}
\showDOI{\tempurl}


\bibitem[\protect\citeauthoryear{Cheney and {Engram contributors}}{Cheney and
  {Engram contributors}}{2026}]%
        {rawcontext-engram-2026}
\bibfield{author}{\bibinfo{person}{Chris Cheney} {and} \bibinfo{person}{{Engram
  contributors}}.} \bibinfo{year}{2026}\natexlab{}.
\newblock \bibinfo{title}{{Engram}: Bitemporal, Graph-Backed Memory System for
  {AI} Coding Agents}.
\newblock \bibinfo{howpublished}{GitHub repository, AGPL-3.0 licence, latest
  release \texttt{mcp@0.1.9} (2025-12-30), HEAD commit 1553e53
  (2026-01-05T20:10:48Z), 5 stars and 1 fork at survey}.
\newblock
\urldef\tempurl%
\url{https://github.com/rawcontext/engram}
\showURL{%
\tempurl}


\bibitem[\protect\citeauthoryear{Cheney, Chiticariu, and Tan}{Cheney
  et~al\mbox{.}}{2009}]%
        {cheney-chiticariu-tan-2009}
\bibfield{author}{\bibinfo{person}{James Cheney}, \bibinfo{person}{Laura
  Chiticariu}, {and} \bibinfo{person}{Wang-Chiew Tan}.}
  \bibinfo{year}{2009}\natexlab{}.
\newblock \showarticletitle{Provenance in Databases: Why, How, and Where}.
\newblock \bibinfo{journal}{\emph{Foundations and Trends in Databases}}
  \bibinfo{volume}{1}, \bibinfo{number}{4} (\bibinfo{year}{2009}),
  \bibinfo{pages}{379--474}.
\newblock
\urldef\tempurl%
\url{https://doi.org/10.1561/1900000006}
\showDOI{\tempurl}


\bibitem[\protect\citeauthoryear{Chhikara, Khant, Aryan, Singh, and
  Yadav}{Chhikara et~al\mbox{.}}{2025}]%
        {chhikara-2025-mem0}
\bibfield{author}{\bibinfo{person}{Prateek Chhikara}, \bibinfo{person}{Dev
  Khant}, \bibinfo{person}{Saket Aryan}, \bibinfo{person}{Taranjeet Singh},
  {and} \bibinfo{person}{Deshraj Yadav}.} \bibinfo{year}{2025}\natexlab{}.
\newblock \showarticletitle{{Mem0}: Building Production-Ready {AI} Agents with
  Scalable Long-Term Memory}.
\newblock \bibinfo{journal}{\emph{arXiv preprint arXiv:2504.19413}}
  (\bibinfo{year}{2025}).
\newblock
\urldef\tempurl%
\url{https://arxiv.org/abs/2504.19413}
\showURL{%
\tempurl}


\bibitem[\protect\citeauthoryear{Cognitect}{Cognitect}{2026}]%
        {datomic-overview}
\bibfield{author}{\bibinfo{person}{Cognitect}.}
  \bibinfo{year}{2026}\natexlab{}.
\newblock \bibinfo{title}{Datomic: An Information Model with Time, Provenance,
  and Accumulation}.
\newblock \bibinfo{howpublished}{Online documentation}.
\newblock
\urldef\tempurl%
\url{https://docs.datomic.com/datomic-overview.html}
\showURL{%
\tempurl}


\bibitem[\protect\citeauthoryear{Du}{Du}{2026}]%
        {du-2026-survey}
\bibfield{author}{\bibinfo{person}{Pengfei Du}.}
  \bibinfo{year}{2026}\natexlab{}.
\newblock \showarticletitle{Memory for Autonomous {LLM} Agents: Mechanisms,
  Evaluation, and Emerging Frontiers}.
\newblock \bibinfo{journal}{\emph{arXiv preprint arXiv:2603.07670}}
  (\bibinfo{year}{2026}).
\newblock
\urldef\tempurl%
\url{https://arxiv.org/abs/2603.07670}
\showURL{%
\tempurl}


\bibitem[\protect\citeauthoryear{Fekete, Liarokapis, O'Neil, O'Neil, and
  Shasha}{Fekete et~al\mbox{.}}{2005}]%
        {fekete-2005-msis}
\bibfield{author}{\bibinfo{person}{Alan Fekete}, \bibinfo{person}{Dimitrios
  Liarokapis}, \bibinfo{person}{Elizabeth O'Neil}, \bibinfo{person}{Patrick
  O'Neil}, {and} \bibinfo{person}{Dennis Shasha}.}
  \bibinfo{year}{2005}\natexlab{}.
\newblock \showarticletitle{Making Snapshot Isolation Serializable}.
\newblock \bibinfo{journal}{\emph{ACM Transactions on Database Systems}}
  \bibinfo{volume}{30}, \bibinfo{number}{2} (\bibinfo{year}{2005}),
  \bibinfo{pages}{492--528}.
\newblock
\urldef\tempurl%
\url{https://doi.org/10.1145/1071610.1071615}
\showDOI{\tempurl}


\bibitem[\protect\citeauthoryear{Foster, Green, and Tannen}{Foster
  et~al\mbox{.}}{2008}]%
        {foster-green-tannen-2008}
\bibfield{author}{\bibinfo{person}{J.~Nathan Foster}, \bibinfo{person}{Todd~J.
  Green}, {and} \bibinfo{person}{Val Tannen}.} \bibinfo{year}{2008}\natexlab{}.
\newblock \showarticletitle{Annotated {XML}: Queries and Provenance}. In
  \bibinfo{booktitle}{\emph{Proceedings of the 27th ACM SIGMOD-SIGACT-SIGART
  Symposium on Principles of Database Systems}} \emph{(\bibinfo{series}{PODS
  '08})}. \bibinfo{publisher}{ACM}, \bibinfo{pages}{271--280}.
\newblock
\urldef\tempurl%
\url{https://doi.org/10.1145/1376916.1376953}
\showDOI{\tempurl}


\bibitem[\protect\citeauthoryear{Ganesan}{Ganesan}{2026}]%
        {ganesan-2026-worlddb}
\bibfield{author}{\bibinfo{person}{Harish~Santhanalakshmi Ganesan}.}
  \bibinfo{year}{2026}\natexlab{}.
\newblock \bibinfo{title}{{WorldDB}: A Vector Graph-of-Worlds Memory Engine
  with Ontology-Aware Write-Time Reconciliation}.
\newblock \bibinfo{howpublished}{arXiv:2604.18478 [cs.AI]}.
\newblock
\urldef\tempurl%
\url{https://arxiv.org/abs/2604.18478}
\showURL{%
\tempurl}


\bibitem[\protect\citeauthoryear{Geerts, Poggi, and Tannen}{Geerts
  et~al\mbox{.}}{2013}]%
        {geerts-poggi-tannen-2013}
\bibfield{author}{\bibinfo{person}{Floris Geerts}, \bibinfo{person}{Antonella
  Poggi}, {and} \bibinfo{person}{Val Tannen}.} \bibinfo{year}{2013}\natexlab{}.
\newblock \showarticletitle{On the Limitations of Provenance for Queries with
  Difference}. In \bibinfo{booktitle}{\emph{Proceedings of the Theory and
  Practice of Provenance Workshop}} \emph{(\bibinfo{series}{TaPP '13})}.
  \bibinfo{publisher}{USENIX Association}.
\newblock
\urldef\tempurl%
\url{https://doi.org/10.1145/2448496.2448516}
\showDOI{\tempurl}


\bibitem[\protect\citeauthoryear{Ghasemirad, Liu, Sprenger, and
  Basin}{Ghasemirad et~al\mbox{.}}{2025}]%
        {ghasemirad-2025-veriso}
\bibfield{author}{\bibinfo{person}{Shabnam Ghasemirad}, \bibinfo{person}{Si
  Liu}, \bibinfo{person}{Christoph Sprenger}, {and} \bibinfo{person}{David
  Basin}.} \bibinfo{year}{2025}\natexlab{}.
\newblock \showarticletitle{{VerIso}: Verifiable Isolation Guarantees for
  Database Transactions}.
\newblock \bibinfo{journal}{\emph{Proc. {VLDB} Endow.}} \bibinfo{volume}{18},
  \bibinfo{number}{5} (\bibinfo{year}{2025}), \bibinfo{pages}{1362--1375}.
\newblock
\urldef\tempurl%
\url{https://doi.org/10.14778/3718057.3718065}
\showDOI{\tempurl}


\bibitem[\protect\citeauthoryear{Ghumare and {AgentMemory
  contributors}}{Ghumare and {AgentMemory contributors}}{2026}]%
        {agentmemory-2026}
\bibfield{author}{\bibinfo{person}{Rohit Ghumare} {and}
  \bibinfo{person}{{AgentMemory contributors}}.}
  \bibinfo{year}{2026}\natexlab{}.
\newblock \bibinfo{title}{{AgentMemory}: Persistent Memory for {AI} Coding
  Agents}.
\newblock \bibinfo{howpublished}{GitHub repository, Apache-2.0 licence, npm
  package \texttt{@agentmemory/agentmemory}, latest release v0.9.17
  (2026-05-16), HEAD commit c93c715 (2026-05-17T11:15:59Z)}.
\newblock
\urldef\tempurl%
\url{https://github.com/rohitg00/agentmemory}
\showURL{%
\tempurl}


\bibitem[\protect\citeauthoryear{Gr{\"a}del and Tannen}{Gr{\"a}del and
  Tannen}{2024}]%
        {gradel-tannen-2024-semiring-fol}
\bibfield{author}{\bibinfo{person}{Erich Gr{\"a}del} {and} \bibinfo{person}{Val
  Tannen}.} \bibinfo{year}{2024}\natexlab{}.
\newblock \showarticletitle{Provenance Analysis and Semiring Semantics for
  First-Order Logic}.
\newblock \bibinfo{journal}{\emph{arXiv preprint arXiv:2412.07986}}
  (\bibinfo{year}{2024}).
\newblock
\urldef\tempurl%
\url{https://arxiv.org/abs/2412.07986}
\showURL{%
\tempurl}


\bibitem[\protect\citeauthoryear{Green, Karvounarakis, and Tannen}{Green
  et~al\mbox{.}}{2007}]%
        {green-karvounarakis-tannen-2007}
\bibfield{author}{\bibinfo{person}{Todd~J. Green}, \bibinfo{person}{Grigoris
  Karvounarakis}, {and} \bibinfo{person}{Val Tannen}.}
  \bibinfo{year}{2007}\natexlab{}.
\newblock \showarticletitle{Provenance Semirings}. In
  \bibinfo{booktitle}{\emph{Proceedings of the 26th {ACM}
  {SIGACT}-{SIGMOD}-{SIGART} Symposium on Principles of Database Systems
  ({PODS})}}. \bibinfo{pages}{31--40}.
\newblock
\urldef\tempurl%
\url{https://doi.org/10.1145/1265530.1265535}
\showDOI{\tempurl}


\bibitem[\protect\citeauthoryear{Green and Tannen}{Green and Tannen}{2017}]%
        {green-tannen-2017-retrospective}
\bibfield{author}{\bibinfo{person}{Todd~J. Green} {and} \bibinfo{person}{Val
  Tannen}.} \bibinfo{year}{2017}\natexlab{}.
\newblock \showarticletitle{The Semiring Framework for Database Provenance}. In
  \bibinfo{booktitle}{\emph{Proceedings of the 36th {ACM}
  {SIGMOD}-{SIGACT}-{SIGAI} Symposium on Principles of Database Systems
  ({PODS})}}. \bibinfo{pages}{93--99}.
\newblock
\urldef\tempurl%
\url{https://doi.org/10.1145/3034786.3056125}
\showDOI{\tempurl}


\bibitem[\protect\citeauthoryear{Guti{\'e}rrez, Shu, Gu, Yasunaga, and
  Su}{Guti{\'e}rrez et~al\mbox{.}}{2024}]%
        {gutierrez-2024-hipporag}
\bibfield{author}{\bibinfo{person}{Bernal~Jim{\'e}nez Guti{\'e}rrez},
  \bibinfo{person}{Yiheng Shu}, \bibinfo{person}{Yu Gu},
  \bibinfo{person}{Michihiro Yasunaga}, {and} \bibinfo{person}{Yu Su}.}
  \bibinfo{year}{2024}\natexlab{}.
\newblock \showarticletitle{{HippoRAG}: Neurobiologically Inspired Long-Term
  Memory for Large Language Models}. In \bibinfo{booktitle}{\emph{Advances in
  Neural Information Processing Systems 37 ({NeurIPS} 2024)}}.
\newblock
\urldef\tempurl%
\url{https://arxiv.org/abs/2405.14831}
\showURL{%
\tempurl}


\bibitem[\protect\citeauthoryear{He and Yu}{He and Yu}{2026}]%
        {2026-verifiable-agentic-infrastructure}
\bibfield{author}{\bibinfo{person}{Jun He} {and} \bibinfo{person}{Deying Yu}.}
  \bibinfo{year}{2026}\natexlab{}.
\newblock \bibinfo{title}{Verifiable Agentic Infrastructure: Proof-Derived
  Authorization for Sovereign {AI} Systems}.
\newblock
\newblock
\showeprint[arxiv]{2605.15228}~[cs.AI]
\urldef\tempurl%
\url{https://arxiv.org/abs/2605.15228}
\showURL{%
\tempurl}


\bibitem[\protect\citeauthoryear{Idreos}{Idreos}{2025}]%
        {idreos-2025-calculators}
\bibfield{author}{\bibinfo{person}{Stratos Idreos}.}
  \bibinfo{year}{2025}\natexlab{}.
\newblock \showarticletitle{Alphabets, Grammars, Calculators, and the End of
  Hand-Crafted Systems}.
\newblock \bibinfo{journal}{\emph{Proceedings of the {VLDB} Endowment
  ({PVLDB})}} \bibinfo{volume}{18}, \bibinfo{number}{12}
  (\bibinfo{year}{2025}), \bibinfo{pages}{5537}.
\newblock
\urldef\tempurl%
\url{https://doi.org/10.14778/3750601.3760522}
\showDOI{\tempurl}


\bibitem[\protect\citeauthoryear{Jensen, Dyreson, B{\"o}hlen, Clifford,
  Elmasri, Gadia, Grandi, Hayes, Jajodia, et~al\mbox{.}}{Jensen
  et~al\mbox{.}}{1998}]%
        {jensen-1998-consensus}
\bibfield{author}{\bibinfo{person}{Christian~S. Jensen},
  \bibinfo{person}{Curtis~E. Dyreson}, \bibinfo{person}{Michael B{\"o}hlen},
  \bibinfo{person}{James Clifford}, \bibinfo{person}{Ramez Elmasri},
  \bibinfo{person}{Shashi~K. Gadia}, \bibinfo{person}{Fabio Grandi},
  \bibinfo{person}{Pat Hayes}, \bibinfo{person}{Sushil Jajodia},
  {et~al\mbox{.}}} \bibinfo{year}{1998}\natexlab{}.
\newblock \showarticletitle{The Consensus Glossary of Temporal Database
  Concepts---{F}ebruary 1998 Version}.
\newblock In \bibinfo{booktitle}{\emph{Temporal Databases: Research and
  Practice}}, \bibfield{editor}{\bibinfo{person}{Opher Etzion},
  \bibinfo{person}{Sushil Jajodia}, {and} \bibinfo{person}{Suryanarayana
  Sripada}} (Eds.). \bibinfo{series}{Lecture Notes in Computer Science},
  Vol.~\bibinfo{volume}{1399}. \bibinfo{publisher}{Springer},
  \bibinfo{pages}{367--405}.
\newblock
\urldef\tempurl%
\url{https://doi.org/10.1007/BFb0053710}
\showDOI{\tempurl}


\bibitem[\protect\citeauthoryear{Khan}{Khan}{2026}]%
        {s-bus-2026-khan}
\bibfield{author}{\bibinfo{person}{Sajjad Khan}.}
  \bibinfo{year}{2026}\natexlab{}.
\newblock \bibinfo{title}{{S-Bus}: Automatic Read-Set Reconstruction for
  Multi-Agent {LLM} State Coordination}.
\newblock
\newblock
\showeprint[arxiv]{2605.17076}~[cs.DC]
\urldef\tempurl%
\url{https://arxiv.org/abs/2605.17076}
\showURL{%
\tempurl}


\bibitem[\protect\citeauthoryear{Kulkarni and Michels}{Kulkarni and
  Michels}{2012}]%
        {kulkarni-michels-2012}
\bibfield{author}{\bibinfo{person}{Krishna Kulkarni} {and}
  \bibinfo{person}{Jan-Eike Michels}.} \bibinfo{year}{2012}\natexlab{}.
\newblock \showarticletitle{Temporal Features in {SQL:2011}}.
\newblock \bibinfo{journal}{\emph{ACM SIGMOD Record}} \bibinfo{volume}{41},
  \bibinfo{number}{3} (\bibinfo{year}{2012}), \bibinfo{pages}{34--43}.
\newblock
\urldef\tempurl%
\url{https://doi.org/10.1145/2380776.2380786}
\showDOI{\tempurl}


\bibitem[\protect\citeauthoryear{{Letta Team}}{{Letta Team}}{2026}]%
        {letta-context-constitution-2026}
\bibfield{author}{\bibinfo{person}{{Letta Team}}.}
  \bibinfo{year}{2026}\natexlab{}.
\newblock \bibinfo{title}{Context Constitution: {Letta} Code's Memory
  Filesystem}.
\newblock \bibinfo{howpublished}{Letta engineering blog}.
\newblock
\urldef\tempurl%
\url{https://letta.com/blog/context-constitution}
\showURL{%
\tempurl}


\bibitem[\protect\citeauthoryear{Li, Zhang, Yang, Ma, and Guo}{Li
  et~al\mbox{.}}{2026}]%
        {li-2026-romem}
\bibfield{author}{\bibinfo{person}{Weixian~Waylon Li}, \bibinfo{person}{Jiaxin
  Zhang}, \bibinfo{person}{Xianan~Jim Yang}, \bibinfo{person}{Tiejun Ma}, {and}
  \bibinfo{person}{Yiwen Guo}.} \bibinfo{year}{2026}\natexlab{}.
\newblock \showarticletitle{Time is Not a Label: Continuous Phase Rotation for
  Temporal Knowledge Graphs and Agentic Memory}.
\newblock \bibinfo{journal}{\emph{arXiv preprint arXiv:2604.11544}}
  (\bibinfo{year}{2026}).
\newblock
\urldef\tempurl%
\url{https://arxiv.org/abs/2604.11544}
\showURL{%
\tempurl}


\bibitem[\protect\citeauthoryear{Li, Xi, Li, Chen, Chen, Song, Niu, Wang, Yang,
  Tang, Yu, Zhao, Wang, Liu, Lin, Wang, Huo, Chen, Chen, Li, Tao, Lai, Wu,
  Tang, Wang, Fan, Zhang, Zhang, Yan, Yang, Xu, Xu, Chen, Wang, Yang, Zhang,
  Xu, Chen, and Xiong}{Li et~al\mbox{.}}{2025}]%
        {li-2025-memos}
\bibfield{author}{\bibinfo{person}{Zhiyu Li}, \bibinfo{person}{Chenyang Xi},
  \bibinfo{person}{Chunyu Li}, \bibinfo{person}{Ding Chen},
  \bibinfo{person}{Boyu Chen}, \bibinfo{person}{Shichao Song},
  \bibinfo{person}{Simin Niu}, \bibinfo{person}{Hanyu Wang},
  \bibinfo{person}{Jiawei Yang}, \bibinfo{person}{Chen Tang},
  \bibinfo{person}{Qingchen Yu}, \bibinfo{person}{Jihao Zhao},
  \bibinfo{person}{Yezhaohui Wang}, \bibinfo{person}{Peng Liu},
  \bibinfo{person}{Zehao Lin}, \bibinfo{person}{Pengyuan Wang},
  \bibinfo{person}{Jiahao Huo}, \bibinfo{person}{Tianyi Chen},
  \bibinfo{person}{Kai Chen}, \bibinfo{person}{Kehang Li},
  \bibinfo{person}{Zhen Tao}, \bibinfo{person}{Huayi Lai}, \bibinfo{person}{Hao
  Wu}, \bibinfo{person}{Bo Tang}, \bibinfo{person}{Zhengren Wang},
  \bibinfo{person}{Zhaoxin Fan}, \bibinfo{person}{Ningyu Zhang},
  \bibinfo{person}{Linfeng Zhang}, \bibinfo{person}{Junchi Yan},
  \bibinfo{person}{Mingchuan Yang}, \bibinfo{person}{Tong Xu},
  \bibinfo{person}{Wei Xu}, \bibinfo{person}{Huajun Chen},
  \bibinfo{person}{Haofen Wang}, \bibinfo{person}{Hongkang Yang},
  \bibinfo{person}{Wentao Zhang}, \bibinfo{person}{Zhi-Qin~John Xu},
  \bibinfo{person}{Siheng Chen}, {and} \bibinfo{person}{Feiyu Xiong}.}
  \bibinfo{year}{2025}\natexlab{}.
\newblock \showarticletitle{{MemOS}: A Memory {OS} for {AI} System}.
\newblock \bibinfo{journal}{\emph{arXiv preprint arXiv:2507.03724}}
  (\bibinfo{year}{2025}).
\newblock
\urldef\tempurl%
\url{https://doi.org/10.48550/arXiv.2507.03724}
\showDOI{\tempurl}


\bibitem[\protect\citeauthoryear{Liu, Ponnapalli, Shankar, Zeighami, Zhu,
  Agarwal, Chen, Suwito, Yuan, Stoica, Zaharia, Cheung, Crooks, Gonzalez, and
  Parameswaran}{Liu et~al\mbox{.}}{2026}]%
        {liu-2026-overlords}
\bibfield{author}{\bibinfo{person}{Shu Liu}, \bibinfo{person}{Soujanya
  Ponnapalli}, \bibinfo{person}{Shreya Shankar}, \bibinfo{person}{Sepanta
  Zeighami}, \bibinfo{person}{Alan Zhu}, \bibinfo{person}{Shubham Agarwal},
  \bibinfo{person}{Ruiqi Chen}, \bibinfo{person}{Samion Suwito},
  \bibinfo{person}{Shuo Yuan}, \bibinfo{person}{Ion Stoica},
  \bibinfo{person}{Matei Zaharia}, \bibinfo{person}{Alvin Cheung},
  \bibinfo{person}{Natacha Crooks}, \bibinfo{person}{Joseph~E. Gonzalez}, {and}
  \bibinfo{person}{Aditya~G. Parameswaran}.} \bibinfo{year}{2026}\natexlab{}.
\newblock \bibinfo{title}{Supporting Our {AI} Overlords: Redesigning Data
  Systems to be Agent-First}.
\newblock \bibinfo{howpublished}{Proc.\ {CIDR} 2026, January 18--21,
  Chaminade}.
\newblock
\urldef\tempurl%
\url{https://www.cidrdb.org/cidr2026/papers/p32-liu.pdf}
\showURL{%
\tempurl}


\bibitem[\protect\citeauthoryear{Lorentzos and Johnson}{Lorentzos and
  Johnson}{1988}]%
        {lorentzos-johnson-1988}
\bibfield{author}{\bibinfo{person}{Nikos~A. Lorentzos} {and}
  \bibinfo{person}{Roger~G. Johnson}.} \bibinfo{year}{1988}\natexlab{}.
\newblock \showarticletitle{Extending Relational Algebra to Manipulate Temporal
  Data}. In \bibinfo{booktitle}{\emph{Proceedings of the 14th International
  Conference on Very Large Data Bases ({VLDB})}}. \bibinfo{publisher}{Morgan
  Kaufmann}, \bibinfo{pages}{289--296}.
\newblock
\urldef\tempurl%
\url{https://vldb.org/conf/1988/P289.PDF}
\showURL{%
\tempurl}


\bibitem[\protect\citeauthoryear{Lu, Cheng, Zhang, and Tang}{Lu
  et~al\mbox{.}}{2026}]%
        {lu-2026-mma}
\bibfield{author}{\bibinfo{person}{Yihao Lu}, \bibinfo{person}{Wanru Cheng},
  \bibinfo{person}{Zeyu Zhang}, {and} \bibinfo{person}{Hao Tang}.}
  \bibinfo{year}{2026}\natexlab{}.
\newblock \bibinfo{title}{{MMA}: {M}ultimodal {M}emory {A}gent}.
\newblock
\newblock
\showeprint[arxiv]{2602.16493}~[cs.AI]
\urldef\tempurl%
\url{https://arxiv.org/abs/2602.16493}
\showURL{%
\tempurl}


\bibitem[\protect\citeauthoryear{Luo, Kang, Ji, Liu, Liu, Li, and Peng}{Luo
  et~al\mbox{.}}{2026}]%
        {luo-2026-shadowmerge}
\bibfield{author}{\bibinfo{person}{Yang Luo}, \bibinfo{person}{Zifeng Kang},
  \bibinfo{person}{Tiantian Ji}, \bibinfo{person}{Xinran Liu},
  \bibinfo{person}{Yong Liu}, \bibinfo{person}{Shuyu Li}, {and}
  \bibinfo{person}{Lingyun Peng}.} \bibinfo{year}{2026}\natexlab{}.
\newblock \showarticletitle{{ShadowMerge}: A Novel Poisoning Attack on
  Graph-Based Agent Memory via Relation-Channel Conflicts}.
\newblock \bibinfo{journal}{\emph{arXiv preprint arXiv:2605.09033}}
  (\bibinfo{year}{2026}).
\newblock
\urldef\tempurl%
\url{https://arxiv.org/abs/2605.09033}
\showURL{%
\tempurl}


\bibitem[\protect\citeauthoryear{Maharana, Lee, Tulyakov, Bansal, Barbieri, and
  Fang}{Maharana et~al\mbox{.}}{2024}]%
        {maharana-2024-locomo}
\bibfield{author}{\bibinfo{person}{Adyasha Maharana}, \bibinfo{person}{Dong-Ho
  Lee}, \bibinfo{person}{Sergey Tulyakov}, \bibinfo{person}{Mohit Bansal},
  \bibinfo{person}{Francesco Barbieri}, {and} \bibinfo{person}{Yuwei Fang}.}
  \bibinfo{year}{2024}\natexlab{}.
\newblock \showarticletitle{Evaluating Very Long-Term Conversational Memory of
  {LLM} Agents}. In \bibinfo{booktitle}{\emph{Proceedings of the 62nd Annual
  Meeting of the Association for Computational Linguistics (Volume 1: Long
  Papers)}}. \bibinfo{pages}{13851--13870}.
\newblock
\urldef\tempurl%
\url{https://doi.org/10.18653/v1/2024.acl-long.747}
\showDOI{\tempurl}


\bibitem[\protect\citeauthoryear{{Mem0 Team}}{{Mem0 Team}}{2026}]%
        {mem0ai-readme-2026}
\bibfield{author}{\bibinfo{person}{{Mem0 Team}}.}
  \bibinfo{year}{2026}\natexlab{}.
\newblock \bibinfo{title}{{Mem0} README: New Memory Algorithm (April 2026)}.
\newblock \bibinfo{howpublished}{GitHub repository README}.
\newblock
\urldef\tempurl%
\url{https://github.com/mem0ai/mem0/blob/main/README.md}
\showURL{%
\tempurl}


\bibitem[\protect\citeauthoryear{Meng, Long, Sioutis, and Zhou}{Meng
  et~al\mbox{.}}{2025}]%
        {meng-2025-belief-algebras}
\bibfield{author}{\bibinfo{person}{Hua Meng}, \bibinfo{person}{Zhiguo Long},
  \bibinfo{person}{Michael Sioutis}, {and} \bibinfo{person}{Zhengchun Zhou}.}
  \bibinfo{year}{2025}\natexlab{}.
\newblock \showarticletitle{On Definite Iterated Belief Revision with Belief
  Algebras}.
\newblock \bibinfo{journal}{\emph{arXiv preprint arXiv:2505.06505}}
  (\bibinfo{year}{2025}).
\newblock
\urldef\tempurl%
\url{https://arxiv.org/abs/2505.06505}
\showURL{%
\tempurl}


\bibitem[\protect\citeauthoryear{Messing}{Messing}{2026}]%
        {messing-2026-tee}
\bibfield{author}{\bibinfo{person}{Solomon Messing}.}
  \bibinfo{year}{2026}\natexlab{}.
\newblock \bibinfo{title}{Hidden Measurement Error in {LLM} Pipelines Distorts
  Annotation, Evaluation, and Benchmarking}.
\newblock
\newblock
\showeprint[arxiv]{2604.11581}~[cs.CL]
\urldef\tempurl%
\url{https://arxiv.org/abs/2604.11581}
\showURL{%
\tempurl}


\bibitem[\protect\citeauthoryear{Mohammadi, Potamitis, Klein, Arora, and
  Bindschaedler}{Mohammadi et~al\mbox{.}}{2026}]%
        {atomix-2026-mohammadi}
\bibfield{author}{\bibinfo{person}{Bardia Mohammadi}, \bibinfo{person}{Nearchos
  Potamitis}, \bibinfo{person}{Lars Klein}, \bibinfo{person}{Akhil Arora},
  {and} \bibinfo{person}{Laurent Bindschaedler}.}
  \bibinfo{year}{2026}\natexlab{}.
\newblock \bibinfo{title}{{Atomix}: Timely, Transactional Tool Use for Reliable
  Agentic Workflows}.
\newblock
\newblock
\showeprint[arxiv]{2602.14849}~[cs.DC]
\urldef\tempurl%
\url{https://arxiv.org/abs/2602.14849}
\showURL{%
\tempurl}


\bibitem[\protect\citeauthoryear{Mohri}{Mohri}{2002}]%
        {mohri-2002-weighted}
\bibfield{author}{\bibinfo{person}{Mehryar Mohri}.}
  \bibinfo{year}{2002}\natexlab{}.
\newblock \showarticletitle{Semiring Frameworks and Algorithms for
  Shortest-Distance Problems}.
\newblock \bibinfo{journal}{\emph{Journal of Automata, Languages and
  Combinatorics}} \bibinfo{volume}{7}, \bibinfo{number}{3}
  (\bibinfo{year}{2002}), \bibinfo{pages}{321--350}.
\newblock
\urldef\tempurl%
\url{http://www.cs.nyu.edu/~mohri/pub/jalc.pdf}
\showURL{%
\tempurl}


\bibitem[\protect\citeauthoryear{Myakala, Agrawal, and Manche}{Myakala
  et~al\mbox{.}}{2026}]%
        {myakala-2026-beliefshift}
\bibfield{author}{\bibinfo{person}{Praveen~Kumar Myakala},
  \bibinfo{person}{Manan Agrawal}, {and} \bibinfo{person}{Rahul Manche}.}
  \bibinfo{year}{2026}\natexlab{}.
\newblock \showarticletitle{{BeliefShift}: Benchmarking Temporal Belief
  Consistency and Opinion Drift in {LLM} Agents}.
\newblock \bibinfo{journal}{\emph{arXiv preprint arXiv:2603.23848}}
  (\bibinfo{year}{2026}).
\newblock
\urldef\tempurl%
\url{https://arxiv.org/abs/2603.23848}
\showURL{%
\tempurl}


\bibitem[\protect\citeauthoryear{Ouyang and Hou}{Ouyang and Hou}{2026}]%
        {ouyang-2026-memlineage}
\bibfield{author}{\bibinfo{person}{Ciyan Ouyang} {and} \bibinfo{person}{Rui
  Hou}.} \bibinfo{year}{2026}\natexlab{}.
\newblock \showarticletitle{{MemLineage}: Lineage-Guided Enforcement for {LLM}
  Agent Memory}.
\newblock \bibinfo{journal}{\emph{arXiv preprint arXiv:2605.14421}}
  (\bibinfo{year}{2026}).
\newblock
\urldef\tempurl%
\url{https://arxiv.org/abs/2605.14421}
\showURL{%
\tempurl}


\bibitem[\protect\citeauthoryear{Packer, Wooders, Lin, Fang, Patil, Stoica, and
  Gonzalez}{Packer et~al\mbox{.}}{2023}]%
        {packer-2023-memgpt}
\bibfield{author}{\bibinfo{person}{Charles Packer}, \bibinfo{person}{Sarah
  Wooders}, \bibinfo{person}{Kevin Lin}, \bibinfo{person}{Vivian Fang},
  \bibinfo{person}{Shishir~G. Patil}, \bibinfo{person}{Ion Stoica}, {and}
  \bibinfo{person}{Joseph~E. Gonzalez}.} \bibinfo{year}{2023}\natexlab{}.
\newblock \showarticletitle{{MemGPT}: Towards {LLMs} as Operating Systems}.
\newblock \bibinfo{journal}{\emph{arXiv preprint arXiv:2310.08560}}
  (\bibinfo{year}{2023}).
\newblock
\urldef\tempurl%
\url{https://arxiv.org/abs/2310.08560}
\showURL{%
\tempurl}


\bibitem[\protect\citeauthoryear{Park, O'Brien, Cai, Morris, Liang, and
  Bernstein}{Park et~al\mbox{.}}{2023}]%
        {park-2023-generative-agents}
\bibfield{author}{\bibinfo{person}{Joon~Sung Park}, \bibinfo{person}{Joseph~C.
  O'Brien}, \bibinfo{person}{Carrie~J. Cai}, \bibinfo{person}{Meredith~Ringel
  Morris}, \bibinfo{person}{Percy Liang}, {and} \bibinfo{person}{Michael~S.
  Bernstein}.} \bibinfo{year}{2023}\natexlab{}.
\newblock \showarticletitle{Generative Agents: Interactive Simulacra of Human
  Behavior}. In \bibinfo{booktitle}{\emph{Proceedings of the 36th Annual {ACM}
  Symposium on User Interface Software and Technology ({UIST})}}.
  \bibinfo{pages}{1--22}.
\newblock
\urldef\tempurl%
\url{https://doi.org/10.1145/3586183.3606763}
\showDOI{\tempurl}


\bibitem[\protect\citeauthoryear{Park}{Park}{2026}]%
        {park-2026-kumiho}
\bibfield{author}{\bibinfo{person}{Young~Bin Park}.}
  \bibinfo{year}{2026}\natexlab{}.
\newblock \showarticletitle{Graph-Native Cognitive Memory for {AI} Agents:
  Formal Belief Revision Semantics for Versioned Memory Architectures}.
\newblock \bibinfo{journal}{\emph{arXiv preprint arXiv:2603.17244}}
  (\bibinfo{year}{2026}).
\newblock
\urldef\tempurl%
\url{https://arxiv.org/abs/2603.17244}
\showURL{%
\tempurl}


\bibitem[\protect\citeauthoryear{Patel, Jha, Pan, Gupta, Asawa, Guestrin, and
  Zaharia}{Patel et~al\mbox{.}}{2025}]%
        {patel-2025-lotus}
\bibfield{author}{\bibinfo{person}{Liana Patel}, \bibinfo{person}{Siddharth
  Jha}, \bibinfo{person}{Melissa Pan}, \bibinfo{person}{Harshit Gupta},
  \bibinfo{person}{Parth Asawa}, \bibinfo{person}{Carlos Guestrin}, {and}
  \bibinfo{person}{Matei Zaharia}.} \bibinfo{year}{2025}\natexlab{}.
\newblock \showarticletitle{Semantic Operators and Their Optimization: Enabling
  {LLM}-Based Data Processing with Accuracy Guarantees in {LOTUS}}.
\newblock \bibinfo{journal}{\emph{Proceedings of the {VLDB} Endowment
  ({PVLDB})}} \bibinfo{volume}{18}, \bibinfo{number}{11}
  (\bibinfo{year}{2025}), \bibinfo{pages}{4171--4184}.
\newblock
\urldef\tempurl%
\url{https://doi.org/10.14778/3749646.3749685}
\showDOI{\tempurl}


\bibitem[\protect\citeauthoryear{Pulipaka, Hlebik, Raghav, Abdelnabi, Raina,
  Sheth, and Fritz}{Pulipaka et~al\mbox{.}}{2026}]%
        {2026-sleeper-memory-poisoning}
\bibfield{author}{\bibinfo{person}{Sidharth Pulipaka},
  \bibinfo{person}{Stanislau Hlebik}, \bibinfo{person}{Leonidas Raghav},
  \bibinfo{person}{Sahar Abdelnabi}, \bibinfo{person}{Vyas Raina},
  \bibinfo{person}{Ivaxi Sheth}, {and} \bibinfo{person}{Mario Fritz}.}
  \bibinfo{year}{2026}\natexlab{}.
\newblock \bibinfo{title}{Hidden in Memory: Sleeper Memory Poisoning in {LLM}
  Agents}.
\newblock
\newblock
\showeprint[arxiv]{2605.15338}~[cs.CR]
\urldef\tempurl%
\url{https://arxiv.org/abs/2605.15338}
\showURL{%
\tempurl}


\bibitem[\protect\citeauthoryear{Rasmussen, Paliychuk, Beauvais, Ryan, and
  Chalef}{Rasmussen et~al\mbox{.}}{2025}]%
        {rasmussen-2025-zep}
\bibfield{author}{\bibinfo{person}{Preston Rasmussen}, \bibinfo{person}{Pavlo
  Paliychuk}, \bibinfo{person}{Travis Beauvais}, \bibinfo{person}{Jack Ryan},
  {and} \bibinfo{person}{Daniel Chalef}.} \bibinfo{year}{2025}\natexlab{}.
\newblock \showarticletitle{{Zep}: A Temporal Knowledge Graph Architecture for
  Agent Memory}.
\newblock \bibinfo{journal}{\emph{arXiv preprint arXiv:2501.13956}}
  (\bibinfo{year}{2025}).
\newblock
\urldef\tempurl%
\url{https://arxiv.org/abs/2501.13956}
\showURL{%
\tempurl}


\bibitem[\protect\citeauthoryear{Roynard}{Roynard}{2026}]%
        {roynard-2026-knowledge-layer}
\bibfield{author}{\bibinfo{person}{Micha{\"e}l Roynard}.}
  \bibinfo{year}{2026}\natexlab{}.
\newblock \showarticletitle{The Missing Knowledge Layer in Cognitive
  Architectures for {AI} Agents}.
\newblock \bibinfo{journal}{\emph{arXiv preprint arXiv:2604.11364}}
  (\bibinfo{year}{2026}).
\newblock
\urldef\tempurl%
\url{https://arxiv.org/abs/2604.11364}
\showURL{%
\tempurl}


\bibitem[\protect\citeauthoryear{Russo, Liu, Sudhir, Vitagliano, Cafarella,
  Kraska, and Madden}{Russo et~al\mbox{.}}{2026}]%
        {russo-2026-abacus}
\bibfield{author}{\bibinfo{person}{Matthew Russo}, \bibinfo{person}{Chunwei
  Liu}, \bibinfo{person}{Sivaprasad Sudhir}, \bibinfo{person}{Gerardo
  Vitagliano}, \bibinfo{person}{Michael Cafarella}, \bibinfo{person}{Tim
  Kraska}, {and} \bibinfo{person}{Samuel Madden}.}
  \bibinfo{year}{2026}\natexlab{}.
\newblock \showarticletitle{{Abacus}: A Cost-Based Optimizer for Semantic
  Operator Systems}.
\newblock \bibinfo{journal}{\emph{Proceedings of the {VLDB} Endowment
  ({PVLDB})}} \bibinfo{volume}{19}, \bibinfo{number}{5} (\bibinfo{year}{2026}),
  \bibinfo{pages}{1060--1073}.
\newblock
\urldef\tempurl%
\url{https://www.vldb.org/pvldb/vol19/p1060-russo.pdf}
\showURL{%
\tempurl}


\bibitem[\protect\citeauthoryear{Sadowski and Chudziak}{Sadowski and
  Chudziak}{2026}]%
        {sadowski-2026-rashomon}
\bibfield{author}{\bibinfo{person}{Albert Sadowski} {and}
  \bibinfo{person}{Jaros{\l}aw~A. Chudziak}.} \bibinfo{year}{2026}\natexlab{}.
\newblock \bibinfo{title}{Rashomon Memory: Towards Argumentation-Driven
  Retrieval for Multi-Perspective Agent Memory}.
\newblock
\newblock
\urldef\tempurl%
\url{https://doi.org/10.48550/arXiv.2604.03588}
\showDOI{\tempurl}
\showeprint[arxiv]{2604.03588}


\bibitem[\protect\citeauthoryear{Salem, Garcia-Molina, and Sands}{Salem
  et~al\mbox{.}}{1989}]%
        {salem-garciamolina-1989}
\bibfield{author}{\bibinfo{person}{Kenneth Salem}, \bibinfo{person}{Hector
  Garcia-Molina}, {and} \bibinfo{person}{Jeannie Sands}.}
  \bibinfo{year}{1989}\natexlab{}.
\newblock \showarticletitle{Altruistic Locking: A Strategy for Coping with Long
  Lived Transactions}. In \bibinfo{booktitle}{\emph{Proceedings of the 2nd
  International Workshop on High Performance Transaction Systems}}.
\newblock
\urldef\tempurl%
\url{https://doi.org/10.1145/64162.64173}
\showDOI{\tempurl}


\bibitem[\protect\citeauthoryear{Shah and {Supermemory contributors}}{Shah and
  {Supermemory contributors}}{2026}]%
        {shah-2026-supermemory}
\bibfield{author}{\bibinfo{person}{Dhravya Shah} {and}
  \bibinfo{person}{{Supermemory contributors}}.}
  \bibinfo{year}{2026}\natexlab{}.
\newblock \bibinfo{title}{{Supermemory}: State-of-the-art Memory and Context
  Engine for {AI}}.
\newblock \bibinfo{howpublished}{GitHub repository
  \texttt{supermemoryai/supermemory}, MIT licence, HEAD commit 36ecf47
  (2026-05-17T07:55:12Z), 22{,}596 stars, zero GitHub releases, Python SDK
  \texttt{supermemory-openai-sdk} on PyPI at alpha v1.0.3}.
\newblock
\urldef\tempurl%
\url{https://github.com/supermemoryai/supermemory}
\showURL{%
\tempurl}


\bibitem[\protect\citeauthoryear{Shankar, Chambers, Shah, Parameswaran, and
  Wu}{Shankar et~al\mbox{.}}{2025}]%
        {shankar-2025-docetl}
\bibfield{author}{\bibinfo{person}{Shreya Shankar}, \bibinfo{person}{Tristan
  Chambers}, \bibinfo{person}{Tarak Shah}, \bibinfo{person}{Aditya~G.
  Parameswaran}, {and} \bibinfo{person}{Eugene Wu}.}
  \bibinfo{year}{2025}\natexlab{}.
\newblock \showarticletitle{{DocETL}: Agentic Query Rewriting and Evaluation
  for Complex Document Processing}.
\newblock \bibinfo{journal}{\emph{Proceedings of the {VLDB} Endowment
  ({PVLDB})}} \bibinfo{volume}{18}, \bibinfo{number}{12}
  (\bibinfo{year}{2025}), \bibinfo{pages}{3920--3932}.
\newblock
\urldef\tempurl%
\url{https://www.vldb.org/pvldb/vol18/p3920-shankar.pdf}
\showURL{%
\tempurl}


\bibitem[\protect\citeauthoryear{Snodgrass and Ahn}{Snodgrass and Ahn}{1986}]%
        {snodgrass-ahn-1986}
\bibfield{author}{\bibinfo{person}{Richard Snodgrass} {and}
  \bibinfo{person}{Ilsoo Ahn}.} \bibinfo{year}{1986}\natexlab{}.
\newblock \showarticletitle{Temporal Databases}.
\newblock \bibinfo{journal}{\emph{IEEE Computer}} \bibinfo{volume}{19},
  \bibinfo{number}{9} (\bibinfo{year}{1986}), \bibinfo{pages}{35--42}.
\newblock
\urldef\tempurl%
\url{https://doi.org/10.1109/MC.1986.1663327}
\showDOI{\tempurl}


\bibitem[\protect\citeauthoryear{Snodgrass}{Snodgrass}{2000}]%
        {snodgrass-1999-tdb}
\bibfield{author}{\bibinfo{person}{Richard~T. Snodgrass}.}
  \bibinfo{year}{2000}\natexlab{}.
\newblock \bibinfo{booktitle}{\emph{Developing Time-Oriented Database
  Applications in {SQL}}}.
\newblock \bibinfo{publisher}{Morgan Kaufmann}, \bibinfo{address}{San
  Francisco, CA}.
\newblock
\showISBNx{1-55860-436-7}
\urldef\tempurl%
\url{https://www2.cs.arizona.edu/~rts/tdbbook.pdf}
\showURL{%
\tempurl}


\bibitem[\protect\citeauthoryear{Su, Guo, Hou, Bai, Li, Zhang, Yin, Lin, Jin,
  Guo, and Cheng}{Su et~al\mbox{.}}{2026}]%
        {su-2026-tsm}
\bibfield{author}{\bibinfo{person}{Miao Su}, \bibinfo{person}{Yucan Guo},
  \bibinfo{person}{Zhongni Hou}, \bibinfo{person}{Long Bai},
  \bibinfo{person}{Zixuan Li}, \bibinfo{person}{Yufei Zhang},
  \bibinfo{person}{Guojun Yin}, \bibinfo{person}{Wei Lin},
  \bibinfo{person}{Xiaolong Jin}, \bibinfo{person}{Jiafeng Guo}, {and}
  \bibinfo{person}{Xueqi Cheng}.} \bibinfo{year}{2026}\natexlab{}.
\newblock \showarticletitle{Beyond Dialogue Time: Temporal Semantic Memory for
  Personalized {LLM} Agents}.
\newblock \bibinfo{journal}{\emph{arXiv preprint arXiv:2601.07468}}
  (\bibinfo{year}{2026}).
\newblock
\urldef\tempurl%
\url{https://arxiv.org/abs/2601.07468}
\showURL{%
\tempurl}


\bibitem[\protect\citeauthoryear{{Tencent}}{{Tencent}}{2026}]%
        {tencent-2026-tencentdb-agent-memory}
\bibfield{author}{\bibinfo{person}{{Tencent}}.}
  \bibinfo{year}{2026}\natexlab{}.
\newblock \bibinfo{title}{{TencentDB} Agent Memory: Fully Local Long-Term
  Memory for {AI} Agents via a 4-Tier Progressive Pipeline}.
\newblock \bibinfo{howpublished}{GitHub repository
  \texttt{Tencent/TencentDB-Agent-Memory}, MIT licence (Tencent header
  wrapper), HEAD commit 5736acc (2026-05-16T12:17:22Z), latest release v0.3.4
  (2026-05-13), 2{,}674 stars, npm package
  \texttt{@tencentdb-agent-memory/memory-tencentdb}, Node >=22.16 required}.
\newblock
\urldef\tempurl%
\url{https://github.com/Tencent/TencentDB-Agent-Memory}
\showURL{%
\tempurl}


\bibitem[\protect\citeauthoryear{{The XTDB Authors}}{{The XTDB
  Authors}}{2026}]%
        {xtdb-docs}
\bibfield{author}{\bibinfo{person}{{The XTDB Authors}}.}
  \bibinfo{year}{2026}\natexlab{}.
\newblock \bibinfo{title}{{XTDB} 2.x: Bitemporal {SQL} for the Real World}.
\newblock \bibinfo{howpublished}{Online documentation}.
\newblock
\urldef\tempurl%
\url{https://docs.xtdb.com/intro/what-is-xtdb.html}
\showURL{%
\tempurl}


\bibitem[\protect\citeauthoryear{Uddin, Shubham, Blanco, Baral, and Wang}{Uddin
  et~al\mbox{.}}{2026}]%
        {uddin-2026-memora}
\bibfield{author}{\bibinfo{person}{Md~Nayem Uddin}, \bibinfo{person}{Kumar
  Shubham}, \bibinfo{person}{Eduardo Blanco}, \bibinfo{person}{Chitta Baral},
  {and} \bibinfo{person}{Gengyu Wang}.} \bibinfo{year}{2026}\natexlab{}.
\newblock \showarticletitle{From Recall to Forgetting: Benchmarking Long-Term
  Memory for Personalized Agents}.
\newblock \bibinfo{journal}{\emph{arXiv preprint arXiv:2604.20006}}
  (\bibinfo{year}{2026}).
\newblock
\urldef\tempurl%
\url{https://arxiv.org/abs/2604.20006}
\showURL{%
\tempurl}


\bibitem[\protect\citeauthoryear{Venkata}{Venkata}{2026}]%
        {three-regimes-conflict-2026}
\bibfield{author}{\bibinfo{person}{Pruthvinath~Jeripity Venkata}.}
  \bibinfo{year}{2026}\natexlab{}.
\newblock \bibinfo{title}{Three Regimes of Context-Parametric Conflict: A
  Predictive Framework and Empirical Validation}.
\newblock
\newblock
\showeprint[arxiv]{2605.11574}~[cs.CL]
\urldef\tempurl%
\url{https://arxiv.org/abs/2605.11574}
\showURL{%
\tempurl}


\bibitem[\protect\citeauthoryear{Ward}{Ward}{2025}]%
        {ward-2025-memoriesdb}
\bibfield{author}{\bibinfo{person}{Joel Ward}.}
  \bibinfo{year}{2025}\natexlab{}.
\newblock \showarticletitle{{MemoriesDB}: A Temporal-Semantic-Relational
  Database for Long-Term Agent Memory}.
\newblock \bibinfo{journal}{\emph{arXiv preprint arXiv:2511.06179}}
  (\bibinfo{year}{2025}).
\newblock
\urldef\tempurl%
\url{https://arxiv.org/abs/2511.06179}
\showURL{%
\tempurl}


\bibitem[\protect\citeauthoryear{Wei, Peng, Dong, Xie, and Wang}{Wei
  et~al\mbox{.}}{2026}]%
        {wei-2026-fademem}
\bibfield{author}{\bibinfo{person}{Lei Wei}, \bibinfo{person}{Xiao Peng},
  \bibinfo{person}{Xu Dong}, \bibinfo{person}{Niantao Xie}, {and}
  \bibinfo{person}{Bin Wang}.} \bibinfo{year}{2026}\natexlab{}.
\newblock \showarticletitle{{FadeMem}: Biologically-Inspired Forgetting for
  Efficient Agent Memory}.
\newblock \bibinfo{journal}{\emph{arXiv preprint arXiv:2601.18642}}
  (\bibinfo{year}{2026}).
\newblock
\urldef\tempurl%
\url{https://arxiv.org/abs/2601.18642}
\showURL{%
\tempurl}


\bibitem[\protect\citeauthoryear{Widiaatmaja, Djeffal, Dandekar, and
  Senellart}{Widiaatmaja et~al\mbox{.}}{2025}]%
        {widiaatmaja-2025-temporal}
\bibfield{author}{\bibinfo{person}{Albert Widiaatmaja}, \bibinfo{person}{Belkis
  Djeffal}, \bibinfo{person}{Ashish Dandekar}, {and} \bibinfo{person}{Pierre
  Senellart}.} \bibinfo{year}{2025}\natexlab{}.
\newblock \showarticletitle{Demonstration of {ProvSQL} Update Provenance
  through Temporal Databases}. In \bibinfo{booktitle}{\emph{Provenance
  Week@SIGMOD}}.
\newblock
\urldef\tempurl%
\url{https://doi.org/10.1145/3736229.3736253}
\showDOI{\tempurl}


\bibitem[\protect\citeauthoryear{Wu, Ji, Kawatkar, Kwan, Gu, Peng, and
  Chang}{Wu et~al\mbox{.}}{2026}]%
        {wu-2026-longmemeval-v2}
\bibfield{author}{\bibinfo{person}{Di Wu}, \bibinfo{person}{Zixiang Ji},
  \bibinfo{person}{Asmi Kawatkar}, \bibinfo{person}{Bryan Kwan},
  \bibinfo{person}{Jia-Chen Gu}, \bibinfo{person}{Nanyun Peng}, {and}
  \bibinfo{person}{Kai-Wei Chang}.} \bibinfo{year}{2026}\natexlab{}.
\newblock \showarticletitle{{LongMemEval-V2}: Evaluating Long-Term Agent Memory
  Toward Experienced Colleagues}.
\newblock \bibinfo{journal}{\emph{arXiv preprint arXiv:2605.12493}}
  (\bibinfo{year}{2026}).
\newblock
\urldef\tempurl%
\url{https://arxiv.org/abs/2605.12493}
\showURL{%
\tempurl}


\bibitem[\protect\citeauthoryear{Wu, Wang, Yu, Zhang, Chang, and Yu}{Wu
  et~al\mbox{.}}{2025}]%
        {wu-2025-longmemeval}
\bibfield{author}{\bibinfo{person}{Di Wu}, \bibinfo{person}{Hongwei Wang},
  \bibinfo{person}{Wenhao Yu}, \bibinfo{person}{Yuwei Zhang},
  \bibinfo{person}{Kai-Wei Chang}, {and} \bibinfo{person}{Dong Yu}.}
  \bibinfo{year}{2025}\natexlab{}.
\newblock \showarticletitle{{LongMemEval}: Benchmarking Chat Assistants on
  Long-Term Interactive Memory}. In \bibinfo{booktitle}{\emph{Proceedings of
  the 13th International Conference on Learning Representations ({ICLR})}}.
\newblock
\urldef\tempurl%
\url{https://openreview.net/forum?id=pZiyCaVuti}
\showURL{%
\tempurl}


\bibitem[\protect\citeauthoryear{Xie and Du}{Xie and Du}{2026}]%
        {xie-2026-neuro-symbolic-sat}
\bibfield{author}{\bibinfo{person}{Shiyao Xie} {and} \bibinfo{person}{Jian
  Du}.} \bibinfo{year}{2026}\natexlab{}.
\newblock \showarticletitle{Neuro-Symbolic Resolution of Recommendation
  Conflicts in Multimorbidity Clinical Guidelines}. In
  \bibinfo{booktitle}{\emph{Proceedings of the 40th Annual {AAAI} Conference on
  Artificial Intelligence, Bridge Program on Logic and {AI}}}.
\newblock
\urldef\tempurl%
\url{https://doi.org/10.48550/arXiv.2604.17340}
\showDOI{\tempurl}


\bibitem[\protect\citeauthoryear{Yang, Lai, Wang, Chang, Harari, and
  Gabrilovich}{Yang et~al\mbox{.}}{2026}]%
        {groupmembench-2026}
\bibfield{author}{\bibinfo{person}{Jingbo Yang}, \bibinfo{person}{Kwei-Herng
  Lai}, \bibinfo{person}{Xiaowen Wang}, \bibinfo{person}{Shiyu Chang},
  \bibinfo{person}{Yaar Harari}, {and} \bibinfo{person}{Evgeniy Gabrilovich}.}
  \bibinfo{year}{2026}\natexlab{}.
\newblock \bibinfo{title}{{GroupMemBench}: Benchmarking {LLM} Agent Memory in
  Multi-Party Conversations}.
\newblock
\newblock
\showeprint[arxiv]{2605.14498}~[cs.CL]
\urldef\tempurl%
\url{https://arxiv.org/abs/2605.14498}
\showURL{%
\tempurl}


\bibitem[\protect\citeauthoryear{Yu, Fang, Liu, and Ma}{Yu
  et~al\mbox{.}}{2026}]%
        {yu-2026-h-mem}
\bibfield{author}{\bibinfo{person}{Jiawei Yu}, \bibinfo{person}{Yixiang Fang},
  \bibinfo{person}{Xilin Liu}, {and} \bibinfo{person}{Yuchi Ma}.}
  \bibinfo{year}{2026}\natexlab{}.
\newblock \bibinfo{title}{{H-Mem}: A Novel Memory Mechanism for Evolving and
  Retrieving Agent Memory via a Hybrid Structure}.
\newblock
\newblock
\showeprint[arxiv]{2605.15701}~[cs.AI]
\urldef\tempurl%
\url{https://arxiv.org/abs/2605.15701}
\showURL{%
\tempurl}


\bibitem[\protect\citeauthoryear{{Zep AI}}{{Zep AI}}{2026}]%
        {getzep-graphiti-2026}
\bibfield{author}{\bibinfo{person}{{Zep AI}}.} \bibinfo{year}{2026}\natexlab{}.
\newblock \bibinfo{title}{{Graphiti}: Build Real-Time Knowledge Graphs for {AI}
  Agents}.
\newblock \bibinfo{howpublished}{GitHub repository}.
\newblock
\urldef\tempurl%
\url{https://github.com/getzep/graphiti}
\showURL{%
\tempurl}


\bibitem[\protect\citeauthoryear{Zhang, Lin, Wu, Sun, Li, Li, and Peng}{Zhang
  et~al\mbox{.}}{2026b}]%
        {zhang-2026-useful-memories}
\bibfield{author}{\bibinfo{person}{Dylan Zhang}, \bibinfo{person}{Yanshan Lin},
  \bibinfo{person}{Zhengkun Wu}, \bibinfo{person}{Yihang Sun},
  \bibinfo{person}{Bingxuan Li}, \bibinfo{person}{Dianqi Li}, {and}
  \bibinfo{person}{Hao Peng}.} \bibinfo{year}{2026}\natexlab{b}.
\newblock \showarticletitle{Useful Memories Become Faulty When Continuously
  Updated by {LLMs}}.
\newblock \bibinfo{journal}{\emph{arXiv preprint arXiv:2605.12978}}
  (\bibinfo{year}{2026}).
\newblock
\urldef\tempurl%
\url{https://arxiv.org/abs/2605.12978}
\showURL{%
\tempurl}


\bibitem[\protect\citeauthoryear{Zhang, Gui, Yang, Chen, and Feng}{Zhang
  et~al\mbox{.}}{2026a}]%
        {zhang-2026-uma}
\bibfield{author}{\bibinfo{person}{Kehao Zhang}, \bibinfo{person}{Shangtong
  Gui}, \bibinfo{person}{Sheng Yang}, \bibinfo{person}{Wei Chen}, {and}
  \bibinfo{person}{Yang Feng}.} \bibinfo{year}{2026}\natexlab{a}.
\newblock \bibinfo{title}{Learning to Remember: End-to-End Training of Memory
  Agents for Long-Context Reasoning}.
\newblock
\newblock
\showeprint[arxiv]{2602.18493}~[cs.CL]
\urldef\tempurl%
\url{https://arxiv.org/abs/2602.18493}
\showURL{%
\tempurl}


\bibitem[\protect\citeauthoryear{Zhang, Wang, Cui, Qiu, Li, Zhu, and He}{Zhang
  et~al\mbox{.}}{2026c}]%
        {zhang-2026-compression-spectrum}
\bibfield{author}{\bibinfo{person}{Xing Zhang}, \bibinfo{person}{Guanghui
  Wang}, \bibinfo{person}{Yanwei Cui}, \bibinfo{person}{Wei Qiu},
  \bibinfo{person}{Ziyuan Li}, \bibinfo{person}{Bing Zhu}, {and}
  \bibinfo{person}{Peiyang He}.} \bibinfo{year}{2026}\natexlab{c}.
\newblock \bibinfo{title}{Experience Compression Spectrum: Unifying Memory,
  Skills, and Rules in {LLM} Agents}.
\newblock
\newblock
\showeprint[arxiv]{2604.15877}~[cs.AI]
\urldef\tempurl%
\url{https://arxiv.org/abs/2604.15877}
\showURL{%
\tempurl}


\bibitem[\protect\citeauthoryear{Zhou, Wang, Peng, Chen, Zhang, and Yu}{Zhou
  et~al\mbox{.}}{2026b}]%
        {zhou-2026-atcc}
\bibfield{author}{\bibinfo{person}{Weixing Zhou}, \bibinfo{person}{Zhiyou
  Wang}, \bibinfo{person}{Zeshun Peng}, \bibinfo{person}{Hetian Chen},
  \bibinfo{person}{Yanfeng Zhang}, {and} \bibinfo{person}{Ge Yu}.}
  \bibinfo{year}{2026}\natexlab{b}.
\newblock \showarticletitle{{ATCC}: Adaptive Concurrency Control for Unforeseen
  Agentic Transactions}.
\newblock \bibinfo{journal}{\emph{arXiv preprint arXiv:2603.13906}}
  (\bibinfo{year}{2026}).
\newblock
\urldef\tempurl%
\url{https://arxiv.org/abs/2603.13906}
\showURL{%
\tempurl}


\bibitem[\protect\citeauthoryear{Zhou, Li, Lu, Wang, Liu, Wang, Yu, Zhang, and
  Li}{Zhou et~al\mbox{.}}{2026a}]%
        {dcrd-context-memory-2026}
\bibfield{author}{\bibinfo{person}{Yigeng Zhou}, \bibinfo{person}{Wu Li},
  \bibinfo{person}{Yifan Lu}, \bibinfo{person}{Yequan Wang},
  \bibinfo{person}{Xuebo Liu}, \bibinfo{person}{Wenya Wang},
  \bibinfo{person}{Jun Yu}, \bibinfo{person}{Min Zhang}, {and}
  \bibinfo{person}{Jing Li}.} \bibinfo{year}{2026}\natexlab{a}.
\newblock \bibinfo{title}{Mitigating Context-Memory Conflicts in {LLMs} through
  Dynamic Cognitive Reconciliation Decoding}.
\newblock
\newblock
\showeprint[arxiv]{2605.12185}~[cs.CL]
\urldef\tempurl%
\url{https://arxiv.org/abs/2605.12185}
\showURL{%
\tempurl}


\bibitem[\protect\citeauthoryear{Zou, Guo, Liang, Wang, Wang, Wen, King, Qu,
  and Xu}{Zou et~al\mbox{.}}{2026}]%
        {zou-2026-demem}
\bibfield{author}{\bibinfo{person}{Mingxi Zou}, \bibinfo{person}{Zhihan Guo},
  \bibinfo{person}{Langzhang Liang}, \bibinfo{person}{Zhuo Wang},
  \bibinfo{person}{Qifan Wang}, \bibinfo{person}{Qingsong Wen},
  \bibinfo{person}{Irwin King}, \bibinfo{person}{Lizhen Qu}, {and}
  \bibinfo{person}{Zenglin Xu}.} \bibinfo{year}{2026}\natexlab{}.
\newblock \showarticletitle{Remember the Decision, Not the Description: A
  Rate-Distortion Framework for Agent Memory}.
\newblock \bibinfo{journal}{\emph{arXiv preprint arXiv:2605.10870}}
  (\bibinfo{year}{2026}).
\newblock
\urldef\tempurl%
\url{https://arxiv.org/abs/2605.10870}
\showURL{%
\tempurl}


\end{thebibliography}

\end{document}